\documentclass[11pt,onecolumn,letterpaper,accepted=2026-05-10]{quantumarticle}

\pdfoutput=1

\usepackage{macros}

\allowdisplaybreaks

\title{Quantum Doeblin Coefficients: Interpretations and Applications}

\author{Ian George}
\affiliation{Centre for Quantum Technologies, National University of Singapore, Singapore 117543, Singapore}
\email{qit.george@gmail.com}
\author{Christoph Hirche}
\affiliation{Institute for Information Processing (tnt/L3S), Leibniz Universit\"at Hannover, Germany}
\email{christoph.hirche@gmail.com}
\author{Theshani Nuradha}
\affiliation{School of Electrical and Computer Engineering, Cornell University, Ithaca, New York 14850, USA}
\affiliation{Department of Mathematics and
Illinois Quantum Information Science and Technology (IQUIST) Center,\\ 
University of Illinois Urbana-Champaign, Urbana, IL 61801, USA}
\email{theshani.gallage@gmail.com}
\author{Mark M.~Wilde}
\affiliation{School of Electrical and Computer Engineering, Cornell University, Ithaca, New York 14850, USA}
\email{wilde@cornell.edu}

\date{}

\begin{document}

\maketitle

\begin{abstract}
In classical information theory, the Doeblin coefficient of a classical channel provides an efficiently computable upper bound on the total-variation contraction coefficient of the channel, leading to what is known as a strong data-processing inequality.
Here, we investigate quantum Doeblin coefficients as a generalization of the classical concept. 
In particular, we define various new quantum Doeblin coefficients, one of which has several desirable properties, including concatenation and multiplicativity, in addition to being efficiently computable. We also develop various interpretations of two of the quantum Doeblin coefficients, including representations as minimal singlet fractions, exclusion values, reverse max-mutual and oveloH informations, reverse robustnesses, and  hypothesis testing reverse mutual and oveloH informations. Our interpretations of quantum Doeblin coefficients as either entanglement-assisted or unassisted exclusion values are particularly appealing, indicating that they are proportional to the best possible error probabilities one could achieve in state-exclusion tasks by making use of the channel. We also outline various applications of quantum Doeblin coefficients, ranging from limitations on quantum machine learning algorithms that use parameterized quantum circuits (noise-induced barren plateaus), on error mitigation protocols, on the sample complexity of noisy quantum hypothesis testing, on the fairness of noisy quantum models, and on mixing, indistinguishability, and decoupling times of time-varying channels. All of these applications make use of the fact that quantum Doeblin coefficients appear in upper bounds on various trace-distance contraction coefficients of a quantum channel. Furthermore,
in all of these applications, our analysis using quantum Doeblin coefficients provides improvements of various kinds over contributions from prior literature, both in terms of generality and being efficiently computable. 
\end{abstract}

\tableofcontents

\section{Introduction}

\subsection{Background}

Given a classical channel $\mathcal{W}$ described by a conditional
probability distribution $\left(  p_{Y|X}(y|x)\right)  _{y\in\mathcal{Y}
,x\in\mathcal{X}}$, its Doeblin coefficient is defined as follows~\cite{Doeblin1937}:
\begin{equation}
\label{eq:doeblin-classical}
\alpha(\mathcal{W})\coloneqq \max_{r\in\mathbb{R}_{\geq0}^{\left\vert \mathcal{Y}
\right\vert }}\left\{  \sum_{y\in\mathcal{Y}}r(y):r(y)\leq p_{Y|X}
(y|x)\ \forall x\in\mathcal{X},y\in\mathcal{Y}\right\}  ,
\end{equation}
where $\mathbb{R}_{\geq0}^{\left\vert \mathcal{Y}\right\vert }$ denotes the
set of all $\left\vert \mathcal{Y}\right\vert $-dimensional vectors with
non-negative entries. As written above, the Doeblin coefficient $\alpha
(\mathcal{W})$ is efficiently computable as a linear program. Its linear
programming dual is given by
\begin{equation}
\alpha(\mathcal{W})  =\min_{q(x|y)}\sum_{y\in\mathcal{Y}} \sum_{x \in \cX} q(x|y)p_{Y|X}
(y|x) , \label{eq:doeblin-dual-classical}
\end{equation}
where the optimization  is over every conditional probability
distribution $\left(  q(x|y)\right)  _{x\in\mathcal{X},y\in\mathcal{Y}}$. Using either of the expressions in~\eqref{eq:doeblin-classical} or~\eqref{eq:doeblin-dual-classical}, one can conclude that
\begin{equation}
    \alpha(\mathcal{W}) = \sum_{y\in\mathcal{Y}}\min_{x\in\mathcal{X}}\left\{  p_{Y|X}(y|x)\right\}
.\label{eq:minimum-likelihood-decoder}
\end{equation}
Ref.~\cite{makur2024doeblin} provides an extensive overview of the Doeblin coefficient of
classical channels, which has found various applications, including change
detection~\cite{Chen2022ChangeDetection}, multi-armed bandits~\cite{Moulos2020}, Markov-chain Monte Carlo methods
\cite{Rosenthal1995}, Markov decision processes~\cite{Alden1992}, and mixing models~\cite{Steinhardt2015}.

As observed in \cite[Theorem~1, Item~8]{makur2024doeblin}, the formulations in
\eqref{eq:doeblin-dual-classical}--\eqref{eq:minimum-likelihood-decoder}
provide the Doeblin coefficient with an operational meaning as being
proportional to the optimal error probability in an information-processing
task known as exclusion~\cite{bandyopadhyay2014conclusive}. The goal of this task is as follows:\ given an input
$x\in\mathcal{X}$ chosen uniformly at random and transmitted through the
channel $\mathcal{W}$ to produce an output $y\in\mathcal{Y}$, upon receiving $y$, pick $x^{\prime
}\in\mathcal{X}$ such that $x^{\prime}\neq x$. The most general strategy the receiver
can use involves a conditional probability distribution $q(x|y)$, and so one
can optimize over all such strategies to conclude that the optimal error
probability is $\frac{1}{\left\vert \mathcal{X}\right\vert }\alpha
(\mathcal{W})$. The equality in~\eqref{eq:minimum-likelihood-decoder}
reflects the fact that an optimal strategy is a minimum likelihood decoding
rule~\cite[Section~II-A]{mishra2023optimal}, in which one guesses a value $x^{\prime}\in\mathcal{X}$ that is least
likely to occur.

Beyond its interpretation given above, the Doeblin coefficient finds possibly
its greatest application in bounding the total-variation contraction
coefficient of a channel, the latter also known as the Dobrushin coefficient~\cite{Dobrushin1956}. To motivate this point, let us recall that the total-variation contraction coefficient of a classical channel $\mathcal{W}$ is
defined as
\begin{align}
\eta_{\operatorname{TV}}(\mathcal{W}) & \coloneqq \sup_{s\neq t}\frac{\left\Vert
\mathcal{W}\circ s-\mathcal{W}\circ t\right\Vert _{1}}{\left\Vert
s-t\right\Vert _{1}}\\
& = \max_{x\neq x'} \frac{1}{2} \left\| p_{Y|X=x} - p_{Y|X=x'}\right\|_1,
\label{eq:simpler-exp-contract-coeff}
\end{align}
where the optimization is over every pair of probability distributions $s$ and
$t$ on the input alphabet $\mathcal{X}$, the probability distribution $\mathcal{W}\circ s$ is defined to be
 $\left(  \sum_{x\in\mathcal{X}}p_{Y|X}
(y|x)s({x})\right)  _{y\in\mathcal{Y}}$, and $\frac{1}{2} \left\Vert s-t\right\Vert
_{1}=\frac{1}{2} \sum_{x\in\mathcal{X}}\left\vert s({x})-t({x})\right\vert $ is the total
variation distance (also called normalized trace distance). The equality in~\eqref{eq:simpler-exp-contract-coeff} was proved in~\cite{Dobrushin1956}. The contraction coefficient
leads to the following strong form of the data-processing inequality for total
variation distance:
\begin{equation}
\eta_{\operatorname{TV}}(\mathcal{W}
) \left\Vert s-t\right\Vert _{1}\geq \left\Vert \mathcal{W}\circ s-\mathcal{W}\circ t\right\Vert _{1}
\end{equation}
and is considered strong when $\eta_{\operatorname{TV}}(\mathcal{W})\in\lbrack0,1)$.

The following inequality relates the Doeblin coefficient to the contraction
coefficient:
\begin{equation}
\eta_{\operatorname{TV}}(\mathcal{W})\leq1-\alpha(\mathcal{W}).
\end{equation}
See, e.g.,~\cite[Theorem~8.17]{wolf2012quantum} and~\cite[Remark~III.2]{Raginsky2016}.
The significance of this inequality is that
one
can use the Doeblin coefficient as a proxy for the contraction coefficient in
various applications. The Doeblin coefficient additionally satisfies various
desirable properties useful in applications, including concatenation~\cite[Corollary~2.2]{Chestnut2010}
\begin{equation}
1-\alpha(\mathcal{W}_{2}\circ\mathcal{W}_{1})\leq\left(  1-\alpha
(\mathcal{W}_{2})\right)  \left(  1-\alpha(\mathcal{W}_{1})\right)  ,
\end{equation}
and multiplicativity~\cite[Theorem~1, Item~5]{makur2024doeblin}:
\begin{equation}
\alpha(\mathcal{W}_{1}\otimes\mathcal{W}_{2})=\alpha(\mathcal{W}_{1}
)\cdot\alpha(\mathcal{W}_{2})
\end{equation}
for two arbitrary classical channels $\mathcal{W}_1$ and $\mathcal{W}_2$. 

\subsection{Summary of Contributions}

In this paper, our main aim is to  explore the notion of Doeblin coefficient in
quantum information theory, as well as its applications therein. As a consequence, we significantly expand on the developments of~\cite[Theorem~8.17 and Eq.~(8.86)]{wolf2012quantum} and
\cite{hirche2024quantum}. In particular, we define several notions of a
quantum Doeblin coefficient, some of which go beyond those already proposed in
\cite[Theorem~8.17 and Eq.~(8.86)]{wolf2012quantum} and
\cite{hirche2024quantum}. Our primary criterion for a quantity to be called a quantum Doeblin
coefficient is that it should reduce to the classical quantity~$\alpha
(\mathcal{W})$ when evaluated for a classical channel $\mathcal{W}$. Beyond
that, we look for connections to the concepts and properties mentioned above: Does the
proposed quantum Doeblin coefficient have an operational meaning?\ Is it
efficiently computable?\ Does it satisfy the concatenation and
multiplicativity properties mentioned above?

This line of questioning has led us to winnow down the various proposed
quantities and focus primarily on the definition of quantum Doeblin
coefficient given in \cref{def:q-doeblin-coeff}, denoted there as $\alpha(\mathcal{N})$. The quantity $\alpha(\mathcal{N})$ reduces to the classical quantity
for a classical channel, it is efficiently computable as a semidefinite
program, and it has an operational meaning generalizing that in
\eqref{eq:doeblin-dual-classical}--\eqref{eq:minimum-likelihood-decoder}.
However, this quantity $\alpha(\mathcal{N})$ is not multiplicative in general. By modifying the definition just
slightly (see \cref{def:b-doeblin-coef}), we define another quantum Doeblin coefficient, denoted by $\alpha_{\wang}(\mathcal{N})$, that is
multiplicative, reduces to the classical case for a classical channel, and is
efficiently computable as a semidefinite program. However, the price to pay
for this slight modification is that it no longer possesses an operational
meaning in general.

A key distinction between the classical and quantum definitions of Doeblin coefficient is the optimization over a Hermitian operator $X$ in \cref{def:q-doeblin-coeff}, \cref{def:alt-doeblin-quantum-induced}, and \cref{def:b-doeblin-coef}, rather than restricting the optimization to positive semidefinite $X$. In the definition of the Doeblin coefficient of a classical channel in~\eqref{eq:doeblin-classical}, we could have instead optimized over $r\in\mathbb{R}^{\left\vert \mathcal{Y}
\right\vert }$; however, it suffices to restrict the optimization in this case to $r\in\mathbb{R}_{\geq 0}^{\left\vert \mathcal{Y}
\right\vert }$. For the quantum case, it is not generally possible to restrict the optimization over Hermitian $X$ to just positive semidefinite $X$. Let us also advocate that the most prudent definition of quantum Doeblin coefficient involves an optimization over Hermitian  $X$. By doing so, not only do the quantities in \cref{def:q-doeblin-coeff} and \cref{def:alt-doeblin-quantum-induced} possess operational meanings in terms of exclusion values of the channel (either with entanglement assistance or without), but one also obtains tighter bounds on the trace distance contraction coefficient (see~\eqref{eq:complete-trace-distance-CC-to-doeblin} and~\eqref{eq:trace-distance-CC-to-induced-doeblin}), which is one of the main applications of a quantum Doeblin coefficient. Interestingly, let us also note that the exclusion task (and thus the exclusion value) is related to the foundations of quantum mechanics, by means of the PBR theorem from~\cite{pusey2012reality}, and thus it brings out distinct features of quantum information when compared to classical information.   

Let us now summarize the main contributions of our paper:

\begin{enumerate}

\item We provide various interpretations of the quantum Doeblin coefficient $\alpha(\mathcal{N})$
from \cref{def:q-doeblin-coeff} (see \cref{sec:doeblin-interpretations}). In particular, we interpret it as a channel's
minimum singlet fraction, its entanglement-assisted exclusion value, its
reverse max-mutual information, its reverse robustness, and its reverse
hypothesis-testing mutual information. See \cref{table:doeblin-expressions} for a summary of these
interpretations, which are explained in far greater detail in \cref{sec:doeblin-interpretations}. We also prove that it possesses various desirable properties,
such as normalization, concatenation, data processing under pre- and
post-processing, submultiplicativity, and concavity, in addition to being efficiently computable.

\item We provide various interpretations of the alternative notion of quantum
Doeblin coefficient $\alpha_I(\mathcal{N})$ from \cref{def:alt-doeblin-quantum-induced} (\cref{sec:interpretations-alt-doeblin}). In particular, we interpret
it as a channel's minimum singlet fraction under entanglement-breaking
decoders, its exclusion value, its reverse max-oveloH information, its reverse
robustness under positive, trace-preserving maps, and its hypothesis-testing
oveloH information. See \cref{table:induced-doeblin-expressions} for a summary of these interpretations, which are explained in far greater detail in \cref{sec:interpretations-alt-doeblin}. We
also prove that it possesses similar desirable properties, as mentioned above.

\item We define an alternative notion of quantum Doeblin coefficient, denoted by $\alpha_{\wang}(\mathcal{N})$ (see
\cref{def:b-doeblin-coef}), and establish that it possesses various desirable properties,
including multiplicativity and concatenation. We summarize its properties and how they compare to other quantum Doeblin coefficients in Table \ref{table:doeblin-properties}.

\item We consider complete contraction coefficients for quantum divergences, with trace distance being an important special case; the latter is defined in~\eqref{eq:trace_distance_cc} using the notation~$\eta_{\Tr}^c(\cN) $.
In particular, we prove that  $\eta_{\Tr}^c(\cN) \leq 1-\alpha(\cN)$ in~\cref{prop:trace_distance_complete_cont_Doeblin_bound}, which finds use in applications where there exist noiseless subsystems in tandem with noisy subsystems.

\item We evaluate the quantum Doeblin coefficient $\alpha(\mathcal{N})$ for several examples of
channels relevant in applications, such as classical--quantum channels,
measurement channels, generalized
dephasing channels, qubit channels, and generalized amplitude damping channels (\cref{sec:doeblin-examples}).
\end{enumerate}

Beyond the contributions listed above, we also provide several applications of
quantum Doeblin coefficients, relevant for quantum computation and information
theory. These are due to the connection between Doeblin coefficients and
contraction coefficients; they include the following:

\begin{enumerate}
\item We establish a connection between noise-induced barren plateaus and
Doeblin coefficients in~\cref{prop:General_noise_BP} of \cref{sec:noise-induced-barren-plateaus}. In particular, a noise-induced barren
plateau~\cite{wang2021noise} occurs when it is difficult to train a parameterized quantum circuit
affected by noise. The terminology means that the optimization landscape for
training becomes extremely flat, making it too difficult to navigate when
conducting an algorithm like gradient descent. In particular, we show that, for a parameterized quantum circuit affected by generic noise, it is quite difficult to train the circuit parameters located in the first few layers of the circuit. To this end, we use quantum Doeblin coefficients to provide efficiently computable bounds on the magnitude of the cost function's gradient, showcasing that the derivatives decay with the depth of the corresponding parameterized gate from the end of the circuit.

Moreover, our results prove that noise-induced barren plateaus occur if the noise channels have non-zero Doeblin coefficients, regardless of the existence of noiseless subsystems, parameter initialization strategy (even with warm starts), or the gradient optimization technique used, further expanding previously known limitations~\cite{wang2021noise,schumann2024emergence,singkanipa2025beyond,mele2024noise}. We should also note that this does not contradict the results of~\cite{mele2024noise}, where non-unital noise can possibly help in escaping barren plateaus, given that their finding holds only ``on average".

\item We establish a connection between limitations on error mitigation and
Doeblin coefficients (\cref{sec:error-mitigation}). In near-term quantum computing, one hope
is that error mitigation protocols can reduce or mitigate the effect of noise
on quantum computations used to estimate the expectation value of certain
observables~\cite{cai2023}. However, expanding on~\cite{takagi2023universal,quek2024exponentially}, we prove that all error mitigation protocols have an overhead
related to the Doeblin coefficient  when the noise channels in the circuit satisfy $\alpha(\cN) >0$ (see~\cref{thm:Error_miti_global}).
Indeed, we prove that the number of noisy data samples required to achieve a fixed error tolerance scales exponentially with the circuit depth. These results are valid for general error mitigation protocols even with the assistance of noiseless subsystems in the circuit, under mild conditions on the observable and the set of states allowed. 

\item We establish limitations on the sample complexity of hypothesis testing
when the states are affected by noise (\cref{sec:noisy-hypothesis-testing}). In particular, our
bounds in~\cref{Prop:SC_noisy} feature the Doeblin coefficient of the noisy channel affecting the
unknown states. This provides a sense of the cost one needs to pay in the task of hypothesis testing when only noisy samples of the states are available.

\item We finally outline connections between the Doeblin coefficient and the fairness of quantum learners (\cref{sec:stability-fairness-q-learners}) and the mixing, indistinguishability, and decoupling times of discrete time homogeneous and inhomogeneous quantum processes (\cref{subS:Mixing_time_of_Quantum_Processes}). Notably, we show how complete contraction coefficients can be used to bound the decoupling time and decoupled indistinguishability time defined in~\cref{def:decoupling-time}. We then prove that the Doeblin coefficients provide efficiently computable upper bounds for these times when the corresponding divergence is trace distance (see~\cref{Cor:decopling_Doeblin} and~\cref{Cor:Distinguishability_time_Doeblin}).
\end{enumerate}

In the conclusion of our paper (\cref{sec:conclusion}), we summarize our main
findings and then outline directions for future research.

\setlength{\tabcolsep}{5pt}
\renewcommand{\arraystretch}{2}

\begin{table}
\begin{tabular}
[c]{c|c}\hline\hline
Interpretation of $\alpha(\mathcal{N})$ & Expression \\\hline\hline
Primal SDP   & $\sup_{X_{B}\in\operatorname{Herm}}\left\{
\operatorname{Tr}[X_{B}]:I_{A}\otimes X_{B}\leq\Gamma_{AB}^{\mathcal{N}
}\right\}  $ \\\hline
Dual SDP  & $\inf_{Y_{AB}\geq 0}\left\{  \operatorname{Tr}
[Y_{AB} \Gamma_{AB}^{\mathcal{N}}]:\operatorname{Tr}_{A}[Y_{AB}]=I_{B}\right\}
$ \\\hline
Minimal singlet fraction & $d^{2}\inf_{\mathcal{D}
\in\operatorname{CPTP}}F(\Phi^{d},\left(  \operatorname{id}\otimes\left(
\mathcal{D}\circ\mathcal{N}\right)  \right)  (\Phi^{d}))$  \\\hline
Entanglement-assisted exclusion value & $\inf_{\cX, \Psi,\left(
\mathcal{E}^{x},\Lambda^{x}\right)
_{x\in\mathcal{X}}}\sum_{x\in\mathcal{X}}\operatorname{Tr}[\Lambda_{RB}
^{x}\mathcal{N}_{A\rightarrow B}(\mathcal{E}_{A^{\prime}\rightarrow A}
^{x}(\Psi_{RA^{\prime}}))]$  \\\hline
Reverse max-mutual information & $\exp\!\left(  -\inf_{\tau\in
\operatorname{aff}(\mathcal{D})}D_{\max}(\mathcal{R}^{\tau}\Vert
\mathcal{N})\right)  $  \\\hline
Reverse robustness & $\sup_{\substack{\lambda \in [0,1],\\ \tau\in\operatorname{aff}(\mathcal{D}),\\ \mathcal{M} \in \operatorname{CPTP}} } \left\{\lambda : \lambda\mathcal{R}^\tau + (1-\lambda) \mathcal{M} = \mathcal{N}\right\}$\\\hline
Reverse hypo.-testing mutual info. & $d^2\exp\!\left(  -  \inf_{\tau
\in\operatorname{aff}(\mathcal{D})}D_{H}^{1-\frac{1}{d^{2}}}(\Phi^{\mathcal{R}^\tau}\Vert\Phi^{\mathcal{N}})  \right)  $  
\\\hline\hline
\end{tabular}
\caption{Various interpretations and expressions for the Doeblin coefficient $\alpha(\mathcal{N})$ of a quantum channel~$\mathcal{N}$. The references to these expressions are as follows: primal SDP~\eqref{eq:alpha_cN_def}, dual SDP~\eqref{eq:alpha-cN-dual}, minimal singlet fraction~\eqref{eq:MASF}, entanglement-assisted exclusion value~\eqref{eq:EA-exclusion-value}, reverse max-mutual information~\eqref{eq:reverse-max-mutual-info}, reverse robustness~\eqref{eq:reverse-robustness}, and reverse hypothesis-testing mutual information~\eqref{eq:reverse-hypothesis-testing-MI}. 
}
\label{table:doeblin-expressions}
\end{table}

\begin{table}
\begin{tabular}
[c]{c|c}\hline\hline
Interpretation of $\alpha_I(\mathcal{N})$  & Expression \\\hline\hline
Primal Conic Program & $\sup_{X_{B}\in\operatorname{Herm}}\left\{
\operatorname{Tr}[X_{B}]:\mathcal{R}^X \leq_P \mathcal{N}
\right\}  $ \\\hline
Dual Conic Program & $\inf_{Y_{AB}\in \operatorname{Sep}(A:B)}\left\{  \operatorname{Tr}
[Y_{AB} \Gamma_{AB}^{\mathcal{N}}]:\operatorname{Tr}_{A}[Y_{AB}]=I_{B}\right\}
$ \\\hline
Minimal singlet fraction & $d^{2}\inf_{\mathcal{D}
\in\operatorname{EB}}F(\Phi^{d},\left(  \operatorname{id}\otimes\left(
\mathcal{D}\circ\mathcal{N}\right)  \right)  (\Phi^{d}))$  \\\hline
Exclusion value & $\inf_{ \cX,  \left(
\psi^{x},\Lambda^{x}\right)
_{x\in\mathcal{X}}}\sum_{x\in\mathcal{X}}\operatorname{Tr}[\Lambda_{B}
^{x}\mathcal{N}_{A\rightarrow B}(\psi_{A}
^{x})]$  \\\hline
Max-oveloH information & $\exp\!\left(  -\inf_{\tau\in
\operatorname{aff}(\mathcal{D})}
\sup_{\rho\in\mathcal{D}}
D_{\max}(\tau\Vert
\mathcal{N}(\rho))\right)  $  \\\hline
Reverse robustness & $\sup_{\substack{\lambda \in [0,1],\\ \tau\in\operatorname{aff}(\mathcal{D}), \mathcal{M} \in \operatorname{PTP}} } \left\{\lambda : \lambda\mathcal{R}^\tau + (1-\lambda) \mathcal{M} = \mathcal{N}\right\}$\\\hline
Hypo.-testing oveloH info. & 
$d^2 \exp\!\left(  -\left[  \sup_{(\psi_x)_{x=1}^{d^2}} \inf_{\tau
\in\operatorname{aff}(\mathcal{D})} D_{H}^{1-\frac{1}{d^{2}}}( \rho_{X} \otimes \tau_B \Vert  \mathcal{N}_{A\to B}(\rho_{XA}))\right]  \right)  $ 
\\\hline\hline
\end{tabular}
\caption{Various interpretations and expressions for the induced Doeblin coefficient $\alpha_I(\mathcal{N})$ of a quantum channel~$\mathcal{N}$. The references to these expressions are as follows: primal conic program~\eqref{eq:alpha-induced-def-pos-cone}, dual conic program~\eqref{eqn:inducedDoeblinDual}, minimal singlet fraction~\eqref{eq:min-singlet-frac-interp-induced-doeblin}, exclusion value~\eqref{eq:induced-doeblin}, max-oveloH information~\eqref{eq:max-oveloh-interp}, reverse robustness~\eqref{eq:reverse-robustness-induced-doeblin}, and hypothesis-testing oveloH information~\eqref{eq:hypo-test-oveloh}. In the above, $\rho_{XA} = \sum_{x=1}^{d^2} \frac{1}{d^2} |x\rangle \!\langle x| \otimes \psi_x$, and each $\psi_x$ is a pure state.}
\label{table:induced-doeblin-expressions}
\end{table}

\begin{table}
\begin{tabular}
[c]{c|c|c|c|}\hline\hline
Property & $\alpha(\cN)$ (Def.~\ref{def:q-doeblin-coeff}) & $\alpha_{I}(\cN)$ (Def.~\ref{def:alt-doeblin-quantum-induced}) & $\alpha_{\wang}(\cN)$ (Def.~\ref{def:b-doeblin-coef}) \\
\hline\hline
Recovers Classical Case & \multicolumn{3}{c|}{$\checkmark$ (\cref{prop:recovers-classical-doeblin})} \\ \hline
Normalized & $\checkmark$ (\cref{prop:q-doeblin-normalization}) & $\checkmark$ (\cref{prop:induced-doeblin-normalized}) & $\checkmark$ (\cref{prop:doeblin-wang-normalized}) \\ \hline
Concatenation & \multicolumn{2}{c|}{$\checkmark$ (\cref{prop:concatenation})} & $\checkmark$ (\cref{prop:concat-of-alpha-wang}) \\ \hline 
$\begin{array}{c}
     \text{Data Processing and}  \\[-5mm]
     \text{Isometric Invariance}
\end{array}$ & \multicolumn{2}{c|}{$\checkmark$ (\cref{prop:DPI-and-LU-invariance})} & $\checkmark$ (\cref{prop:DPI-and-LU-invariance-wang}) \\ \hline
Submultiplicative & \multicolumn{2}{c|}{$\checkmark$ (\cref{prop:q-doeblin-sub-multiplicative})} & $\checkmark$ (\cref{prop:multiplicativity-b-N}) \\ \hline
Multiplicative & \multicolumn{2}{c|}{$\times$ (\cref{example:strict-submultiplicativity-for-qubit-channels})} & $\checkmark$ (\cref{prop:multiplicativity-b-N}) \\ \hline
Concavity & \multicolumn{2}{c|}{$\checkmark$ (\cref{prop:concavity-doeblin})} & $\checkmark$ (\cref{prop:doeblin-wang-concavity})\\ 
\hline \hline
\end{tabular}
\caption{Summary of the properties of the main quantum Doeblin coefficients and where to find their proofs. $\checkmark$ means the property is satisfied and $\times$ means it is not. Note this shows that the main advantage of $\alpha_{\wang}$ is that it is multiplicative at the expense of lacking the operational interpretations that the other two quantities possess.
}
\label{table:doeblin-properties}
\end{table}

\renewcommand{\arraystretch}{1}

\section{Preliminaries}

We begin by establishing some notation used throughout the rest of our paper. For further background on quantum information, see \cite{Wilde-Book,Hayashi17,Watrous-Book,Holevo19book,khatri2024principlesquantumcommunicationtheory}.
Define the maximally entangled state $\Phi
_{d}$ of Schmidt rank $d\in \mathbb{N}$ as
\begin{equation} \label{eq:max_entangled}
\Phi_{d}\coloneqq \frac{1}{d}\sum_{i,j}|i\rangle\!\langle j|\otimes
|i\rangle\!\langle j|.
\end{equation}

The Choi operator of a superoperator $\mathcal{M}_{A\to B}$ is defined as
\begin{equation}
\label{eq:choi_operator}
    \Gamma^{\cM}_{AB} \coloneqq  \sum_{i,j} |i\rangle\!\langle j|_A \otimes \cM_{A'\to B}(|i\rangle\!\langle j|_{A'}),
\end{equation}
where system $A'$ is isomorphic to system $A$.

For two Hermitian operators $N$ and $M$ we use the notation $N\geq M$ to indicate that $N-M$ is a positive semidefinite operator. For two Hermiticity preserving superoperators $\mathcal{N}$ and $\mathcal{M}$, we use the notation $\mathcal{N}\geq \mathcal{M}$ to indicate that $\mathcal{N}-\mathcal{M}$ is a completely positive map, and we use the notation $\mathcal{N}\geq_P \mathcal{M}$ to indicate that $\mathcal{N}-\mathcal{M}$ is a positive map. The latter condition is equivalent to $\mathcal{N}(\rho) \geq \mathcal{M}(\rho)$ for every density operator $\rho$.

\subsection{Quantum Divergences}

Here we define some quantum divergences, which we take to be functions of a quasi-state (Hermitian operator with trace equal to one) and a positive semidefinite operator. In particular, we define the hockey-stick divergence and the hypothesis testing divergence (also called smoothed min-divergence).
These were originally defined for the first argument being a state in~\cite[Eq.~(9)]{sharma2012strongconversesquantumchannel} and\cite{BD10,BD11,WR12}, 
 respectively, and the hypothesis-testing divergence was recently extended to the case of a quasi-state first argument in~\cite[Eq.~(62)]{ji2024barycentric}. 

Let $\tau \in \operatorname{aff}(\mathcal{D})$ be a quasi-state, where $\operatorname{aff}(\mathcal{D})$ denotes the affine hull of the set of density operators
(i.e., Hermitian operators with unit trace), and let $\sigma $ be a positive semidefinite operator.
For $\gamma \geq 0$, the hockey-stick divergence is defined as 
\begin{equation}
    E_\gamma(\tau \Vert \sigma) \coloneqq \Tr\!\left[ (\tau-\gamma \sigma)_+\right] -(1-\gamma)_+,
    \label{eq:hockey-stick-div-def}
\end{equation}
where $(A)_+$ denotes the positive part of a Hermitian operator $A$. To define this, first define  $(x)_+ \coloneqq \max\{0,x\}$ for scalars. Then $(A)_+ \coloneqq \sum_{i} (a_i)_+ |i\rangle\!\langle i|$, with $A = \sum_i a_i |i\rangle\!\langle i|$ a spectral decomposition of $A$.
We also have the following equivalent formulation:
\begin{equation}\label{eq:HS_sup_form}
     E_\gamma(\tau \Vert \sigma) = \sup_{0 \leq M \leq I} \Tr\!\left[M(\tau-\gamma \sigma) \right]-(1-\gamma)_+.
\end{equation}
At $\gamma=1$, the hockey-stick divergence  reduces to the normalized trace distance:
\begin{equation}
    E_1(\tau \Vert \sigma) = T(\tau, \sigma) \coloneqq \frac{1}{2} \left \| \tau - \sigma\right\|_1.
    \label{eq:norm-TD-def}
\end{equation}
For $\varepsilon\in\left[  0,1\right]  $,  the hypothesis testing relative entropy is defined as
\begin{align}
D_{H}^{\varepsilon}(\tau\Vert\sigma)  & \coloneqq -\ln\inf_{\Lambda\geq0}\left\{
\operatorname{Tr}[\Lambda\sigma]:\operatorname{Tr}[\Lambda\tau]\geq
1-\varepsilon,\ \Lambda\leq I\right\} \label{eq:extended_ht_Relative_Entropy} \\
& =-\ln\sup_{\mu,Z\geq0}\left\{  \mu\left(  1-\varepsilon\right)
-\operatorname{Tr}[Z]:\mu\tau\leq\sigma+Z\right\}  .
\label{eq:dual-SDP-hypo-test-rel-ent}
\end{align}
The dual expression in \eqref{eq:dual-SDP-hypo-test-rel-ent} was given in \cite[Eq.~(B2)]{WangWilde2019} for the case of $\tau$ being a state, and the derivation easily extends to the case of $\tau$ being a quasi-state.

We have two different definitions of max-relative entropy of $\tau \in \operatorname{aff}(\mathcal{D})$ and $\sigma$ positive semidefinite, which was originally defined for states in \cite{datta2009}.
These are as follows:
\begin{align}
     D_{\max}(\tau\|\sigma) & \coloneqq \ln \inf_{\lambda \geq 0} \left\{ \lambda: \tau \leq \lambda \sigma\right\} \label{eq:dmax-1-herm-maps}\\
     \widetilde{D}_{\infty}(\tau\|\sigma) & \coloneqq \ln \inf_{\lambda \geq 0} \left\{ \lambda: -\lambda \sigma \leq  \tau \leq \lambda \sigma\right\} \\
     & = \ln \left\| \sigma^{-1/2} \tau \sigma^{-1/2} \right\|_\infty.
\end{align}
These quantities are related as follows:
\begin{equation}
    D_{\max}(\tau\|\sigma) \leq \widetilde{D}_{\infty}(\tau\|\sigma).
\end{equation}
We also define two different definitions of max-relative entropy of Hermiticity preserving maps~$\mathcal{T}$ and $\mathcal{S}$ as
\begin{align}
    D_{\max}(\mathcal{T} \Vert \mathcal{S}) & \coloneqq \ln \inf_{\lambda \geq 0} \left\{ \lambda: \mathcal{T} \leq \lambda \mathcal{S}\right\} \label{eq:max-rel-ent-channels}\\
    \widetilde{D}_{\infty}(\mathcal{T} \Vert \mathcal{S}) & \coloneqq \ln \inf_{\lambda \geq 0} \left\{ \lambda: -\lambda \mathcal{S} \leq \mathcal{T} \leq \lambda \mathcal{S}\right\},
\end{align}
where the superoperator inequality $\mathcal{T} \leq \lambda \mathcal{S}$ holds if and only if $\lambda \mathcal{S} - \mathcal{T}  $ is a completely positive map. These quantities are likewise ordered as
\begin{equation}
    D_{\max}(\mathcal{T} \Vert \mathcal{S}) \leq \widetilde{D}_{\infty}(\mathcal{T} \Vert \mathcal{S}).
\end{equation}

\subsection{Conic Programming and the Separable Cone}

Throughout this work we use conic programs, which include the more commonly used semidefinite programs as a special case. Of particular importance are cone programs over the separable cone. We provide a brief review of these concepts as used in this work. We refer the reader to \cite{Watrous-Conic-Program-Notes,boyd2004convex} for further details on cone programs.

\paragraph{Cone Programs} We begin by defining a cone program generically in a way that aligns with the rest of this work's notation. Let $(V,\langle \cdot, \cdot \rangle_{V})$, $(W,\langle \cdot, \cdot \rangle_{W})$ be real-valued inner product spaces. Let $\cK \subset V$ be a closed, convex cone, and let $\cK^{\ast}$ denote its dual cone. Let $a \in V$ and $b \in W$. Let $\Phi:V \to W$ be a linear map. The primal and dual of the cone program may then be written as\footnote{Note that this is a relabeling of what is the primal and dual in \cite[Section 1.4]{Watrous-Conic-Program-Notes}. Thus, everything stated in this subsection is a relabeling of statements in \cite{Watrous-Conic-Program-Notes}.}
 \begin{center}
    \begin{miniproblem}{0.5}
      \emph{Primal problem}\\[-7mm]
      \begin{equation}      \begin{aligned}\label{eqn:cone-primal}
        \text{maximize:}\quad & \langle b , x \rangle_{W} \\
        \text{subject to:}\quad & \Phi^{\ast}(x)-a \in \cK^{\ast}  \\
        & x \in W \ 
      \end{aligned}
      \end{equation}
    \end{miniproblem}
    \hspace*{2mm}
    \begin{miniproblem}{0.4}
      \emph{Dual problem}\\[-7mm]
      \begin{equation}
      \begin{aligned}\label{eqn:cone-dual}
        \text{minimize:}\quad & \langle a , y \rangle_{V} \\
        \text{subject to:}\quad & \Phi(y) = b \\
        & y \in \cK \ ,
      \end{aligned}
      \end{equation}
    \end{miniproblem}
 \end{center}
 where $\Phi^{\ast}: W \to V$ is the unique linear map satisfying $\langle w, \Phi(v) \rangle_{W} = \langle \Phi^{\ast}(w), v \rangle_V$ for all $v \in V, w \in W$.\footnote{We reserve the notation $\cE^{\dagger}$ for the adjoint map with respect to the Hilbert--Schmidt inner product as is standard in quantum information theory.} Such cone programs always satisfy weak duality; i.e., defining the feasible sets
 \begin{align}
 F & \coloneq \{ x \in W: \Phi^{\ast}(x) - a \in \cK^{\ast}\}, \\
 G & \coloneq \{y \in \cK: \Phi(y) = b)\},    
 \end{align}
and the optimal values
\begin{align}
\alpha & \coloneq \sup_{x \in F} \; \langle b , x \rangle_{W} ,\\
\beta & \coloneq \; \inf_{y \in G} \langle a , y \rangle_{V},
\end{align}
which are taken over the extended real line, we have $\alpha \leq \beta$. A sufficient condition for strong duality ($\beta = \alpha$) is Slater's criterion, which we state here (see also \cite[Theorem 1.3]{Watrous-Conic-Program-Notes}).
 
\begin{proposition}
    Let $\operatorname{Relint}(C)$ denote the relative interior of a convex set $C$. With respect to the notation given,
    \begin{enumerate}
        \item If $F \neq \emptyset$ and $\operatorname{Relint}(\cK) \cap G \neq \emptyset$, then there exists $x \in F$ such that $\langle b , x \rangle_{W} = \beta$; i.e., $\alpha = \beta$.
        
        \item If $G \neq \emptyset$ and there exists $x \in W$ such that $\Phi^{\ast}(x) - a \in \operatorname{Relint}(\cK^{\ast})$, then there exists $y \in G$ such that $\langle a , y \rangle_{V} = \alpha$; i.e., $\alpha = \beta$.
    \end{enumerate}
\end{proposition}

 \begin{remark}
     While we have presented the above more generally, we note throughout this work that we will always use the inner product to be the Hilbert--Schmidt product, $\langle X, Y \rangle = \Tr[X^{\dag}Y]$ for all linear operators $X,Y \in \operatorname{Lin}(A,B)$, which will simplify some of the above notation.
 \end{remark}

\paragraph{Separable Cone}
The cone that we primarily work with concretely is the separable cone. We summarize key properties to which we will appeal.
\begin{definition}
    Given Hilbert spaces $A$ and $B$, the separable cone (equivalently, set of separable operators) are positive semidefinite operators $P_{AB} \geq 0$ such that there exists a finite alphabet~$\cX$ and tuples of positive operators $(Q_{x})_{x \in \cX} \subset \Pos(A)$, $(R_{x})_{x \in \cX} \subset \Pos(B)$ such that
    \begin{align}
        P = \sum_{x \in \cX} Q_{x} \otimes R_{x} \, \ .
    \end{align}
\end{definition}
It is straightforward to verify that the separable cone is indeed a cone, and a proof that it is closed is established in \cite[Proposition 6.8]{Watrous-Book}.
Using the definition of a dual cone and standard properties of separable operators, one may verify the relation
\begin{align}\label{eq:block-positive-relation}
    X \in \Sep^{\ast}(A:B) \Leftrightarrow \bra{\alpha}_{A}\bra{\beta}_{B}X\ket{\alpha}_{A}\ket{\beta}_{B} \geq 0 \quad \forall \ket{\alpha}_{A}\ket{\beta}_{B} \ .
\end{align}
The set of operators in $\Sep^{\ast}(A:B)$ are sometimes referred to as being `block-positive.' 

Lastly, we state the following well-known relation between a linear map being positive and the Choi operator being block-positive (see \cite{Johnston-2012norms} for an elementary proof).
\begin{proposition}\label{prop:positive-to-block-positive}
    A linear map $\Phi_{A \to B}$ is positive if and only if $\Gamma^{\Phi}_{AB} \in \Sep^{\ast}(A:B)$.
\end{proposition}

Beyond the separable cone itself, we also use the following well-known SDP-representable cones that contain the separable cone. First, let us recall the positive partial transpose (PPT) cone:
\begin{align}\label{eq:PPT-cone}
    \operatorname{PPT}(A:B) \coloneqq \{P_{AB} \geq 0: P_{AB}^{\Gamma} \geq 0 \} \ , 
\end{align}
where $X^{\Gamma} \coloneqq (\id_{A} \otimes T_{B})(X)$ is the partial transpose as $T_{B}$ is the transpose map. Second, let us recall the family of cones that interpolate between the separable and positive cone \cite{doherty2004complete}.
\begin{definition}\label{def:k-sym-ext-cone}
     For $k \in \mbb{N}$, the $k$-symmetrically extendible cone is defined as
    \begin{align*}
        \operatorname{Sym}_{k}(A:B) \coloneq \left\{P_{AB} \geq 0 : \exists P_{AB^{k}_{1}} \geq 0 : \begin{matrix} P_{AB^{k}_{1}} = (I_{A} \otimes W_{\pi})P_{AB^{k}_{1}}(I_{A} \otimes W_{\pi}^{\dagger}) \, \, \forall \pi \in \cS([k]) \\
        \Tr_{B^{k-1}_{1}}[P_{AB^{k}_{1}}] = P_{AB}
        \end{matrix} \right\} \ ,
    \end{align*}
    where $\cS_{k}$ is the set of permutations on $k$ elements, and for every $\pi \in \cS_{k}$, $W_{\pi}$ is the unitary that permutes the $k$ copies of $B$ according to $\pi$; i.e., for every unit vector
    $\{\ket{v_{i}} \in B_{i}\}_{i \in \{1,\ldots ,k\}}$, 
    $$W_{\pi}\ket{v_{1}}_{B_{1}}\ket{v_{2}}_{B_{2}}\cdots \ket{v_{k}}_{B_{k}} = \ket{v_{1}}_{B_{\pi(1)}}\ket{v_{2}}_{B_{\pi(2)}}\cdots\ket{v_{k}}_{B_{\pi(k)}} \ . $$ 
\end{definition}
\noindent This means that $\operatorname{Sym}_{k}(A:B)$ is the set of positive semidefinite operators, labeled by $P_{AB}$, that admit extensions where $P_{AB_{i}} = P_{AB}$ for all $i \in \{1,\ldots ,k\}$. It is well-known that $\Pos(A \otimes B) = \operatorname{Sym}_{1}(A:B) \supsetneq \operatorname{Sym}_{2}(A:B) \supsetneq \cdots$ and $\Sep(A:B) = \lim_{k \to \infty} \operatorname{Sym}_{k}(A:B)$.

\subsection{Quantum State Exclusion}

The task of quantum state exclusion \cite{bandyopadhyay2014conclusive} appears in various forms throughout our paper.
Quantum state exclusion was considered in~\cite{pusey2012reality} to study whether quantum states are $\psi$-ontic or $\psi$-epistemic, and it was studied directly subsequently, e.g.,~\cite{bandyopadhyay2014conclusive,McIrvin2024QSE}. More recently quantum state exclusion has been studied in a number of papers~\cite{Havlivcek2020simple,Leifer2020noncontextuality,uola2020all,russo2023inner,mishra2023optimal,johnston2023tight}.

Let us recall the task for a given ensemble of states.

\begin{definition}
\label{def:state-exclusion-err-prob}
    Consider a finite alphabet $\cX$, a probability distribution over $\cX$, $p$, and a tuple of quantum states $(\rho^{x}_{A})_x$. For a given POVM with $\vert \cX \vert$ outcomes, $(\Lambda_{x})_{x}$, the error probability of state exclusion is $\sum_{x} p(x)\Tr[\Lambda_{x}\rho^{x}]$. We then define the minimal error probability of state exclusion as follows:
    \begin{align}\label{eq:QSE}
        P_{\operatorname{err}}((p_{x},\rho^{x})_x) \coloneqq \min_{(\Lambda_{x})_{x} \in \operatorname{POVM}} \sum_{x} p(x)\Tr[\Lambda_{x}\rho^{x}] \ . 
    \end{align}
\end{definition}

Note that this is quantum state exclusion, as the agent `wins' whenever they do not guess the true value of the state; thus they win whenever they successfully exclude the state they were sent.

\section{Quantum Doeblin Coefficients}

\subsection{Main Definition of Quantum Doeblin Coefficient}

The main definition of a quantum channel's Doeblin coefficient that we use in this paper was essentially introduced in~\cite[Eq.~(8.86)]{wolf2012quantum} and revisited in~\cite[Corollary~III.7]{hirche2024quantum}. 

\begin{definition}[Quantum Doeblin Coefficient]
\label{def:q-doeblin-coeff}
    The quantum Doeblin coefficient of a channel~$\mathcal{N}_{A \to B}$ is
defined as
    \begin{align}
        \alpha(\cN) \coloneqq \sup_{ X_B\in\operatorname{Herm}} \left\{\Tr[ X_B]: I_A\otimes  X_B 
        \leq \Gamma^{\cN}_{AB} \right\}, \label{eq:alpha_cN_def}
    \end{align}
where the Choi operator of $\mathcal{N}_{A\to B}$ is defined in~\eqref{eq:choi_operator}.
\end{definition} 

We note here that \cref{def:q-doeblin-coeff} is equivalent to
\begin{equation}
    \alpha(\mathcal{N}) = \sup_{X \in \operatorname{Herm}}\left\{ \operatorname{Tr}[X] : \mathcal{R}^X \leq \mathcal{N}\right\},
    \label{eq:doeblin-CP-order}
\end{equation}
where $\mathcal{R}^X(\cdot) \coloneqq \operatorname{Tr}[\cdot] X$ is a replacer map.

By observing that~\eqref{eq:alpha_cN_def} is a semidefinite program, we obtain the following dual expression for the Doeblin coefficient, which is helpful in our analysis carried out later. 

\begin{proposition}[Dual SDP Formula for Doeblin Coefficient]
\label{prop:dual-SDP-doeblin-main}
    Let $\cN_{A \to B}$ be a quantum channel. Then  $\alpha(\cN)$ defined in~\eqref{eq:alpha_cN_def} is equal to the following expression:
    \begin{align}
\alpha(\cN) &= \inf_{Y_{AB} \geq 0} \left\{ \Tr[Y_{AB} \Gamma^{\cN}_{AB}] \, : \, \Tr_{A}[Y_{AB}] = I_{B} \right\} \label{eq:alpha-cN-dual},
\end{align}
where $\Gamma^{\cN}_{AB}$ is defined in~\eqref{eq:choi_operator}.
\end{proposition}

\begin{proof}
In the following proof, we omit system labels for simplicity.
Consider that
\begin{align}
\alpha(\mathcal{N})  &  =\sup_{X\in\operatorname{Herm}}\left\{
\operatorname{Tr}[X]:I\otimes X\leq\Gamma^{\mathcal{N}}\right\} \\
&  =\sup_{X\in\operatorname{Herm}}\left\{  \operatorname{Tr}[X]+\inf_{Y\geq
0}\operatorname{Tr}[Y(\Gamma^{\mathcal{N}}-I\otimes X)]\right\} \\
&  =\sup_{X\in\operatorname{Herm}}\inf_{Y\geq0}\left\{  \operatorname{Tr}
[X]+\operatorname{Tr}[Y(\Gamma^{\mathcal{N}}-I\otimes X)]\right\} \\
&  =\sup_{X\in\operatorname{Herm}}\inf_{Y\geq0}\left\{  \operatorname{Tr}
[Y\Gamma^{\mathcal{N}}]+\operatorname{Tr}[X\left(  I-\operatorname{Tr}
_{A}[Y]\right)  ]\right\} \\
&  \leq\inf_{Y\geq0}\sup_{X\in\operatorname{Herm}}\left\{  \operatorname{Tr}
[Y\Gamma^{\mathcal{N}}]+\operatorname{Tr}[X\left(  I-\operatorname{Tr}
_{A}[Y]\right)  ]\right\} \\
&  =\inf_{Y\geq0}\left\{  \operatorname{Tr}[Y\Gamma^{\mathcal{N}
}]:I=\operatorname{Tr}_{A}[Y]\right\}  .
\end{align}
Strong duality holds by picking strongly feasible solution $X=-I_B$ in the primal and the feasible solution $Y=\frac{1}
{d_A}I_A\otimes I_B$ in the dual. Then it follows that
\begin{equation}
\alpha(\mathcal{N})=\inf_{Y\geq0}\left\{  \operatorname{Tr}[Y\Gamma
^{\mathcal{N}}]:I=\operatorname{Tr}_{A}[Y]\right\}.
\end{equation}
This concludes the proof.
\end{proof}

\subsection{Alternative Notions of Quantum Doeblin Coefficients}

In this section, we introduce some alternative notions of quantum Doeblin coefficient. All of these quantities reduce to the classical definition when the channel is a classical channel, as stated in Proposition \ref{prop:recovers-classical-doeblin} at the end of this section. 

The first alternative definition we consider is from~\cite[Theorem~8.17]{wolf2012quantum}, subsequently used in~\cite[Definition~1]{angrisani2023differentialprivacyamplificationquantum}:

\begin{definition}
\label{def:alt-doeblin-quantum-induced}
    For a quantum channel $\mathcal{N}$, define the following quantity:
    \begin{align}
        \alpha_I(\mathcal{N}) & \coloneqq \sup_{X\in \operatorname{Herm}}\left\{ \operatorname{Tr}[X] :    \mathcal{N} - \mathcal{R}^X \in \operatorname{Pos}\right\} \label{eq:alpha-induced-def-pos-cone} \\
        & = \sup_{X\in \operatorname{Herm}}\left\{ \operatorname{Tr}[X] :    \mathcal{R}^X(\rho) \leq \mathcal{N}(\rho) \quad \forall \rho \in \mathcal{D} \right\},
    \end{align}
    where $\operatorname{Pos}$ denotes the set of positive maps.
\end{definition}
This quantity leads to an improved upper bound on the trace-distance contraction coefficient than that given by $\alpha(\mathcal{N})$, as shown in~\cite[Theorem~8.17]{wolf2012quantum}. However, the main drawback of it is that it is not efficiently computable, because it is computationally difficult to optimize over the set of positive maps.

We now introduce another alternative definition of quantum Doeblin coefficient, denoted by $\alpha_{\wang}(\mathcal{N})$ and inspired by methods employed in \cite{WFD17,wang2018semidefinite}. Although it does not give an improved upper bound on the trace distance contraction coefficient, it has several desirable properties, including 1) it is efficiently computable as a semidefinite program and 2) it is multiplicative under tensor products of quantum channels. As such, the quantity $\alpha_{\wang }(\cN)$ is useful in applications.

\begin{definition} 
\label{def:b-doeblin-coef}
    For a quantum channel $\mathcal{N}$, let us define $\alpha_{\wang }(\cN)$  as
\begin{equation}
    \alpha_\wang(\cN)  \coloneqq  \sup_{X_B \in \operatorname{Herm}} \left\{ \operatorname{Tr}[X_B] : - \Gamma^{\mathcal{N}}_{AB} \leq I_A \otimes X_B \leq \Gamma^{\mathcal{N}}_{AB}\right\}
    \label{eq:primal-b-relaxed-doeblin} .
\end{equation}
\end{definition}

We note that alternative expressions for $ \alpha_{\wang}(\mathcal{N})$ are as follows:
\begin{align}
    \alpha_{\wang}(\mathcal{N}) & = 
    \sup_{X_B \in \operatorname{Herm}} \left\{ \operatorname{Tr}[X_B] : - \mathcal{N} \leq \mathcal{R}^X \leq \mathcal{N} \right\}\\
    & = \left[ \inf_{\lambda \geq 0, \tau \in \operatorname{aff}(\mathcal{D})} \left\{\lambda : -\lambda  \Phi^{\mathcal{N}} \leq \pi_{d}\otimes\tau \leq \lambda \Phi^{\mathcal{N}} \right\}\right]^{-1}\\
    & = \left[ \inf_{\tau \in \operatorname{aff}(\mathcal{D})} \left \| (\Phi^{\mathcal{N}})^{-1/2}(\pi_{d}\otimes\tau) (\Phi^{\mathcal{N}})^{-1/2}\right \|_\infty\right]^{-1} .
\end{align}

The following inequality is an immediate consequence of \cref{def:q-doeblin-coeff} and \cref{def:b-doeblin-coef}, because the optimization in \cref{def:b-doeblin-coef} has more constraints than that in \cref{def:q-doeblin-coeff}:
\begin{equation}
    \alpha(\mathcal{N}) \geq \alpha_{\wang}(\mathcal{N}).
    \label{eq:a-b-doeblin-ineq}
\end{equation}

\begin{proposition}[Dual of $\alpha_{\wang}(\mathcal{N})$]
For a quantum channel $\mathcal{N}_{A\to B}$, the following equality holds:
\begin{equation}
    \alpha_{\wang}(\mathcal{N}) = \inf_{Y^{1}_{AB},Y^{2}_{AB}\geq0}\left\{  \operatorname{Tr}\!\left[  \left(  Y^{1}_{AB}+Y^{2}_{AB}\right)\Gamma
^{\mathcal{N}}_{AB}  \right]  :I_B=\operatorname{Tr}
_{A}[Y^{2}_{AB}-Y^{1}_{AB}]\right\},
\label{eq:dual-of-bN-quantity}
\end{equation}
where $\Gamma
^{\mathcal{N}}_{AB}$ is defined in \eqref{eq:choi_operator}.
\end{proposition}

\begin{proof}
    Consider that
\begin{align}
\alpha_{\wang}(\mathcal{N}) & =\sup_{X\in\operatorname{Herm}}\left\{  \operatorname{Tr}
[X]:-\Gamma^{\mathcal{N}}\leq I\otimes X\leq\Gamma^{\mathcal{N}}\right\}  \\
& =\sup_{X\in\operatorname{Herm}}\left\{  \operatorname{Tr}[X]+\inf_{Y_{1},Y_{2}\geq
0}\operatorname{Tr}\!\left[  \left(  I\otimes X+\Gamma^{\mathcal{N}}\right)
Y_{1}\right]  +\operatorname{Tr}\!\left[  \left(  \Gamma^{\mathcal{N}}-I\otimes
X\right)  Y_{2}\right]  \right\}  \\
& =\sup_{X\in\operatorname{Herm}}\inf_{Y_{1},Y_{2}\geq0}\left\{  \operatorname{Tr}
[X]+\operatorname{Tr}\!\left[  \left(  I\otimes X+\Gamma^{\mathcal{N}}\right)
Y_{1}\right]  +\operatorname{Tr}\!\left[  \left(  \Gamma^{\mathcal{N}}-I\otimes
X\right)  Y_{2}\right]  \right\}  \\
& =\sup_{X\in\operatorname{Herm}}\inf_{Y_{1},Y_{2}\geq0}\left\{  \operatorname{Tr}
\left[  \Gamma^{\mathcal{N}}\left(  Y_{1}+Y_{2}\right)  \right]
+\operatorname{Tr}\!\left[  X\left(  I+\operatorname{Tr}_{A}[Y_{1}
-Y_{2}]\right)  \right]  \right\}  \\
& \leq\inf_{Y_{1},Y_{2}\geq0}\sup_{X\in\operatorname{Herm}}\left\{  \operatorname{Tr}
\left[  \Gamma^{\mathcal{N}}\left(  Y_{1}+Y_{2}\right)  \right]
+\operatorname{Tr}\!\left[  X\left(  I+\operatorname{Tr}_{A}[Y_{1}
-Y_{2}]\right)  \right]  \right\}  \\
& =\inf_{Y_{1},Y_{2}\geq0}\left\{  \operatorname{Tr}\!\left[  \Gamma
^{\mathcal{N}}\left(  Y_{1}+Y_{2}\right)  \right]  :I=\operatorname{Tr}
_{A}[Y_{2}-Y_{1}]\right\}  .
\end{align}
Strong duality holds by picking $X=0$ (primal feasible)$\ $and $Y_{2}=\left(
1+\varepsilon\right)  \frac{I}{\left\vert A\right\vert }\otimes I$ and
$Y_{1}=\varepsilon\frac{I}{\left\vert A\right\vert }\otimes I$, for
$\varepsilon>0$, so that these choices are strictly feasible for the dual.
\end{proof}

We define more alternative Doeblin coefficients by restricting the $X \in \Herm$ condition to $X \geq 0$.
\begin{definition}
    For a quantum channel $\cN$, define the following quantities:
\begin{align}
    \alpha_{+}(\cN) & \coloneqq  \sup_{X \geq 0} \left\{ \Tr[X]: I \otimes X \leq \Gamma^\cN \right\}\label{def:alpha-plus-doeblin}\\
    \alpha_{+,I}(\cN) & \coloneqq  \sup_{X \geq 0} \left\{ \Tr[X]: {\cR}^X \leq_P \mathcal{N} \right\},\label{def:alpha_I-plus-doeblin}
\end{align}
where $\cR^X(\cdot)= \Tr[\cdot] X$.
\end{definition}

Note that
\begin{equation}
    \alpha_+(\mathcal{N}) \leq \alpha_{\wang}(\mathcal{N})
\end{equation}
because $X\geq 0$ and $I\otimes X \leq \Gamma^{\mathcal{N}}$ imply that $-\Gamma^{\mathcal{N}} \leq I\otimes X \leq \Gamma^{\mathcal{N}}$, so that optimizing over Hermitian $X$ satisfying $-\Gamma^{\mathcal{N}} \leq I\otimes X \leq \Gamma^{\mathcal{N}}$ can only give the same or a larger value for~$\operatorname{Tr}[X]$.

\subsubsection{Induced and Cone-Restricted Doeblin Coefficients}

In this section we consider even more variants of quantum Doeblin coefficients. We begin by defining an induced version of the Doeblin coefficient, which arises from considering the classical Doeblin coefficient evaluated over all classical channels induced by a given quantum channel. We then establish a cone program for this quantity, and one of our findings is that the induced Doeblin coefficient is equal to the expression in~\eqref{eq:alpha-induced-def-pos-cone}. By generalizing the cone program, we obtain `cone-restricted' Doeblin coefficients, which recover the quantum Doeblin coefficient $\alpha(\mathcal{N})$ as a special case.

Recall the expression for the classical Doeblin coefficient from~\eqref{eq:minimum-likelihood-decoder}. 
We can then consider the `induced Doeblin coefficient' of a quantum channel by the use of prepared states and POVMs to map the quantum output system to a classical system. This is analogous to what was done in~\cite{chitambar2022communication}, where the communication value of a quantum channel is induced by the communication value of a classical channel.

\begin{definition}\label{def:induced-doeblin-coefficient}
    For a general channel $\cN_{A \to B}$ and $n,n'\in\mathbb{N}$, define the $n\to n'$ best-case induced Doeblin coefficient as
    \begin{align}
    \label{eq:n-to-np-induced-doeblin}
    \alpha_{I}^{n \to n'}(\cN) \coloneqq  \min_{\substack{(\Lambda_{y})_{y \in [n']} \in \operatorname{POVM} \\ (\rho_{x})_{x \in [n]}}} \left\{ \sum_{y \in [n']} \min_{x\in [n]} \tr[\Lambda_{y} \cN(\rho_{x})] \right \}\ .
\end{align}
Moreover, the best-case  
induced Doeblin coefficient is defined as
    \begin{align}
    \label{eq:induced-doeblin}
    \alpha_{I}(\cN) \coloneqq  \inf_{n,n' \in \mathbb{N}} \alpha_{I}^{n \to n'}(\cN) \ .
    \end{align}
\end{definition}

As we show later on (\cref{prop:DPI-and-LU-invariance}), the Doeblin coefficient $\alpha(\mathcal{N})$ satisfies a data-processing inequality. Using this now, we conclude that the induced Doeblin coefficient $\alpha_I(\mathcal{N})$ cannot be smaller than the quantum Doeblin coefficient $\alpha(\mathcal{N})$.

\begin{proposition}
\label{prop:induced-doeblin-larger-than-doeblin}
    For a quantum channel $\cN$, the following inequality holds:
    \begin{equation}
        \alpha_{I}(\cN) \geq \alpha(\cN),
    \end{equation}
where $\alpha_I(\cN)$ and $\alpha (\cN)$ are defined in~\eqref{eq:induced-doeblin} and~\eqref{eq:alpha_cN_def}, respectively.
\end{proposition}

\begin{proof}
Consider preparation and measurement channels $\cF_{X \to A}$, and $\cE_{B \to Y}$ respectively. Then by \cref{prop:DPI-and-LU-invariance}, it follows that $\alpha(\cE \circ \cN \circ \cF) \leq \alpha(\cN)$. We finally conclude that
\begin{align}
    \alpha_{I}(\cN) = \inf_{\cF_{X \to A},\cE_{B \to Y}} \alpha(\cF \circ \cN \circ \cE) \geq \alpha(\cN) \ , 
\end{align}
where we used the definition of the induced Doeblin coefficient.
\end{proof}

Next, we aim to show that $\alpha_{I}(\cN)$ can be expressed as a conic program. This follows in steps and closely follows proof methods from~\cite{Chitambar-2021a}.

\begin{proposition}
\label{prop:induced-Doeblin-structure}
    Let $\cN_{A \to B}$ be a quantum channel, and fix $n\in\mathbb{N}$. Then
    \begin{align}\label{eq:alpha-n-to-n}
        \alpha_{I}^{n \to n}(\cN) & = \inf_{\substack{(\Lambda_{x})_{x \in [n]} \\ (\ket{\psi_{x}})_{x \in [n]}}} \sum_{x \in [n]} \Tr[\Lambda_{x} \cN(\dyad{\psi_{x}})] \ , \\
        \alpha_{I}(\cN) &= \alpha_{I}^{d_{B}^{2} \to d_{B}^{2}}(\cN) \ .
    \end{align}
\end{proposition}

\begin{proof}
Identifying $\cY = [n']$, $\cX = [n]$, for every choice of strategy, we may consider the function $f: \cY \to \cX$ defined by $f(y) \coloneq \text{argmin}_{x} \Tr[\Lambda_{y}\cN(\rho_{x})]$, so that $\alpha^{n \to n'} = \sum_{y} \Tr[\Lambda_{y}\cN(\rho_{f(y)})]$. First we will show that, without loss of generality, the induced Doeblin coefficient $\alpha_{I}(\cN)$ is achieved using a strategy where $f$ is bijective, at which point a relabeling will allow us to reduce the induced quantum Doeblin coefficient to the expression $\alpha_{I}(\cN) = \sup_{n \in \mbb{N}} \alpha_{I}^{n \to n}(\cN)$.

Consider an induced classical channel $\cW$. If there is $\hat{x} \in \cX$ such that $\hat{x} \not \in \operatorname{Im}(f)$, then one may remove $\rho_{\hat{x}}$ to obtain a new induced channel with the same induced Doeblin coefficient, as $\rho_{\hat{x}}$ is not the minimizer for any $y \in \cY$. Thus, $f$ is surjective. 
    
    Now imagine there exist $y_{1} \neq y_{2}$ such that $f(y_1) = f(y_2) = \hat{x}$; i.e., $f$ is not injective. Define $\ol{W}_{\ol{Y} \vert X}(\ol{y} \vert x)$ where $\ol{\cY} \coloneq \cY \setminus \{y_{1} , y_{2}\} \cup \{\hat{y}\}$,  which is induced by altering the POVM by $\Lambda_{\wt{y}} \coloneq \Lambda_{y}$ for $y \in \cY \setminus \{y_{1} , y_{2}\}$ and $\Lambda_{\hat{y}} \coloneq \Lambda_{y_{1}} + \Lambda_{y_{2}}$. Then,
    \begin{align}
        \alpha(\cW) & = \left(\sum_{y \in \cY \setminus \{y_{1} , y_{2}\}} \min_{x} \Tr[\Lambda_{y}\cN(\rho_{x})]\right) + \Tr[\Lambda_{y_{1}}\cN(\rho_{\hat{x}})] + \Tr[\Lambda_{y_{2}}\cN(\rho_{\hat{x}})]   \\
        & \geq  \sum_{y \in \cY \setminus \{y_{1} , y_{2}\}} \min_{x} \Tr[\Lambda_{y}\cN(\rho_{x})] + \min_{x} \Tr[\Lambda_{\hat{y}}\cN(\rho_{x})] \\ 
        & = \alpha(\ol{\cW}) \ ,
    \end{align}
    where the inequality follows from linearity and the definition of $\Lambda_{\hat{y}}$. Since the induced Doeblin coefficient involves an infimization, we can repeat this procedure iteratively to obtain an injective function. Thus, an optimal $f$ is always bijective. Taking everything together, this implies that
    $\alpha_{I}(\cN) = \inf_{n \in \mbb{N}} \alpha_{I}^{n \to n}(\cN)$, where
    \begin{equation}
    \alpha_{I}^{n \to n}(\cN) = \inf_{\substack{(\Lambda_{x})_{x \in [n]} \\ (\rho_{x})_{x \in [n]}}} \sum_{x \in [n]} \Tr[\Lambda_{x} \cN(\rho_{\sigma(x)})] \ ,
    \end{equation}
    and $\sigma$ is a permutation on $n$ objects, as $f$ is a bijection and thus a permutation. Lastly, we may relabel the states with the inverse permutation $\sigma^{-1}$ to obtain the form
    \begin{align}
         \alpha_{I}^{n \to n}(\cN) = \inf_{\substack{(\Lambda_{x})_{x \in [n]} \\ (\rho_{x})_{x \in [n]}}} \sum_{x \in [n]} \Tr[\Lambda_{x} \cN(\rho_{x})] \ .
    \end{align}
    Now, we show that we can restrict the optimization to be over pure states. To do this, for each $x \in \cX$ let $\rho_{x} = \sum_{y=1}^{r_{x}} p^{x}(y) \dyad{\phi^{x}_{y}}$ where $r_{x}$ is the rank of $\rho_{X}$ and $p^{x}(y)$ is a probability distribution. The existence of this decomposition follows from the spectral theorem. Then we have
    \begin{align}
        \sum_{x} \Tr[\Lambda_{x} \cN(\rho_{x})] = \sum_{x} \Tr\!\left[\Lambda_{x} \sum_{y} p^{x}(y)\cN(\dyad{\phi^{x}_{y}})\right] \geq \sum_{x} \Tr[\Lambda_{x} \cN(\dyad{\phi_{y^{\star}}^{x}}] \ , 
    \end{align}
    where $\ket{\phi^{x}_{y^{\star}}} \coloneq \text{argmin}_{y} \Tr[\Lambda_{x}\dyad{\phi^{x}_{y}}]$, which justifies the inequality. Thus, one can always decrease the value by replacing $\rho_{x}$ with one of the pure states in its spectral decomposition, and so the restriction to pure states is without loss of generality. This establishes~\eqref{eq:alpha-n-to-n}. 
    
    Finally, we will establish that $\alpha(\cN) = \alpha_{I}^{d_{B}^{2} \to d_{B}^{2}}(\cN)$. We do this by a common convexity argument for bounding the cardinality of the necessary POVM outcome size. For clarity, we provide the argument in full. First, recall that a POVM on a space $B$, $(\Lambda_{x})_{x \in [n]}$, is extremal if it cannot be expressed as a convex combination of other POVMs on that space, in the sense that it cannot be expressed as $\Lambda_{x} = \sum_{y} p(y)\Lambda_{x}^{y}$ where $p$ is a non-degenerate probability distribution and $((\Lambda_{x}^{y})_{x})_{y}$ are POVMs. Second, recalling Minkowski's theorem~\cite[Theorem 1.10]{Watrous-Book}, and noting that the set of POVMs on (finite-dimensional) $B$ is compact, all POVMs may be expressed as convex combinations of extremal POVMs; i.e., $(\Lambda_{x})_x$ is such that $\Lambda_{x} = \sum_{y} p(y)\Lambda^{y}_{x}$ where $p$ is a distribution, and for each $y$,  the POVM $(\Lambda^{y}_{x})_{x}$  
    is an extremal POVM on $B$. Finally, recall that an extremal POVM on $B$ has at most $d_{B}^{2}$ non-zero outcomes~\cite[Corollary~2.48]{Watrous-Book}. It follows that for all $n \geq d_{B}^{2}$, the inequality $\alpha_{I}^{n \to n} \geq \alpha_{I}^{d_{B}^{2} \to d_{B}^{2}}$ holds,  following from linearity in the choice of POVM. That is, for all $n \geq d_{B}^{2}$, every set $(\ket{\psi_{x}})_{x \in [n]}$ and (possibly non-extremal) POVM $(\Lambda_{x})_{x \in [n]}$ such that $\Lambda_{x} \neq 0$ for all $x \in [n]$,
    \begin{align}
        \sum_{x} \Tr[\Lambda_{x}\cN(\dyad{\psi_{x}})] & = \sum_{x} \Tr\!\left[\sum_{y}p(y)\Lambda^{y}_{x}\cN(\dyad{\psi_{x}})\right] \\
        & = \sum_{y}p(y) \sum_{x}\Tr[\Lambda_{x}^{y}\cN(\dyad{\psi_{x}})] \\
        & \geq \min_{y} \sum_{x}\Tr[\Lambda_{x}^{y}\cN(\dyad{\psi_{x}})] \\
        & \geq \alpha_{I}^{d_{B}^{2} \to d_{B}^{2}}(\cN) \ , 
    \end{align}
    where we used the definition of $\alpha_{I}^{n \to n}$ in the final inequality and that an extremal POVM $\left\{\Lambda^{y}_{x}: \Lambda^{y}_{x} \neq 0\right\}$ with the $\{\ket{\psi_{x}}: \Lambda_{x}^{y} \neq 0\}$ is always feasible for $\alpha_{I}^{d_{B}^{2} \to d_{B}^{2}}$. As $\alpha_{I}(\cN) = \inf_{n} \alpha_{I}^{n \to n}(\cN)$ and $\alpha_{I}^{n \to n}(\cN)$ is non-increasing in $n$, this completes the proof.
\end{proof}

\begin{proposition}
\label{prop:induced-Doeblin}
    Given a quantum channel $\cN_{A \to B}$, the induced Doeblin coefficient $\alpha_{I}(\cN)$ from~\eqref{eq:induced-doeblin} is equal to both of the following conic optimization programs:
    \begin{center}
    \begin{miniproblem}{0.5}
      \emph{Primal problem}\\[-7mm]
      \begin{equation}
      \begin{aligned}\label{eqn:inducedDoeblinPrimal}
        \text{maximize:}\quad & \tr[X] \\
        \text{subject to:}\quad & \Gamma^\cN- {I}_{A} \otimes X \in \Sep^{\ast}(A:B)  \\
        & X\in \Herm(B) \ .
      \end{aligned}
      \end{equation}
    \end{miniproblem}
    \hspace*{2mm}
    \begin{miniproblem}{0.4}
      \emph{Dual problem}\\[-7mm]
      \begin{equation}
      \begin{aligned}\label{eqn:inducedDoeblinDual}
        \text{minimize:}\quad & \Tr[Y\Gamma^{\cN}] \\
        \text{subject to:}\quad & \tr_{A}[Y] = I_{B} \\
        & Y \in \Sep(A:B)
      \end{aligned}
      \end{equation}
    \end{miniproblem}
 \end{center}
\end{proposition}
\begin{proof}
   The proof is nearly identical to that of~\cite[Proposition 2]{chitambar2022communication}, but we provide it for completeness. By \cref{prop:induced-Doeblin-structure}, letting $n = d_{B}^{2}$, we conclude that
   \begin{align}
    \alpha_{I}(\cN)= &\min_{\substack{(\Lambda_{x})_{x \in [n]} \in \operatorname{POVM} \\ (\rho_{x})_{x \in [n]}}}  \sum_{x \in [n]} \tr[\Lambda_{x} \cN(\rho_{x})]  \\
      &= \min_{\substack{(\Lambda_{x})_{x \in [n]} \in \operatorname{POVM} \\ (\rho_{x})_{x \in [n]}}} \sum_{x \in [n]} \tr[(\rho_{x}^{T} \otimes \Lambda_{x}) \Gamma^{\cN}] \ ,
   \end{align}
   where we used the action of the channel in terms of its Choi operator. Note that the operator $Y \coloneq \sum_{x} \rho^{T}_{x} \otimes \Lambda_{x} \geq 0$ is separable and satisfies $\Tr_{A}[Y] = \sum_{x} \Lambda_{x} = I_{B}$. Similarly, if $Y$ is separable, then there exists a finite alphabet $\cY$ such that $Y = \sum_{y} \dyad{\psi_{y}} \otimes P_{y}$ where $P_{y} \geq 0$ for all $y \in \cY$. Moreover, $\Tr_{A}[Y] = I_{B}$ implies that $\sum_{y} P_{y} = I_{B}$. Thus, this defines a feasible strategy. This derives what we labeled the dual problem. This is already in the form of \eqref{eqn:cone-dual} where we let $\Phi \coloneq \Tr_{A}$, so we obtain the primal problem immediately by recalling the adjoint map of $\Tr_{A}$ consists of tensoring the identity on the $A$ space.
\end{proof}

With this cone program established, we are able to show that the induced Doeblin coefficient of this section is indeed the same as the quantity defined in \cref{def:alt-doeblin-quantum-induced}.
\begin{proposition}
\label{prop:equality-of-induced-doeblins}
    The quantities denoted by $\alpha_{I}(\cN)$ in Definitions~\ref{def:alt-doeblin-quantum-induced} and \ref{def:induced-doeblin-coefficient} are equal.
\end{proposition}

\begin{proof}
   As the objective functions of \cref{def:alt-doeblin-quantum-induced} and the primal problem of \cref{prop:induced-Doeblin} are the same, it suffices to show that the constraints of these optimization programs are equivalent. We have the following equivalences:
   \begin{align}
         \cN - \cR^{X} \in \Pos \quad 
       \Leftrightarrow \quad & \; \Gamma^{\cN - \cR^{X}} \in \Sep^{\ast}(A:B) \\
       \Leftrightarrow \quad & \; \Gamma^{\cN} - I_{A} \otimes X \in \Sep^{\ast}(A:B) \ , 
   \end{align}
   where the first equivalence is given in \cref{prop:positive-to-block-positive} and the second equivalence follows from the definition of the Choi operator and a replacer map. As in both cases, $X \in \Herm(A)$, this shows that the constraints are equivalent and completes the proof.
\end{proof}

Finally, following~\cite{george2024cone}, by replacing the separable cone with the positive cone $\cK \subset \Pos\!\left(A \otimes B\right)$, we define the cone-restricted Doeblin coefficient $\alpha_{I,\cK}(\cN)$ as follows:
 \begin{center}
    \begin{miniproblem}{0.5}
      \emph{Primal problem}\\[-7mm]
      \begin{equation}
      \begin{aligned}\label{eqn:inducedDoeblinPrimal_cone}
        \text{maximize:}\quad & \tr[X] \\
        \text{subject to:}\quad & \Gamma^\cN- {I}_{A} \otimes X \in \cK^{\ast}(A:B)  \\
        & X\in \Herm(B) \ .
      \end{aligned}
      \end{equation}
    \end{miniproblem}
    \hspace*{2mm}
    \begin{miniproblem}{0.4}
      \emph{Dual problem}\\[-7mm]
      \begin{equation}
      \begin{aligned}\label{eqn:inducedDoeblinDual_cone}
        \text{minimize:}\quad & \Tr[Y \Gamma^\cN ] \\
        \text{subject to:}\quad & \tr_{A}[Y] = {I}_{B} \\
        & Y \in \cK
      \end{aligned}
      \end{equation}
    \end{miniproblem}
 \end{center}
We note here that we need extra structure to guarantee strong duality in this case. A sufficient condition is for $I_{AB}$ to be in the relative interior of the cone $\cK$ in order to apply Slater's criterion for strong duality. This sufficient condition will be used later to work with $\alpha_{I,\cK}$.

We end this section by stressing that all these quantities studied above recover the classical Doeblin coefficient for classical channels, which justifies identifying them as quantum Doeblin coefficients.
\begin{proposition}\label{prop:recovers-classical-doeblin}
    Let $\cW_{Y \vert X}$ be a classical-to-classical channel. Then for $\Sep(A:B) \subseteq \cK \subset \Pos(A \otimes B)$,
    \begin{align}
        \alpha_{\operatorname{cl}}(\cW) = \alpha_{I,\cK}(\cW) = \alpha(\cW) = \alpha_{\wang}(\cW) \ , 
    \end{align}
    where $\alpha_{\operatorname{cl}}$ denotes the classical Doeblin coefficient.
\end{proposition}
\begin{proof}
    The key point is that $\cW$ is invariant under dephasing, and all the conic constraints we consider are preserved under this action, so we can reduce back to the classical case. There are various forms of such a proof. For concreteness, we provide one such proof for $\alpha_{I} = \alpha_{I,\Sep}$. The proofs for $\alpha$ and $\alpha_{\wang}$ are similar, and then the cones listed in the proposition statement take the claimed value, as they are bounded from below and above by $\alpha_{I}(\cW)$ and $\alpha_{\wang}(\cW)$, respectively.

    Let $X$ be an optimizer for $\alpha_{I}(\cW)$. First we show, without loss of generality, that it is diagonal in the computational basis. By Proposition \ref{prop:positive-to-block-positive}, $\cW - \cR^{X}$ is a positive map. Letting $\Delta_{Y}$ denote the completely dephasing map, we have that $\Delta_{Y} \circ (\cW - \cR^{X}) = \cW - \cR^{\Delta_{Y}(X)}$ is positive. Moreover $\Tr[X] = \Tr[\Delta_{Y}(X)]$. Thus, without loss of generality, an optimizer $X$ of $\alpha_{I}(\cW)$ is diagonal. As $X$ is Hermitian, its diagonal entries are all real. Thus, expressing the $y^{\mathrm{th}}$ diagonal entry as $r(y)$, we are interested in $ \Gamma^{\cW} - I \otimes X = \sum_{x,y} (W(y\vert x)-r(y))\dyad{x} \otimes \dyad{y}$. By \eqref{eq:block-positive-relation}, $\Sep^{\ast}(X:Y) \ni \Gamma^{\cW}-I \otimes X$ if and only if $r(y) \leq W(y \vert x)$ for all $x$ and $y$. Noting that $\Tr[X] = \sum_{y} r(y)$ and $W(y \vert x) \geq 0$, without loss of generality $r(y) \geq 0$ for all $y$. Putting these points together,
    \begin{align}
        \alpha_{I}(\cW) = \max_{r\in\mathbb{R}_{\geq0}^{\left\vert \mathcal{Y}
\right\vert }}\left\{  \sum_{y\in\mathcal{Y}}r(y):r(y)\leq W
(y|x)\ \forall x\in\mathcal{X},y\in\mathcal{Y}\right\} = \alpha_{\operatorname{cl}}(\cW) \ , 
    \end{align}
    where the final equality follows from~\eqref{eq:doeblin-classical}.
\end{proof}

\section{Interpretations of the Quantum Doeblin Coefficient}

\label{sec:doeblin-interpretations}

In the following subsections, we establish various interpretations of and alternate expressions for the quantum Doeblin coefficient $\alpha(\mathcal{N})$, as summarized in \cref{table:doeblin-expressions}. After that, we do the same for the induced Doeblin coefficient $\alpha_I(\mathcal{N})$, and note that these latter interpretations are summarized in \cref{table:induced-doeblin-expressions}.

\subsection{Minimum Achievable Singlet Fraction}

 In this section, we show that the quantum Doeblin coefficient in~\eqref{eq:alpha_cN_def} is proportional to the minimum achievable singlet fraction, which is a counterpart of the maximum achievable singlet fraction considered in~\cite{konig2009operational}. Alternatively, due to the connection with a task known as quantum state exclusion~\cite{bandyopadhyay2014conclusive}, this task could also be called fully quantum exclusion.
 
\begin{theorem} \label{thm:doeblin_antiDistinguish}
The quantum Doeblin coefficient of a channel $\mathcal{N}$ is proportional to the minimum
achievable singlet fraction:
\begin{align}
\alpha(\mathcal{N})  &  =d^{2}\inf_{\mathcal{D}\in\operatorname{CPTP}
}\operatorname{Tr}[\Phi_{d}(\operatorname{id}\otimes(\mathcal{D}
\circ\mathcal{N}))(\Phi_{d})]\\
&  =d^{2}\inf_{\mathcal{D}\in\operatorname{CPTP}}F(\Phi_{d},(\operatorname{id}\otimes(\mathcal{D}
\circ\mathcal{N}))(\Phi_{d})), \label{eq:MASF}
\end{align}
where $F(\omega,\tau)\coloneqq \left\Vert \sqrt{\omega}\sqrt{\tau}\right\Vert
_{1}^{2}$, $\Phi_d$ is defined in~\eqref{eq:max_entangled}, and $d$ is the dimension of the input system of $\mathcal{N}$.
\end{theorem}

\begin{proof}
Recall from~\eqref{eq:alpha-cN-dual} that the Doeblin coefficient can be expressed as 
\begin{equation}
\alpha(\mathcal{N})=\inf_{Y_{AB}\geq0}\left\{  \operatorname{Tr}[Y_{AB}\Gamma
^{\mathcal{N}}_{AB}]:I=\operatorname{Tr}_{A}[Y_{AB}]\right\}.
\end{equation}
The conditions $Y_{AB}\geq0$ and $I=\operatorname{Tr}_{A}[Y_{AB}]$ imply that $Y_{AB}$ is the
Choi operator of a completely positive unital (CPU)\ map from $A$ to $B$. Thus,
\begin{align}
&  \inf_{Y_{AB}\geq0}\left\{  \operatorname{Tr}[Y_{AB}\Gamma^{\mathcal{N}}_{AB}
]:I=\operatorname{Tr}_{A}[Y_{AB}]\right\} \nonumber\\
&  =\inf_{\mathcal{M}\in\operatorname{CPU}}\operatorname{Tr}[\Gamma
^{\mathcal{M}}_{AB}\Gamma^{\mathcal{N}}_{AB}]\label{eq:min-singlet-frac-proof-1}\\
&  =d^{2}\inf_{\mathcal{M}\in\operatorname{CPU}}\operatorname{Tr}
[(\operatorname{id}\otimes\mathcal{M})(\Phi_{d})(\operatorname{id}
\otimes\mathcal{N})(\Phi_{d})]\\
&  =d^{2}\inf_{\mathcal{M}\in\operatorname{CPU}}\operatorname{Tr}[\Phi
_{d}(\operatorname{id}\otimes(\mathcal{M}^{\dag}\circ\mathcal{N}))(\Phi
_{d})]\\
&  =d^{2}\inf_{\mathcal{D}\in\operatorname{CPTP}}\operatorname{Tr}[\Phi
_{d}(\operatorname{id}\otimes(\mathcal{D}\circ\mathcal{N}))(\Phi_{d})] \\
& =d^{2}\inf_{\mathcal{D}\in\operatorname{CPTP}}F(\Phi_{d},(\mathcal{D}
\circ\mathcal{N})(\Phi_{d})),
\label{eq:min-singlet-frac-proof-last}
\end{align}
where in the penultimate equality we used the fact that the Hilbert--Schmidt adjoint of a
unital map is a trace-preserving map.
\end{proof}

\begin{remark}
    By taking the decoding channel $\cD$ in~\eqref{eq:MASF} to be completely depolarizing, a direct calculation verifies that the value is equal to one. Thus one may use the above to see that $\alpha(\cN) \leq 1$ via this interpretation.
\end{remark}

Note that the analysis accommodates channels with a different input and output dimension (i.e., it is not required that $d_A = d_B$).

As indicated previously, the quantity on the right side of~\eqref{eq:MASF} is a counterpart of \textquotedblleft
quantum correlation\textquotedblright\ or \textquotedblleft maximum achievable
singlet fraction\textquotedblright\ considered in~\cite[Eq.~(20)]{konig2009operational}. As such, 
it is thus reasonable to call it the \textquotedblleft minimum achievable singlet
fraction,\textquotedblright as we have done.

\begin{remark}[Fully Quantum Counterpart of Exclusion]
    If the channel $\mathcal{N}$ is a classical--quantum channel of the
form
\begin{equation}
\mathcal{N}(\omega)=\sum_{x}\langle x|\omega|x\rangle\rho_{x},
\end{equation}
then after some direct calculations we obtain that the right-hand side of~\eqref{eq:MASF} becomes
\begin{equation}
\inf_{\mathcal{D}\in\operatorname{CPTP}}F(\Phi_{d},(\operatorname{id}\otimes (\mathcal{D}\circ
\mathcal{N}))(\Phi_{d}))=\frac{1}{d^{2}}\inf_{\left(  \Lambda_{x}\right)
_{x}}\sum_{x}\operatorname{Tr}[\Lambda_{x}\rho_{x}],
\end{equation}
where $\left(  \Lambda_{x}\right)
_{x}$ is a POVM,
so that indeed we have recovered a fully quantum counterpart of the
exclusion task, as considered in~\cite{bandyopadhyay2014conclusive}.
\end{remark}

\subsection{Exclusion Value of a Quantum Channel}

The induced Doeblin coefficient of a quantum channel and its cone-restricted variants have similar mathematical structure to the communication value of a quantum channel~\cite{chitambar2022communication} and its cone-restricted variants~\cite{george2024cone}, respectively. Thus, one may expect that there is an operational interpretation for the quantum Doeblin coefficient analogous to the communication value. Indeed, the communication value of a quantum channel $\cN$ may be expressed in terms of the minimal error probability of state discrimination over the quantum channel $\cN$, which is one reason that it is a measure of the ability to communicate over the channel $\cN$. It also generalizes to entanglement-assisted variations by using the cone-restricted variant~\cite{george2024cone}.

In this section, we relate the quantum Doeblin coefficient to performing quantum state exclusion (QSE) over the channel. We do this by first observing that the induced Doeblin coefficient can be expressed in terms of the error probability of QSE, then extending generalizing this to the entanglement-assisted setting for cone-restricted Doeblin coefficients (\cref{thm:game-interpretation}). In particular, this proves that $\alpha(\cN)$ is the optimal classical Doeblin coefficient taken over all classical channels achievable using $\cN$ and arbitrary entanglement assistance.

Recalling the definition of error probability for QSE (\cref{def:state-exclusion-err-prob}), we observe that the induced Doeblin coefficient may be expressed as an infimization involving quantum state exclusion.
\begin{proposition}\label{prop:induced-Doeblin-QSE-game}
    The induced Doeblin coefficient optimizes a rescaled version of  state exclusion (with uniform prior) over all possible ensembles:
    \begin{align}
        \alpha_{I}(\cN) = \inf_{n \in \mbb{N}} n \cdot \inf_{(\rho^{x})_{x \in [n]}} P_{\operatorname{err}}\!\left((n^{-1},\cN(\rho_{x}))_x\right) \ . 
    \end{align}
\end{proposition}

\begin{proof}
    By \cref{prop:induced-Doeblin-structure},
    \begin{align}
        \alpha_{I}(\cN) &= \inf_{n \in \mbb{N}} \inf_{\substack{(\Lambda_{x})_{x \in [n]} \\ (\rho_{x})_{x \in [n]}}} \sum_{x \in [n]} \Tr[\Lambda_{x} \cN(\rho_{x})] \\
        &= \inf_{n \in \mbb{N}}  n \inf_{(\rho_{x})_{x \in [n]}} \inf_{(\Lambda_{x})_{x \in [n]}} \frac{1}{n} \sum_{x \in [n]} \Tr[\Lambda_{x} \cN(\rho_{x})] \\
        &= \inf_{n \in \mbb{N}}  n \inf_{(\rho_{x})_{x \in [n]}} P_{\operatorname{err}}\!\left(\left(\frac{1}{n}, \cN(\rho_{x})\right)_x\right) \ ,
     \end{align} 
     concluding the proof.
\end{proof}

\begin{figure}
    \begin{center}
    \begin{tikzpicture}
        \tikzstyle{porte} = [draw=black!50, fill=black!20]
        \draw
            (0,0) node (EA) {$\ket{\varphi}_{A'B'}$}
            ++(-0.2,-1) node (input) {$x \sim \cX$}
            ++(2.2,0.2) node[porte, minimum size =1cm] (Encoder) {$\cE^{x}_{A' \to A}$}
            ++(2,0) node[porte, minimum height = 1cm, minimum width = 1.5cm] (Channel) {$\cN_{A \to B}$}
            ++(2.2,0.65) node[porte,minimum height = 2.3cm, minimum width = 1cm] (Decoder) {$\left(\Lambda_{BB'}^{\hat{x}}\right)_{\hat{x}}$}
            ++(1.5,0) node (output) {$\hat{x}$}
            ;
            \path[draw=black, -] (0.4,-0.25) -- (0.5,-0.6)
            ;
            \path[draw=black, ->] (0.5,-0.6) -- (1.34,-0.6);
            \path[draw=black, ->] (0.4,-1) -- (1.34,-1);
            \path[draw=black, ->] (Encoder) -- (Channel)
            ;
            \path[draw=black, ->] (Channel) -- (5.3,-0.8);
            \path[draw=black,-] (0.4,0.2) -- (0.5,0.5);
            \path[draw=black,->] (0.5,0.5) -- (5.3,0.5);
            \path[draw=black,->] (Decoder) -- (output);
    \end{tikzpicture}
    \end{center}
    \caption{Depiction of the game that characterizes the value of $\alpha(\cN)$. One infimizes over choices of encoders and decoders, so it is minimizing the quantum state exclusion over all codes.}
    \label{fig:game-interpretation}
\end{figure}

Our goal is now to generalize the above proposition by adding entanglement assistance and calling it a `channel state exclusion game.' The game is depicted in \cref{fig:game-interpretation}.
\begin{definition}
\label{def:exclusion-game-cone}
     Given a cone $\cK(B \otimes B') \subseteq \Pos(B \otimes B')$, the channel state exclusion game (see \cref{fig:game-interpretation}) is characterized by
     \begin{align}
        \cG^{\cK}_{\operatorname{QSE}}(\cN) \coloneq \inf_{\cX,(\Lambda^{x})_x, (\cE^{x})_x, \ket{\varphi}} \sum_{x} P(x \vert x) \ , 
     \end{align}
     where the infimum is over every decoder (POVM) $(\Lambda^{x})_{x}$ such that each element is in the cone~$\cK$, every encoding channel $(\cE^{x})_x$, and
     \begin{equation}
     P(x' \vert x) \coloneq \Tr[\Lambda^{x'}_{BB'} (\cN_{A \to B} \circ \cE^{x}_{A' \to A})(\varphi_{A'B'})]    ,
     \end{equation}
     where $\varphi_{A'B'}$ is an arbitrary pure bipartite state.
 \end{definition}
 
 This is a state exclusion game as, up to a rescaling by the input dimension, it optimizes the state exclusion error probability, with uniform prior, when using the channel $\cN$. In other words, it is a generalization of \cref{prop:induced-Doeblin-QSE-game} to the setting where there is entanglement-assistance.

 \begin{theorem}
 \label{thm:game-interpretation}
 Let $\cK(A \otimes B) \subseteq \Pos(A \otimes B)$ be a closed, convex cone such that it is invariant under local CP maps on the $A$ space and $\I \in \operatorname{relint}(\cK)$. Then $\alpha_{I,\cK}(\cN) = \cG_{\operatorname{QSE}}^{\cK}(\cN)$.
\end{theorem}

\begin{proof}
    See \cref{app:proof-game-interpretation-exclusion-value}.
\end{proof}

As the proof is nearly identical to one in \cite{george2024cone} which itself is a rather direct generalization of a proof in \cite{chitambar2022communication}, we relegate it to \cref{app:proof-game-interpretation-exclusion-value}. We now state the most important case of \cref{thm:game-interpretation}, which makes use of the definition of the game (\cref{def:exclusion-game-cone}) and~\eqref{eq:alpha-n-to-n}. 

\begin{corollary}
\label{cor:q-doeblin-as-EA-assisted-doeblin}
    The quantum Doeblin coefficient $\alpha(\cN)$ is equal to the best-case Doeblin coefficient optimized over all classical channels induced by $\cN$ and arbitrary entanglement assistance. That is to say, 
    \begin{align}
        \alpha(\cN) = \inf_{\substack{\cX, \varphi \in \cS(A'B') \\
        (\cE^{x}_{A' \to A})_x, \cD_{BB' \to X}}} \alpha(\cW^{(\cE^{x}),\cD,\varphi}) \ ,
        \label{eq:EA-exclusion-value}
    \end{align}
    where the infimization is over every finite alphabet $\cX$, pure state $\varphi_{A'B'}$, encoder $(\cE^{x}_{A' \to A})_{x \in \cX}$, and decoder $\cD_{BB' \to X}$, such that they induce a classical channel with conditional probabilities given by
    \begin{align}
        \cW^{(\cE^{x}),\cD,\phi}(x' \vert x) \coloneq \Tr[\Lambda^{x'}_{BB'} (\id_{B'} \otimes (\cN_{A\to B} \circ \cE_{A'\to A}^{x}))(\varphi_{A'B'})] \ .
    \end{align}
\end{corollary}

We remark that, in the same way as identified in \cite[Corollary 19]{george2024cone} for the communication value, what \cref{thm:game-interpretation} and \cref{cor:q-doeblin-as-EA-assisted-doeblin} show is that $\alpha(\cN)$ and $\alpha_{I}(\cN)$ are interpolated between by providing arbitrary entanglement assistance, but limiting the entangling power of the decoder. We state this formally using the relevant cones, which generalizes the previous corollary.
\begin{corollary}\label{cor:entangling-power}
    For every $k \in \mbb{N}$, define the entanglement-rank-$k$ cone \cite{Watrous-Book}
    \begin{align*}
        \operatorname{Ent}_{k}(A:B) \coloneq \{P_{AB} \geq 0: \; \exists \, \cI \, : \exists \{X_{i}\} \subset \operatorname{Lin}(A,B) : P = \sum_{i \in \cI} \operatorname{vec}(X_{i})\operatorname{vec}(X_{i})^{\dagger} \} \ ,
    \end{align*}
    where $\operatorname{vec}$ is the `vec mapping' (see, e.g.,~\cite[Eq.~(1.128)]{Watrous-Book}). Then for channel $\cN_{A \to B}$ and every $k \in \{1,\ldots ,d\}$ where $d = \min\{\vert A \vert , \vert B \vert\}$
    \begin{align} 
        \alpha_{I}(\cN) = \alpha_{I,\operatorname{Ent}_{1}}(\cN) \geq \alpha_{I,\operatorname{Ent}_{k}}(\cN) \geq \alpha_{I,\operatorname{Ent}_{d}}(\cN) = \alpha_{I,\Pos}(\cN) = \alpha(\cN) \ .
    \end{align}
    In other words, one interpolates between the quantum Doeblin coefficients ${\alpha_I}(\cN)$ and $\alpha(\cN)$ operationally by considering the channel state exclusion game (\cref{def:exclusion-game-cone}) and interpolating between the decoder POVM being a separable measurement, $(\Lambda^{x})_{x} \subset \cK_{1} = \Sep(A:B)$, and the decoder elements being arbitrary, $(\Lambda^{x})_{x} \subset \cK_{d}(A:B) = \Pos(A \otimes B)$ by restricting the POVM elements to a given entanglement rank cone $k$.
\end{corollary}
\begin{proof}
    First, it is known \cite[Section 6.1]{Watrous-Book} that for Hilbert spaces $A$ and $B$  
    $$\Sep(A:B) = \operatorname{Ent}_{1}(A:B) \subseteq \operatorname{Ent}_{k}(A:B) \subseteq \Pos(A \otimes B) = \operatorname{Ent}_{\min\{\vert A \vert, \vert B \vert\}}(A:B) \ , $$
    which implies the ordering between coefficients by \eqref{eqn:inducedDoeblinDual}. Second, it is known that the entanglement rank cones are closed under completely positive maps acting on one system (special case of \cite[Theorem 6.23]{Watrous-Book}). It is also known $I \in \operatorname{relint}(\operatorname{Ent}_{k}(A:B))$ for all $k$ \cite[Corollary 1]{george2024cone}. Thus, the cones satisfy the conditions needed to apply \cref{thm:game-interpretation}. The proof is then completed by considering \cref{prop:induced-Doeblin} and \cref{def:q-doeblin-coeff}.
\end{proof}

While the above claim uses a finite-sized hierarchy, it is not known to be efficiently computable. A similar claim may be made using the cones generated by the intersection of the PPT and symmetrically extendible cones, which are SDP representable. Such a claim was not made explicit for the communication value.
\begin{corollary}\label{cor:entangling-power-in-terms-of-k-sym}
    For every $k \in \mbb{N}$, define the cone $\cK_{k}(A:B) \coloneq \operatorname{PPT}(A:B) \cap \operatorname{Sym}_{k}(A:B)$ where these sets are defined in \eqref{eq:PPT-cone} and \cref{def:k-sym-ext-cone}. Then for every $k \in \mbb{N}$,
    \begin{align}
        \alpha_{I}(\cN) = \alpha_{I,\Sep}(\cN) \geq \alpha_{I,\cK_{k}}(\cN) \geq \alpha_{I,\cK_{1}}(\cN) = \alpha_{I,\Pos}(\cN) = \alpha(\cN) \ .
    \end{align}
    In other words, one operationally interpolates from the quantum Doeblin coefficient $\alpha(\cN)$ to ${\alpha_I}(\cN)$ by considering the channel state exclusion game (\cref{def:exclusion-game-cone}) with the decoder being made up of POVM elements that are PPT and $k$-symmetrically extendible and increasing the value of $k$. 
\end{corollary}

\begin{proof}
    One can verify using the definition of the PPT cone and $\operatorname{Sym}_{k}(A:B)$ that if an element is contained in the cone, then it still is after applying a completely positive map to the $A$ space. It follows that $\cK_{k}(A:B)$ also satisfies this property. Moreover, for every $k \in \mbb{N}$, $\Sep(A:B) \subseteq \cK_{k}(A:B)$ and $I \in \operatorname{relint}(\Sep(A:B))$ due to \cite{gurvits2002largest} (see \cite[Corollary 1]{george2024cone} to see this proven explicitly). Thus, for all $k \in \mbb{N}$, $\cK_{k}(A:B)$ satisfies the conditions on the cone to apply \cref{thm:game-interpretation}. The claimed inequalities and equalities then follow from \cref{prop:induced-Doeblin}, \cref{def:q-doeblin-coeff}, and that $\Pos(A \otimes B) = \operatorname{Sym}_{1}(A:B) \supseteq \operatorname{Sym}_{2}(A:B) \supseteq \cdots \supseteq \operatorname{Sym}_{k}(A:B) \supseteq \Sep(A:B)$ for every $k \in \mbb{N}$.
\end{proof}

\subsection{Reverse Max-Mutual Information}

In this section, we observe that the Doeblin coefficient $\alpha(\mathcal{N})$ can be written in terms of the reverse max-mutual information $\inf_{\tau\in\operatorname{aff}(\mathcal{D}
)}D_{\max}(\mathcal{R}^{\tau}\Vert\mathcal{N})$ of the channel $\mathcal{N}$ (alternatively, this quantity can be called the Doeblin information of the channel). This is a counterpart of the fact that the entanglement-assisted communication value can be written in terms of the max-mutual information of the channel (see \cref{rem:comm-val-max-mutual} below). The formulation given in \cref{prop:reverse-max-mutual-info} can also be compared to the concept of resource weight \cite{Lewenstein1998,Regula2022prl,Regula2022tightconstraints}, studied in the context of quantum resource theories, especially if one considers the replacer map $\mathcal{R}^\tau$ to play the role of a free map.

\begin{proposition}
\label{prop:reverse-max-mutual-info}
The following equality holds for a quantum channel $\mathcal{N}$:
\begin{equation}
\alpha(\mathcal{N})=\exp\!\left(  -\inf_{\tau\in\operatorname{aff}(\mathcal{D}
)}D_{\max}(\mathcal{R}^{\tau}\Vert\mathcal{N})\right)  ,
\label{eq:reverse-max-mutual-info}
\end{equation}
where $D_{\max}(\mathcal{R}^{\tau}\Vert\mathcal{N})$ is defined from \eqref{eq:max-rel-ent-channels} and $\mathcal{R}^{\tau}(\cdot) \coloneqq \operatorname{Tr}[\cdot]\tau$ is a replacer superoperator.

\end{proposition}

\begin{proof}
Consider that
\begin{align}
\alpha(\mathcal{N})  & =\sup_{X\in\operatorname{Herm}}\left\{
\operatorname{Tr}[X]:\mathcal{R}^{X}\leq\mathcal{N}\right\}  \\
& =\sup_{\substack{\lambda\geq0,\\\tau\in\operatorname{aff}(\mathcal{D}
)}}\left\{  \operatorname{Tr}[\lambda\tau]:\mathcal{R}^{\lambda\tau}
\leq\mathcal{N}\right\}  \\
& =\sup_{\substack{\lambda\geq0,\\\tau\in\operatorname{aff}(\mathcal{D}
)}}\left\{  \lambda:\lambda\mathcal{R}^{\tau}\leq\mathcal{N}\right\} \label{eq:key-robustness-step} \\
& =\sup_{\substack{\mu\geq0,\\\tau\in\operatorname{aff}(\mathcal{D})}}\left\{
\frac{1}{\mu}:\frac{1}{\mu}\mathcal{R}^{\tau}\leq\mathcal{N}\right\}  \\
& =\left[  \inf_{\substack{\mu\geq0,\\\tau\in\operatorname{aff}(\mathcal{D}
)}}\left\{  \mu:\mathcal{R}^{\tau}\leq\mu\mathcal{N}\right\}  \right]
^{-1}\\
& =\left[  \inf_{\tau\in\operatorname{aff}(\mathcal{D})}\exp\!\left(  D_{\max
}(\mathcal{R}^{\tau}\Vert\mathcal{N})\right)  \right]  ^{-1}\\
& =\exp\!\left(  -\inf_{\tau\in\operatorname{aff}(\mathcal{D})}D_{\max
}(\mathcal{R}^{\tau}\Vert\mathcal{N})\right)  .
\end{align}
The second equality follows because $X=0$ is a feasible solution, so that it suffices to restrict the optimization to be over $X\in\operatorname{Herm}$ such that $\operatorname{Tr}[X] \geq 0$. Furthermore, all such $X$ can be written as $X = \lambda \tau $, where $\lambda\geq 0$ and $\tau\in\operatorname{aff}(\mathcal{D})$. The penultimate equality follows from the definition of $D_{\max}$ in~\eqref{eq:dmax-1-herm-maps}.
\end{proof}

\begin{remark}
\label{rem:comm-val-max-mutual}
By following a proof similar to the above, we note that the entanglement-assisted communication value $\operatorname{cv}^*(\mathcal{N})$ of a quantum channel $\mathcal{N}$ (see \cite[Theorem~6]{chitambar2022communication} for a mathematical expression) can be written in terms of the max-mutual information:
\begin{equation}
    \ln \operatorname{cv}^*(\mathcal{N}) = \inf_{\sigma\in\mathcal{D}} D_{\max
}(\mathcal{N}\Vert\mathcal{R}^{\sigma}).
\end{equation}
\end{remark}

\subsection{Reverse Robustness}

In this section, we interpret the quantum Doeblin coefficient $\alpha(\mathcal{N})$ as a reverse robustness, related to robustness quantities previously considered in quantum information (see, e.g., \cite{Regula2022prl,Regula2022tightconstraints}). Here, the goal is to express the original channel as a convex combination of a trace-preserving replacer map and an another arbitrary channel, with as much weight as possible on the replacer map. As before, we allow the replacer map to be non-physical and replace the input state with an arbitrary quasi-state. 
This characterization is helpful in establishing $1-\alpha(\mathcal{N})$ as an upper bound on the trace distance contraction coefficient of $\mathcal{N}$ (\cref{prop:trace_distance_complete_cont_Doeblin_bound}), as it neatly separates $\mathcal{N}$ into a component that has a trivial action on an arbitrary input and a component that does not.
This characterization is also helpful in establishing \cref{prop:concatenation} below, where we use the following slightly different representation of the quantum Doeblin coefficient:
\begin{equation}
    1-\alpha(\mathcal{N}) = \inf_{\substack{\lambda \in [0,1],\\ \tau\in\operatorname{aff}(\mathcal{D}),\\ \mathcal{M} \in \operatorname{CPTP}} } \left\{\lambda : (1-\lambda)\mathcal{R}^\tau + \lambda \mathcal{M} = \mathcal{N}\right\},
    \label{eq:reverse-robustness-opposite}
\end{equation}
which is equivalent to \eqref{eq:reverse-robustness} below.

\begin{proposition}
\label{prop:reverse-robustness}
The following equality holds for a quantum channel $\mathcal{N}$:
\begin{equation}
\alpha(\mathcal{N})=\sup_{\substack{\lambda\in\left[  0,1\right]  ,\\\tau
\in\operatorname{aff}(\mathcal{D}),\\\mathcal{M}\in\operatorname{CPTP}}
}\left\{  \lambda:\mathcal{N}=\lambda\mathcal{R}^{\tau}+\left(  1-\lambda
\right)  \mathcal{M}\right\}  .
\label{eq:reverse-robustness}
\end{equation}

\end{proposition}

\begin{proof}
Consider that
\begin{align}
\alpha(\mathcal{N})  & =\sup_{\substack{\lambda\geq0,\\\tau\in
\operatorname{aff}(\mathcal{D})}}\left\{  \lambda:\lambda\mathcal{R}^{\tau
}\leq\mathcal{N}\right\}  \\
& =\sup_{\substack{\lambda\geq0,\\\tau\in\operatorname{aff}(\mathcal{D}
),\\\mathcal{M}^{\prime}\in\operatorname{CP}}}\left\{  \lambda:\lambda
\mathcal{R}^{\tau}+\mathcal{M}^{\prime}=\mathcal{N}\right\}  \\
& =\sup_{\substack{\lambda\in\left[  0,1\right]  ,\\\tau\in
\operatorname{aff}(\mathcal{D}),\\\mathcal{M}\in\operatorname{CPTP}}}\left\{
\lambda:\mathcal{N}=\lambda\mathcal{R}^{\tau}+\left(  1-\lambda\right)
\mathcal{M}\right\}  .
\end{align}
The first equality follows from~\eqref{eq:key-robustness-step}. The second equality follows because the
inequality $\lambda\mathcal{R}^{\tau}\leq\mathcal{N}$ is equivalent to the
existence of a completely positive map $\mathcal{M}^{\prime}$ such that
$\mathcal{M}^{\prime}=\mathcal{N}-\lambda\mathcal{R}^{\tau}$. The final
equality follows because the trace-preserving constraint on $\mathcal{N}$
implies that the following equality holds for every state $\rho$:
\begin{align}
1  & =\operatorname{Tr}[\rho]=\operatorname{Tr}[\mathcal{N}(\rho)]\\
& =\operatorname{Tr}[\left(  \lambda\mathcal{R}^{\tau}+\mathcal{M}^{\prime
}\right)  (\rho)]\\
& =\lambda\operatorname{Tr}[\mathcal{R}^{\tau}(\rho)]+\operatorname{Tr}
[\mathcal{M}^{\prime}(\rho)]\\
& =\lambda\operatorname{Tr}[\tau]+\operatorname{Tr}[\mathcal{M}^{\prime}
(\rho)]\\
& =\lambda+\operatorname{Tr}[\mathcal{M}^{\prime}(\rho)]. \label{eq:reverse-robustness-normalization-argument}
\end{align}
Since $\mathcal{M}^{\prime}$ is completely positive, it follows that
$\operatorname{Tr}[\mathcal{M}^{\prime}(\rho)]\geq0$, so that $\lambda
\in\left[  0,1\right]  $. Furthermore, it then follows that $\operatorname{Tr}
[\mathcal{M}^{\prime}(\rho)]=1-\lambda\in\left[  0,1\right]  $ for every input
state $\rho$. Thus, in the case that $\lambda\in\lbrack0,1)$, it follows that
$\mathcal{M}\coloneqq \frac{1}{1-\lambda}\mathcal{M}^{\prime}$ is a completely
positive trace-preserving map, and so the last equality holds, concluding the
proof for $\lambda\in\lbrack0,1)$. In the case that $\lambda=1$, it follows
that $\operatorname{Tr}[\mathcal{M}^{\prime}(\rho)]=0$ for every input state
$\rho$, and thus $\mathcal{M}^{\prime}$ is the zero map in this case (i.e.,
such that $\mathcal{M}^{\prime}(\rho)=0$ for every input state). Thus,
$\mathcal{N}=\mathcal{R}^{\tau}$ in this case (i.e., $\mathcal{N}$ is a
replacer channel), concluding the proof.
\end{proof}

\subsection{Reverse Hypothesis Testing Mutual Information}

In this section, we establish a relationship between the Doeblin coefficient and the extended hypothesis testing relative entropy defined in~\eqref{eq:alpha_cN_def} and~\eqref{eq:extended_ht_Relative_Entropy}, respectively. Indeed, expressing the Doeblin coefficient in this way is helpful for relating it to other distinguishability measures, such as the extended sandwiched R\'enyi relative entropy, as done in \cref{app:relate-to-extended-sandwiched}.

\begin{proposition} \label{thm:alpha_N_equality_Hypo_Test}
    Let $\cN$ be a quantum channel. We have that
    \begin{equation}
-\ln  \alpha(\mathcal{N})   = -2 \ln d + \inf_{\tau
\in\operatorname{aff}(\mathcal{D})}D_{H}^{1-\frac{1}{d^{2}}}\!\left(\pi_{d}
\otimes\tau\Vert\Phi^{\mathcal{N}}\right),
\label{eq:reverse-hypothesis-testing-MI}
\end{equation}
where $\operatorname{aff}(\mathcal{D})$ is the affine hull of the set of density operators
(i.e., Hermitian operators with unit trace), $D_{H}^{1-\frac{1}{d^{2}}}$ is given in~\eqref{eq:extended_ht_Relative_Entropy}, $\pi_d \coloneqq I/d$, and $\Phi^\cN \coloneqq \frac{1}{d} \Gamma_{AB}^\cN$ with $\Gamma_{AB}^\cN$ in~\eqref{eq:choi_operator}.
\end{proposition}

\begin{proof}
    See~\cref{Proof:thm:alpha_N_equality_Hypo_Test}.
\end{proof}

\subsection{Interpretations of the Induced Quantum Doeblin Coefficient}

\label{sec:interpretations-alt-doeblin}

In this subsection, we show how to lift the interpretations for $\alpha(\cN)$ to $\alpha_{I}(\cN)$.

\subsubsection{Minimum Achievable Singlet Fraction}

 We begin by showing that $\alpha_{I,\cK}(\cN)$, as defined in~\eqref{eqn:inducedDoeblinDual_cone}, inherits a singlet fraction interpretation almost immediately in the same manner as~\cite[Proposition 12]{george2024cone}. This recovers~\cref{thm:doeblin_antiDistinguish} as a specific case. It also establishes $\alpha_{I}(\cN)$ as having the same singlet-fraction interpretation but with the minimization over \textit{entanglement-breaking channels} \cite{HSR03}.
 
 \begin{proposition}\label{prop:singlet-fraction-of-cone-induced-do}
     Let $\cK(A:B) \subset \Pos(A \otimes B)$ be such that $I_{AB}$ is in the relative interior of $\cK(A:B)$. Then,
     \begin{align}
         \alpha_{I,\cK}(\cN) \coloneqq  d^{2} \min_{\cE \in \operatorname{CPTP}(B,A): \Gamma^{\cE^{\dagger}} \in \cK} \tr[\Phi_d (\operatorname{id}_{A} \otimes (\cE \circ \cN))(\Phi_d)] \ ,
     \end{align}
     where $d$ denotes the dimension of the input system of $\mathcal{N}$.
 In particular, 
 \begin{align}
 \alpha_{I}(\cN) = d^{2} \min_{\cE \in \operatorname{EB}} \Tr[\Phi_{d}(\operatorname{id}_{A} \otimes (\cE \circ \cN))(\Phi_{d})] \ ,
 \label{eq:min-singlet-frac-interp-induced-doeblin}
 \end{align} 
 where $\operatorname{EB}$ denotes the set of entanglement-breaking channels.
  \end{proposition}
  
 \begin{proof}
     The condition that $I_{AB}$ is in the relative interior of $\cK(A:B)$ guarantees that strong duality holds, so that we may work with the dual problem. Note that $\tr_{A}[Y] = {I}_{B}$ implies that $Y = \Gamma^\cM$ for some completely positive, unital (CPU) map $\cM: A \to B$. Thus, similar to \eqref{eq:min-singlet-frac-proof-1}--\eqref{eq:min-singlet-frac-proof-last}, we have that
     \begin{align}
        \alpha_{I,\cK}(\cN) & = \min_{\cM \in \operatorname{CPU}: \Gamma^\cM \in \cK} \Tr[\Gamma^\cN \Gamma^\cM]  \\
        & = d \min_{\cM \in \operatorname{CPU}: \Gamma^\cM \in \cK} \tr[\Gamma^\cN(\operatorname{id}_{A} \otimes \cM)(\Phi_d)] \\
        & = d \min_{\cM \in \operatorname{CPU}: \Gamma^\cM  \in \cK} \tr[\cM^\dag (\Gamma^\cN) \Phi_d] \\ 
        & = d^2 \min_{\cM \in \operatorname{CPU}: \Gamma^\cM  \in \cK} \tr\!\left[\left(\cM^\dag_{B\to A} \circ \cN_{A \to B}(\Phi_d)\right )(\Phi_d)\right] \ . 
     \end{align}
    To conclude the proof, note that $\cM^\dag$ is a CPTP map as it is the adjoint of a CPU map and we replace optimizing over $\cM^{\dagger}$ with $\cE$ for notational clarity.

    To obtain the statement in \eqref{eq:min-singlet-frac-interp-induced-doeblin}, recall that $\Gamma^{\cE^{\dagger}}_{AB} = T_{AB}(\Gamma_{BA}^{\cE})$ where $T_{AB}$ the transpose map acting on the joint space, and $\Gamma_{BA}^{\cE}$ is $\Gamma^{\cE}_{AB}$, with the spaces swapped~\cite{george2024cone}. It follows that if $\Gamma^{\cE^{\dagger}} \in \Sep(A:B)$, then $\Gamma^{\cE} \in \Sep(A:B)$. As a CPTP map $\cE$ is entanglement breaking if and only if $\Gamma^{\cE} \in \Sep(A:B)$, this completes the proof.
 \end{proof}

\subsubsection{Max-oveloH Information}

In this section, we represent the induced Doeblin coefficient $\alpha_I(\mathcal{N})$ in terms of the max-oveloH information of the channel $\mathcal{N}$. This is a counterpart of the fact that the communication value of a channel can be written in terms of the max-Holevo information \cite[Theorem~1]{chitambar2022communication}.

\begin{proposition}
The following equality holds for a quantum channel $\mathcal{N}$:
\begin{equation}
\alpha_{I}(\mathcal{N})=\exp\!\left(  -\inf_{\tau\in\operatorname{aff}
(\mathcal{D})}\sup_{\rho \in \mathcal{D}} D_{\max}(\tau\Vert\mathcal{N}(\rho))\right)  .
\label{eq:max-oveloh-interp}
\end{equation}

\end{proposition}

\begin{proof}
Consider that
\begin{align}
\alpha_{I}(\mathcal{N}) &  =\sup_{X\in\operatorname{Herm}}\left\{
\operatorname{Tr}[X]:\mathcal{R}^{X}\leq_{P}\mathcal{N}\right\}  \\
&  =\sup_{\substack{\lambda\geq0,\\\tau\in\operatorname{aff}(\mathcal{D}
)}}\left\{  \operatorname{Tr}[\lambda\tau]:\mathcal{R}^{\lambda\tau}\leq
_{P}\mathcal{N}\right\}  \\
&  =\sup_{\substack{\lambda\geq0,\\\tau\in\operatorname{aff}(\mathcal{D}
)}}\left\{  \lambda:\lambda\mathcal{R}^{\tau}\leq_{P}\mathcal{N}\right\}  \\
&  =\sup_{\substack{\mu\geq0,\\\tau\in\operatorname{aff}(\mathcal{D}
)}}\left\{  \frac{1}{\mu}:\frac{1}{\mu}\mathcal{R}^{\tau}\leq_{P}
\mathcal{N}\right\}  \\
&  =\left[  \inf_{\substack{\mu\geq0,\\\tau\in\operatorname{aff}(\mathcal{D}
)}}\left\{  \mu:\mathcal{R}^{\tau}\leq_{P}\mu\mathcal{N}\right\}  \right]
^{-1}\\
&  =\left[  \inf_{\substack{\mu\geq0,\\\tau\in\operatorname{aff}(\mathcal{D}
)}}\left\{  \mu:\mathcal{R}^{\tau}(\rho)\leq\mu\mathcal{N}(\rho)\ \forall
\rho\in\mathcal{D}\right\}  \right]  ^{-1}\\
&  =\left[  \inf_{\tau\in\operatorname{aff}(\mathcal{D})}\inf_{\mu\geq
0}\left\{  \mu:\tau\leq\mu\mathcal{N}(\rho)\ \forall\rho\in\mathcal{D}
\right\}  \right]  ^{-1}\\
&  =\left[  \inf_{\tau\in\operatorname{aff}(\mathcal{D})}\inf_{\mu\geq0}
\sup_{\rho\in\mathcal{D}}\left\{  \mu:\tau\leq\mu\mathcal{N}(\rho)\right\}
\right]  ^{-1} \label{eq:for-all-rho-to-sup-rho} \\
&  =\left[  \inf_{\tau\in\operatorname{aff}(\mathcal{D})}\sup_{\rho
\in\mathcal{D}}\inf_{\mu\geq0}\left\{  \mu:\tau\leq\mu\mathcal{N}
(\rho)\right\}  \right]  ^{-1} \label{eq:minimax-switch}\\
&  =\left[  \inf_{\tau\in\operatorname{aff}(\mathcal{D})}\sup_{\rho
\in\mathcal{D}}\exp\!\left(  D_{\max}(\tau\Vert\mathcal{N}(\rho))\right)
\right]  ^{-1}\\
&  =\exp\!\left(  -\inf_{\tau\in\operatorname{aff}(\mathcal{D})}\sup_{\rho
\in\mathcal{D}}D_{\max}(\tau\Vert\mathcal{N}(\rho))\right)  .
\end{align}
The equality in \eqref{eq:for-all-rho-to-sup-rho} holds because the constraint should hold for all $\rho \in \mathcal{D}$. The equality in \eqref{eq:minimax-switch} holds by applying the extended Sion minimax theorem from \cite[Theorem~2.11]{BB2023}; indeed, the objective function is linear in $\mu$ and $\rho$, the set of $\mu$ satisfying $\mu \geq 0, \tau \leq \mu \mathcal{N}(\rho)$ is convex, and the set of $\rho $ satisfying $\rho \in \mathcal{D},  \tau \leq \mu \mathcal{N}(\rho)$ is convex and compact. The penultimate equality follows from applying the definition of $D_{\max} $ in \eqref{eq:dmax-1-herm-maps}.
\end{proof}

\subsubsection{Reverse Robustness}

Here we write the induced Doeblin coefficient $\alpha_I(\mathcal{N})$ as a reverse robustness, analogous to what we found in \cref{prop:reverse-robustness}. The main difference between \cref{prop:induced-doeblin-reverse-rob} and \cref{prop:reverse-robustness} is that \cref{prop:induced-doeblin-reverse-rob} involves an optimization over the set of positive, trace-preserving maps. 

\begin{proposition}
\label{prop:induced-doeblin-reverse-rob}
The following equality holds for a quantum channel $\mathcal{N}$:
\begin{equation}
\alpha_I(\mathcal{N})=\sup_{\substack{\lambda\in\left[  0,1\right]  ,\\\tau
\in\operatorname{aff}(\mathcal{D}),\\\mathcal{M}\in\operatorname*{PTP}
}}\left\{  \lambda:\mathcal{N}=\lambda\mathcal{R}^{\tau}+\left(
1-\lambda\right)  \mathcal{M}\right\}  ,
\label{eq:reverse-robustness-induced-doeblin}
\end{equation}
where $\operatorname{PTP}$ denotes the set of positive, trace-preserving maps. 
\end{proposition}

\begin{proof}
The proof is identical to the proof of \cref{prop:reverse-robustness}, with the main exception being
that we use the ordering induced by positive maps rather than completely
positive maps.
\end{proof}

\subsubsection{Hypothesis-testing oveloH information}

Here we write the induced Doeblin coefficient $\alpha_I(\mathcal{N})$ in terms of the hypothesis testing oveloH information of the channel $\mathcal{N}$, analogously to what we did in \cref{thm:alpha_N_equality_Hypo_Test} for the Doeblin coefficient $\alpha(\mathcal{N})$.

\begin{proposition}
The following equality holds for a quantum channel $\mathcal{N}$:
\begin{equation}
-\ln\alpha_{I}(\mathcal{N})=-2\ln d+\sup_{\left(  \psi^{x}\right)
_{x=1}^{d^{2}}}\inf_{\tau\in\operatorname{aff}(\mathcal{D})}D_{H}^{1-\frac
{1}{d^{2}}}(\rho_{X}\otimes\tau_{B}\Vert\rho_{XB}),
\label{eq:hypo-test-oveloh}
\end{equation}
where $d$ is the dimension of the channel input system and each $\psi_{A}^{x}$ is a pure state for all $x\in\left[  d^{2}\right]  $
and
\begin{equation}
\rho_{XB}\coloneqq \frac{1}{d^{2}}\sum_{x=1}^{d^{2}}|x\rangle\!\langle x|\otimes
\mathcal{N}_{A\rightarrow B}\left(  \psi_{A}^{x}\right)  .
\end{equation}

\end{proposition}

\begin{proof}
Consider that
\begin{align}
-\ln\alpha_{I}(\mathcal{N})  & =-\ln\inf_{\left(  \psi^{x},\Lambda^{x}\right)
_{x=1}^{d^{2}}}\sum_{x=1}^{d^{2}}\operatorname{Tr}[\Lambda_{B}^{x}
\mathcal{N}_{A\rightarrow B}(\psi_{A}^{x})]\\
& =-\ln\!\left(  d^{2}\inf_{\left(  \psi^{x}\right)  _{x=1}^{d^{2}}}
\inf_{\left(  \Lambda^{x}\right)  _{x=1}^{d^{2}}}\sum_{x=1}^{d^{2}}\frac
{1}{d^{2}}\operatorname{Tr}[\Lambda_{B}^{x}\mathcal{N}_{A\rightarrow B}
(\psi_{A}^{x})]\right)  \\
& =-2\ln d+\sup_{\left(  \psi^{x}\right)  _{x=1}^{d^{2}}}\left(  -\ln
\inf_{\left(  \Lambda^{x}\right)  _{x=1}^{d^{2}}}\sum_{x=1}^{d^{2}}\frac
{1}{d^{2}}\operatorname{Tr}[\Lambda_{B}^{x}\mathcal{N}_{A\rightarrow B}
(\psi_{A}^{x})]\right)  \\
& =-2\ln d+\sup_{\left(  \psi^{x}\right)  _{x=1}^{d^{2}}}\inf_{\tau
\in\operatorname{aff}(\mathcal{D})}D_{H}^{1-\frac{1}{d^{2}}}(\rho_{X}
\otimes\tau_{B}\Vert\rho_{XB})
\end{align}
The first equality follows from \cref{prop:induced-Doeblin-structure}, and the last equality follows from~\cite[Proposition~7]{ji2024barycentric}.
\end{proof}

\section{Properties of Quantum Doeblin Coefficients}

\subsection{Normalization}

\begin{proposition}[Normalization]
\label{prop:q-doeblin-normalization}
    For a quantum channel $\mathcal{N}$, the following inequalities hold:
    \begin{equation}
        0 \leq \alpha(\mathcal{N}) \leq 1.
    \end{equation}
    Moreover, $\alpha(\mathcal{N}) = 1$ if and only if $\cN$ is a replacer channel, and $\alpha(\mathcal{N}) =0$ if and only if $\cN$ is orthogonal to a completely positive unital map under the Hilbert--Schmidt inner product through the Choi isomorphism. 
\end{proposition}

\begin{proof}
    To establish the inequality $0 \leq \alpha(\mathcal{N})$, note that $X_B=0 $ is feasible in the primal SDP in~\eqref{eq:alpha_cN_def}, and the objective function evaluates to zero for this choice. To establish the inequality $\alpha(\mathcal{N})\leq 1$, we note that $Y_{AB} = \frac{1}{d_A}I_A \otimes I_B$ is feasible for the dual SDP in~\eqref{eq:alpha-cN-dual}, and the objective function evaluates to one for this choice.

    We now establish the saturation conditions. We begin by establishing that $\alpha(\cN) = 1$ if and only if $\cN$ is a replacer channel. Let $\cN$ be a replacer channel. Then $\Gamma^{\cN} = I_{A} \otimes \sigma_{B}$ for some $\sigma_{B} \in \Density(B)$. Then $X_{B} = \sigma_{B}$ is feasible for~\eqref{eq:alpha_cN_def}, and thus $\alpha(\cN) = 1$. For the other direction, let $\alpha(\cN) = 1$. Recall that $1 - \alpha(\cN) \geq \eta_{\Tr}(\cN)$~\cite{wolf2012quantum,hirche2024quantum} (see also  \cref{prop:trace_distance_complete_cont_Doeblin_bound} for a strengthening), and so $\eta_{\Tr}(\cN) = 0$. It follows that $\sup_{\rho,\sigma \in \Density} \left\Vert \cN(\rho) - \cN(\sigma) \right\Vert_{1} = 0$. Thus, $\cN(\rho) = \cN(\sigma)$ for all $\rho,\sigma\in \mathcal{D}$. It follows that $\cN$ is a replacer channel.

    The case that $\alpha(\cN) = 0$ is in effect an immediate consequence of the dual problem~\eqref{eq:alpha-cN-dual}. By the Choi isomorphism, the set $\{Y_{AB} \geq 0: \Tr_{A}[Y_{AB}] = I_{B}\}$ consists of Choi operators of the space of completely positive, unital maps from $A$ to $B$. Note that, for fixed spaces $A$ and $B$, this is a compact set, and so we may minimize over it. Thus, by~\eqref{eq:alpha-cN-dual}, $\alpha(\cN) = 0$ if and only if $\min_{\cE \in \text{CPU}} \Tr[\Gamma^{\cN}\Gamma^{\cE}] = 0$, which means it is orthogonal to some CPU map under the Hilbert--Schmidt inner product.
\end{proof}

We note that while the saturation for the lower bound is more mathematical than physical, it is relevant. Namely, an example of such a channel that is noisy but satisfies this condition will be considered in \cref{example:strict-submultiplicativity-for-qubit-channels}. Another way of stating the saturation condition is based on the minimum achievable singlet fraction interpretation from \cref{thm:doeblin_antiDistinguish}: the Doeblin coefficient $\alpha(\mathcal{N}) = 0 $ if and only if there exists a decoding channel $\mathcal{D}$ acting on the channel output of the Choi state $\Phi^{\mathcal{N}}$ to make it orthogonal to the maximally entangled state $\Phi$.

\subsection{Concatenation}

In this section, we consider how the quantum Doeblin coefficients $\alpha(\mathcal{N})$ and $\alpha_I(\mathcal{N})$ behave under sequential composition of channels. In particular, we find that the quantities $1-\alpha\!\left(\cN_2 \circ \cN_1\right)$ and $1-\alpha_I\!\left(\cN_2 \circ \cN_1\right)$ are submultiplicative under such a sequential composition, and we note that this property is strongly related to the behavior of the contraction coefficient of a quantum channel, studied later on in \cref{sec:contraction-coefficient}. 

\begin{proposition}[Concatenation] \label{prop:concatenation}
For channels $\mathcal{N}_{1}$ and $\mathcal{N}_{2}$, the following inequalities
hold:
\begin{align}
1-\alpha(\mathcal{N}_{2}\circ\mathcal{N}_{1}) & \leq\left(  1-\alpha
(\mathcal{N}_{1})\right)  \left(  1-\alpha(\mathcal{N}_{2})\right), \\
1-\alpha_I(\mathcal{N}_{2}\circ\mathcal{N}_{1}) & \leq\left(  1-\alpha_I
(\mathcal{N}_{1})\right)  \left(  1-\alpha_I(\mathcal{N}_{2})\right).
\label{eq:concatenation-induced}
\end{align}
\end{proposition}

\begin{proof}
The proof is analogous to that given for \cite[Corollary~2.2]{Chestnut2010} in the classical case, and it makes use of the representation in \eqref{eq:reverse-robustness-opposite}.
For $i\in\left\{  1,2\right\}  $, let $\lambda_{i}\in\left[  0,1\right]  $,
$\tau_{i}\in\operatorname{aff}(\mathcal{D})$, and $\mathcal{M}_{i}\in \operatorname{CPTP}$ be such that
$\left(  1-\lambda_{i}\right)  \mathcal{R}_{\tau_{i}}+\lambda_{i}
\mathcal{M}_{i}=\mathcal{N}_{i}$. Then it follows that
\begin{align}
  \mathcal{N}_{2}\circ\mathcal{N}_{1}
&  =\left[  \left(  1-\lambda_{2}\right)  \mathcal{R}_{\tau_{2}}+\lambda
_{2}\mathcal{M}_{2}\right]  \circ\left[  \left(  1-\lambda_{1}\right)
\mathcal{R}_{\tau_{1}}+\lambda_{1}\mathcal{M}_{1}\right] \\
&  =\left(  1-\lambda_{2}\right)  \left(  1-\lambda_{1}\right)  \mathcal{R}
_{\tau_{2}}\circ\mathcal{R}_{\tau_{1}}+\lambda_{2}\left(  1-\lambda
_{1}\right)  \mathcal{M}_{2}\circ\mathcal{R}_{\tau_{1}}\nonumber\\
&  \qquad+\left(  1-\lambda_{2}\right)  \lambda_{1}\mathcal{R}_{\tau_{2}}
\circ\mathcal{M}_{1}+\lambda_{2}\lambda_{1}\mathcal{M}_{2}\circ\mathcal{M}
_{1}\\
&  =\left(  1-\lambda_{2}\right)  \left(  1-\lambda_{1}\right)  \mathcal{R}
_{\tau_{2}}+\lambda_{2}\left(  1-\lambda_{1}\right)  \mathcal{R}
_{\mathcal{M}_{2}(\tau_{1})}\nonumber\\
&  \qquad+\left(  1-\lambda_{2}\right)  \lambda_{1}\mathcal{R}_{\tau_{2}
}+\lambda_{2}\lambda_{1}\mathcal{M}_{2}\circ\mathcal{M}_{1}\\
&  =\left(  1-\lambda_{2}\right)  \mathcal{R}_{\tau_{2}}+\lambda_{2}\left(
1-\lambda_{1}\right)  \mathcal{R}_{\mathcal{M}_{2}(\tau_{1})}+\lambda
_{2}\lambda_{1}\mathcal{M}_{2}\circ\mathcal{M}_{1}\\
&  =\left(  1-\lambda_{1}\lambda_{2}\right)  \mathcal{R}_{\tau^{\prime}
}+\lambda_{2}\lambda_{1}\mathcal{M}_{2}\circ\mathcal{M}_{1},
\end{align}
where
\begin{equation}
\tau^{\prime}\coloneqq \frac{1-\lambda_{2}}{1-\lambda_{1}\lambda_{2}}\tau
_{2}+\frac{\lambda_{2}\left(  1-\lambda_{1}\right)  }{1-\lambda_{1}\lambda
_{2}}\mathcal{M}_{2}(\tau_{1}).
\end{equation}
Thus, the choices $\lambda_{2}\lambda_{1}$, $\tau^{\prime}$, and
$\mathcal{M}_{2}\circ\mathcal{M}_{1}$ are feasible for $1-\alpha
(\mathcal{N}_{2}\circ\mathcal{N}_{1})$, and it follows that
\begin{equation}
1-\alpha(\mathcal{N}_{2}\circ\mathcal{N}_{1})\leq\lambda_{2}\lambda_{1}.
\end{equation}
Since $\lambda_{i}\in\left[  0,1\right]  $, $\tau_{i}\in \operatorname{aff}(\mathcal{D})$,
and $\mathcal{M}_{i}\in \operatorname{CPTP}$ are arbitrary feasible choices for $1-\alpha
(\mathcal{N}_{i})$, the desired inequality follows.

The proof for \eqref{eq:concatenation-induced} proceeds in exactly the same way, but instead using the representation in \cref{prop:induced-doeblin-reverse-rob}.
\end{proof}

\begin{corollary}
    For a channel $\mathcal{N}$ and $n\in \mathbb{N}$, the following inequalities hold:
    \begin{align}
        1-\alpha(\mathcal{N}^{\circ n}) & \leq (1-\alpha(\mathcal{N}))^{n}, \\
        1-\alpha_I(\mathcal{N}^{\circ n}) &  \leq (1-\alpha_I(\mathcal{N}))^{n}.
    \end{align}
\end{corollary}

\begin{proof}
    Direct consequence of \cref{prop:concatenation}.
\end{proof}

\subsection{Data Processing under Pre- and Post-Processing}

Next, we show that both Doeblin coefficients, $\alpha(\mathcal{N})$ and $\alpha_I(\mathcal{N})$, of a channel $\mathcal{N}$ can only stay the same or increase under pre- and post-processing, and they are  invariant if both these pre- and post-processing channels are isometric-reversal and isometric channels, respectively. 

Given an isometric channel $\mathcal{U}$, define the corresponding isometric-reversal channel as \cite[Eq.~(4.264)]{Wilde-Book} 
\begin{equation}
\mathcal{P}^{\mathcal{U}}(\cdot) \coloneqq \mathcal{U}^{\dag}(\cdot) + \operatorname{Tr}[(\operatorname{id}-\mathcal{U}^\dag)(\cdot)] \omega,
\label{eq:isometric-reversal-channel}
\end{equation}
where $\omega$ is an arbitrary state. It then follows that $\mathcal{P}^{\mathcal{U}}$ indeed reverses the action of $\mathcal{U}$; i.e.,
\begin{equation}    \mathcal{P}^{\mathcal{U}}\circ\mathcal{U} = \operatorname{id}.
    \label{eq:isometry-and-reversal}
\end{equation}

\begin{proposition}[Data Processing under Pre- and Post-Processing and Isometric Invariance]
\label{prop:DPI-and-LU-invariance}
Let $\cE_{B \to B'}$, $\cG_{A' \to A}$, and $\cN_{A \to B}$ be quantum channels.
   Then
   \begin{align}
       \alpha(\cE \circ \cN \circ \cG) & \geq \alpha(\cN) ,\\
       \alpha_I(\cE \circ \cN \circ \cG) & \geq \alpha_I(\cN) . \label{eq:data-proc-induced-doeblin}
   \end{align}
  That is, the quantum Doeblin coefficient does not decrease under pre- and post-processing. 
   Moreover, the quantum Doeblin coefficient is invariant under isometric-reversal and isometric pre- and post-processing, respectively; i.e.,
\begin{align}
    \alpha(\mathcal{N}) & = \alpha(\mathcal{U} \circ \mathcal{N} \circ \mathcal{P}^{\mathcal{V}}),
    \label{eq:isometric-inv-doeblin}\\
    \alpha_I(\mathcal{N}) & = \alpha_I(\mathcal{U} \circ \mathcal{N} \circ \mathcal{P}^{\mathcal{V}}),
    \label{eq:isometric-inv-induced-doeblin}
\end{align}
where $\mathcal{U}$ and $\mathcal{V}$ are isometric channels.
\end{proposition}

\begin{proof}
    To establish the inequality in \eqref{eq:data-proc-induced-doeblin}, consider that
\begin{align}
    \alpha_I(\mathcal{E}\circ \mathcal{N} \circ \mathcal{F}) &  = \inf_{(\rho_x, \Lambda_x)_x} \sum_x \operatorname{Tr}[\Lambda_x (\mathcal{E}\circ \mathcal{N} \circ \mathcal{F})(\rho_x)]\\
    &  = \inf_{(\rho_x, \Lambda_x)_x} \sum_x \operatorname{Tr}[\mathcal{E}^\dag(\Lambda_x) \mathcal{N} ( \mathcal{F}(\rho_x))]
    \\
    & \geq \inf_{(\Lambda_y, \rho_y)_y} \sum_y \operatorname{Tr}[\Lambda_y \mathcal{N}(\rho_y)]\\
    & = \alpha_I(\mathcal{N}),
\end{align}
where we made use of \eqref{eq:induced-doeblin}.
The equality in \eqref{eq:isometric-inv-induced-doeblin} follows because
\begin{align}
    \alpha_I(\mathcal{U} \circ \mathcal{N} \circ \mathcal{P}^{\mathcal{V}}) & \geq \alpha_I(\mathcal{N}) \\
    \alpha_I(\mathcal{N})& = \alpha_I(\mathcal{P}^{\mathcal{U}} \circ \mathcal{U} \circ \mathcal{N} \circ \mathcal{P}^{\mathcal{V}}\circ \mathcal{V}) \geq  \alpha_I( \mathcal{U} \circ \mathcal{N} \circ \mathcal{P}^{\mathcal{V}}),
\end{align}
as a consequence of~\eqref{eq:data-proc-induced-doeblin} and~\eqref{eq:isometry-and-reversal}.
By appealing to \cref{cor:q-doeblin-as-EA-assisted-doeblin}, we can employ the same proofs to establish the claims for $\alpha(\mathcal{N})$ .
\end{proof}

\subsection{Submultiplicativity and Multiplicativity}

In \cref{prop:concatenation}, we considered channels being applied sequentially. We now consider when they are applied in parallel. This means that we are interested in the multiplicativity of the quantum Doeblin coefficient $\alpha(\mathcal{N})$. First, we show that it is submultiplicative.

\begin{proposition}[Submultiplicativity]
\label{prop:q-doeblin-sub-multiplicative}
Let $\cN$ and $\cM$ be quantum channels. The following submultiplicativity inequalities hold: 
\begin{align}
    \alpha(\cN \otimes \cM) &\leq \alpha(\cN)\alpha(\cM) ,\\
    \alpha_{I}(\cN \otimes \cM) & \leq \alpha_{I}(\cN)\alpha_{I}(\cM) .
\end{align}
\end{proposition}

\begin{proof}
    Let $Y^{1},Y^{2}$ be optimizers for $\cN,\cM$, respectively, to the dual program given in~\eqref{eq:alpha-cN-dual}. Then $Y^{1} \otimes Y^{2}$ is feasible for $\cN \otimes \cM$, satisfying $\operatorname{Tr}_{A_1A_2}[Y^1 \otimes Y^2]=I_{B_1} \otimes I_{B_2} = I_{B_1B_2}$,  with value
    \begin{equation}
        \Tr[\Gamma^{\cN \otimes \cM} (Y^{1} \otimes Y^{2})] = \Tr[(\Gamma^{\cN} \otimes \Gamma^{\cM})(Y^{1} \otimes Y^{2})] = \alpha(\cN)\alpha(\cM).
    \end{equation}
     As~\eqref{eq:alpha-cN-dual} involves a minimization, we conclude that this is an upper bound. The same proof works for $\alpha_{I}$ using \eqref{eqn:inducedDoeblinDual} and that if $Y^{1} \in \Sep(A_{1}:B_{1})$ and $Y^{2} \in \Sep(A_{2}:B_{2})$, then $Y^{1} \otimes Y^{2} \in \Sep(A_{1}A_{2}:B_{1}B_{2})$.
\end{proof}

\begin{corollary}
    Let $\cN$ be a quantum channel satisfying $\alpha(\cN)=0$. Then, for every channel~$\cM$, we have
    \begin{equation}
        \alpha(\cN \otimes \cM) =0.
    \end{equation}
    Consequently for $n \in \mathbb{N}$, 
    \begin{equation}
         \alpha(\cN^{\otimes n}) =0.
    \end{equation}   
\end{corollary}

\begin{proof}
    The proof follows by applying~\cref{prop:q-doeblin-sub-multiplicative} together with the assumption $\alpha(\cN)=0$ and the non-negativity of the quantum Doeblin coefficient.
\end{proof}

Before continuing forward, we show that even qubit channels with classical inputs can be strictly submultiplicative. This differs from the communication value, which is multiplicative for qubit channels \cite{Chitambar-2021a}. This also differs from the minimal entropy and Holevo capacity (resp.~maximal $p$-norms) which are additive (resp.~multiplicative) for all entanglement-breaking channels \cite{shor2002additivity,king2002maximal}. Interestingly, the example uses the same sort of construction as given for the PBR theorem \cite{pusey2012reality} and is directly related to the operational task of state exclusion considered there, given the connection of $\alpha(\mathcal{N}) $ to state exclusion for classical-quantum channels (see \cref{rem:doeblin-exclusion-cq}).

\begin{example}
\label{example:strict-submultiplicativity-for-qubit-channels}
    The qubit-to-qubit classical-quantum channel 
    \begin{align}\label{eq:PBR-channel}
        \cN(X) = \bra{0} X \ket{0} \dyad{0} + \bra{1}X \ket{1} \dyad{+}
    \end{align}
    is such that
    \begin{align}\label{eq:strict-submult-for-qubit-channel}
        \alpha_{I}(\cN \otimes \cN) = \alpha(\cN \otimes \cN) < \alpha(\cN)^{2} = \alpha_{I}(\cN)^{2} \ .
    \end{align}
\end{example}

\begin{proof}
    By \cref{cor:cq-chan-induced-doeb-is-doeb}, the equality $\alpha_{I}(\cN) = \alpha(\cN)$ holds  for a classical-quantum channel. Note that applying classical-quantum channels in parallel realizes another classical-quantum channel. This justifies the equalities in~\eqref{eq:strict-submult-for-qubit-channel}. One can verify that the Hermitian operator $X \coloneq \begin{bmatrix} 0.3964 & 0.25 \\ 0.25 & -0.1036 \end{bmatrix}$ satisfies $X \leq \dyad{0}$, $X \leq \dyad{+}$. Thus, $\alpha(\cN) \geq 0.2929$. Using  the dual problem \eqref{eq:alpha-cN-dual} directly, one can verify that
    \begin{align}
        Y_{AB} &= \dyad{00} \otimes \frac{1}{2} \begin{bmatrix} 
        0 & 0 & 0 & 0 \\ 0 & 1 & 1 & 0 \\ 0 & 1 & 1 & 0 \\ 0 & 0 & 0 & 0 
        \end{bmatrix}
        + \dyad{01} \otimes \frac{1}{4} \begin{bmatrix} 
        1 & -1 & 1 & 1 \\ -1 & 1 & -1 & -1 \\ 1 & -1 & 1 & 1 \\ 1 & -1 & 1 & 1 
        \end{bmatrix} \notag \\
        & \hspace{1cm} + \dyad{10} \otimes \frac{1}{4} \begin{bmatrix} 
        1 & 1 & -1 & 1 \\ 1 & 1 & -1 & 1 \\ -1 & -1 & 1 & -1 \\ 1 & 1 & -1 & 1
        \end{bmatrix}
        + \dyad{11} \otimes \frac{1}{2} \begin{bmatrix} 
        1 & 0 & 0 & -1 \\
        0 & 0 & 0 & 0 \\
        0 & 0 & 0 & 0 \\
        -1 & 0 & 0 & 1 
        \end{bmatrix} \ 
    \end{align}
    is such that $\Tr_{A}[Y] = I_{B}$ and $\Tr[\Gamma^{\cN \otimes \cN}Y] = 0$, which implies $\alpha(\cN \otimes \cN) =0$.\footnote{Alternatively, one may verify this using a standard SDP solver.}
\end{proof}

\begin{remark}
    We note the above shows that $\alpha(\cN)$ can be a (trivially) loose bound on $\eta_{\Tr}(\cN)$ as, for $\cN$ as defined in \eqref{eq:PBR-channel}, $\eta_{\Tr}(\cN \otimes \cN) < 1$ as all outputs are overlapping, but $\alpha(\cN \otimes \cN) = 0$.
\end{remark}

\begin{remark}[Multiplicativity of Doeblin] 
 In~\cref{prop:q-doeblin-sub-multiplicative}, we showed that 
 $
     \alpha(\cN \otimes \cM) \leq \alpha(\cN) \alpha(\cM)$.
 However, it is not possible to obtain the reverse inequality in general, i.e., $ \alpha(\cN \otimes \cM) \not\geq \alpha(\cN)\alpha(\cM)$. We obtain counterexamples that violate the reverse inequality by choosing generalized amplitude damping channels (defined using the Kraus operators in~\eqref{eq:kraus_GADC} and~\eqref{eq:kraus_GADC_rest}) as $\cN$ and $\cM$ with the choice of parameters $p=1$ and $\eta=0.1$. 
\end{remark}

We remark that the issue for establishing multiplicativity is that it is not clear that the optimization program in~\eqref{eq:alpha_cN_def} is multiplicative because if for $X_1, X_2 \in \Herm$, $\Gamma^{\cN} \geq I_{A_{1}} \otimes X_{1} $ and $\Gamma^{\cM} \geq I_{A_{2}} \otimes X_{2}$, it does not follow that $\Gamma^{\cN} \otimes \Gamma^{\cM} \geq I_{A_{1}A_{2}} \otimes X_{1} \otimes X_{2}$ as may be verified with real numbers. In other words, the basic concern is that $\alpha(\mathcal{N})$ and $\alpha(\mathcal{M})$ can be truly optimized by Hermitian operators $ X_{1} $ and $ X_{2}$, respectively, such that the negative eigenspaces  cause an issue when tensored together. For instance, for the generalized amplitude damping channel with the parameters $p=1$ and $\eta=0.1$, through numerical evaluation, we find that the optimal  $X=\begin{bmatrix}
    0.6838 & 0 \\ 
    0 & -0.2162
\end{bmatrix}$. Then, with that, we have that $I_{A_1A_2} \otimes X \otimes X \nleq \Gamma^{\cN_{p,\eta}} \otimes \Gamma^{\cN_{p,\eta}} $.

In contrast, if there exist optimizers of $\alpha(\cN)$ and $\alpha(\cM)$, given by $X_{1}$ and $X_{2}$, respectively, that are positive semidefinite, then multiplicativity does hold, as we will now prove. To that end, let us begin with the following simple lemma.

\begin{lemma}\label{prop:tensor-operator-dom}
    Let $A,B,C \geq 0$. If $A \leq B$, then $A \otimes C \leq B \otimes C$.
\end{lemma}
\begin{proof}
    Consider the map $\Phi(\cdot) = (I \otimes \sqrt{C}) \cdot (I \otimes \sqrt{C})^{\dag}$. This is CP by the characterization of CP maps, thus $0 \leq \Phi(B-A) = \Phi(B) - \Phi(A) = B \otimes C - A \otimes C$. Re-ordering completes the proof.
\end{proof}
Using this, we obtain the following general result on multiplicativity, which relies on $\alpha_{+}(\cN)$ defined in  \cref{def:alpha-plus-doeblin}.
\begin{lemma}
\label{lem:suff-cond-for-multiplicativity}
    Given a set of channels $\cN_{A_{i} \to B_{i}}$ where $i \in [n]$. If $\alpha_{+}(\cN_{i}) = \alpha(\cN_{i})$ for all $i \in [n]$, then 
    \begin{align}
        \alpha\!\left(\bigotimes_{i \in [n]} \cN_{i}\right) = \prod_{i \in [n]} \alpha(\cN_{i}) \ . 
    \end{align}
\end{lemma}

\begin{proof}
    We first prove the case when $n = 2$  and then explain how the proof generalizes. Recall that \cref{def:alpha-plus-doeblin} is the same as \cref{def:q-doeblin-coeff} except the optimization is restricted to a positive semidefinite $X$. Thus, if $\alpha_{+}(\cN_{1}) = \alpha(\cN_{1})$ and $\alpha_{+}(\cN_{2}) = \alpha(\cN_{2})$, then $\alpha(\cN_{1}), \alpha(\cN_{2})$ can be optimized by optimizers of $\alpha_{+}(\cN_{1}),\alpha_{+}(\cN_{2})$, respectively, given by $X^{1}_{B_{1}},X^{2}_{B_{2}} \geq 0$. Then,
    \begin{align}
        (I_{A_{1}} \otimes X^{1}_{B_{1}}) \otimes (I_{A_{2}} \otimes X^{2}_{B_{2}}) & \leq \Gamma^{\cN_{1}} \otimes  (I_{A_{2}} \otimes X^{2}_{B_{2}}) \\
         & \leq \Gamma^{\cN_{1}} \otimes \Gamma^{\cN_{2}} \ , 
    \end{align}
    where the first inequality follows from tensoring $(I_{A_{2}} \otimes X^{2}_{B_{2}})$ to $(I_{A_{1}} \otimes X^{1}_{B_{1}}) \leq \Gamma^{\cN_{1}}$ and the second follows by tensoring $\Gamma^{\cN_{1}}$ to $(I_{A_{2}} \otimes X^{2}_{B_{2}}) \leq \Gamma^{\cN_{2}}$. Both of these inequalities hold via  \cref{prop:tensor-operator-dom}. Combining these relations, we obtain 
    \begin{align}
        I_{A_{1}} \otimes I_{A_{2}} \otimes X^{1}_{B_{1}} \otimes X^{2}_{B_{2}} \leq \Gamma^{\cN_1} \otimes \Gamma^{\cN_2} \ . 
    \end{align}
    Thus, $X_{B_{1}B_{2}} \coloneq X^{1} \otimes X^{2}$ is feasible for $\alpha(\cN_{1} \otimes \cN_{2})$ via \eqref{eq:alpha_cN_def} and achieves $\Tr[X^{1} \otimes X^{2}] = \Tr[X^{1}]\Tr[X^{2}] = \alpha(\cN_1)\alpha(\cN_2)$. As \eqref{eq:alpha_cN_def} is a supremization, combining this with \cref{prop:q-doeblin-sub-multiplicative} allows us to conclude $\alpha(\cN_{1} \otimes \cN_{2}) = \alpha(\cN_{1})\alpha(\cN_{2})$. This completes the proof for the $n = 2$ case. For $n > 2$, note that we can use the argument above iteratively to construct feasible $X_{B^{n}_{1}} \coloneq \otimes_{i \in [n]} X_{B_{i}}^{i}$ that achieves $\prod_{i \in [n]} \alpha(\cN_{i})$. By  \cref{prop:q-doeblin-sub-multiplicative}, this completes the proof.
\end{proof}

Note that as both $\alpha(\cN)$ and $\alpha_{+}(\cN)$ are semidefinite programs, the conditions of  \cref{lem:suff-cond-for-multiplicativity} can be checked numerically. However, it is reasonable to look for sufficient structure of the channel to know without numerically checking. We provide two cases. First, as will be shown in  \cref{cor:measurement-channel-Doeblin-collapse}, for every measurement channel $\cM$, it is the case that $\alpha_{+}(\cM) = \alpha(\cM)$. Thus we obtain the following:
\begin{theorem}
    For measurement channels $\cN_{A \to X}$ and $\cM_{B \to Y}$, 
    \begin{equation}
    \alpha(\cN \otimes \cM) = \alpha(\cN)\alpha(\cM).
    \end{equation}
\end{theorem}
We remark this case is important as it recovers that the Doeblin coefficient is multiplicative for arbitrary classical channels \cite[Theorem 1]{makur2024doeblin}. Also note that  \cref{example:strict-submultiplicativity-for-qubit-channels} shows that one cannot make a similar claim for preparation channels.

The second case where we obtain multiplicativity involves channels that are covariant with respect to some irreducible representation on their outputs. We refer to such channels as mictodiactic covariant channels, due to the fact that, for such channels, the output is a maximally mixed state if the input is maximally mixed.\footnote{The term ``mictodiactic'' comes from Greek for ``mixtures passing through.''}
\begin{definition}
    Let $\cN_{A \to B}$ be a quantum channel, and let $G$ be a group. Let $g\to U_{g}$ and $g\to V_{g}$ be unitary representations on $A$ and $B$ respectively, where $g\in G$. We say $\cN$ is $(G,U_{g},V_{g})$-covariant if $V_{g}\cN(X)V_{g}^{\dagger} = \cN(U_{g}XU_{g}^{\dagger})$ for all $g \in G$ and $X \in \operatorname{Lin}(A)$.
\end{definition}

\begin{lemma}\label{lem:q-doeblin-for-irreduc-covar-chan}
    Let $\cN_{A \to B}$ be a $(G,U_{g},V_{g})$-covariant channel, where $ g \to V_{g}$ is an irreducible representation of a compact group $G$. Then
    \begin{align}\label{eq:q-Doeblin-for-covariant-channel}
        \alpha(\cN) = \max\{c:  I_{A} \otimes c \pi_{B} \leq \Gamma^{\cN} \} = \alpha_{+}(\cN), 
    \end{align}
    where $\pi_B$ is the maximally mixed state of system $B$.
\end{lemma}

\begin{proof}
    Let $Y$ be an optimizer for~\eqref{eq:alpha_cN_def}. Then, using that the twirling map
    \begin{equation}  
    \cT(\cdot) \coloneq \int_{G} (\overline{U}_{g} \otimes V_{g}) \cdot  (\overline{U}_{g} \otimes V_{g})^{\dag} \ d\mu(g)
    \end{equation}
    is CP, we find that
    \begin{align}
        0 & \leq \cT(\Gamma^{\cN} - I_{A} \otimes Y) \\
        & = \Gamma^{\cN} -\int_{G} (\overline{U}_{g} \otimes V_{g}) I_{A} \otimes Y (\overline{U}_{g} \otimes V_{g})^{\dag}  \ d\mu(g) \\
        & = \Gamma^{\cN} - I_{A} \otimes \int_{G} V_{g}YV_{g}^{\dag} \ d\mu(g) \\
        & = \Gamma^{\cN} - I_{A} \otimes \Tr[Y]\pi_B \ , 
    \end{align}
    where the first equality follows from  linearity of the twirling map and covariance of the channel, the second from linearity and a standard property of unitaries, and the third follows from the fact that an irreducible representation of a compact group $G$ forms a one-design. It follows from the last equation that $\Tr[Y]\pi_B$  is feasible. Since $Y$ is optimal, this is also optimal as $\Tr[\Tr[Y]\pi_B] = \Tr[Y]$. Finally, note that $\Tr[Y] \geq 0$ without loss of generality given \eqref{eq:alpha_cN_def}. Thus, $\Tr[Y]\pi_{B}$ is feasible for $\alpha_{+}(\cN)$. As $\alpha_{+}(\cN) \leq \alpha(\cN)$, this shows that $\alpha_{+}(\cN) = \alpha(\cN)$. Combining these points completes the proof.
\end{proof}

\begin{theorem}
    \sloppy For all $i \in [n]$, let $\cN_{i} \in \Channel(A_{i},B_{i})$ be $(G_{i},U_{g_{i}},V_{g_{i}})$-covariant where $g_i \to V_{g_{i}}$ is an irreducible representation of a compact group $G_{i}$. Then 
    \begin{equation}
        \alpha\!\left(\bigotimes_{i \in [n]} \cN_{i}\right) = \prod_{i \in [n]} \alpha(\cN_{i}).
    \end{equation}
     That is, multiplicativity holds in this setting.
\end{theorem}

\begin{proof}
    This follows from  \cref{lem:q-doeblin-for-irreduc-covar-chan} and \cref{lem:suff-cond-for-multiplicativity}.
\end{proof}

\subsection{Concavity}

Next, we show that the quantum Doeblin coefficients are concave under a convex combination of channels. This may find use in applications where a set of noise channels are applied randomly.

\begin{proposition}[Concavity]
\label{prop:concavity-doeblin}
Let $\lambda_i \geq 0$ for all $i$ and $\sum_i {\lambda_i}=1$, and let $(\cN_i)_i$ be a tuple of quantum channels.
    The quantum Doeblin coefficients $\alpha(\cN) $ and $\alpha_I(\cN)$ are concave; i.e., for $\cN\coloneqq \sum_i\lambda_i\cN_i$, we have
    \begin{align}
        \alpha(\cN) \geq \sum_i \lambda_i \alpha(\cN_i),\\
        \alpha_I(\cN) \geq \sum_i \lambda_i \alpha_I(\cN_i).
    \end{align}
\end{proposition}

\begin{proof}
First, we see that $\Gamma^{\cN}=\sum_i \lambda_i \Gamma^{\cN_i}$. Let $X_{\cN_i}$ be chosen such that $I \otimes X_{\cN_i} \leq \Gamma^{\cN_i} $ for all $i$. With that, we have 
\begin{equation}
    \sum_{i} \lambda_i I \otimes X_{\cN_i}  \leq \sum_i \lambda_i \Gamma^{\cN_i} =\Gamma^{\cN}.
\end{equation}
This leads to 
    \begin{align}
        \alpha(\cN) \geq \tr\!\left[\sum_i \lambda_i X_{\cN_i}\right] =  \sum_i \lambda_i \alpha(\cN_i).
    \end{align}
Finally, we conclude the proof by optimizing $X_{\cN_i}$ on the right-hand side, together with~\eqref{eq:alpha_cN_def}. 

The claim for $\alpha_I(\mathcal{N})$ follows from similar reasoning.
\end{proof}

\subsection{Properties of Alternative Notions of Quantum Doeblin Coefficient}

\begin{proposition}\label{prop:induced-doeblin-normalized}
    For a quantum channel $\cN$, the following bounds hold:
    \begin{equation}
    0 \leq \alpha_{I}(\cN) \leq 1,   
    \end{equation}
    where the upper bound holds with equality if and only if $\cN$ is a replacer channel.
\end{proposition}

\begin{proof}
    That $\alpha_{I}(\cN) \geq 0$ follows from  \cref{prop:induced-doeblin-larger-than-doeblin} and \cref{prop:q-doeblin-normalization}. That $\alpha_{I}(\cN) \leq 1$ follows from  \cref{def:induced-doeblin-coefficient} and that the classical Doeblin coefficient is bounded from above by~1.
    
    We now focus on establishing the saturation condition. If $\cN$ is a replacer channel, it is straightforward to conclude that $\alpha_{I}(\cN) = 1$. The other direction follows from the inequality $\eta_{\operatorname{Tr}}(\mathcal{N}) \leq 1- \alpha_I(\mathcal{N})$ (see \eqref{eq:trace-distance-CC-to-induced-doeblin}), and an argument similar to that given in the proof of \cref{prop:q-doeblin-normalization}. Indeed, this inequality and $\alpha_{I}(\cN) = 1$ imply  that $\eta_{\Tr}(\cN) = 0$. It follows that $\sup_{\rho,\sigma \in \Density} \left\Vert \cN(\rho) - \cN(\sigma) \right\Vert_{1} = 0$. Thus, $\cN(\rho) = \cN(\sigma)$ for all $\rho,\sigma\in \mathcal{D}$. It follows that $\cN$ is a replacer channel.
\end{proof}

\begin{proposition}[Normalization]\label{prop:doeblin-wang-normalized}
    For a quantum channel $\mathcal{N}$, the following inequalities hold:
    \begin{equation}
        0 \leq \alpha_{\wang}(\mathcal{N}) \leq 1.
    \end{equation}
    Moreover, it is equal to one if and only if $\cN$ is a replacer channel.
\end{proposition}

\begin{proof}
    To establish the inequality $0 \leq \alpha_{\wang}(\mathcal{N})$, we note that $X_B=0 $ is feasible in the primal SDP in~\eqref{eq:primal-b-relaxed-doeblin}, and the objective function evaluates to zero for this choice. To establish the inequality $\alpha_{\wang}(\mathcal{N})\leq 1$, we note that $Y^2_{AB} = \frac{1}{d_A}I_A \otimes I_B$ and $Y^1_{AB}=0$ are feasible for the dual SDP in~\eqref{eq:dual-of-bN-quantity}, and the objective function evaluates to one for these choices.
    
    For the equality $\alpha_{\wang}(\mathcal{N}) = 1$, note that it is clear that if $\cN$ is a replacer channel, then $\alpha_{\wang}(\cN) = 1$ is achievable. Furthermore, if $\alpha_{\wang}(\cN) = 1$, then $1 = \alpha_{\wang}(\cN) \leq \alpha(\cN) \leq 1$ and, by  \cref{prop:q-doeblin-normalization}, the channel $\cN$ must be a replacer channel.
\end{proof}

\begin{proposition}[Multiplicativity of $\alpha_{\wang}(\mathcal{N})$]
\label{prop:multiplicativity-b-N}
For channels $\mathcal{N}$ and $\mathcal{M}$, the following multiplicativity equality holds:
\begin{equation}
    \alpha_{\wang}(\mathcal{N} \otimes \mathcal{M}) = \alpha_{\wang}(\mathcal{N} )\cdot  \alpha_{\wang}( \mathcal{M}).
\end{equation}
\end{proposition}

\begin{proof}
To show this equality, we employ proof techniques similar to those presented in~\cite[Proposition~6]{WFD17}. 

The inequality $\alpha_{\wang}(\mathcal{N} \otimes \mathcal{M}) \geq \alpha_{\wang}(\mathcal{N} )\cdot  \alpha_{\wang}( \mathcal{M})$ follows by picking $X_1 $ to be primal feasible for $\alpha_{\wang}(\mathcal{N} )$ and $X_2$ to be primal feasible for $\alpha_{\wang}(\mathcal{M} )$ (as given in~\eqref{eq:primal-b-relaxed-doeblin}), and noting that $X_1 \otimes X_2 $ is feasible for $\alpha_{\wang}(\mathcal{N} \otimes \mathcal{M})$, by~\cite[Lemma~12.35]{khatri2024principlesquantumcommunicationtheory} and the fact that the Choi operator of $\mathcal{N} \otimes \mathcal{M}$ is $\Gamma^{\mathcal{N}} \otimes \Gamma^{\mathcal{M}}$.

For the opposite inequality, suppose that $Z_{A_{1}B_{1}}^{1}$ and
$Z_{A_{1}B_{1}}^{2}$ are dual feasible for $\alpha_{\wang}(\mathcal{N})$ and that
$W_{A_{2}B_{2}}^{1}$ and $W_{A_{2}B_{2}}^{2}$ are dual feasible for
$\alpha_{\wang}(\mathcal{M})$, as defined in~\eqref{eq:dual-of-bN-quantity}. Thus, the following constraints hold:
\begin{align}
Z_{A_{1}B_{1}}^{1},Z_{A_{1}B_{1}}^{2},W_{A_{2}B_{2}}^{1},W_{A_{2}B_{2}}^{2}  &
\geq0,\\
\operatorname{Tr}_{A_{1}}\!\left[  Z_{A_{1}B_{1}}^{2}-Z_{A_{1}B_{1}}^{1}\right]
& =I_{B_{1}},\\
\operatorname{Tr}_{A_{2}}\!\left[  W_{A_{2}B_{2}}^{2}-W_{A_{2}B_{2}}^{1}\right]
& =I_{B_{2}}.
\end{align}
Then pick
\begin{align}
Y_{A_{1}B_{1}A_{2}B_{2}}^{1}  & \coloneqq Z_{A_{1}B_{1}}^{1}\otimes W_{A_{2}B_{2}}
^{1}+Z_{A_{1}B_{1}}^{2}\otimes W_{A_{2}B_{2}}^{2},\\
Y_{A_{1}B_{1}A_{2}B_{2}}^{2}  & \coloneqq Z_{A_{1}B_{1}}^{1}\otimes W_{A_{2}B_{2}}
^{2}+Z_{A_{1}B_{1}}^{2}\otimes W_{A_{2}B_{2}}^{1}.
\end{align}
It follows that $Y_{A_{1}B_{1}A_{2}B_{2}}^{1}$ and $Y_{A_{1}B_{1}A_{2}B_{2}
}^{2}$ are dual feasible for $\alpha_{\wang}(\mathcal{N}\otimes\mathcal{M})$ because
\begin{align}
Y_{A_{1}B_{1}A_{2}B_{2}}^{1},Y_{A_{1}B_{1}A_{2}B_{2}}^{2}  & \geq0,\\
\operatorname{Tr}_{A_{1}A_{2}}\!\left[  Y_{A_{1}B_{1}A_{2}B_{2}}^{2}
-Y_{A_{1}B_{1}A_{2}B_{2}}^{1}\right]    & =\operatorname{Tr}_{A_{1}A_{2}
}\!\left[  \left(  Z_{A_{1}B_{1}}^{2}-Z_{A_{1}B_{1}}^{1}\right)  \otimes\left(
W_{A_{2}B_{2}}^{2}-W_{A_{2}B_{2}}^{1}\right)  \right]  \\
& =I_{B_{1}}\otimes I_{B_{2}}\\
& =I_{B_{1}B_{2}}.
\end{align}
Furthermore, the objective function value $\operatorname{Tr}[\Gamma
^{\mathcal{N}\otimes\mathcal{M}}(Y_{A_{1}B_{1}A_{2}B_{2}}^{1}+Y_{A_{1}
B_{1}A_{2}B_{2}}^{2})]$ is given by
\begin{align}
& \operatorname{Tr}\!\left[  \Gamma^{\mathcal{N}\otimes\mathcal{M}}\left(
Y_{A_{1}B_{1}A_{2}B_{2}}^{1}+Y_{A_{1}B_{1}A_{2}B_{2}}^{2}\right)  \right]   \notag  \\
&
=\operatorname{Tr}\!\left[  \left(  \Gamma^{\mathcal{N}}\otimes\Gamma
^{\mathcal{M}}\right)  \left(  \left(  Z_{A_{1}B_{1}}^{1}+Z_{A_{1}B_{1}}
^{2}\right)  \otimes\left(  W_{A_{2}B_{2}}^{1}+W_{A_{2}B_{2}}^{2}\right)
\right)  \right]  \\
& =\operatorname{Tr}\!\left[  \Gamma^{\mathcal{N}}\left(  Z_{A_{1}B_{1}}
^{1}+Z_{A_{1}B_{1}}^{2}\right)  \otimes\Gamma^{\mathcal{M}}\left(
W_{A_{2}B_{2}}^{1}+W_{A_{2}B_{2}}^{2}\right)  \right]  \\
& =\operatorname{Tr}\!\left[  \Gamma^{\mathcal{N}}\left(  Z_{A_{1}B_{1}}
^{1}+Z_{A_{1}B_{1}}^{2}\right)  \right]  \operatorname{Tr}\!\left[
\Gamma^{\mathcal{M}}\left(  W_{A_{2}B_{2}}^{1}+W_{A_{2}B_{2}}^{2}\right)
\right]  .
\end{align}
As such, we conclude that
\begin{align}
& \operatorname{Tr}\!\left[  \Gamma^{\mathcal{N}}\left(  Z_{A_{1}B_{1}}
^{1}+Z_{A_{1}B_{1}}^{2}\right)  \right]  \operatorname{Tr}\!\left[
\Gamma^{\mathcal{M}}\left(  W_{A_{2}B_{2}}^{1}+W_{A_{2}B_{2}}^{2}\right)
\right] \notag  \\
& =\operatorname{Tr}\!\left[  \Gamma^{\mathcal{N}\otimes\mathcal{M}}\left(
Y_{A_{1}B_{1}A_{2}B_{2}}^{1}+Y_{A_{1}B_{1}A_{2}B_{2}}^{2}\right)  \right]  \\
& \geq \alpha_{\wang}(\mathcal{N}\otimes\mathcal{M})  .
\end{align}
Since the inequality holds for every $Z_{A_{1}B_{1}}^{1}$ and $Z_{A_{1}B_{1}
}^{2}$ dual feasible for $\alpha_{\wang}(\mathcal{N})$ and $W_{A_{2}B_{2}}^{1}$ and
$W_{A_{2}B_{2}}^{2}$ dual feasible for $\alpha_{\wang}(\mathcal{M})$, we conclude the
desired inequality:
\begin{equation}
\alpha_{\wang}(\mathcal{N}) \cdot \alpha_{\wang}(\mathcal{M})
\geq \alpha_{\wang}(\mathcal{N}\otimes\mathcal{M})  .
\end{equation}
This concludes the proof.
\end{proof}

\begin{corollary}\label{Cor:mult-alpha-from-alpha-wang}
   For channels $\mathcal{N}$ and $\mathcal{M}$, if $\alpha_\wang(\cN)=\alpha(\cN)$ and $\alpha_\wang(\cM)=\alpha(\cM)$, then the following multiplicativity equality holds 
   \begin{align}
       \alpha(\mathcal{N} \otimes \mathcal{M}) = \alpha(\mathcal{N} )\cdot  \alpha( \mathcal{M}).
   \end{align}
\end{corollary}

\begin{proof}
    Note the following chain of arguments, 
    \begin{align}
        \alpha(\mathcal{N} )\cdot  \alpha( \mathcal{M}) \geq \alpha(\mathcal{N} \otimes \mathcal{M}) \geq \alpha_\wang(\mathcal{N} \otimes \mathcal{M}) = \alpha_\wang(\mathcal{N} )\cdot  \alpha_\wang( \mathcal{M}) = \alpha(\mathcal{N} )\cdot  \alpha( \mathcal{M}). 
    \end{align}
    The first inequality follows from \cref{prop:q-doeblin-sub-multiplicative}, and the first equality follows from \cref{prop:multiplicativity-b-N}.
\end{proof}

\begin{proposition}[Data Processing under Pre- and Post-Processing and Isometric Invariance]
\label{prop:DPI-and-LU-invariance-wang}
Let $\cE_{B \to B'}$, $\cG_{A' \to A}$, and $\cN_{A \to B}$ be quantum channels.
   Then
   \begin{align}
       \alpha_{\wang}(\cE \circ \cN \circ \cG) & \geq \alpha_{\wang}(\cN)  . \label{eq:data-proc-alpha-wang}
   \end{align}
   Moreover, the quantity $\alpha_{\wang}(\mathcal{N})$ is invariant under isometric-reversal and isometric pre- and post-processing, respectively; i.e.,
\begin{align}
    \alpha(\mathcal{N}) & = \alpha(\mathcal{U} \circ \mathcal{N} \circ \mathcal{P}^{\mathcal{V}}),
    \label{eq:isometric-inv-alpha-wang},
\end{align}
where $\mathcal{U}$ and $\mathcal{V}$ are isometric channels and $\mathcal{P}^{\mathcal{V}}$ is defined in \eqref{eq:isometric-reversal-channel}.
\end{proposition}

\begin{proof}
Let $X\in \operatorname{Herm}$ be feasible for $\alpha_{\wang}(\mathcal{N})$. Then $-\mathcal{N} \leq \mathcal{R}^X \leq \mathcal{N}$ holds, and this implies that $-\mathcal{E}\circ \mathcal{N} \circ \mathcal{G} \leq \mathcal{E}\circ\mathcal{R}^X\circ \mathcal{G} \leq \mathcal{E}\circ\mathcal{N}\circ \mathcal{G}$ holds. Observe that 
\begin{align}
\mathcal{E}\circ\mathcal{R}^X\circ \mathcal{G} & = \mathcal{E}\circ\mathcal{R}^X \\
& = \mathcal{R}^{\mathcal{E}(X)}\\
\operatorname{Tr}[X] & = \operatorname{Tr}[\mathcal{E}(X)],
\end{align}
where the last equality follows because $\mathcal{E}$ is trace preserving. Then $\mathcal{E}(X)$ is feasible for $\alpha_{\wang}(\mathcal{E}\circ \mathcal{N} \circ \mathcal{G})$, and we conclude that
\begin{equation}
    \operatorname{Tr}[X] = \operatorname{Tr}[\mathcal{E}(X)] \leq \alpha_{\wang}(\mathcal{E}\circ \mathcal{N} \circ \mathcal{G}).
\end{equation}
Since the inequality holds for every $X$ feasible for $\alpha(\mathcal{N})$, we conclude that $\alpha_{\wang}( \mathcal{N} ) \leq \alpha_{\wang}\!\left(\mathcal{E}\circ \mathcal{N} \circ \mathcal{G}\right)$. The statement about isometric invariance follows the same proof given for \eqref{eq:isometric-inv-induced-doeblin}. 
\end{proof}

\begin{proposition}[Concavity]\label{prop:doeblin-wang-concavity} Let $\lambda_i \geq 0$ for all $i$ and $\sum_i {\lambda_i}=1$, and let $(\cN_i)_i$ be a tuple of quantum channels.
    Then $\alpha_{\wang}$ is concave; i.e., for $\cN\coloneqq \sum_i\lambda_i\cN_i$, we have
    \begin{align}
        \alpha_{\wang}(\cN) \geq \sum_i \lambda_i \alpha_{\wang}(\cN_i).
    \end{align}
\end{proposition}

\begin{proof}
    The proof is similar to the proof of \cref{prop:concavity-doeblin} but incorporates the extra constraint $-\Gamma^{\mathcal{N}} \leq X$.
\end{proof}

We lastly establish that $\alpha_{\wang}$ satisfies the same concatenation property that $\alpha$ and $\alpha_{I}$ satisfy. This follows the same proof method as for the other quantities, as given in \cref{prop:concatenation}. For this reason, the complete proof is provided in \cref{app:proof-of-alpha-wang-concat}, and the approach is summarized here. First, we establish the following alternative representation of $\alpha_{\wang}$:

\begin{lemma}
\label{lem:alpha-wang-CP-ordering-form}
For a quantum channel $\cN$,
\begin{align}
    \alpha_{\wang}(\cN) =\max_{\substack{\lambda \in \left[0,1\right], \\ \tau \in \aff(\cD), \\ \cM,\cG \in \operatorname{CPTP}}} \left\{ \lambda :\lambda\cR_{\tau} + (1-\lambda)\cM = \cN \, , \, \cN = (1+\lambda)\cG - \lambda \cR_{\tau} \right\} \ . 
\end{align}
\end{lemma}

\begin{proof}
    The proof is provided in \cref{app:proof-of-alpha-wang-concat}.
\end{proof}
\noindent This result and its proof method are analogous to \eqref{eq:key-robustness-step}, which was used to establish the robustness interpretation of $\alpha(\cN)$.

Second, we use optimizers of the above representation for $\alpha_{\wang}(\cN_{1})$ and $\alpha_{\wang}(\cN_{2})$ to construct a feasible solution to the above representation such that the concatenation property follows as a consequence.

\begin{proposition}
\label{prop:concat-of-alpha-wang}
    For quantum channels $\cN_{1}$ and $\cN_{2}$ such that $\cN_{2} \circ \cN_{1}$ is well defined,
    \begin{align}
        1-\alpha_{\wang}(\cN_{2}\circ \cN_{1}) \leq (1-\alpha_{\wang}(\cN_{2}))(1-\alpha_{\wang}(\cN_{1})) \ . 
    \end{align}
\end{proposition}
\begin{proof}
    The proof is provided in \cref{app:proof-of-alpha-wang-concat}.
\end{proof}

\section{Contraction and Expansion Coefficients}

\subsection{Contraction Coefficients and Quantum Doeblin Coefficients}

\label{sec:contraction-coefficient}

Let $\DD(\rho\|\sigma)$ be a generalized divergence, i.e., a function of a pair of quantum states onto the reals that obeys the data-processing inequality for all states $\rho$ and $\sigma$ and every channel~$\mathcal{N}$ \cite{Polyanskiy2010,sharma2012strongconversesquantumchannel}:
\begin{equation}
    \DD(\rho\|\sigma) \geq \DD(\cN(\rho)\|\cN(\sigma)).
\end{equation}
We call a generalized divergence faithful if $\DD(\rho\|\rho)\geq0$ holds for every quantum state $\rho$. Contraction coefficients give strong data-processing inequalities for these divergences and hence have numerous applications. There is a large body of literature exploring contraction coefficients for different divergences (see, e.g., \cite{Hiai15,Hirche2022contraction}). Here, we define a generalized contraction coefficient that allows for a reference system. 

\begin{definition}[Generalized Contraction Coefficient]
    For a quantum channel $\cN$, a generalized divergence $\DD$, and $r\in \mathbb{N}$, we define
    \begin{align}
        \eta_\DD(\cN,r) \coloneqq \sup_{\substack{\rho_{RA}\neq\sigma_{RA},\\ \rho_R=\sigma_R, \\ |R|=r}} \frac{\DD((\id\otimes\cN)(\rho_{RA})\|(\id\otimes\cN)(\sigma_{RA}))}{\DD(\rho_{RA}\|\sigma_{RA})}, 
    \end{align}
    leading to the contraction coefficient and the complete contraction coefficient, 
    \begin{align}
        \eta_\DD(\cN) &\coloneqq  \eta_\DD(\cN,1) = \sup_{\substack{\rho_{A}\neq \sigma_{A}}} \frac{\DD(\cN(\rho_{A})\|\cN(\sigma_{A}))}{\DD(\rho_{A}\|\sigma_{A})}, \label{eq:def-gen-contra-coef}\\
        \eta^c_\DD(\cN) & \coloneqq \sup_{r\geq1} \eta_\DD(\cN,r), \label{eq:Comp_CC}
    \end{align}
    respectively. 
\end{definition}
We also define a generalized input-dependent contraction coefficient that allows for a reference system.
\begin{definition}[Input-Dependent Generalized Contraction Coefficient] \label{def:input-dep-contraction-coeffs}
    For a quantum channel $\cN_{A \to B}$, a reference state $\omega_{A}$, a generalized divergence $\DD$, and $r\in \mathbb{N}$, we define
    \begin{align}
        \eta_\DD(\cN,\omega,r) \coloneqq \sup_{\substack{\rho_{RA} \neq \sigma_{RA},\\ \sigma_{A} = \omega_{A}, \\ \rho_R = \sigma_R, \\ |R| = r}} \frac{\DD((\id\otimes\cN)(\rho_{RA})\|(\id\otimes\cN)(\sigma_{RA}))}{\DD(\rho_{RA}\|\sigma_{RA})}, 
    \end{align}
    leading to the input-dependent contraction coefficient and the input-dependent complete contraction coefficient, 
    \begin{align}
        \eta_\DD(\cN,\omega) &\coloneqq  \eta_\DD(\cN,1) = \sup_{\substack{\rho_{A}\neq \omega_{A}}} \frac{\DD(\cN(\rho_{A})\|\cN(\omega_{A}))}{\DD(\rho_{A}\|\omega_{A})}, \label{eq:input_d_CC} \\
        \eta^c_\DD(\cN,\omega) & \coloneqq \sup_{r\geq1} \eta_\DD(\cN,\omega,r), \label{eq:input_d_Com_CC}
    \end{align}
    respectively. 
\end{definition}

The idea of allowing for additional reference systems was previously investigated in~\cite[Section 4]{gao2022complete}.
In the following, we provide a number of properties of these coefficients. 
\begin{proposition}[Properties of Contraction Coefficients] \label{prop:complete-contraction-coeff-properties}
    For quantum channels $\cN$ and $ \cM$, the following hold: 
    \begin{enumerate}
        \item For faithful $\DD$, the contraction coefficients are bounded as follows: 
        \begin{align}
    0\leq\eta_\DD(\cN) &\leq \eta^c_\DD(\cN)\leq 1 
    \label{eq:contraction-coeff-to-complete}, \\
    0\leq\eta_\DD(\cN,\omega) &\leq \eta^c_\DD(\cN,\omega)\leq 1 .
\end{align}
\item Under concatenation, the following inequality holds for all $r\in\mathbb{N}$: 
\begin{align}
\eta_\DD(\cN\circ\cM,r)  
    \leq \eta_\DD(\cN,r)\, \eta_\DD(\cM,r).
\end{align}
\item For isometric channels $\cU$  and $\cV$, the following equality holds for all $r\in\mathbb{N}$:
\begin{align}
    \eta_\DD(\cU\circ\cN\circ\cP^{\cV},r) = \eta_\DD(\cN,r),
\end{align}
where $\cP^{\cV}$ is defined in \eqref{eq:isometric-reversal-channel}. 
    \end{enumerate}
\end{proposition}
\begin{proof}
    1. In~\eqref{eq:contraction-coeff-to-complete}, the first inequality follows from faithfulness of the generalized divergence $\mathbb{D}$, the second by definition, and the third by the data-processing inequality. The same reasoning holds for the chain of inequalities in the input-dependent case.
    
    2. Here, observe that for all $\rho_{RA}$ and $\sigma_{RA}$ with $\rho_R=\sigma_R$ and $|R|=r$, we have
    \begin{align}    &\frac{\DD((\id\otimes(\cN\circ\cM))(\rho_{RA})\|(\id\otimes(\cN\circ\cM))(\sigma_{RA}))}{\DD(\rho_{RA}\|\sigma_{RA})} \notag \\
        &= \frac{\DD((\id\otimes(\cN\circ\cM))(\rho_{RA})\|(\id\otimes(\cN\circ\cM))(\sigma_{RA}))}{\DD((\id\otimes\cM)(\rho_{RA})\|(\id\otimes\cM)(\sigma_{RA}))} \frac{\DD((\id\otimes\cM)(\rho_{RA})\|(\id\otimes\cM)(\sigma_{RA}))}{\DD(\rho_{RA}\|\sigma_{RA})} \\
        &\leq \eta_\DD(\cN,r)\, \eta_\DD(\cM,r). 
    \end{align}
    
    3. As a consequence of 1.~and 2., the following inequality holds for every pre- and post-processing channel $\mathcal{E}$ and $\mathcal{G}$:
    \begin{equation}
        \eta_{\DD}(\mathcal{G} \circ \mathcal{N} \circ \mathcal{E},r) \leq \eta_{\DD}(\mathcal{N},r).
    \end{equation}
    The claim then follows from reasoning similar to that given for \eqref{eq:isometric-inv-doeblin}--\eqref{eq:isometric-inv-induced-doeblin}. 
\end{proof}
Based on the above, the complete trace-distance contraction coefficient of a quantum channel is defined as 
\begin{equation}
    \eta^c_{\Tr}(\mathcal{N}) \coloneqq \sup_{r\geq1} \sup_{\substack{\rho_{RA},\sigma_{RA}\\ \rho_R=\sigma_R \\ |R|=r}} \frac{T\!\left((\id\otimes\cN)(\rho_{RA}),(\id\otimes\cN)(\sigma_{RA})\right)}{T(\rho_{RA},\sigma_{RA})},  \label{eq:trace_distance_cc}
\end{equation}
where $T(\rho,\sigma) $ is defined in \eqref{eq:norm-TD-def}.
The non-complete version was previously discussed in~\cite{LR1999,Hiai15}. One main observation from~\cite[Eq.~(8.86)]{wolf2012quantum} and~\cite[Eq.~(III.57)]{hirche2024quantum} is that $1-\alpha(\cN)$ is an upper bound  on $\eta_{\Tr}(\mathcal{N})$.
Here we prove the stronger statement that it is also an upper bound on the complete contraction coefficient.

\begin{proposition} \label{prop:trace_distance_complete_cont_Doeblin_bound} For $\cN$ a quantum channel, the following inequality holds:
    \begin{align}
        \eta^c_{\Tr}(\mathcal{N}) \leq 1-\alpha(\cN). 
        \label{eq:complete-trace-distance-CC-to-doeblin}
    \end{align}
\end{proposition}

\begin{proof}
Via the reverse robustness representation of the Doeblin coefficient (\cref{prop:reverse-robustness}), let $\lambda \in [0,1],\tau\in\operatorname{aff}(\mathcal{D}), \mathcal{M} \in \operatorname{CPTP}$ be  optimizers for $\alpha(\cN)$, such that $ \lambda\mathcal{R}^\tau + (1-\lambda) \mathcal{M} = \mathcal{N}$. Then, 
\begin{align}
    & \left\| (\id\otimes\cN)(\rho_{RA}) - (\id\otimes\cN)(\sigma_{RA}) \right\|_1 \notag \\
        &=  \left\| (\id\otimes(\lambda\mathcal{R}^\tau + (1-\lambda) \mathcal{M}))(\rho_{RA}) - (\id\otimes(\lambda\mathcal{R}^\tau + (1-\lambda) \mathcal{M}))(\sigma_{RA}) \right\|_1 \\
    & = \left\| \lambda   \rho_R \otimes \tau - \lambda   \sigma_R \otimes \tau + (1-\lambda) [(\id\otimes\mathcal{M})(\rho_{RA}) -  (\id\otimes\mathcal{M})(\sigma_{RA})] \right\|_1\\
   &= (1-\lambda) \left\|  (\id\otimes\mathcal{M})(\rho_{RA}-\sigma_{RA}) \right\|_1 \\
   &\leq (1-\lambda) \left\|  \rho_{RA} - \sigma_{RA} \right\|_1. 
\end{align}
Here, the third equality makes use of the identical marginals on $R$ (i.e., $\rho_R = \sigma_R$), and the inequality follows from the data-processing inequality for the trace distance. Noting that $\lambda=\alpha(\cN)$ concludes the proof. 
\end{proof}

Let us employ the shorthand $\eta^c_\gamma \equiv \eta^c_{E_\gamma}$, when choosing the generalized divergence to be the hockey-stick divergence (see \eqref{eq:hockey-stick-div-def}). Next, we show that  the following upper bound holds.

\begin{proposition}[Complete Contraction of Hockey-Stick Divergence]\label{prop:CC_HS}
    For $\cN$  a quantum channel and $\gamma \geq 1$, the following inequality holds:
    \begin{equation}
        \eta^c_\gamma(\mathcal{N}) \leq 1- \alpha_{+}(\cN),
    \end{equation}
where $\alpha_{+}(\cN)$ is defined in~\eqref{def:alpha-plus-doeblin}.
\end{proposition}
\begin{proof}
Recall from~\cite[Proposition~III.3]{hirche2024quantum} that 
\begin{equation}
\alpha_{+}(\mathcal{N})=\sup_{\substack{\lambda\in\left[  0,1\right]  ,\\\sigma
\in \cD,\\\mathcal{M}\in\operatorname{CPTP}}
}\left\{  \lambda:\mathcal{N}=\lambda\mathcal{R}^{\sigma}+\left(  1-\lambda
\right)  \mathcal{M}\right\}  .
\label{eq:reverse-robustness-alpha_+}
\end{equation}
Let $\lambda \in [0,1]$, $\sigma \in \cD$, and  $\mathcal{M} \in \operatorname{CPTP}$ be  optimizers for $\alpha_+(\cN)$, such that $ \lambda\mathcal{R}^\sigma + (1-\lambda) \mathcal{M} = \mathcal{N}$.

Furthermore, let $\rho_{RA}$ and $\omega_{RA}$ be quantum states such that $\rho_R=\omega_R$. Then consider that
\begin{align}
   & E_\gamma\!\left((\id_R \otimes \cN  )(\rho_{RA}) \Vert (\id_R \otimes \cN  )(\omega_{RA})\right) \notag \\ 
   & =E_\gamma\!\left((\id\otimes(\lambda\mathcal{R}^\sigma + (1-\lambda) \mathcal{M}))(\rho_{RA}) \Vert (\id\otimes(\lambda\mathcal{R}^\sigma + (1-\lambda) \mathcal{M}))(\omega_{RA}) \right)  \\ 
   &= \sup_{0 \leq M_{BR} \leq I_{BR}} \Tr\!\left[ M_{BR}\left[(\id\otimes(\lambda\mathcal{R}^\sigma + (1-\lambda) \mathcal{M}))(\rho_{RA}) - \gamma(\id\otimes(\lambda\mathcal{R}^\sigma + (1-\lambda) \mathcal{M}))(\omega_{RA}) \right]  \right] \\
   &= \sup_{0 \leq M_{BR} \leq I_{BR}} \Tr\!\left[ M_{BR}\left(\lambda   \rho_R \otimes \sigma -\gamma \lambda   \omega_R \otimes \sigma + (1-\lambda) [(\id\otimes\mathcal{M})(\rho_{RA}) - \gamma  (\id\otimes\mathcal{M})(\omega_{RA})]  \right)  \right] \\
   & \leq  \sup_{0 \leq M_{BR} \leq I_{BR}} \Tr\!\left[ M_{BR}\left( (1-\lambda) [(\id\otimes\mathcal{M})(\rho_{RA}) - \gamma  (\id\otimes\mathcal{M})(\omega_{RA})]  \right)  \right] \\ 
   &=(1-\lambda)  E_\gamma\!\left((\id_R \otimes  \cM )(\rho_{RA}) \Vert (\id_R \otimes  \cM )(\omega_{RA})\right) \\
   &\leq (1-\lambda)  E_\gamma(\rho_{RA} \Vert \omega_{RA}).
\end{align}
The second equality follows by using~\eqref{eq:HS_sup_form} for the hockey-stick divergence; the first inequality follows by upper bounding the negative term by zero, because $1-\gamma \leq 0$, $\rho_R=\omega_R$, and $\sigma \otimes \rho_R \geq 0$; and the last inequality by the data-processing inequality for the hockey-stick divergence.

Finally, by supremizing over quantum states  $\rho_{RA}$ and $\omega_{RA}$ satisfying $\rho_R=\omega_R$, we conclude the proof.
\end{proof}

Let $f: (0,\infty) \to \RR$ be a convex and twice differentiable function 
satisfying $f(1)=0$. Then, for all quantum states $\rho$ and $\sigma$, recall the quantum $f$-divergence  defined in \cite[Definition~2.3]{hirche2024quantumDivergences}: 
\begin{equation}\label{eq:f_divergence}
    D_f(\rho \Vert \sigma) \coloneqq \int_{1}^{\infty} f''(\gamma) E_\gamma(\rho \Vert \sigma) + \gamma^{-3} f''(\gamma^{-1}) E_\gamma(\sigma \Vert \rho) \ \mathrm{d} \gamma.
\end{equation}
\begin{corollary}[Complete Contraction of $f$-Divergences] \label{Cor:CC_f_div}
    Let $D_f$ be an $f$-divergence of the form in~\eqref{eq:f_divergence}. Then, we have that 
    \begin{equation}
        \eta^c_{D_f}(\cN) \leq  1- \alpha_{+}(\cN),
    \end{equation}
where $\alpha_{+}(\cN)$ is defined in~\eqref{def:alpha-plus-doeblin}.
\end{corollary}
\begin{proof}
    Let $\rho_{RA}$ and $\sigma_{RA}$ be quantum states such that $\rho_R=\sigma_R$. 
    Then, consider that
    \begin{align}
       & D_f\!\left( (\id_R \otimes \cN  ) (\rho_{RA}) \Vert (\id_R \otimes \cN  )(\sigma_{RA})\right) \notag \\
       & =\int_{1}^{\infty} f''(\gamma) E_\gamma\!\left((\id_R \otimes \cN  )(\rho_{RA}) \Vert (\id_R \otimes \cN  )(\sigma_{RA})\right) \notag \\ 
       & \quad \quad + \gamma^{-3} f''(\gamma^{-1}) E_\gamma\!\left((\id_R \otimes \cN  )(\sigma_{RA}) \Vert ( \id_R \otimes \cN  )(\rho_{RA})\right) \mathrm{d} \gamma \\
       & \leq \left(1-\alpha_+(\cN)\right) \int_{1}^\infty f''(\gamma) E_\gamma(\rho_{RA} \Vert \sigma_{RA}) + \gamma^{-3} f''(\gamma^{-1}) E_\gamma(\sigma_{RA} \Vert \rho_{RA}) \ \mathrm{d} \gamma \\
       &= \left(1-\alpha_+(\cN)\right) D_f(\rho_{RA} \Vert \sigma_{RA}),
    \end{align}
    where the last inequality follows from~\cref{prop:CC_HS}.
    With this and supremizing over quantum states $\rho_{RA}$ and $\sigma_{RA}$  such that $\rho_R=\sigma_R$, we conclude the proof.
\end{proof}

If we do not allow for reference systems, one can in fact deduce from~\cite[Theorem~8.17]{wolf2012quantum} that also
\begin{align}
    \eta_{\Tr}(\mathcal{N}) \leq 1-\alpha_I(\cN)
    \label{eq:trace-distance-CC-to-induced-doeblin}
\end{align}
holds. 
Bounds on the trace-distance contraction coefficient are particularly interesting because they give upper bounds on the contraction coefficient for a wide class of quantum $f$-divergences~\cite[Lemma~4.1]{hirche2024quantumDivergences}. As $\alpha_{I}(\cN)$ is not known to be computationally efficient to evaluate in general, as it requires optimizing over the separable cone, we note the following hierarchy of SDP computable upper bounds that can only improve upon relaxing to considering $\alpha(\cN)$.
\begin{proposition}
    Consider quantum channel $\cN_{A \to B}$. For every $k \in \mbb{N}$, define the cone $\cK_{k}(A:B) \coloneq \operatorname{PPT}(A:B) \cap \operatorname{Sym}_{k}(A:B)$ where the cones on the right are defined in \eqref{eq:PPT-cone} and \cref{def:k-sym-ext-cone}. Then for every $k \in \mbb{N}$, $\alpha_{I,\cK_{k}}(\cN)$, as defined via \eqref{eqn:inducedDoeblinDual_cone}, is SDP representable and satisfies the inequality
    \begin{align}
        \eta_{\Tr}(\cN) \leq 1 - \alpha_{I,\cK_{k}}(\cN) \ ,
    \end{align}
    which can only ever be a tighter bound than $\eta_{\Tr}(\cN) \leq 1 - \alpha(\cN)$.
\end{proposition}

\begin{proof}
    By \cref{cor:entangling-power-in-terms-of-k-sym}, $\alpha_{I}(\cN) \geq \alpha_{I,\cK_{k}}(\cN) \geq \alpha(\cN)$. Combining this with \eqref{eq:trace-distance-CC-to-induced-doeblin} establishes the inequality and the claim that this bound can only be tighter. That $\alpha_{I,\cK_{k}}(\cN)$ is SDP representable follows from the constraints that define $\operatorname{PPT}(A:B)$ and $\operatorname{Sym}_{k}(A:B)$ being SDP representable as is well known. So in this case, \eqref{eqn:inducedDoeblinDual} is SDP representable.
\end{proof}
\noindent We remark that $\alpha_{I,\cK_{k}}(\cN)$ has the appealing property of having an operational interpretation given in \cref{cor:entangling-power-in-terms-of-k-sym}.

\bigskip 
We define the contraction coefficient for the hockey-stick divergence $E_\gamma(\rho\|\sigma)$ for $\gamma \geq 1$ from \eqref{eq:def-gen-contra-coef}, resulting in the following~\cite[Eq.~(II.4)]{hirche2022quantum}:
\begin{align}
    \eta_\gamma(\cN) &\coloneqq  \sup_{\substack{\rho,\sigma \in \mathcal{D}: \\ E_\gamma(\rho,\sigma) \neq 0} } \frac{E_\gamma(\cN(\rho)\|\cN(\sigma))}{E_\gamma(\rho\|\sigma)}  \label{eq:contraction_coe_E_gamma}\\
    &= \sup_{\psi\perp\phi} E_\gamma(\cN(\psi)\|\cN(\phi)), \label{eq:contraction_E_gamma_reduction}
\end{align}
where the second equality is a consequence of~\cite[Theorem~II.2]{hirche2022quantum}  and the notation $\psi\perp\phi$ indicates that $\psi$ and $\phi$ are orthogonal pure states. (See also~\cite[Proposition~11]{nuradha2024contraction} for an alternative proof of~\eqref{eq:contraction_E_gamma_reduction}.) Note that 
\begin{equation}
    \eta_1(\cN) = \eta_{\Tr}(\cN)
\end{equation}
because $E_1(\rho \Vert \sigma)= T(\rho,\sigma)$.

Here, we provide an alternative proof for a refined connection between $\eta_\gamma(\cN)$ and $\alpha(\cN)$, originally considered in~\cite[Theorem~8.17]{wolf2012quantum} and~\cite[Lemma~III.6]{hirche2024quantum}. First we give an alternative expression for the hockey-stick divergence. 

\begin{lemma}
\label{lem:hockey_stick_var}
Let $\rho$ be a Hermitian operator with unit trace, and let $\sigma $ be a positive semidefinite operator. For $\gamma \geq 0$,
    we have
    \begin{align}
        E_\gamma(\rho\|\sigma) +(1-\gamma)_{+} = 1 - \sup_{Y\in\Herm}\{ \Tr[Y]:Y\leq\rho, Y\leq\gamma\sigma \}\label{Eq:HS-var}.
    \end{align}
\end{lemma}
\begin{proof}
A similar expression was previously given for the minimum probability of error in a state discrimination problem~(see~\cite[Eq.~(III.17)]{YKL1975} and~\cite[Section II]{bandyopadhyay2014conclusive}), 
\begin{align}
    p_e(p,\rho,\sigma) = \sup_{Y\in\Herm}\{ \tr[Y]:Y\leq p\rho, Y\leq(1-p)\sigma \}. 
\end{align}
Observing that the minimum error probability can be written as a function of the hockey-stick divergence as follows~\cite[Eq.~(2.51)]{hirche2024quantumDivergences} 
\begin{align}
    p_e(p,\rho,\sigma) &\coloneqq  \inf_{M : 0\leq M \leq I} \left\{p \tr (\Id-M)\rho + (1-p)\tr M\sigma\right\} \\
    &= p - p E_{\frac{1-p}p}(\rho\|\sigma), 
\end{align}
and fixing $p=\frac1{1+\gamma}$, we have
\begin{align}
    E_\gamma(\rho\|\sigma) + (1-\gamma)_{+} &= 1-(\gamma+1)p_e(p,\rho,\sigma) \\ 
   &= 1-(\gamma+1)\sup_{Y\in\Herm}\{ \tr[Y]:Y\leq \frac1{1+\gamma}\rho, Y\leq\frac\gamma{1+\gamma}\sigma \}\\ 
    &= 1-\sup_{Y\in\Herm}\{ \tr[Y]:Y\leq \rho, Y\leq\gamma\sigma \},
\end{align}
concluding the proof.
\end{proof}

\cref{cor:contr-hs-div} below provides an alternative formulation of the hockey-stick contraction coefficient $\eta_\gamma(\cN)$. 
\begin{corollary}[Contraction Coefficient of Hockey-Stick Divergence]
\label{cor:contr-hs-div}
Let $\cN$ be a quantum channel.
    The following equality holds for $\gamma \geq 1$:
    \begin{align}
    \eta_\gamma(\cN) 
    &= 1 - \inf_{\psi\perp\phi}\sup_{Y\in\Herm}\{ \tr[Y]:Y\leq\cN(\psi), Y\leq\gamma\cN(\phi) \} ,
    \label{eq:contraction-pure-stick-dual}
\end{align}
where $\eta_\gamma(\cN)$ is defined in~\eqref{eq:contraction_coe_E_gamma}.
\end{corollary}

\begin{proof}
The proof of~\eqref{eq:contraction-pure-stick-dual} follows by adapting~\cref{lem:hockey_stick_var} together with~\eqref{eq:contraction_E_gamma_reduction}.   
\end{proof}

\begin{remark}[Efficient Computability of $\alpha(\cN^{\otimes n})$] \label{rem:effi_com_alpha}
   Due to
    \begin{equation}
        \eta_{\operatorname{Tr}}(\mathcal{N}^{\otimes n}) \leq 1-\alpha(\cN^{\otimes n}),
    \end{equation}
the computation of $\alpha(\cN^{\otimes n})$ is useful in providing bounds on $\eta_{\operatorname{Tr}}(\mathcal{N}^{\otimes n}) $. Recall from the dual SDP in~\eqref{eq:alpha-cN-dual} that 
    \begin{align}
\alpha(\cN^{\otimes n}) &= \inf_{Y_{A^nB^n} \geq 0} \left\{ \Tr\left[(\Gamma^{\cN}_{AB})^{\otimes n}Y_{A^nB^n}\right] \, : \, \Tr_{A^n}[Y_{A^nB^n}] = I_{B^n} \right\}.
\end{align}
By observing the SDP as it is, it seems that the time complexity to evaluate this SDP would be $O(\exp(n))$. However, one can utilize permutation symmetry to reduce its complexity to $O\!\left( \operatorname{poly}(n)\right)$, by following the approach demonstrated in~\cite{FST22}. In particular, we can establish the claimed efficiency with the use of properties in~\cite[Proposition~3]{singh2024extendible}.
\end{remark}

\begin{corollary}
    Let $\cN$ be a quantum channel.
The following inequality holds for all $n\in \mathbb{N}$:
\begin{equation}
\label{eq:tighter_contraction_n_tensor_d_max_alt}
\eta_{\operatorname{Tr}}(\mathcal{N}^{\otimes n}) \leq 1-\alpha_{\wang}(\mathcal{N})^n \leq n(1-\alpha_{\wang}(\mathcal{N})).
\end{equation}
\end{corollary}

\begin{proof}
The first inequality follows from~\eqref{eq:trace-distance-CC-to-induced-doeblin}, \cref{prop:induced-doeblin-larger-than-doeblin},~\eqref{eq:a-b-doeblin-ineq}, and multiplicativity of $\alpha_{\wang}(\mathcal{N})$ (\cref{prop:multiplicativity-b-N}). The second inequality follows because $1-x^n \leq n(1-x)$ for all $n\in \mathbb{N}$ and $x\in[0,1]$.
\end{proof}

\subsection{Expansion Coefficients and Reverse Doeblin Coefficients}

So far, we considered contraction coefficients and Doeblin coefficients. Next, we focus on expansion coefficients and reverse Doeblin coefficients. Expansion coefficients were mentioned in~\cite{ramakrishnan2020computing} and then further explored in~\cite{laracuente2023information,hirche2024quantum,belzig2024reversetypedataprocessinginequality}. We give a general definition here that allows for arbitrary divergences and additional reference systems. 
\begin{definition}
    For a quantum channel $\cN$, a generalized divergence $\DD$, and $r\in \mathbb{N}$, we define
    \begin{align}
        \check\eta_\DD(\cN,r) \coloneqq \inf_{\substack{\rho_{RA}\neq\sigma_{RA},\\ \rho_R=\sigma_R, \\ |R|=r}} \frac{\DD((\id\otimes\cN)(\rho_{RA})\|(\id\otimes\cN)(\sigma_{RA}))}{\DD(\rho_{RA}\|\sigma_{RA})}, 
    \end{align}
    leading to the expansion coefficient and the complete expansion coefficient, 
    \begin{align}
        \check\eta_\DD(\cN) &\coloneqq  \check\eta_\DD(\cN,1) = \inf_{\substack{\rho_{A}\neq \sigma_{A}}} \frac{\DD(\cN(\rho_{A})\|\cN(\sigma_{A}))}{\DD(\rho_{A}\|\sigma_{A})} \\
        \check\eta^c_\DD(\cN) & \coloneqq \inf_{r\geq1} \check\eta_\DD(\cN,r), 
    \end{align}
    respectively. 
\end{definition}
It should be remarked that, in this generality, expansion coefficients are often trivial, i.e., equal to zero~\cite{laracuente2023information,hirche2024quantum,belzig2024reversetypedataprocessinginequality}. A notable exception is the case of the trace distance. 
It was shown in~\cite[Lemma~IV.1]{hirche2024quantum} that the expansion coefficient of the trace distance is lower bounded as follows by the reverse Doeblin coefficient $\check{\alpha}(\cN)$:
\begin{equation} \label{eq:expansion_trace}
    \check{\eta}_{\Tr}(\cN)  \geq 1- \check{\alpha}(\cN),
\end{equation}
where 
\begin{equation} \label{eq:reverse_Doeblin}
   \check{\alpha}(\cN) \coloneqq \inf_{X \in \Herm} \left\{p: \cN  \geq_{\operatorname{deg}}\cQ_{p,X} \right\}
\end{equation}
with 
\begin{equation}
    \cQ_{p,X}(\rho) \coloneqq (1-p) \rho +p \operatorname{Tr}[\rho] X
\end{equation}
for $X \in \Herm$ and for all $\rho$. The notation $\cN  \geq_{\operatorname{deg}} \cM $ is equivalent to the existence of a channel $\cA$ such that $\cA \circ \cN = \cM$.
Also, note that~\eqref{eq:reverse_Doeblin} is SDP computable by~\cite[Corollary~IV.2]{hirche2024quantum}.
Here, we state a stronger result that includes the complete expansion coefficient for the trace distance. 
\begin{proposition} \label{prop:trace_distance_complete_exp_Doeblin_bound} Let $\cN$ be a quantum channel. Then, we have
    \begin{align}
        \check\eta_{\Tr}(\mathcal{N}) \geq \check\eta^c_{\Tr}(\mathcal{N}) \geq 1-\check\alpha(\cN). 
        \label{eq:complete-trace-distance-EC-to-doeblin}
    \end{align}
\end{proposition}

\begin{proof}
    The first inequality holds by definition. Now, let $p$, $\cA$, and $X$ be optimizers for $\check\alpha(\cN)$. Consider the following: 
    \begin{align}
        &\left\| (\id\otimes\cN)(\rho_{RA}) - (\id\otimes\cN)(\sigma_{RA}) \right\|_1 \notag \\
        &\geq \left\| (\id\otimes\cA\circ\cN)(\rho_{RA}) - (\id\otimes\cA\circ\cN)(\sigma_{RA}) \right\|_1 \\
        &= \left\| (\id\otimes\cQ_{p,X})(\rho_{RA}) - (\id\otimes\cQ_{p,X})(\sigma_{RA}) \right\|_1 \\
        &= (1-p) \left\| \rho_{RA} - \sigma_{RA} \right\|_1. 
    \end{align}
    The second inequality follows directly from the fact that we have shown that
    \begin{equation}
    \frac{\left\| (\id\otimes\cN)(\rho_{RA}) - (\id\otimes\cN)(\sigma_{RA}) \right\|_1}{\left\| \rho_{RA} - \sigma_{RA} \right\|_1} \geq 1-p
    \end{equation}
    holds for all states $\rho_{RA}$ and $\sigma_{RA}$ satisfying $\rho_R = \sigma_R$.
\end{proof}

Although not the main focus of this work, expansion coefficients and reverse Doeblin coefficient will play a role again later when we discuss applications. 

\section{Quantum Doeblin Coefficients for Example Channels} 

\label{sec:doeblin-examples}

In this section, we consider quantum Doeblin coefficients for a variety of quantum channels, in some cases finding exact expressions. 

\subsection{Classical--Quantum Channels}

First, we consider classical--quantum channels of the form
\begin{align}
    \cN(\cdot) = \sum_i \langle i| (\cdot) |i\rangle  \rho_i, \label{Eq:m-and-p}
\end{align}
where $(\rho_i)_i$ is a tuple of quantum states and $\{|i\rangle\}_i$ is an orthonormal basis. For channels of this form, we obtain the following characterization of the Doeblin coefficient $\alpha(\mathcal{N})$, which demonstrates that, for such channels, it is equal to a quantity previously considered in~\cite[Eq.~(185)]{mishra2023optimal}, in the context of error exponents for quantum state exclusion.

\begin{proposition}
    \label{lem:cq_channel_doeblin}
    For a classical--quantum channel $\cN$ of the form in~\eqref{Eq:m-and-p}, the following equality holds:
    \begin{align}
        \alpha(\cN) &= \sup_{ X \in\Herm} \left\{\tr X :X  \leq \rho_i\;\forall i \right\}.
    \end{align}
\end{proposition}

\begin{proof}
For such channels, the Choi operator is given by
\begin{align}
    \Gamma^{\cN} = \sum_{i,j}|i\rangle\!\langle j| \otimes \mathcal{N}(|i\rangle\!\langle
j|) = \sum_{i}|i\rangle\!\langle i| \otimes \rho_i.
\end{align}
Hence, 
\begin{align}
        \alpha(\cN) &= \sup_{X\in\Herm} \left\{\tr X: I \otimes X 
        \leq \Gamma^{\cN} \right\} \\
        &= \sup_{X\in\Herm} \left\{\tr X: I \otimes X \leq \sum_{i} |i\rangle\!\langle i| \otimes \rho_i \right\} \\
        &= \sup_{X\in\Herm} \left\{\tr X:X  \leq \rho_i\;\forall i \right\},
    \end{align}
    concluding the proof.
\end{proof}

\begin{remark}[Connection to Trace Distance and its Multivariate Generalizations]
Comparing~\cref{lem:cq_channel_doeblin} with~\eqref{Eq:HS-var} for $\gamma=1$ and applying \cref{lem:hockey_stick_var}, we notice that for a channel~$\cN$ as in~\eqref{Eq:m-and-p} and associated with a pair $(\rho,\sigma)$ of states, we have
\begin{align}
    E_1(\rho\|\sigma) = \frac12 \left\|\rho-\sigma\right\|_1 = 1 - \alpha(\cN). \label{eq:E_1_2_ensemble}
\end{align}
For classical--quantum channels with more than two states (i.e., such that the dimension of the input system is larger than two), we can view the quantum Doeblin coefficient as a multivariate trace distance generalization. This generalizes classical results from~\cite{makur2024doeblin}.     
\end{remark}

\begin{remark}[State Discrimination and Exclusion]
\label{rem:doeblin-exclusion-cq}
    We can relate~\cref{lem:cq_channel_doeblin} and~\eqref{eq:E_1_2_ensemble} to the probability of error in state discrimination and state exclusion. Consider an equal prior $p=\frac12$ setting. Then the probability of error evaluates to 
\begin{align}
    p_e(p,\rho,\sigma) = \frac12 \alpha(\cN) = \frac12(1-E_1(\rho\|\sigma)). 
\end{align}

Analogous considerations hold for non-equal priors and state exclusion between multiple states; see~\cite[Eq.~(181)]{mishra2023optimal}. 
\end{remark}

\begin{corollary}
\label{cor:cq-chan-induced-doeb-is-doeb}
    For a classical-quantum channel $\cN$, the following equality holds:
    \begin{equation}
    \alpha_{I}(\cN) = \alpha(\cN).
    \end{equation}
\end{corollary}
\begin{proof}
As \cref{prop:induced-doeblin-larger-than-doeblin} shows that $\alpha_{I}(\cN) \geq \alpha(\cN)$, it suffices to prove the other inequality. Let $X$ be an optimizer for $\alpha_{I}(\cN)$. We will show that it is feasible for $\alpha(\cN)$. By \cref{prop:induced-Doeblin} and \eqref{eq:block-positive-relation}, we have that
\begin{align}
    & \bra{\alpha}\bra{\beta}(\Gamma^{\cN} - I \otimes X) \ket{\alpha}\ket{\beta} \geq 0 \,  \qquad \forall \ket{\alpha}\ket{\beta} \notag \\
    \Leftrightarrow  \quad & \sum_{i} \vert \langle \alpha \vert i \rangle \vert^{2} \bra{\beta}\rho_{i}\ket{\beta} - \bra{\beta}X\ket{\beta} \geq 0 \, \qquad \forall \ket{\alpha}\ket{\beta} \ ,
\end{align}
where the equivalence uses the Choi representation of the classical-quantum channel. By considering the case $\ket{\alpha}=\ket{i'}$, we find that $\bra{\beta}\rho_{i'}\ket{\beta} \geq \bra{\beta}X\ket{\beta}$ for all $\ket{\beta}$. This is the condition for positive semidefiniteness, i.e., $\rho_{i'} \geq X$. As this holds for each $i'$, we have that $\rho_{i} \geq X$ for all $i$. By \cref{lem:cq_channel_doeblin}, this proves that $X$ is feasible for $\alpha(\cN)$. As $\alpha(\cN)$ involves a maximization, this completes the proof.
\end{proof}

\subsection{Measurement Channels}

A measurement channel is defined as follows for a measurement comprised of POVM elements~$(\Lambda_y)_{y\in \mathcal{Y}}$:
\begin{equation}
    \mathcal{M}(\rho) \coloneqq \sum_{y\in \mathcal{Y}} \operatorname{Tr}[\Lambda_y \rho] |y\rangle \! \langle y|,
    \label{eq:def-meas-channel}
\end{equation}
where $\{|y\rangle\}_{y\in\mathcal{Y}}$ is an orthonormal basis.

\begin{proposition}\label{prop:doeblin-of-measurement-channel}
    The following formula holds for the Doeblin coefficient of a measurement channel~$\mathcal{M}$, as defined in~\eqref{eq:def-meas-channel}:
    \begin{equation}
        \alpha (\mathcal{M}) =  \sum_{y\in \mathcal{Y}} \lambda_{\min}(\Lambda_y).
    \end{equation}
\end{proposition}

\begin{proof}
One can check that the Choi operator of a measurement channel is as follows:
\begin{equation}
\Gamma^{\mathcal{M}} = \sum_{y\in \mathcal{Y}} \Lambda_y^T \otimes |y\rangle \! \langle y|.
\end{equation}
Thus, the Doeblin coefficient for such channels is as follows:
\begin{equation}
    \alpha (\mathcal{M}) = \sup_{X \in \operatorname{Herm}} \left\{\operatorname{Tr}[X] : I \otimes X \leq \sum_{y\in \mathcal{Y}} \Lambda_y^T \otimes |y\rangle \! \langle y| \right\}.
    \label{eq:doeblin-meas-ch}
\end{equation}
Defining $\Delta(\cdot) \coloneqq \sum_{y\in \mathcal{Y}} |y\rangle \! \langle y| (\cdot) |y\rangle \! \langle y| $ and given that $\operatorname{Tr}[X] = \operatorname{Tr}[\Delta(X)]$ and
\begin{equation}
    I \otimes X \leq \sum_{y\in \mathcal{Y}} \Lambda_y^T \otimes |y\rangle \! \langle y| \quad \implies \quad I \otimes \Delta(X) \leq \sum_{y\in \mathcal{Y}} \Lambda_y^T \otimes \Delta(|y\rangle \! \langle y|) = \sum_{y\in \mathcal{Y}} \Lambda_y^T \otimes |y\rangle \! \langle y|,
\end{equation}
it suffices to optimize over diagonal $X$ in~\eqref{eq:doeblin-meas-ch}. Then the SDP reduces to
\begin{align}
    \alpha (\mathcal{M}) & = \sup_{(\mu_y \in \mathbb{R})_{y\in \mathcal{Y}}} \left\{\sum_y \mu_y :  I \otimes \sum_{y\in \mathcal{Y}} \mu_y |y\rangle\!\langle y| \leq  \sum_{y\in \mathcal{Y}} \Lambda_y^T \otimes  |y\rangle\!\langle y| \right\} \\
    & = \sup_{(\mu_y \in \mathbb{R})_{y\in \mathcal{Y}}} \left\{\sum_y \mu_y : \mu_y I  \leq  \Lambda_y^T \quad \forall {y\in \mathcal{Y}} \right\} \\
    & = \sum_{y\in \mathcal{Y}} \lambda_{\min}(\Lambda_y^T) \\
    & = \sum_{y\in \mathcal{Y}} \lambda_{\min}(\Lambda_y),
\end{align}
in analogy with the classical Doeblin coefficient, as expressed in~\eqref{eq:minimum-likelihood-decoder}.
\end{proof}

The following shows that all of the considered quantum Doeblin coefficients are equal for measurement channels.
\begin{corollary}\label{cor:measurement-channel-Doeblin-collapse}
    For a measurement channel $\cM$,
\begin{equation}
\alpha_{+}(\cM) = \alpha_{\wang}(\cM) = \alpha(\cM) = \alpha_{I}(\cM).
\end{equation}
\end{corollary}

\begin{proof}
    In general, $\alpha_{+}(\cM) \leq \alpha_{\wang}(\cM) \leq \alpha(\cM) \leq \alpha_{I}(\cM)$. The equality $\alpha_{+}(\cM) = \alpha(\cM)$ follows from \cref{prop:doeblin-of-measurement-channel}, showing that an optimizer $X$ of $\alpha(\cM)$ is positive semidefinite without loss of generality. Therefore all that remains is to show that $\alpha(\cM) \geq \alpha_{I}(\cM)$. To that end, let $X$ be an optimizer of $\alpha_{I}(\cM)$. By~\eqref{eqn:inducedDoeblinPrimal} and using $\Gamma^{\cM} = \sum_{y} \Lambda_{y}^{T} \otimes \dyad{y}$, consider that
    \begin{align}
        \bra{\alpha}\bra{\beta}(\Gamma^{\cM} - I \otimes X) \ket{\alpha}\ket{\beta}\geq 0 \, \forall \ket{\alpha}\ket{\beta} \Leftrightarrow \sum_{y} \bra{\alpha}\Lambda_{y}^{T}\ket{\alpha} \vert \langle \beta \vert y \rangle \vert^{2} - \bra{\beta}X\ket{\beta} \geq 0 \, \forall \ket{\alpha}\ket{\beta} \ .
    \end{align}
    As the above holds for an arbitrary tensor-product pure state $\ket{\alpha}\ket{\beta}$, by making the choice $\ket{\beta} = \ket{y'}$ and $\ket{\alpha}$ the minimum eigenvector of $\Lambda^{T}_{y'}$, we obtain $\lambda_{\min}(\Lambda^{T}_{y'}) \geq \bra{y'}X\ket{y'}$ for all $y' \in \cY$. Thus, using that the trace may be written as the sum over the diagonal in any basis, we have
    \begin{align}
        \alpha_{I}(\cM) = \Tr[X] = \sum_{y'} \bra{y'}X\ket{y'} \leq \sum_{y'}\lambda_{\min}(\Lambda^{T}_{y'}) = \alpha(\cM) \ , 
    \end{align}
    where the last equality follows from \cref{prop:doeblin-of-measurement-channel}. 
\end{proof}

\subsection{Entanglement-Breaking Channels}

An entanglement-breaking channel has the following form \cite{HSR03}:
\begin{equation}
    \mathcal{E}(\rho) \coloneqq \sum_x \operatorname{Tr}[\Lambda_x \rho] \sigma_x,
\end{equation}
where $(\Lambda_x)_x$ is a POVM and $(\sigma_x)_x$ is a tuple of states.

\begin{proposition}
    Let $\cN$ be an entanglement-breaking channel. Then for every channel $\cM$, we have that 
    \begin{equation}
        \alpha_{I,+}(\cN \otimes \cM) \leq \alpha_{I,+}(\cN) \alpha_{I,+}(\cM).
    \end{equation}
\end{proposition}

\begin{proof}
     Let $\cN: \cL(A) \to \cL(B)$ be an entanglement-breaking channel. Then, 
\begin{equation}
    \cN(\rho_A)= \sum_x \Tr[ \Lambda_x \rho_A] \sigma_x.
\end{equation}
Also, $\cM: \cL(A') \to \cL(B')$.
Considering $\alpha_{I,+}(\cN \otimes \cM)$ in~\cref{def:alpha_I-plus-doeblin}, a feasible solution to the primal $X_{\cN \otimes \cM}$ needs to satisfy the following  for every state $\rho_{AA'}$
\begin{equation}
    X_{\cN \otimes \cM} \leq \sum_x \Tr[ \Lambda_x \rho_A] \sigma_x \otimes \cM(\omega_x),
\end{equation}
where 
\begin{equation}
    \omega_x \coloneqq \frac{\Tr_A \!\left[(\Lambda_x \otimes I_{A'}) \rho_{AA'} \right]}{\Tr[\Lambda_x \rho_A]}
\end{equation}
Choose an optimizer $X_\cN \geq 0$ of the primal of $\alpha_{I,+}(\cN)$ and $X_\cM \geq 0$ for $\alpha_{I,+}(\cM)$. Thus, they satisfy, for all states $\rho_A$ and $\omega_{A'}$,
\begin{equation}
    X_\cN \leq \cN(\rho_A), \qquad X_{\cM} \leq \cM(\omega_{A'}).
\end{equation}
Then, consider that
\begin{align}
    \sum_x \Tr[ \Lambda_x \rho_{A}] \sigma_x\otimes \cM(\omega_x) & \geq  \sum_x \Tr[ \Lambda_x \rho_{A}] \sigma_x\otimes X_{\cM} \\
    & \geq X_\cN \otimes X_\cM.
\end{align}
With this, we have that $X_\cN \otimes X_\cM \leq (\cN \otimes \cM) (\rho_{AA'})$ for all $\rho_{AA'}$, leading to $\Tr[X_{\cN} \otimes X_{\cM}] \leq \alpha_{I,+}(\cN \otimes \cM)$, so as to arrive at 
\begin{equation}
    \alpha_{I,+}(\cN)  \alpha_{I,+}(\cM)  \leq  \alpha_{I,+}(\cN \otimes \cM),
\end{equation}
by recalling that $X_\cN$ and $X_\cM$ are optimizers for the primal of $\alpha_{I,+}(\cN)$ and $\alpha_{I,+}(\cM)$.
\end{proof}

\subsection{Generalized Dephasing Channels}

Define the generalized qudit dephasing channel $  \cD$ as follows:
\begin{equation}
    \cD(\rho)\coloneqq   \sum_{i=0}^{d-1} \langle i| \rho| j \rangle  |i \rangle\! \langle j| \langle\psi_j | \psi_i\rangle, 
\end{equation}
where $(|\psi_i\rangle)_i$ is a tuple of pure quantum states. An isometric extension of this channel is specified by $|i\rangle \to |i\rangle |\psi_i\rangle$.
Then the Choi operator of the generalized dephasing channel is as follows: 
\begin{equation}
\label{eq:Choi_dephasing}
    \Gamma^{  \cD} =  \sum_{i , j=1}^d \langle\psi_j | \psi_i\rangle |i\rangle\!\langle j| \otimes |i\rangle\!\langle j|   .
\end{equation}

\begin{proposition}[Doeblin Coefficients for the Generalized Dephasing Channel]
For  a generalized dephasing channel $\cD$, the following equalities hold:
\begin{align}
    \alpha(\cD) &=0, \label{eq:dephasing_alpha} \\
    \alpha_{I}(\cD) &= 0.
\end{align}
\end{proposition}

\begin{proof}
The equality $\eta_{\Tr}(\cD) = 1$ holds because
\begin{equation}
    \Vert \cD(\dyad{0}) - \cD(\dyad{1}) \Vert_{1} = \Vert \dyad{0} - \dyad{1} \Vert_{1}. 
\end{equation}
As $\eta_{\Tr}(\cD) \leq 1 - \alpha(\cD)$ (see \eqref{eq:contraction-coeff-to-complete} and \cref{prop:trace_distance_complete_cont_Doeblin_bound}), we  conclude that $\alpha(\cD) = 0$. Furthermore, by~\eqref{eq:trace-distance-CC-to-induced-doeblin}, we conclude that $\alpha_I(\cD)=0$.
\end{proof}

\subsection{Qubit Channels}

In this section, our main purpose is to analyze Doeblin coefficients and contraction and expansion coefficients for the trace distance for qubit channels.
Before doing so, let us first analyse conditions such that the Doeblin coefficient is non-trivial (i.e., $\alpha(\cN) >0$) for general channels. Then we study the conditions such that $\alpha(\cN) >0$ for qubit channels.

We begin by noting whenever the Choi matrix is full rank, the quantum Doeblin coefficient~$\alpha(\mathcal{N})$ is strictly greater than zero; i.e., this guarantees that the channel is contractive.
\begin{proposition} Let $\cN\colon \cL(\cH_A) \to \cL(\cH_B)$ be a quantum channel. 
    If $\Gamma^{\cN}>0$, then $\alpha(\cN) \geq d_{B}\lambda_{\min}(\Gamma^{\cN}) > 0$. 
\end{proposition}
\begin{proof}
    If $\Gamma^{\cN} > 0$, then $t \coloneqq  \lambda_{\min}(\Gamma^{\cN}) > 0$. Then letting $X = t I_B$, we have for all $\ket{\omega} \in A \otimes B$, $
      0 <  \bra{\omega}I_A \otimes X\ket{\omega} = t = \lambda_{\min}(\Gamma^{\cN}) \leq \bra{\omega} \Gamma^{\cN}\ket{\omega}$. Thus, $X$ is feasible for the optimization in~\eqref{eq:alpha_cN_def} and has $\Tr[X] = d_{B}\lambda_{\min}(\Gamma^{\cN})$. As~\eqref{eq:alpha_cN_def} involves a supremization, this completes the proof.
\end{proof}
 One may take being full rank as a qualitative sign that the channel disturbs a quantum state strongly. A natural question is if, with more structure, more can be said about when the quantum Doeblin coefficient is non-trivial. Indeed, we are able to establish that for qubit channels, the quantum Doeblin coefficient is `faithful' in the sense that it is non-zero if and only if the trace-distance contraction coefficient is not equal to one.
 
 \begin{theorem}
 \label{thm:qubit-faithfulness}
For a qubit channel $\cN$, the following are equivalent:
\begin{itemize}
    \item $\alpha(\cN)>0$, 
    \item $\eta_{\Tr}(\cN)<1$. 
\end{itemize}
\end{theorem}

To establish \cref{thm:qubit-faithfulness}, we will rely on a variety of parameterizations of qubit channels, which are reviewed in \cref{sec:qubit-channel-parameterizations} for clarity. Of particular import will be the Stokes parameterization, which tells us that any qubit channel can be expressed as an affine
transformation on the Bloch sphere according to $w_0 I +w\cdot\vec{\sigma} \rightarrow w_0 I+(t+Tw)\cdot\vec{\sigma}$.
The proof of \cref{thm:qubit-faithfulness} breaks into two steps. First, we find necessary and sufficient conditions for the quantum Doeblin coefficient to be strictly non-zero in terms of the Stokes parameterization of the channel (\cref{lem:qubit-stokes-nec-and-suff}). Second, we establish necessary and sufficient conditions for the trace distance contraction coefficient to be strictly less than one in terms of the Stokes parameterization, by slightly strengthening a result of~\cite{Hiai15} (\cref{Lem:cc-tr-qubit}). Lastly, we establish these two claims in fact select for the same conditions on the Stokes parameters, which will complete the proof of \cref{thm:qubit-faithfulness}. A key idea in both lemmas is to realize that the `non-unital part' of a qubit channel is irrelevant for determining the  value.

\cref{lem:qubit-stokes-nec-and-suff} below establishes necessary and sufficient conditions for $\alpha(\cN) > 0$. While in effect purely mathematical, we can give a physical explanation for what the proof shows. The proof is made up of two reductions. The first reduction shows that only the `unital part' of the qubit channel $\cN$ is relevant for computing $\alpha(\cN)$. What we mean by this is that, under some pre- and post-processing, the linear shift of the Paulis that a qubit channel can realize does not affect the value of the contraction coefficient. The second reduction in effect shows that the `unital part' of the channel is contractive if and only if it is a non-trivial mixture of all four Pauli channels, i.e., $\Phi(\cdot) = p_{0}\operatorname{id}(\cdot) + p_{1}X \cdot X^{\dagger} + p_{2}Y \cdot Y^{\dagger} + p_{3}Z \cdot Z^{\dagger}$ where $p_{0},p_{1},p_{2},p_{3} > 0$. We stress that this is not the same as saying that a unital channel is contractive if and only if it is a Pauli channel with all four Paulis, because the first step in the reduction involves pre- and post-processing to convert the channel to have a specific form in the Stokes representation.

\begin{lemma}
\label{lem:qubit-stokes-nec-and-suff}
    Let $\cN$ be a qubit channel given by the transformation $w_0 I +w\cdot\vec{\sigma} \rightarrow w_0 I+(t+Tw)\cdot\vec{\sigma}$. Let $\cN'$ be the unitary pre- and post-processed form of $\cN$ such that
    \begin{align}\label{eq:processed-qubit-channel}
        \Gamma^{\cN'} = \frac{1}{2}\left[I_A \otimes I_B + I_{A} \otimes \vecp{r} \cdot \vec{\sigma} + \sum_{k} t_{k}' \sigma_{k} \otimes \sigma_{k}\right] \ .
    \end{align}
    Then, $\alpha(\cN) > 0$ if and only if $\vec{t}'$ is a full support convex combination of the set 
    \begin{equation}
     \{(1,1,1), (1,-1,-1),(-1,1,-1),(-1,-1,1)\} \ .
    \end{equation}
\end{lemma}

\begin{proof}
    By \cref{prop:DPI-and-LU-invariance}, without loss of generality, we can focus on any version of $\cN$ that has undergone unitary pre- and post-processing. We thus choose the unitary pre- and post-processing from \cref{prop:simplified-choi-under-processing} so that, without loss of generality, we just need to consider $\alpha(\cN')$ where 
    \begin{align}\label{eq:simplified-correlation-matrix}
    \Gamma^{\cN'} = \frac{1}{2}\left[I_{A} \otimes I_{B} + I_{A} \otimes \vecp{r} \cdot \vec{\sigma} + \sum_{k} t_{k}' \sigma_{k} \otimes \sigma_{k} \right] \ .
    \end{align}
    Now, note that by~\eqref{eq:alpha-cN-dual} any feasible $Y$ is the Choi operator of a unital CP map by the Choi isomorphism. Thus, by~\eqref{eq:Choi-Unital-parameterize}, we may express $Y$ in the form
    \begin{align}
        Y = \frac{1}{2}\left[I_{A} \otimes I_{B} + \vec{s} \cdot \vec{\sigma} \otimes I_{B} + \sum_{i,j} b_{i,j} \sigma_{i} \otimes \sigma_{j} \right] \  
    \end{align}
    and the positivity constraints are $\vert B_{11} \pm B_{22} \vert \leq \vert 1 + B_{33} \vert$ by \cref{prop:unital-positivity-constraints}. Now, 
    \begin{align}
        \Tr[\Gamma^{\cN'}Y] & = \frac{1}{4}\Tr\Bigg[\left(I_{A} \otimes I_{B} + I_{A} \otimes \vecp{r} \cdot \vec{\sigma} + \sum_{k} t_{k}' \sigma_{k} \otimes \sigma_{k}\right) \notag \\
        & \hspace{2cm} \cdot \left(I_{A} \otimes I_{B} + \vec{s} \cdot \vec{\sigma} \otimes I_{B} + \sum_{i,j} b_{i,j} \sigma_{i} \otimes \sigma_{j}\right) \Bigg] \\
        & = \frac{1}{4}\left(4 + 4\sum_{k} t_{k}'b_{k,k}\right) \ ,
    \end{align}
    where the second equality uses the fact that Pauli operators are traceless and the commutation relations between Pauli operators to simplify. Thus, by~\eqref{eq:alpha-cN-dual} and the positivity constraints being $\vert B_{11} \pm B_{22} \vert \leq \vert 1 + B_{33} \vert$, we find that
    \begin{align}
        \alpha(\cN') = 1 + \inf\left\{ \sum_{k} t_{k}'b_{k,k} : \vert b_{1,1} \pm b_{2,2} \vert \leq \vert 1 \pm b_{3,3} \vert \right \} \ . \label{eq:alpha_qubit_simplified_b}
    \end{align}
    Note this shows that only the unital part of $\cN'$ and $Y$ (if we consider it also as a Choi operator) matter. In other words, we have reduced the problem to showing $\alpha$ is faithful on qubit channels to showing it is faithful on the equivalence class of unital qubit channels.\footnote{The equivalence class is with respect to unitary pre- and post-processing.} For clarity, we rewrite the optimization for (presumably CP) unital $\cM$:
    \begin{align}
        \alpha(\cM) = 1 + \inf \left\{ \sum_{k} t_{k}b_{k} : \vert b_{1} \pm b_{2} \vert \leq \vert 1 \pm b_{3} \vert \right\} \ . \label{eq:alpha-reduction-step-1}
    \end{align}
    
    We are now just interested in what values of $\vec{t}$ satisfying the CP constraints in~\eqref{eq:positivity-conditions} are such that the infimum equals unity. Using $\eta_{\Tr}(\cN) \leq 1 - \alpha(\cN) = 1 - (1 +\inf\{ \sum_{k} t_{k}b_{k} : \vert b_{1} \pm b_{2} \vert \leq \vert 1 \pm b_{3} \vert \})$, we are interested in when the infimum is equal to negative one.
    
    As established in~\cite{King-2000a}, 
    the conditions on the vector $\vec{b}$, which are positivity constraints, means that it lies in the tetrahedron defined by extreme points $\{(1,1,1),$ $(1,-1,-1),$ $(-1,1,-1),$ $(-1,-1,1)\}$. By convexity, it will always be optimized on one of these extreme points. So we are in principle interested in
    \begin{align}\label{eq:satisfying-set}
        \min\left\{ \sum_{k} t_{k}, t_{1} - (t_{2}+t_{3}), t_{2} - (t_{1}+t_{3}), t_{3} - (t_{1}+t_{2}) \right\} = -1  
    \end{align}
    over the $\vec{t}$ satisfying
    $\vert t_{1} \pm t_{2} \vert \leq \vert 1 \pm t_{3} \vert$. Again, this means 
    \begin{align}
        \vec{t} \in \operatorname{conv}(\{(1,1,1),(1,-1,-1),(-1,1,-1),(-1,-1,1)\}) \ .
    \end{align}
    We may then consider $\vec{t}$ as a convex combination of these points, which using $0 \leq p_{0},p_{1},p_{2}$ such that $p_{0}+p_{1}+p_{2} \leq 1$, we can write 
    \begin{align}\label{eq:vec-t-as-convex-comb}
        \vec{t} \coloneq \begin{bmatrix} 2(p_0 + p_1) - 1 \\
        2(p_0 + p_2) - 1 \\ 
        1 - 2(p_1 + p_2) 
        \end{bmatrix} \ ,
    \end{align}
    where the parameterization is implicit. We are then interested when each function in~\eqref{eq:satisfying-set} takes the value negative one as a function of $p$. Direct calculations lead to
    \begin{align}
        \sum_{k} \vec{t}_{k} = -1 \Leftrightarrow p_{0} = 0, &\quad
        t_{1} - (t_{2}+t_{3}) \Leftrightarrow p_{1} = 0, \\
        t_{2} - (t_{1}+t_{3}) \Leftrightarrow p_{2} = 0, &\quad  t_{3} - (t_{1}+t_{2}) \Leftrightarrow p_{0}+p_{1}+p_{2} = 0 \ .
    \end{align}
    In other words, the set of $\vec{t}$ \textit{not} satisfying~\eqref{eq:satisfying-set} are those that are a mixture of \textit{all} extreme points.
\end{proof}

The first part of the previous lemma reduced the consideration of $\alpha(\cN)$ for qubit channels to the `unital part.' The following lemma shows that something similar holds for the trace distance contraction coefficient.

\begin{lemma}\label{Lem:cc-tr-qubit}\label{lem:qubit_trace_contr}
    Let $\cN$ be a qubit channel given by the transformation $w_0 I +w\cdot\vec{\sigma} \rightarrow w_0 I+\left(t+Tw\right)\cdot\vec{\sigma}$. Then we have 
    \begin{align}
        \eta_{\Tr}(\cN)=\|T\|_\infty.
    \end{align}
\end{lemma}
\begin{proof}
    The same statement was proven in~\cite{Hiai15} for unital maps. It solely remains to show that for every channel $\cN$, the non-unital part $t$ does not contribute to the contraction coefficient, which is clear by direct calculation:
    \begin{align}
        \eta_{\Tr}(\cN)&= \frac{1}{2} \sup_{w,w'} \left\| w_0 I+(t+Tw)\cdot\vec{\sigma} - (w_0 I +(t+Tw')\cdot\vec{\sigma}) \right\|_1 \\
        &= \frac{1}{2} \sup_{w,w'} \left\| Tw \cdot \vec{\sigma} -Tw'\cdot\vec{\sigma} \right\|_1 \ \\
        &= \frac{1}{2} \sup_{w,w’} \left \Vert T(w - w’) \right\Vert \\
        &= \sup_{w} \left\Vert Tw \right\Vert \\
        &= \left\Vert T \right \Vert_{\infty} \ ,
    \end{align}
where $\left\Vert w \right\Vert, \left\Vert w' \right\Vert \leq 1$, as we optimize over quantum states and the final equality uses the optimal choice $w'= -w$ with $\Vert w' \Vert = 1$ and $\Vert X \Vert_{1} = \Tr \vert X \vert$ to simplify. 
\end{proof}

Finally, we use the above lemma to show that the conditions for $\eta_{\Tr}(\cN) < 1$ are the same as $\alpha(\cN) > 0$. 
\begin{proof}[Proof of \cref{thm:qubit-faithfulness}] 
By the unitary invariance of the trace-distance contraction coefficient, we may convert $\cN$ to $\cN'$, where the latter has a Choi operator of the form in~\eqref{eq:processed-qubit-channel}. That is to say, $\cN'$ is given by the transformation $w_0 I +w\cdot\vec{\sigma} \rightarrow w_0 I+(\vecp{r}+T')\cdot\vec{\sigma}$ where $T' = \operatorname{diag}(t'_{1},t'_{2},t_{3}')$. It follows that $\eta_{\Tr}(\cN) = \eta_{\Tr}(\cN') = \Vert T' \Vert_{\infty} =  \max\{t_{1}',t_{2}',t_{3}'\} \eqqcolon \vecp{t}$, where the second equality used \cref{Lem:cc-tr-qubit} and the third used the definition of the operator norm.
    
    Re-parameterizing $\vecp{t}$ in terms of the extreme points of the CP qubit maps as was done in~\eqref{eq:vec-t-as-convex-comb}, we have that the contraction coefficient is the maximal entry of
        \begin{align}
        \vecp{t} \coloneq \begin{bmatrix} 2(p_0 + p_1) - 1 \\
        2(p_0 + p_2) - 1 \\ 
        1 - 2(p_1 + p_2) 
        \end{bmatrix}  ,
    \end{align}
    where $0 \leq p_0, p_1, p_2$ and $p_0 + p_1 + p_2 \leq 1$. It is now easy to see that any of the entries being one, is equivalent to at least one coefficient being zero, e.g.
    \begin{align}
        2(p_0 + p_1) - 1 =1 \Leftrightarrow p_0+p_1& =1 \\
        \Rightarrow p_2& =0. 
    \end{align}
    Hence the support of the corresponding convex combination can not be full. These are the same conditions that \cref{lem:qubit-stokes-nec-and-suff} specifies for $\alpha(\cN) > 0$. Thus the conditions are equivalent. 
\end{proof}

\begin{proposition}[Doeblin Coefficients for Qubit Channels]
    Let $\cN$ be a qubit channel given by the transformation $w_0 I +w\cdot\vec{\sigma} \rightarrow w_0 I+(t+Tw)\cdot\vec{\sigma}$. Then, we have that 
\begin{equation}
    \alpha(\cN)= 1- \lambda_1(T)-\lambda_2(T) +\lambda_3(T),
\end{equation}    
where $\lambda_1(T)$, $\lambda_2(T)$, $\lambda_3(T)$ are the ordered eigenvalues of matrix $T$ (i.e., such that $\lambda_1(T) \geq \lambda_2(T) \geq \lambda_3(T)$). 

Furthermore, 
\begin{equation}
    \alpha_I(\cN)=1-\lambda_1(T) =1- \|T\|_\infty= 1-\eta_{\Tr}(\cN).
\end{equation}
\end{proposition}

\begin{proof}
    By~\eqref{eq:alpha_qubit_simplified_b} and the assumption on the channel $\cN$ such that the following transformation happens $w_0 I +w\cdot\vec{\sigma} \rightarrow w_0 I+(t+Tw)\cdot\vec{\sigma}$, we have that 
    \begin{align}
        \alpha(\cN) = 1 + \inf\left\{ \sum_{k} t_{k}b_{k,k} : \vert b_{1,1} \pm b_{2,2} \vert \leq \vert 1 \pm b_{3,3} \vert \right \}  , 
    \end{align}
where $T=\operatorname{diag}(b_{1,1}, b_{2,2}, b_{3,3})$ and $t=(t_1,t_2,t_3)^T$.
  As established in~\cite{King-2000a}, 
    the conditions on the vector $\vec{b}$, which are positivity constraints, means that it lies in the tetrahedron defined by the extreme points $\{(1,1,1),$ $(1,-1,-1),$ $(-1,1,-1),$ $(-1,-1,1)\}$. Then, the infimum is achieved when negative weights ($-1$) are assigned for the two largest eigenvalues and positive weight ($+1$) is assigned for the smallest eigenvalue. This leads to the desired expression of 
    \begin{equation}
        \alpha(\cN)= 1- \lambda_1(T)-\lambda_2(T) +\lambda_3(T).
    \end{equation}

To prove the next statement, we employ ideas presented around~\cite[Eq.~(41)]{chitambar2022communication}. Recall the dual program for $\alpha_I(\cN)$ in~\cref{prop:induced-Doeblin}.
Note that a separable operator on two qubits can be generally written as $Y= \sum_k \alpha_k^T \otimes \beta_k$. Then, to demand the constraints of the dual, we need $\Tr[\alpha_k]=1$ and $\sum_k \beta_k=I$. With that we also choose $\alpha_k \coloneqq \frac{1}{2} (I + a_k \cdot \vec{\sigma})$ and $\beta_k \coloneqq \gamma_k (I + b_k\cdot \vec{\sigma})$ with $\sum_k \gamma_k b_k =0$, $\sum_k \gamma_k =1$, and $\|b_k\|_2, \|a_k\|_2 \leq 1$. With these consider, 
\begin{align}
    \Tr[ \Gamma^\cN Y] &= \Tr\!\left[ \Gamma^\cN \left( \sum_k \alpha_k^T \otimes \beta_k \right) \right ] \\ & = \sum_k \Tr\!\left[ \beta_k \cN(\alpha_k) \right] \\
    & =\sum_k \frac{\gamma_k}{2} \Tr\!\left[ (I + b_k\cdot \vec{\sigma}) \left(I+ (Ta_k +t) \cdot \vec{\sigma} \right) \right] \\
    &=1+ \sum_k \gamma_k b_k \cdot  (Ta_k +t) \\
    &= 1+ \sum_k \gamma_k b_k^T Ta_k,
\end{align}
where the last equality follows due to the condition $\sum_k \gamma_k b_k=0$. To find $\alpha_I(\cN)$ we need to infimize $1+ \sum_k \gamma_k b_k^T Ta_k$, which leads to infimizing $ \sum_k \gamma_k b_k^T Ta_k$. Then, the function of our interest is lower bounded by $\sum_k \gamma_k b_k^T Ta_k \geq 1- \lambda_1(T)$ since $\sum_k \gamma_k =1$ and $b_k^T T a_k \geq -\sigma_{\max}(T)=-\lambda_1(T)$, where $\sigma_{\max}(\cdot)$ corresponds to the largest singular value.  The inequality there is achieved by the following choice: choose $\gamma_k=1/2$ for $k=\{1,2\}$ and zero otherwise; $b_1= -b_2=-a_1=a_2=-u_1$,  where $u_1$ is the eigenvector corresponding to the maximum eigenvalue of $T$ (i.e., corresponding to $\lambda_1(T)$). Also, note that the above choice satisfies the constraints $\sum_k \gamma_k b_k=0$, $\sum_k \gamma_k=1$, and $\|b_k\|_2, \|a_k\|_2 \leq 1$.
Then, we arrive at
\begin{equation}
    \alpha_I(\cN)= 1-\lambda_1(T)= 1- \|T\|_\infty,
\end{equation}
concluding the proof by also recalling~\cref{lem:qubit_trace_contr}.
\end{proof}

\begin{proposition} \label{prop:n-noise-qubit}
    Let $\cN$ be a qubit channel given by the transformation $w_0 I +w\cdot\vec{\sigma} \rightarrow w_0 I+(t+Tw)\cdot\vec{\sigma}$, and let $n \in \mathbb{N}$. 
    Then, we have that, for all $n$-qubit states $\rho$ and $\sigma$,
    \begin{equation}
        T\!\left( \cN^{\otimes n}(\rho), \cN^{\otimes n}(\sigma)\right) \leq 4 n \|T\|_\infty \, T(\rho,\sigma).
    \end{equation}
As such, we have that 
\begin{equation}
    \eta_{\Tr}\!\left( \cN^{\otimes n}\right) \leq 4  n \|T\|_\infty.
\end{equation}
\end{proposition}

\begin{proof}
First, recall that all quantum channels have a quantum state that is a fixed point \cite[Theorem~6.11]{wolf2012quantum}. With that, let $\tau$ be a state that is a fixed point of the qubit channel $\cN$ (i.e.; $\cN(\tau)= \tau$). With that, consider that, for an arbitrary qubit state~$\omega$,
\begin{align}
    \left\| \cN(\omega)- \tau \right\|_1 &=  \left\| \cN(\omega)- \cN(\tau) \right \|_1  \\
    & \leq \left\|T\right\|_\infty \left\|\omega-\tau\right\|_1 \\
    & \leq 2 \left\|T \right\|_\infty, \label{eq:trace_norm_bound_tau}
\end{align}
where we use~\cref{Lem:cc-tr-qubit} for the first inequality.

By \cite[Proposition~11]{Palma_wasserstein21}, we have that 
\begin{equation}
    \left\| \cN^{\otimes n}\right\|_{W_1} \leq  \sup_{\rho_{RA}} \left\| (\id_R \otimes  (\cN-\cE))(\rho_{RA}) \right\|_1
\end{equation}
where $\cE$ is a replacer channel that acts as $\cE(X)= \Tr[X] \tau$, and for a channel $\cP$,
\begin{equation} \label{eq:waase_cont}
     \left\| \cP \right\|_{W_1} \coloneqq  \sup_{\rho \neq \sigma}  \frac{\left\|\cP(\rho-\sigma) \right\|_{W_1}}{\left\|\rho-\sigma\right\|_{W_1}},
\end{equation}
where $\left\|\cdot\right\|_{W_1}$ is the quantum Wasserstein distance of order 1 \cite[Definition~7]{Palma_wasserstein21}.
Then using~\cref{lem:diamond_to_one_norm} with $d_A=d_B=2$ therein, we arrive at
\begin{equation}
    \left\| \cN^{\otimes n}\right\|_{W_1} \leq 2  \sup_{\omega} \left\|   (\cN-\cE)(\omega) \right\|_1.
    \label{eq:w1-bound-to-trace-norm}
\end{equation}
See also \cite[Proposition~33]{mele2024noise}.

Combining \eqref{eq:w1-bound-to-trace-norm} with the previous bound on the trace norm with the fixed state $\tau$ in~\eqref{eq:trace_norm_bound_tau}, we see that 
\begin{equation}\label{eq:wassers_n_qubit_cont}
    \left\| \cN^{\otimes n}\right\|_{W_1} \leq 4  \left\|T\right\|_\infty.
\end{equation}

Also recall that, from~\cite[Proposition~2]{Palma_wasserstein21}, for a traceless Hermitian operator $X$, 
\begin{equation}
    \frac{1}{2} \left\|X\right\|_1 \leq \left\|X\right\|_{W_1} \leq \frac{n}{2} \left\|X\right\|_1.
\end{equation}
Utilizing this together with
\begin{equation}
    \left\| \cN^{\otimes n}(\rho-\sigma) \right\|_{W_1} \leq \left\| \cN^{\otimes n}\right\|_{W_1} \left \|\rho-\sigma \right\|_{W_1} \leq 4 \left\| T\right\|_{\infty} \left \|\rho-\sigma \right\|_{W_1},
\end{equation}
we arrive at 
\begin{equation}
  \frac{1}{2}  \left\|  \cN^{\otimes n}(\rho-\sigma) \right \|_1 \leq 4 \left\| T\right\|_\infty \cdot \frac{n}{2} \left\|\rho-\sigma\right\|_1,
\end{equation}
concluding the proof, together with $T(\rho,\sigma)= \left\|\rho-\sigma\right\|_1/2$ and the definition of $\eta_{\Tr}(\cdot)$.
\end{proof}

\begin{lemma}
\label{lem:diamond_to_one_norm} Let $\cN$ and $\cM$ be two quantum channels that map $\cL(\cH_A) \to \cL(\cH_B)$, with $d_A$ and $d_B$ the dimensions of the input and the output systems, respectively. Then, we have that 
\begin{equation}
   \sup_{\rho_{RA}} \left\| (\id_R \otimes  (\cN-\cM))(\rho_{RA}) \right\|_1 \leq \sqrt{d_A d_B}  \sup_{\rho} \left\|   (\cN-\cM)(\rho) \right\|_1,
\end{equation}   
where system $R$ is isomorphic to system $A$ (i.e., $d_R=d_A$).
\end{lemma}

\begin{proof}
    Consider that
    \begin{align}
       \sup_{\rho_{RA}} \left\| (\id_R \otimes  (\cN-\cM))(\rho_{RA}) \right\|_1 &=   \sup_{\rho_{RA}} \frac{\left\| (\id_R \otimes  (\cN-\cM))(\rho_{RA}) \right\|_1}{\| \rho_{RA}\|_1} \\ 
       & \leq \sqrt{d_A d_B} \sup_{\rho_{RA}} \frac{\left\| (\id_R \otimes  (\cN-\cM))(\rho_{RA}) \right\|_2}{\| \rho_{RA}\|_1} \\
       & = \sqrt{d_A d_B}  \sup_{\rho} \frac{\left\| (\cN-\cM)(\rho) \right\|_2}{\| \rho\|_1} \\ 
       & \leq  \sqrt{d_A d_B}  \sup_{\rho} \frac{\left\| (\cN-\cM)(\rho) \right\|_1}{\| \rho\|_1} \\
       &= \sqrt{d_A d_B}   \sup_{\rho} \left\|   (\cN-\cM)(\rho) \right\|_1,
    \end{align}
    where the first inequality follows because $\left\|A\right\|_1 \leq \sqrt{d} \left\|A\right\|_2$ with $d=d_A \cdot d_B$ the dimension of~$A$; the second equality follows from the stability of the superoperator norms induced by the two norm, by using~\cite[Theorem~4]{watrous2004notes} with the choice $p=2 \geq 2$ and $q=1 \leq 2$ therein; and the last inequality again by the change of two norm to one norm using $\left\|A\right\|_2 \leq \left\|A\right\|_1$.
\end{proof}

\begin{remark}[Comparison to Existing Work]
For a qubit channel $\cN$ specified by the transformation $w_0 I +w\cdot\vec{\sigma} \rightarrow w_0 I+(t+Tw)\cdot\vec{\sigma}$, by applying the techniques used in the proof of~\cite[Proposition~3]{mele2024noise}, 
one can obtain that 
\begin{equation}
    T\!\left( \cN^{\otimes n}(\rho), \cN^{\otimes n}(\sigma) \right) \leq 8 n  \max_{p \in \{X,Y,Z\}} \left | \frac{D_p}{D_p -1} \right| T(\rho,\sigma),
\end{equation}
where $T=\operatorname{diag}(D_X, D_Y,D_Z)$.
Comparing it with~\cref{prop:n-noise-qubit}, we see that \cref{prop:n-noise-qubit} is tighter since $\|T\|_\infty \leq  \max_{p \in \{X,Y,Z\}} \left | \frac{D_p}{D_p -1} \right| $. With that~\cref{prop:n-noise-qubit} provides a non-trivial worst-case bound for the trace-distance contraction for a larger set of noise channels than previously known. Moreover, we improve the contraction coefficient of Wasserstein distance of order 1 of $\cN^{\otimes n}$ defined in~\eqref{eq:waase_cont} given in~\cite[Lemma~32]{mele2024noise} by establishing~\eqref{eq:wassers_n_qubit_cont}.
\end{remark}

Next, we provide an analytical expression for the expansion coefficient in~\eqref{eq:expansion_trace} for qubit channels.
\begin{lemma}[Expansion Coefficient for Qubit Channels]\label{prop:expansion_coeff_qubit}

   Let $\cN$ be a qubit channel specified by the transformation $w_0 I +w\cdot\vec{\sigma} \rightarrow w_0 I+(t+Tw)\cdot\vec{\sigma}$. Then, we have that 
   \begin{equation}
       \check{\eta}_{\Tr}(\cN)= \sigma_{\min}(T),
   \end{equation}
where $\check{\eta}_{\Tr}$ is defined in~\eqref{eq:expansion_trace} and $\sigma_{\min}(A)$ is the minimum singular value of the operator $A$.
\end{lemma}

\begin{proof}
    By recalling the definition of $\check{\eta}_{\Tr}$ in~\eqref{eq:expansion_trace}, we have 
    \begin{align}
        \check{\eta}_{\Tr}(\cN) = \inf_{\rho_1,\rho_2: T(\rho_1,\rho_2) \neq 0} \frac{T\!\left( \cN(\rho_1), \cN(\rho_2) \right)}{T(\rho_1,\rho_2)}.
    \end{align}
Also denote the qubit states as follows for $i \in \{1,2\}$:
\begin{align}
    \rho_i & = \frac{1}{2} (I + w_i\cdot  \vec{\sigma}), \\ 
  \cN(\rho_i) & =  \frac{1}{2} \left(I + (t+ T w_i)\cdot  \vec{\sigma} \right). 
\end{align}
Then, we arrive at the following:
\begin{align}
   \check{\eta}_{\Tr}(\cN) & = \inf_{w_1, w_2: \left\| w_1- w_2\right\|_2 \neq 0}  \frac{\left\| T( w_1 -w_2) \right\|_2}{\left\| w_1- w_2\right\|_2} \\
   &=\inf_{w': \left\| w'\right\|_2 \neq 0}  \frac{\left\| Tw' \right\|_2}{\left\| w'\right\|_2} \\ 
   &= \sigma_{\min}(T),
\end{align}
where the first equality follows by substituting qubit representations to the trace distance and using the fact that $T(\rho_1,\rho_2)= \left\| w_1- w_2\right\|_2$. 
\end{proof}

\subsection{Generalized Amplitude-Damping Channels}

Consider the generalized amplitude damping (GAD) channel $\cN_{p,\eta}$, defined for $p \in [0,1]$ and $\eta \in [0,1]$ in terms of the following Kraus operators:
\begin{align} \label{eq:kraus_GADC}
    &K_1 \coloneqq \sqrt{p} \begin{bmatrix}
        1 &  0 \\
        0 & \sqrt{\eta}
    \end{bmatrix},  \quad K_2\coloneqq \sqrt{p} \begin{bmatrix}
        0 & \sqrt{1-\eta} \\
        0 &  \quad 0
    \end{bmatrix}, \\ & K_3\coloneqq \sqrt{1-p} \begin{bmatrix}
        \sqrt{\eta} &  0  \\
        0 &   1
    \end{bmatrix}, \quad K_4\coloneqq \sqrt{1-p} \begin{bmatrix}
        0 &  0 \\
        \sqrt{1-\eta} &   0
    \end{bmatrix}, \label{eq:kraus_GADC_rest}
\end{align}
such that $\cN_{p,\eta}(\cdot) \coloneqq \sum_{i=1}^4 K_i (\cdot) K_i^\dag$. See~\cite[Sections~III-C and III-D]{KSW2020} for various interpretations and discussions of the generalized amplitude damping channel.

\begin{lemma}
\label{lem:Doeblin_GAD}
     For $p \in [0,1]$ and $\eta \in [0,1]$, the Doeblin coefficient of the generalized amplitude damping channel is given by
    \begin{align}
        \alpha(\cN_{p,\eta}) = (1-\sqrt{\eta})^2. 
        \label{eq:doeblin-coeff-GADC}
    \end{align}
\end{lemma}

\begin{proof}
    By definition of the Doeblin coefficient, i.e., using the expression in~\eqref{eq:alpha_cN_def}, it is sufficient to find a feasible Hermitian operator $X$ in order to give a lower bound on the coefficient. We choose
    \begin{align}
        X = \begin{bmatrix}
            p + (1-p)\eta-\sqrt{\eta} & 0 \\
            0 & p\eta+(1-p)-\sqrt{\eta}
        \end{bmatrix}.
    \end{align}
    For this choice of $X$ to be feasible, the following matrix inequality should hold:
    \begin{align}
        0 \leq \Gamma^{\cN}_{AB} - I_A\otimes  X_B = \begin{bmatrix}
            \sqrt{\eta} & 0 & 0 & \sqrt{\eta} \\
            0 & \sqrt{\eta}-\eta & 0 & 0 \\
            0 &  0 & \sqrt{\eta}-\eta & 0 \\
            \sqrt{\eta} & 0 & 0 & \sqrt{\eta}
        \end{bmatrix}. 
    \end{align}
    It can easily be seen to hold, because $0\leq\eta\leq1$. The Choi matrix for the generalized amplitude damping channel can be found in~\cite[Eq.~(105)]{KSW2020} and~\cite[Appendix C]{hirche2024quantum}. The following inequality thus holds: 
    \begin{align}
        \alpha(\cN_{p,\eta}) \geq \tr [X] = (1-\sqrt{\eta})^2.
        \label{eq:doeblin-coeff-GADC-LB}
    \end{align}
    
   To arrive at the opposite inequality, let us employ the dual program in~\eqref{eq:alpha-cN-dual} with the specific choice  of 
\begin{equation}\label{eq:Doeblin_Y_GAD}
   Y_{AB}= \begin{bmatrix}
        1 & 0 & 0 & -1 \\
         0 & 0 & 0 & 0 \\
          0 & 0 & 0 & 0 \\
           -1 & 0 & 0 & 1 
    \end{bmatrix}.
\end{equation}
Here, we have that $\Tr_A[Y_{AB}] = I$ and $Y_{AB} \geq 0$ (eigenvalues being $0$ and $2$), making it a feasible solution for the dual program. For that choice,
\begin{equation}
    \Tr\!\left[\Gamma^{\cN_{p,\eta}}_{AB} Y_{AB}\right] = p+ (1-p) \eta - 2\sqrt{\eta} + p\eta +(1-p) = 1- 2\sqrt{\eta} + \eta =\left (1-\sqrt{\eta} \right)^2
\end{equation}
so that 
\begin{equation}
   \left (1-\sqrt{\eta} \right)^2  \geq \alpha(\cN_{p,\eta}).
        \label{eq:doeblin-coeff-GADC-UB}
\end{equation}
Combining~\eqref{eq:doeblin-coeff-GADC-LB} and~\eqref{eq:doeblin-coeff-GADC-UB}, we conclude~\eqref{eq:doeblin-coeff-GADC}.
\end{proof}

\cref{lem:Doeblin_GAD} implies the following inequality:
\begin{align}
    \eta_{\tr}(\cN_{p,\eta}) \leq 1-(1-\sqrt{\eta})^2. 
\end{align}
The right-hand side is strictly smaller than one for all $\eta<1$, and it is independent of $p$.

\begin{lemma} \label{lem:b_n_GADC}
         For $p \in [0,1]$ and $\eta \in [0,1]$, the quantity $\alpha_{\wang}(\cN_{p,\eta})$ for the generalized amplitude damping channel $\cN_{p,\eta}$ is given by
    \begin{align}
        \alpha_{\wang}(\cN_{p,\eta}) = \begin{cases}
            \frac{2p(1-2p)(1-\eta)^2}{\eta+2p(1-\eta)} & p\leq \frac{\sqrt{\eta}}{2(1+\sqrt{\eta})} \\
            (1-\sqrt{\eta})^2 & \frac{\sqrt{\eta}}{2(1+\sqrt{\eta})} \leq p \leq 1 -\frac{\sqrt{\eta}}{2(1+\sqrt{\eta})} \\
            \frac{2(1-p)(2p-1)(1-\eta)^2}{\eta+2(1-p)(1-\eta)} & p\geq 1-\frac{\sqrt{\eta}}{2(1+\sqrt{\eta})}
        \end{cases}. 
        \label{eq:b-coeff-GADC}
    \end{align}
\end{lemma}
\begin{proof}
For the middle case, it is sufficient to show that the solution to $\alpha(\cN_{p,\eta})$ also fulfills the stronger conditions here.  
For the first case, we start by showing that the given solution is a lower bound. Consider the matrix
    \begin{align}
        X = \begin{bmatrix}
            -p(1-\eta) & 0 \\
            0 & p\eta +1 -p - \frac{\eta}{\eta+2p(1-\eta)}
        \end{bmatrix}.
    \end{align}
    For this to be a feasible solution we need to show that,
    \begin{align}
        0 &\leq \Gamma^{\cN}_{AB} - I_A\otimes  X_B, \\
        0 &\leq \Gamma^{\cN}_{AB} + I_A\otimes  X_B,
        \end{align} 
        which can again be done by checking positivity by explicit calculation. The bound is then exactly given by $\tr X$. For the third case, the lower bound follows similarly by choosing 
        \begin{align}
        X = \begin{bmatrix}
            (1-p)\eta +p - \frac{\eta}{\eta+2(1-p)(1-\eta)} & 0 \\
            0 & -(1-p)(1-\eta)
        \end{bmatrix}.
    \end{align}

    Now we move on to prove the converse. For the middle case, the upper bound is established by choosing the same $Y$ in~\eqref{eq:Doeblin_Y_GAD} as a specific choice for $Y_2=Y$ and $Y_1=0$ in the dual SDP of $\alpha_{\wang}(\cN)$ presented in~\eqref{eq:dual-of-bN-quantity}.

    For the third case, consider \footnote{Note that the following choices were made with the support of numerical simulations and the respective optimizers.}
    \begin{equation}
      Y_1=  \begin{bmatrix}
            0 & 0 & 0 & 0\\
             0 & x_1 & 0 & 0\\
              0 & 0 & 0 & 0\\
               0 & 0 & 0 & 0
        \end{bmatrix}, \quad Y_2=  \begin{bmatrix}
            1 & 0 & 0 & x_2\\
             0 & 0 & 0 & 0\\
              0 & 0 & 0 & 0\\
               x_2 & 0 & 0 & 1+x_1
        \end{bmatrix},
           \end{equation}
with 
\begin{align}
    x_1 &\coloneqq \frac{(1-\eta) \left( \eta (1-2p)^2 -4(1-p)^2\right)}{\left((2p-1)\eta +2(1-p)\right)^2}, \\
    x_2 &\coloneqq \frac{-\sqrt{\eta}}{(2p-1)\eta +2(1-p)}.
\end{align}
The chosen $Y_1$ and $Y_2$ are feasible choices for the dual SDP since $Y_1, Y_2 \geq 0$ and $\Tr_A[Y_2-Y_1]=I_B$. This follows by the fact that $x_1 \geq 0$ for $1 \geq p \geq 1- (\sqrt{\eta}- \eta)/(2(1-\eta))$ and $1+x_1 \geq 0$ with $(1+x_1)-x_2^2\geq 0$. Furthermore, we have 
\begin{equation}
    \Tr[\Gamma^\cN_{AB} (Y_1 +Y_2)]= \frac{2(1-p)(2p-1) (1-\eta)^2} {(2p-1)\eta +2(1-p)} \geq b(\cN_{p,\eta}).
\end{equation}

 For the first case to obtain the desired upper bound, choose 
 \begin{equation}
      Y_1=  \begin{bmatrix}
            0 & 0 & 0 & 0\\
             0 & 0 & 0 & 0\\
              0 & 0 & x_1 & 0\\
               0 & 0 & 0 & 0
        \end{bmatrix}, \quad Y_2=  \begin{bmatrix}
            1+x_1& 0 & 0 & x_2\\
             0 & 0 & 0 & 0\\
              0 & 0 & 0 & 0\\
               x_2 & 0 & 0 & 1
        \end{bmatrix}
           \end{equation}
        and the proof follows by replacing $p$ with $(1-p)$.
\end{proof}

\begin{remark}[Multiplicativity]
    Due to multiplicativity of $\alpha_\wang(\mathcal{N})$ (\cref{prop:multiplicativity-b-N}), we also conclude from~\cref{lem:b_n_GADC} and~\cref{Cor:mult-alpha-from-alpha-wang} that 
$\alpha\!\left(\cN_{p,\eta}^{\otimes n} \right) =\left(\alpha(\cN_{p,\eta})\right)^n$  whenever 
\begin{equation}
    \frac{\sqrt{\eta}}{2(1+\sqrt{\eta})} \leq p \leq 1 -\frac{\sqrt{\eta}}{2(1+\sqrt{\eta})} 
\end{equation}
    since $\alpha(\cN_{p,\eta})=\alpha_\wang(\cN_{p,\eta})$ in that regime.
\end{remark}

For comparison, we also give exact expressions for the contraction and expansion coefficients. 
\begin{lemma}[Contraction Coefficient for the GAD Channel]\label{lem:GADC_contr_coef}
    Let $p \in [0,1]$ and $\eta \in [0,1]$. Then, the contraction coefficient and expansion coefficient of the GAD channel is characterized as follows:
    \begin{align}
        \eta_{\Tr}(\cN_{p,\eta}) &= \sqrt{\eta} ,\\ 
        \check{\eta}_{\Tr}(\cN_{p,\eta}) &= {\eta}.
    \end{align}
\end{lemma}
\begin{proof}
    Note that for the GAD channel with parameters $p$ and $\eta$, $T=\operatorname{diag}(\sqrt{\eta}, \sqrt{\eta},\eta)$ (recall the qubit channel mapping $w_0 I +w\cdot\vec{\sigma} \rightarrow w_0 I+(t+Tw)\cdot\vec{\sigma}$.). The, the proof follows by directly applying \cref{Lem:cc-tr-qubit} and~\cref{prop:expansion_coeff_qubit} since $\|T\|_\infty= \sqrt{\eta}$ and $\sigma_{\min}(T)= \eta$.
\end{proof}

\section{Applications of Quantum Doeblin Coefficients} \label{Sec:Applications}

In this section, we discuss how quantum Doeblin coefficients and their properties can be utilized in a diverse range of applications, including limitations on noise-induced barren plateaus, error mitigation protocols, noisy hypothesis testing, fairness of quantum learning models, and mixing times of time-inhomogeneous quantum Markov chains. Our analysis in terms of quantum Doeblin
coefficients provides insights over previous works, and we obtain efficiently computable limitations in all of these applications when the  systems involved are noisy.  

\subsection{Noise-Induced Barren Plateaus}

\label{sec:noise-induced-barren-plateaus}

Variational quantum algorithms and quantum neural networks have been considered in the context of near-term quantum computing~\cite{Cerezo2021vqa}, with the possibility of using quantum and classical resources in tandem for learning and inference tasks. There, the quantum computational resources are (mostly) used to estimate a cost function and then classical computational resources are used to optimize the parameters of the variational quantum circuits by minimizing the cost function estimated. 
The barren plateau phenomenon occurs when the optimization landscape becomes extremely flat and thus the parameters of a quantum circuit  become untrainable when using gradient-based optimizers~\cite{mcclean2018barren,larocca2024review_barren_plateaus}. 

In practice, near-term quantum devices are noisy. As such, it is a natural question to understand how well one can train parameterized quantum circuits in the presence of noisy gates. This was originally studied under a model in which  Pauli qubit noise or depolarizing noise acts on  each qubit before and after each unitary in a quantum circuit~\cite{wang2021noise}. The authors of~\cite{wang2021noise} showed that, when noise is present, it is  harder to train the parameters because the partial derivatives of the cost function with respect to the parameters decay exponentially with the depth of the circuit. This phenomenon was termed ``noise-induced barren plateaus''. Then, Ref.~\cite{schumann2024emergence} explored this phenomenon when the parameter shift rule is used as the gradient estimation technique and when the noisy channel $\cN$ has a non-zero  trace distance contraction coefficient 
 (i.e., $\eta_{\Tr}(\cN) <1$). 
 Moreover, Ref.~\cite{schumann2024emergence} showed that noise-induced barren plateaus occur when the parameterized unitary gates in
each layer form a (global) unitary two-design (i.e., when $\cU_i(\vec{\theta}_i)$ in \eqref{eq:noisy_QC} realizes a unitary two-design, for all $i$).
Ref.~\cite{mele2024noise} extended the study of noise-induced barren plateaus with non-unital noise, for random quantum circuits that form local unitary two-designs and such that qubit local noise channels are interleaved with unitary gates. Note that~\cite{mele2024noise}  focused on the average case and showed that non-unital noise may help in avoiding vanishing variance of a (local) cost function under the selected random circuit structure; they arrived at the same conclusions on the trainability of the parameters with the depth of the circuit for cost functions with global observables. 

However, previous works focused only on a subset of circuit architectures and noise models (e.g., random quantum circuits, circuits with Pauli qubit noise) 
and most of those bounds are hard to  compute efficiently (e.g., estimating $\eta_{\Tr}(\cN)$ for noise channels). In this work, we extend the study of the impact of noise on the trainability of quantum circuits to more general noise models and architectures, and we devise efficiently computable bounds using Doeblin coefficients.

Let us consider the following model for the quantum channel realized by a noisy parameterized quantum circuit:
\begin{equation}
\label{eq:noisy_QC} \cM_{(\vec{\theta}_1,\ldots, \vec{\theta}_D )}\coloneqq \id_R \otimes 
\left(\cN_D \circ \cU_D(\vec{\theta}_D) \circ \cdots \circ \cN_1 \circ \cU_1(\vec{\theta}_1) \right),
\end{equation}
where each $\cN_i$ is a noisy channel and each $\cU_i$ is a parameterized unitary channel for $i \in \{1,\ldots,D\}$. Also, suppose that each unitary channel $ \cU_i(\vec{\theta}_i)=  U_i(\vec{\theta}_i)(\cdot) U_i(\vec{\theta}_i)^\dag $ is comprised of the following unitary operators:
\begin{equation} \label{eq:unitary_operator}
    U_i(\vec{\theta}_i) = \prod_{j=1}^{J_i} \exp\!\left(-\frac{i \theta_i^j H_i^j} {2}\right),
\end{equation}
with each $H_i^j$ a Hermitian operator satisfying $\|H_{i}^{j}\|_\infty \leq 1$. Also, note that our analysis applies to the case in which these unitaries act on $n$ qudits (our results boil down to the qubit case simply by setting $d=2$).

Let the cost function to be evaluated be as follows:
\begin{equation}
\label{eq:cost_function}
    C(\vec{\theta}_1,\ldots, \vec{\theta}_D ) \coloneqq \Tr\!\left[ O \cM_{(\vec{\theta}_1,\ldots, \vec{\theta}_D )}(\rho_0)\right],
\end{equation}
for an observable $O$ and an initial state $\rho_0$ acting on the Hilbert space~$\cH_R \otimes \cH_A$ of the joint system of $R$ and $A$.

Next, we establish an upper bound on the partial derivatives of the cost function with respect to the circuit parameters (e.g., rotation angles) given a noisy quantum circuit taking the form of~\eqref{eq:noisy_QC}. In the proof of \cref{prop:General_noise_BP}, we make use of some techniques from~\cite{schumann2024emergence,mele2024noise} (trace distance contraction and Taylor expansions of partial derivatives around a point), together with the contraction bounds based on Doeblin coefficients, as studied in our paper.

\begin{proposition}[Barren Plateaus with General Noisy Circuits and Gradient Based Optimization] \label{prop:General_noise_BP}
Fix $D\in\mathbb{N}$, and let $i \in \{1,\ldots, D\}$ and $j \in \{1,\ldots, J_i\}$. Then, we have the following upper bound on the partial derivatives of the cost function in~\eqref{eq:cost_function}:
\begin{align}
 \left| \frac{\partial C(\vec{\theta}_1,\ldots, \vec{\theta}_D )}{\partial{\theta_i^j}} \right| 
 & \leq C_1 \left\| O \right\|_\infty \left( \prod_{k=i}^{D} \left( 1-\alpha(\cN_k) \right) \right)^{2/3},
 \label{eq:1st-ineq-barren-plateaus-general-noisy}
\end{align}  
where 
\begin{equation}
    C_1 \coloneqq   \left(\frac{8}{3} \right)^{1/3} + \frac{4}{3} \left(\frac{3}{8} \right)^{2/3}\leq 2.0801.
\end{equation}

For the setting in which $\cN_i= \cN$ for all $i\in \{1,\ldots, D\}$,
we have that
\begin{align}
 \left| \frac{\partial C(\vec{\theta}_1,\ldots, \vec{\theta}_D )}{\partial{\theta_i^j}} \right| 
 & \leq C_1 \left\| O \right\|_\infty \left( 1-\alpha(\cN) \right)^{2(D-i+1)/3}.
 \label{eq:2nd-ineq-barren-plateaus-general-noisy}
\end{align}  
\end{proposition}

\begin{proof}
Recall from~\cite[Eq.~(C32)]{mele2024noise}
that 
\begin{equation}\label{eq:partial_upper_bound}
    \left|\frac{\partial f(x)}{\partial x} \right| \leq  \left| \frac{f(x+h)-f(x-h)}{2h}\right| +  \frac{h^2}{6} \sup_{y \in (x-h,x+h)} \left|\frac{\partial^3 f(y)}{\partial^3 x} \right|. 
\end{equation}
Fix $i \in\{1,\ldots, D\}$ and $j \in \{1, \ldots, J_i\}$. Then also pick two values for $\theta_i^j$, given by $\theta_i^j=\theta+ h$ and $\theta_i^j=\theta- h$, where $\theta\in\mathbb{R}$ and $h > 0$. With that, choose $\vec{\theta_i^{+}}$ and $\vec{\theta_i^{-}}$ such that all other entries are the same except for $\theta_i^j$, which takes on the fixed two values.
Define the shorthand 
\begin{align}
    \cP_{[D,i]}^{j,\theta_i^j } & \coloneqq  \id_R \otimes \left(\cN_D \circ \cU_D(\vec{\theta}_D) \circ \cdots \circ \cN_i \circ \cU_i(\vec{\theta}_i) \right), \\ 
     \cP_{[k,1]}  & \coloneqq  \id_R \otimes \left(\cN_k \circ \cU_k(\vec{\theta}_k) \circ \cdots \circ \cN_1 \circ \cU_1(\vec{\theta}_1) \right),
\end{align}
and 
\begin{equation}
    \rho_{i-1} \coloneqq \cP_{[i-1,1]} (\rho_0)
\end{equation}
for $i>1$, with $\rho_0$ being the initial input state, which lies in $\mathcal{D}(\cH_R \otimes \cH_A)$. Note that we have removed the dependence on the circuit parameters for brevity.

Now consider that
\begin{align}
    & \Tr\!\left[O \left(  \cP_{[D,i]}^{j,\theta+h }(\rho_i)- \cP_{[D,i]}^{j,\theta -h}(\rho_i) \right) \right] \notag \\
    &\leq \left\| O \right\|_\infty \left\|\cP_{[D,i]}^{j,\theta+h }(\rho_i)- \cP_{[D,i]}^{j,\theta -h}(\rho_i) \right \|_1 \\
    & \leq\left\|O\right\|_\infty \prod_{k={i}}^{D} \left(1- \alpha(\cN_k) \right)\left\|\left(\id_R \otimes \cU_i\!\left(\vec{\theta_i^{+}}\right)\right)(\rho_{i-1}) - \left( \id_R \otimes \cU_i\!\left(\vec{\theta_i^{-}}\right)\right)(\rho_{i-1}) \right\|_1 \label{eq:start_unital}
     \\ &\leq 2\left\|O\right\|_\infty \prod_{k={i}}^{D} \left(1- \alpha(\cN_k) \right), \label{eq:first_order_bound}
\end{align}
where the first inequality follows from H\"older's inequality; the second from the repetitive application of the complete contraction of trace distance under each noise layer until the $i$-th layer in~\eqref{eq:complete-trace-distance-CC-to-doeblin}, together with the unitary invariance property of trace distance (recall that the marginal state at the subsystem $R$ is $ \Tr_A[\rho_0]$ for all states considered); and the last inequality follows because $\| \rho -\sigma\|_1 \leq 2$ for all states $\rho$ and $\sigma$.

Also, from~\cite[Lemma~45]{mele2024noise},
we have that  
\begin{equation}
    \left| \frac{\partial^3 C(\vec{\theta}_1,\ldots, \vec{\theta}_D )}{\partial^3{\theta_i^j}} \right| \leq 8 \left\|O\right\|_\infty \left\|I_R \otimes H_{\theta_i^j} \right\|_\infty^3 =8 \left\|O\right\|_\infty \left\| H_{\theta_i^j} \right\|_\infty^3.
\end{equation}
With that and together with the assumption that $\left\|H_{\theta_i^j} \right\|_\infty \leq 1$, we arrive at
\begin{equation} \label{eq:3rd_order_bound}
     \left| \frac{\partial^3 C(\vec{\theta}_1,\ldots, \vec{\theta}_D )}{\partial^3{\theta_i^j}} \right| \leq 8 \left\|O\right\|_\infty. 
\end{equation}

Recalling~\eqref{eq:partial_upper_bound}, consider that
\begin{align}
     & \left| \frac{\partial C(\vec{\theta}_1,\ldots, \vec{\theta}_D )}{\partial{\theta_i^j}} \right| \notag \\
     &\leq  \frac{1}{2h} \left|  \Tr\!\left[O \left(  \cP_{[D,i]}^{j,\theta+h }(\rho_i)- \cP_{[D,i]}^{j,\theta -h}(\rho_i) \right) \right] \right| + \frac{h^2}{6} \sup_{\theta_i^j} \left| \frac{\partial^3 C(\vec{\theta}_1,\ldots, \vec{\theta}_D )}{\partial^3{\theta_i^j}} \right|  \label{eq:start_general}\\
     & \leq \frac{1}{2h}2\left\|O\right\|_\infty \prod_{k={i}}^{D} \left(1- \alpha(\cN_k) \right)+ \frac{8h^2}{6} \left\|O\right\|_\infty
     \\
     &= \left\| O \right\|_\infty  \left(\frac{8}{3} \right)^{1/3} \left(\prod_{k={i}}^{D} \left(1- \alpha(\cN_k) \right)\right)^{2/3} + \frac{4}{3} \left(\frac{3}{8}\prod_{k={i}}^{D} \left(1- \alpha(\cN_k) \right)\right)^{2/3} \left\|O\right\|_\infty \\
     & \leq  C_1 \left\| O \right\|_\infty  \left( \prod_{k={i}}^{D} \left(1- \alpha(\cN_k) \right) \right)^{2/3},
\end{align}
where the second inequality follows by~\eqref{eq:first_order_bound}
and~\eqref{eq:3rd_order_bound}; the equality follows by substituting 
\begin{equation}
    h= \left(\frac{3}{8}\prod_{k={i}}^{D} \left(1- \alpha(\cN_k) \right)\right)^{1/3},
\end{equation}
which minimizes the right-hand side.
This concludes the proof of~\eqref{eq:1st-ineq-barren-plateaus-general-noisy}.

The inequality in~\eqref{eq:2nd-ineq-barren-plateaus-general-noisy} follows by substituting $\cN_k=\cN$ in \eqref{eq:1st-ineq-barren-plateaus-general-noisy}, concluding the proof.
\end{proof} 

\begin{corollary}[Barren Plateaus with Product Noisy Channels]\label{Cor:BP_product_channels}
Let $i \in \{1,\ldots, D\}$ and $j \in \{1,\ldots, J_i\}$. Then, we have the following upper bound on the partial derivatives of the cost function given in~\eqref{eq:cost_function}:
\begin{align}
 \left| \frac{\partial C(\vec{\theta}_1,\ldots, \vec{\theta}_D )}{\partial{\theta_i^j}} \right| 
 & \leq C_1 \left\| O \right\|_\infty \left( \prod_{k=i}^{D} \left( 1-\left(\alpha_{\wang}(\cP_k) \right)^\ell \right) \right)^{2/3},
\end{align}  
where  $\mathcal{N}_i = \mathcal{P}_i^{\otimes \ell}$ and
\begin{equation}
    C_1 \coloneqq   \left(\frac{8}{3} \right)^{1/3} + \frac{4}{3} \left(\frac{3}{8} \right)^{2/3}\leq 2.0801.
\end{equation}

For the setting in which $\cN_i= \cN= \cP^{\otimes \ell}$, we have that 
\begin{align}
 \left| \frac{\partial C(\vec{\theta}_1,\ldots, \vec{\theta}_D )}{\partial{\theta_i^j}} \right| 
 & \leq C_1 \left\| O \right\|_\infty  \left( 1-\left(\alpha_{\wang}(\cP)\right)^\ell \right)^{2(D-i+1)/3}.
\end{align}  
\end{corollary}

\begin{proof}
    The proof follows by adapting together~\cref{prop:General_noise_BP} with 
    $\alpha(\cN)\geq \alpha_{\wang}(\cN)$ and $\alpha_{\wang}(\cP^{\otimes \ell})= \left(\alpha_{\wang}(\cP)\right)^\ell$ by the multiplicativity of $\alpha_{\wang}(\cN)$ (\cref{prop:multiplicativity-b-N}).
\end{proof}

\begin{remark}[Gradient Computation Using Parameter Shift Rule]
Similar to the approach used in~\cite{schumann2024emergence} for the case when the parameter shift rule \cite{Li2017,Mitarai2018,Schuld2019,crooks2019gradients} can be used as the gradient estimation technique, one can arrive at a stronger statement in terms of the depth of the circuit compared to the results obtained in~\cref{prop:General_noise_BP}. 
In particular, we find that 
\begin{equation}
    \left| \frac{\partial C(\vec{\theta}_1,\ldots, \vec{\theta}_D )}{\partial{\theta_i^j}} \right| 
  \leq C' \left\| O \right\|_\infty  \prod_{k=i}^{D} \left( 1-\alpha(\cN_k) \right),
\end{equation}
where $C'$ is a constant depending on the parameter shift rule used.
\end{remark}

Next, we elaborate on the impact of the unital noise channels (i.e., a channel $\cN$ such that $\cN(I)=I$) on the partial derivatives of the cost function. 
\begin{proposition}[Barren Plateaus with Unital Noise]
\label{prop:Unital_noise_BP}
Fix $D\in\mathbb{N}$, and let $i \in \{1,\ldots, D\}$ and $j \in \{1,\ldots, J_i\}$. Then, we have the following upper bound on the partial derivatives of the cost function given in~\eqref{eq:cost_function}:
\begin{align} \label{eq:unital_noise_gradient_bound_k}
 \left| \frac{\partial C(\vec{\theta}_1,\ldots, \vec{\theta}_D )}{\partial{\theta_i^j}} \right| 
 & \leq 3 \left(\frac{4}{3} \right)^{1/3} \left\| O \right\|_\infty \left( \prod_{k=1}^{D} \left( 1-\alpha(\cN_k) \right) \right)^{2/3}.
\end{align}  
Furthermore,
\begin{align} \label{eq:cost_concentration}
\left|C(\vec{\theta}_1,\ldots, \vec{\theta}_D )  - \frac{1}{|R|d^n}\Tr[O] \right| 
    &\leq 2\left\|O\right\|_\infty \left( \prod_{k=1}^{D} \left( 1-\alpha(\cN_k) \right) \right),
\end{align}
where $|R|$ is the size of the system $R$.

For the setting for which $\cN_i= \cN$ for all $i\in \{1,\ldots,D\}$,
we have that 
\begin{align}\label{eq:unital_noise_gradient_bound_all}
 \left| \frac{\partial C(\vec{\theta}_1,\ldots, \vec{\theta}_D )}{\partial{\theta_i^j}} \right| 
 & \leq  3 \left(\frac{4}{3} \right)^{1/3}  \left\| O \right\|_\infty \left( \left( 1-\alpha(\cN) \right)\right)^{2D/3}.
\end{align}  
\end{proposition}

\begin{proof}
For unital noise channels (i.e., satisfying $\cN_i(I)=I$), consider the following, starting from~\eqref{eq:start_unital}:
\begin{align}
     & \Tr\!\left[O \left(  \cP_{[D-i+1,i]}^{j,\theta+h }(\rho_i)- \cP_{[D-i+1,i]}^{j,\theta -h}(\rho_i) \right) \right] \nonumber \\ 
    & \leq\left\|O\right\|_\infty \prod_{k={i}}^{D} \left(1- \alpha(\cN_k) \right)\left\|\left(\id_R \otimes \cU_i\!\left(\vec{\theta_i^{+}}\right)\right)(\rho_{i-1}) - \left( \id_R \otimes \cU_i\!\left(\vec{\theta_i^{-}}\right)\right)(\rho_{i-1}) \right\|_1  \\
    & \leq \left\|O\right\|_\infty  \prod_{k={i}}^{D} \left(1- \alpha(\cN_k) \right) \Bigg(\left\|\left(\id_R \otimes \cU_i\!\left(\vec{\theta_i^{+}}\right)\right)(\rho_{i-1}) - \frac{I}{|R| d^n} \right\|_1 + \notag
    \\
   & \quad \quad  \quad \left\|\left(\id_R \otimes \cU_i\!\left(\vec{\theta_i^{-}}\right)\right)(\rho_{i-1})- \frac{I}{|R| d^n} \right\|_1 \Bigg) \\ 
    & \leq \left\|O\right\|_\infty  \prod_{k={i}}^{D} \left(1- \alpha(\cN_k) \right) \Bigg(\left\|\left(\id_R \otimes \cU_i\!\left(\vec{\theta_i^{+}}\right)\right)(\rho_{i-1}) - \frac{I}{|R| d^n} \right\|_1 + \notag
    \\
   & \quad \quad  \quad \left\|\left(\id_R \otimes \cU_i\!\left(\vec{\theta_i^{+}}\right)\right)(\rho_{i-1}) - \left( \id_R \otimes \cU_i\!\left(\vec{\theta_i^{-}}\right)\right)\left(\frac{I}{|R| d^n}  \right) \right\|_1 \Bigg) \\ 
    &\leq  \left\|O\right\|_\infty \prod_{k={i}}^{D} \left(1- \alpha(\cN_k) \right) \prod_{k={i-1}}^{1} \left(1- \alpha(\cN_k) \right)\left(\left\|\rho_0 -\frac{I}{|R| d^n}
    \right\|_1 +\left\| \rho_0 -\frac{I}{|R| d^n }\right\|_1 \right) \\
    &\leq 4 \left\|O\right\|_\infty \prod_{k={1}}^{D} \left(1- \alpha(\cN_k) \right), \label{eq:end_unital}
\end{align}
where the second inequality follows from the triangular inequality of the trace norm; the first equality follows because the noise channels $\cN$ are unital by assumption and all other channels are unitary; the third inequality follows by considering the complete contraction of trace distance under the rest of the noise channels after the $i$-th layer in~\eqref{eq:complete-trace-distance-CC-to-doeblin}, together with the unitary invariance of the trace norm and recalling $ \rho_{i-1} \coloneqq \cP_{[i-1,1]} (\rho_0)$; and finally the last inequality follows because $\|\rho-\sigma\|_1 \leq 2$ for all states $\rho$ and $\sigma$.

We conclude the inequality~\eqref{eq:unital_noise_gradient_bound_k} on the partial derivative for the case of unital channels by combining~\eqref{eq:end_unital} together with the steps followed after~\eqref{eq:start_general}, with the choice 
\begin{equation}
    h= \left(\frac{3}{4}\prod_{k={i}}^{D} \left(1- \alpha(\cN_k) \right)\right)^{1/3}.
\end{equation}

The second inequality~\eqref{eq:cost_concentration}
follows similarly to the proof of the previous claim in~\eqref{eq:end_unital}  by the application of H\"older's inequality, complete contraction of the trace distance together with unitary invariance of the trace norm. In particular, it proceeds as above by considering that 
\begin{align}
    \left|C(\vec{\theta}_1,\ldots, \vec{\theta}_D )  - \frac{1}{|R|d^n}\Tr[O] \right| &= \left|  \Tr\!\left[O \left(  \cP_{[D,1]}(\rho_0)- \cP_{[D,i]}\left(\frac{I}{|R| d^n} \right) \right) \right]   \right|
\end{align}

The last inequality~\eqref{eq:unital_noise_gradient_bound_all} follows by substituting $\cN_i= \cN$ for all $i \in\{1,\ldots, D\}$, concluding the proof.
\end{proof}

\paragraph{Discussion}
Our results presented in~\cref{prop:General_noise_BP} demonstrate that the learnability of the circuit parameters decreases with the depth of the corresponding parameterized unitary from the end of the circuit. In fact, the first few layers are not trainable whenever $\alpha(\cN_i) >0$ for all $i \in\{1,\ldots,D\}$. To the best of our knowledge, our finding here is the first to demonstrate this phenomenon even with \textit{noiseless subsystems} available within the circuit.  In previous works~\cite{wang2021noise,schumann2024emergence,mele2024noise}, all qubits are considered to be impacted by noise channels. Thus, our work shows that it is not necessary for all the qubits to be noisy in order to observe this trainability issue in terms of the decaying of partial derivatives of the cost function. These noiseless subsystems can even be understood as corresponding to the existence of a noiseless feedback system. Furthermore, our results hold regardless of the size of the noiseless subsystem $R$. To this end, considering the trivial subsystem $R$, our results recover previously known results from~\cite{wang2021noise,schumann2024emergence}. Also, note that our results can be extended to the setting where there exists a fixed unitary in the noiseless subsystem $R$.

The noise models used in our analysis are quite general, satisfying only $\alpha(\cN) >0$ or $\alpha_\wang(\cN) >0$, without any assumption on the specific structure of noise or the gradient-optimization technique used. For instance, Ref.~\cite{wang2021noise} analysed noise-induced barren plateaus under a special class of i.i.d.~qubit Pauli noise channels, while Ref.~\cite{mele2024noise} worked with qubit noise channels (both unital and non-unital). For general noise channels (without noiseless subsystems), Ref.~\cite{schumann2024emergence} provided analytical upper bounds on the partial derivatives for noise channels satisfying $\eta_{\Tr}(\cN) <1$, when the parameter shift rule is used as the gradient estimation technique (when the generators $H_i^j$ of the parameterized unitary operators in~\eqref{eq:unitary_operator} have only two unique eigenvalues or unitary operators can be decomposed into product of operators satisfying the said property~\cite{crooks2019gradients}), while Ref.~\cite{singkanipa2025beyond} has considered noise channels where the Hilbert--Schmidt norm is contractive (one norm in $\eta_{\Tr}(\cN)$ replaced by the two nrom). A  possible downside of those approaches is that the estimation of $\eta_{\Tr}(\cN)$ or the two-norm contraction are computationally difficult. However, our bounds obtained in~\cref{prop:General_noise_BP} and~\cref{Cor:BP_product_channels} are efficiently computable by SDPs. With that, it is possible to compute the bounds in order to obtain insights about how exactly the depth of the circuit and the noise channels impact the decaying of the derivative magnitudes and the trainability of the circuit parameters.

Moreover, our results hold for all parameterized quantum circuits, regardless of how the initialization of parameters is carried out. Some previous works~\cite{mele2024noise,schumann2024emergence} on noise-induced barren plateaus have focused on random quantum circuits (e.g., global and two-qubit unitary two-designs) and analysed the impact on the variance of the cost function to arrive at conclusions similar to ours. Also, in~\cite{mele2024noise}, it was shown that, with non-unital noise, a lower bound on the variance can be obtained, which in turn may avoid barren plateaus for such circuits for local cost functions. Note that these results hold on average. However, in practice there are use-cases where the parameters are not initialized randomly and the circuit structures are decided based on the prior knowledge and domain expertise. In those settings, results obtained under the assumption of random circuit initializations are not applicable. Also, it is suspected that smart initializations through warm starts and pre-training of parameters may serve as a strategy to overcome barren plateaus~\cite{rudolph2023synergistic}. However, our results show a negative result here \textit{when the circuits are noisy}. In fact, we show that, even if one uses clever initialization procedures to select parameters and unitary circuit structures, with noisy circuits we cannot harness the suspected advantage. This further highlights that, regardless of the expressivity of the circuit (e.g., unitary two-designs have higher expressivity), noise within the parameterized quantum circuits leads to barren plateaus, hindering trainability.

From~\cref{prop:General_noise_BP} and~\cref{Cor:BP_product_channels} for general noise channels, we deduce that even though the last few layers can be potentially trained, the first few layers are very hard to train whenever $\alpha(\cN_i) > 0$ or $\alpha_\wang(\cP_i) >0$ for $i \in \{1,\ldots,D\}$. This also suggests that one can only (potentially) train the last few layers and fix the circuit parameters of the first few layers without optimizing them. In contrast, if the noise channels are unital, \cref{prop:Unital_noise_BP} demonstrates that the trainability of the circuit parameters decays exponentially with the overall depth of the circuit regardless of the location of the unitary in the circuit. Our upper bounds on the partial derivatives from~\cref{prop:Unital_noise_BP} provide insights for selecting the depth of the circuit for learning and inference tasks that require optimization of parameters, when the noise is unital.

Furthermore, we observe that some noisy circuits may be classically simulable. With reference to~\eqref{eq:cost_concentration} in~\cref{prop:Unital_noise_BP}, we see that, with increasing depth $D$ of the circuit and $\alpha(\cN_i) >0$ for all $i \in\{1,\ldots,D\}$, the cost function value $C(\vec{\theta}_1,\ldots, \vec{\theta}_D )$ concentrates around $ \frac{1}{|R|d^n}\Tr[O] $. To this end, whenever $\Tr[O]=0$, one can estimate the cost function value to be zero. In that case, it is possible to classically simulate the cost function, which is an expectation value of the observable $O$ with a state dependent on the circuit parameters. Moreover, it is possible to obtain a confidence interval for the cost function value as follows: consider the scenario in which $\cN_i=\cN$ for all $i \in \{1,\ldots,D\}$; then by rewriting~\eqref{eq:cost_concentration} we arrive at
\begin{equation}
  \frac{1}{|R|d^n}\Tr[O]-2\left\|O\right\|_\infty \left(1-\alpha(\cN)\right)^{D}  \leq C(\vec{\theta}_1,\ldots, \vec{\theta}_D )  \leq  \frac{1}{|R|d^n}\Tr[O] +2\left\|O\right\|_\infty \left(1-\alpha(\cN)\right)^{D}.
\end{equation}
Note that~\eqref{eq:cost_concentration} is valid for unital noise channels. For qubit noise channels, Ref.~\cite[Proposition~35]{mele2024noise} developed a procedure to classically simulate the cost function when $O$ is replaced by a Pauli operator, up to an absolute error $\varepsilon>0$ with a high probability over random quantum circuits (on average).
We leave further exploration on the classical simulability of the cost functions under general noise channels and general circuit architectures for future work.

In summary, our study of noise-induced barren plateaus extends to new noisy circuits with the possibility of having noiseless subsystems and circuit architectures independent of the parameter initialization and the gradient estimation technique. Also, we provide efficiently computable limits on the scaling of partial derivatives of the cost function and the overall concentration of the cost function under noisy parametrized quantum circuits.

\subsection{Error Mitigation}

\label{sec:error-mitigation}

The goal of error mitigation is to learn expectation values of observables while mitigating the impact of errors caused by noise \cite{cai2023}. Here, we specifically discuss the task of estimating the expectation value $\Tr[A\rho]$ of an observable $A$ up to some error and with high probability, by mitigating the impact of the errors caused by noise in the quantum circuits. 

\begin{figure}
    \centering
    \includegraphics[width=0.7\linewidth]{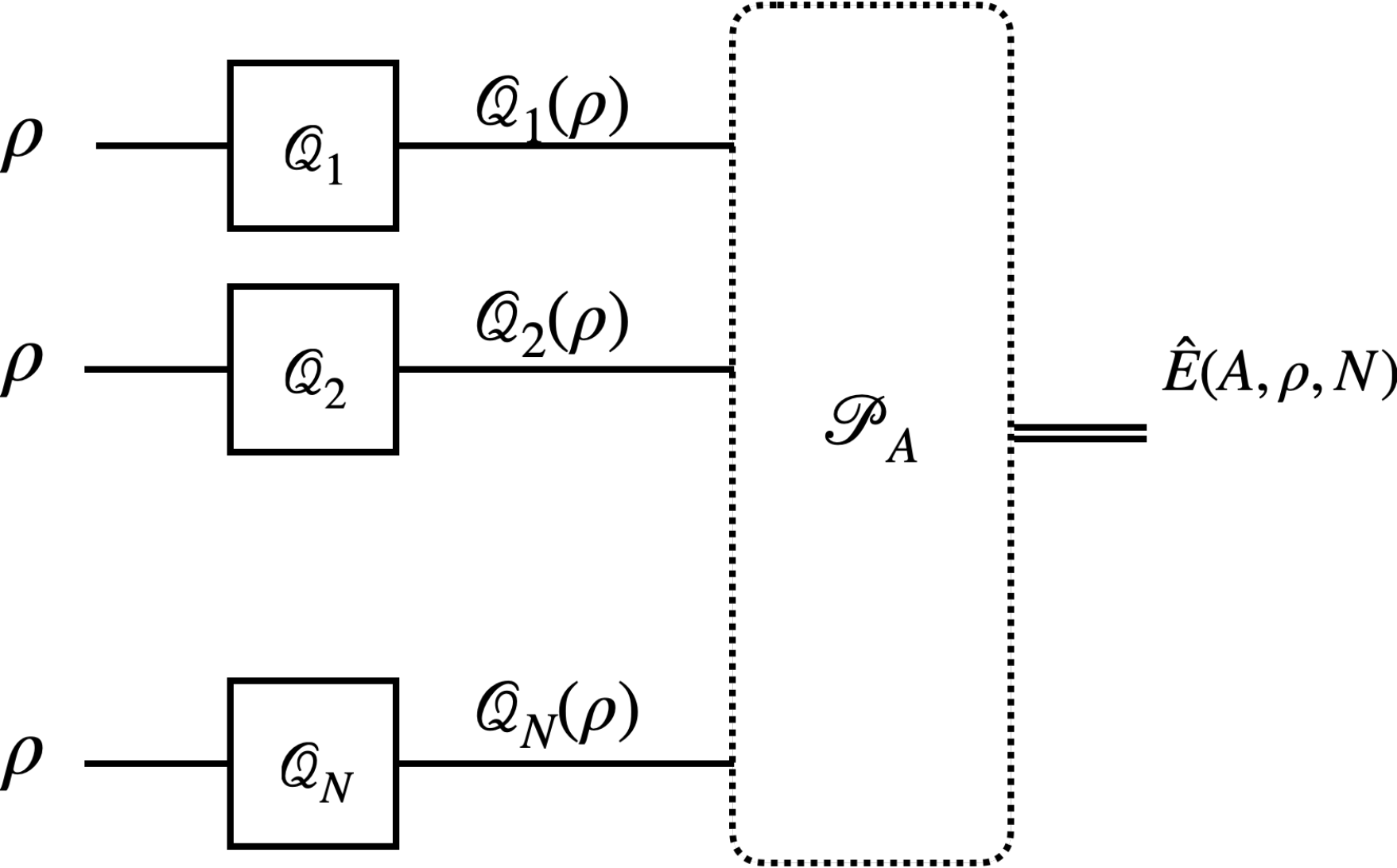}
    \caption{Error mitigation protocol $\cP_A$: The goal of the protocol is to provide an estimate for $\Tr[A\rho]$ for an observable $A$ and all states $\rho$ in a certain class $\cS$. In particular, $\cP_A$ takes as input the $N$ distorted input states $\bigotimes_{i=1}^N \cQ_i(\rho)$, where each $\cQ_i$ is a quantum channel, and outputs an estimate for $\Tr[A\rho]$ denoted as $\hat{E}(A,\rho,N)$ such that $  \left|  \Tr[ A \rho] - \hat{E}(\rho, A, N)\right| \leq \varepsilon$ with probability not less than $1-\delta$ for $\varepsilon>0$ and $\delta \in (0,1)$.}
    \label{fig:error_mitigation}
\end{figure}

\begin{definition}[Error Mitigation Protocol] \label{def:error_mitigation}
Fix $\varepsilon >0$ and $\delta \in (0,1)$. Fix $\cS$ to be a set of states and $\cO$ to be a set of observables.
Let $A \in \cO$ be a Hermitian operator. 
An error mitigation protocol $\cP_A$ (as depicted in~\cref{fig:error_mitigation}) takes as input $N$ noisy quantum states, $\bigotimes_{i=1}^N \cQ_i(\rho)$, where each $\cQ_i$ is a noisy quantum channel and $\rho \in \cS$ is the original state of interest. Then, the channel $\cP_A$ produces an estimate of $\Tr[A\rho]$, which we denote by $\hat{E}(\rho, A, N)$, as an output. We say that the error mitigation protocol is successful if 
\begin{equation} \label{eq:error_mitig_def}
    \Pr \! \left(\left|  \Tr[ A \rho] - \hat{E}(\rho, A, N)\right| \leq \varepsilon\right) \geq 1- \delta,
\end{equation}
for all $\rho \in \cS$ and $A \in \cO$.
   
\end{definition}

Note that $\cP_A$ in~\cref{fig:error_mitigation} includes many error mitigation protocols known and also protocols for which classical post-processing is allowed (see~\cite{takagi2023universal} and references therein).

Also define, 
\begin{equation}
    D_{\cO}(\rho,\sigma) \coloneqq \sup_{A \in \cO} \left|\Tr[ A (\rho -\sigma)]\right|.
\end{equation}

 Using the proof of Theorem~1 and Theorem~3 in~\cite{takagi2023universal} along with our bounds on (complete) contraction of trace distance, we next analyse the error mitigation capabilities and limitations of a broader class of noisy quantum channels where the \textit{noisy channels are not necessarily unital}. Also note that our bounds represent fundamental limitations that are independent of which particular error mitigation protocol is used to obtain the estimate $\hat{E}(\rho, A, N)$.

Let the noise channels $\cQ_i$ in~\cref{fig:error_mitigation} for all $i\in\{1,\ldots,N\} $ be characterized as follows:
\begin{equation}\label{eq:Q_i}
    \cQ_i=\id_R \otimes \left( \cN_D \circ \cU_D \circ \cdots \circ \cN_j \circ \cU_j \circ \cdots \circ \cN_1 \circ \cU_1 \right),
\end{equation}
where each $\cN_j$ is a noisy channel and each $\cU_j$ is a unitary channel, for $j \in \{1,\ldots,D\}$.

 \begin{theorem} [Sample Complexity of Error Mitigation]\label{thm:Error_miti_global}
     Suppose that an error-mitigation strategy in~\cref{fig:error_mitigation} satisfies~\eqref{eq:error_mitig_def} for some $\varepsilon > 0$ and $0\leq \delta \leq 1/2$. Also fix $\cQ_i=\cQ_k$ for all $i,k \in \{1,\ldots,N\}$, and suppose that $\cQ_i$ is given by~\eqref{eq:Q_i}. If there exist $\rho_{RA},\sigma_{RA} \in \cS$ such that $D_\cO(\rho_{RA},\sigma_{RA}) \geq 2 \varepsilon$ and $\rho_R=\sigma_R$, then an error mitigation protocol requires $N$ noisy states such that the following inequality should hold
      \begin{equation}
         N \geq \frac{(1-2\delta)} {\prod_{i=1}^D \left(1- \alpha(\cN_i)\right)}
     \end{equation}
     in order to satisfy~\eqref{eq:error_mitig_def} for all $\rho_{RA} \in \cS$ and $A_{RA} \in \cO$.
 \end{theorem}

\begin{proof}
    The proof follows by starting in a similar fashion to the proof of Theorem~1 and Theorem~3 in~\cite{takagi2023universal}, with the assumption that there exist $\rho_{RA},\sigma_{RA} \in \cS$ such that $D_\cO(\rho_{RA},\sigma_{RA}) \geq 2 \varepsilon$.

    Consider Eq.~(A5) of~\cite{takagi2023universal}, where $(\mathcal{Q}_i)_{i=1}^N$ denotes the tuple of noisy channels, each of which is applied on one of the $N$ copies of the state. Let us suppose that the same noisy channel (with $D$ noisy layers) is applied on each copy of the given state. Then
    \begin{align}
        1-2\delta &\leq T\!\left( \bigotimes_{i=1}^N \cQ_i(\rho_{RA}) , \bigotimes_{i=1}^N
        \cQ_i(\sigma_{RA})  \right) \\
        & \leq N \ T\!\left( \cQ(\rho_{RA}) , \cQ(\sigma_{RA})  \right),
    \end{align}
where the last inequality holds by the subadditivity of trace distance and choosing $\cQ_i= \cQ$.

Note that $\cQ = \id\otimes \left( \cN_D \circ \cU_D \circ \cdots \circ \cN_j \circ \cU_j \circ \cdots \circ \cN_1 \circ \cU_1 \right)$, 
where each $\cU_i$, for all $i\in\{1, \ldots, D\}$, is a unitary circuit applied in between each noisy channel. 
Also, let us define the following shorthand for $i\in\{1, \ldots, D\}$:
\begin{equation}
    \cP_i \coloneqq \id_R \otimes (\cN_i \circ \cU_i \circ \cdots \circ \cN_1 \circ \cU_1).
\end{equation}

Then, we have that 
\begin{align}
   & T\!\left( \cQ(\rho_{RA}) , \cQ(\sigma_{RA})  \right) \notag \\
   & \leq  \left(1- \alpha(\cN_D)\right) T\!\left( (\id_R \otimes \cU_D) \circ \cP_{D-1}(\rho_{RA}) , (\id_R \otimes \cU_D) \circ \cP_{D-1}(\sigma_{RA})  \right) \\
  & = \left(1- \alpha(\cN_D)\right) T\!\left( \cP_{D-1}(\rho_{RA}) ,  \cP_{D-1}(\sigma_{RA})  \right),  
\end{align}
where the first inequality follows from applying \cref{prop:trace_distance_complete_cont_Doeblin_bound} for the channel $\cN_D$, along with the equality $\Tr_A[ \cP_{D-1}(\rho_{RA})]=\Tr_A[ \cP_{D-1}(\sigma_{RA})]$ due to $\rho_R =\sigma_R$; and the
last equality follows by unitary invariance of the trace distance. 

Next, by the repetitive application of the Doeblin coefficient bound in order to bound the complete contraction of trace distance, due to $D-1$ noisy channels and by using unitary invariance of trace distance, we have that
\begin{align}
  T\!\left( \cQ(\rho_{RA}) , \cQ(\sigma_{RA})  \right) &\leq  {\prod_{i=1}^D \left(1- \alpha(\cN_i)\right)} \ T(\rho_{RA},\sigma_{RA}) \\
  & \leq {\prod_{i=1}^D \left(1- \alpha(\cN_i)\right)},
\end{align}
where the second inequality follows because $T(\rho_{RA},\sigma_{RA}) \leq 1$.
Finally, putting everything together, we arrive at the desired result.
\end{proof}

The lower bound in~\cref{thm:Error_miti_global} is SDP-computable. Thus, when we are given the description of the set of noisy channels, it can be evaluated efficiently when the dimension of the systems are small. However, in general, the computational complexity increases exponentially with an increasing number of qubits. Instead, for a channel $\cN= \cP^{\otimes n}$, one can use $\alpha_{\wang}(\cP)$ to provide an efficiently computable bound by employing multiplicativity of $\alpha_\wang$ (\cref{prop:multiplicativity-b-N}), as presented next.

\begin{remark}[Error Mitigation Limits with Local Noisy Channels]\label{rem:Error_miti_local}
     By considering local noisy qubit channels instead of global noisy channels  as in~\cref{thm:Error_miti_global}, we need 
\begin{align}
   N &\geq \frac{(1-2\delta)} {\prod_{i=1}^D \left(1- \alpha(\cP_i^{\otimes n})\right)}  \label{eq:local_miti_alpha}\\
   &\geq \frac{(1-2\delta)} {\prod_{i=1}^D \left(1- \alpha_{\wang}(\cP_i)^n \right)}, \label{eq:local_miti_b}
\end{align}
where $\alpha_{\wang}(\cP_i)$ is defined in~\eqref{eq:primal-b-relaxed-doeblin}.
\end{remark}

\paragraph{Discussion} Error mitigation under deep noisy quantum circuits has been studied in~\cite{takagi2023universal,quek2024exponentially}. Those works showed that there is an exponential dependence of the depth of the circuit on the number of noisy data samples required for error mitigation. This suggests that, without access to an exponential number of noisy data samples, it is not possible to mitigate errors when estimating the expectation values of observables (see \cref{def:error_mitigation} for the success criterion of an error-mitigation protocol). 

In particular,~\cite[Theorem~3]{takagi2023universal} showed this exponential scaling of the required number of noisy input samples when the noisy quantum channels are depolarizing noise channels (i.e., consider $\cQ_i$ in~\cref{fig:error_mitigation} as given by~\eqref{eq:noisy_QC} with $R$ system being the trivial system and each $\cN_i$ being a depolarizing channel with flip parameter $p$). In~\cite{quek2024exponentially}, the study was extended to general unital and non-unital noise channels with the assumption of random unitaries in the circuit (i.e., consider $\cQ_i$ in~\cref{fig:error_mitigation} as given by~\eqref{eq:noisy_QC} with $R$ system being the trivial system and each $\cU_i$ being (independently) randomly sampled from the set of unitaries that form a unitary two-design). In this setting with random unitaries and qubit noise channels, Ref.~\cite{quek2024exponentially} showed that there exists an exponential overhead with respect to both the depth and the width (number of qubits) of the circuit.

However, in practice, the unitaries are not always random quantum circuits (i.e., they can be fixed or chosen from a different sampling or optimization approach), and the noise channels that can arise in quantum devices are far beyond depolarizing noise or even qubit noise channels. In our work, using Doeblin coefficients, we extend the analysis of the resource requirements for successful error mitigation to a large class of noisy channels and quantum circuit architectures. \cref{thm:Error_miti_global} demonstrates that this exponential dependence holds for a large class of noise channels whenever $\alpha(\cN_i) >0$. 

Let us consider an example setting in which each $\cN_i=\cN_{p,\eta}$ for $p \in [0,1]$ and $\eta \in (0,1)$. Then, using the previously known bound from~\cite[Eq.~(I22)]{takagi2023universal} (by choosing $k=1$ therein), we have the following lower bound
\begin{equation} \label{eq:SC_GAD_previous}
    N \geq \frac{(1-2\delta)^2 \min\{p,1-p\} \lambda_{\min}(\Gamma^{N_{p,\eta}})}{ \eta^D 8 \ln (2)},
\end{equation}
where 
\begin{equation}
    \lambda_{\min}(\Gamma^{\cN_{p, \eta}}) = \frac{1+\eta}{2}- \sqrt{\eta +(1-\eta)^2 (1-2p)^2}.
\end{equation}
 This then raises the question whether the exponential dependence of the depth $D$ of the circuit on the sample complexity of error mitigation techniques also vanishes in the case of $p\in \{0,1\}$. However, our result in~\cref{thm:Error_miti_global} shows that there is still an exponential dependence on $D$, as long as $\eta \in (0,1)$ for all $p \in [0,1]$, because $\alpha(\cN_{p,\eta})= (1-\sqrt{\eta})^2$. Also, note that $p=1$ is a practically relevant noise model,  known as the amplitude damping channel. This highlights a practically relevant scenario where our results provide useful insights.

Our results in~\cref{thm:Error_miti_global} hold even if there is a noiseless subsystem (denoted by $R$ in~\eqref{eq:noisy_QC}), showcasing that there exists an exponential scaling of the required resources,  characterized by $\alpha(\cN_i)$. Thus, our result is valid for a large class of noise models. One example towards this is the gate-based noise models, where the noise of the circuit is significant when only non-identity unitary circuits are applied. In this scenario, we can assume that in the noiseless system $R$, there are no non-identity unitaries in operation. Most, importantly this conveys that having a noiseless subsystem is not sufficient to circumvent the exponential resource requirement for error mitigation.

Another advantage of our results in~\cref{thm:Error_miti_global} is that the lower bound given there is SDP computable. This allows one to get a sense of how the required number of noisy data samples scales for an error mitigation protocol to achieve a desired accuracy depending on the depth of the circuit and the noise channels related to specific applications of interest. When there are local noise channels in the system (i.e., $\cN_i =\cP_i^{\otimes n}$), we see that the exponential scaling depends on $\alpha(\cP_i^{\otimes n})$ as given in~\cref{rem:Error_miti_local}. By~\cref{rem:effi_com_alpha}, the lower bound in~\eqref{eq:local_miti_alpha} can be computed with the time complexity scaling as $O\!\left(\operatorname{poly} (n)\right)$ by utilizing  permutation symmetry of $\cP_i^{\otimes n}$. Furthermore, a somewhat weaker bound given in~\eqref{eq:local_miti_b} can be efficiently computed by using $\alpha_{\wang}(\cN)$, without any computational complexity scaling with $n$.

\medskip
In summary, our results provide efficiently computable limits on the sample complexity of a broad range of error-mitigation protocols under various noise models satisfying $\alpha(\cN) >0 $ or $\alpha_{\wang}(\cN) >0$ and circuit architectures with general quantum circuits (not necessarily random circuits) and circuits with noiseless subsystems.

\subsection{Noisy Quantum Hypothesis Testing}

\label{sec:noisy-hypothesis-testing}

In this section, we study the impact of noise on the hypothesis testing of two quantum states. In particular, this refers to the scenario where we need to guess whether a given state is $\rho$ or $\sigma$, when we have access to only multiple samples of the states $\cN(\rho)$ or $\cN(\sigma)$ (finite), where $\cN$ is a noisy quantum channel. To this end, first let us recall the noiseless scenario, for which $\cN$ is the identity channel, and then utilize those results together with contraction coefficients and Doeblin coefficients to quantify the impact of having access only to noisy samples of quantum data.

\paragraph{Sample Complexity of Noiseless Quantum Hypothesis Testing} 

\textit{Problem setup}: 
Suppose that there are two states $\rho$ and $\sigma$,   $\rho^{\otimes n}$ is selected with probability $\beta\in (0,1)$, and $\sigma^{\otimes n}$ is selected with probability $1-\beta$. The sample complexity is equal to the minimum value of $n$ needed to reach a constant error probability in deciding which state was selected. To define this quantity formally, let us recall that the
Helstrom--Holevo theorem~\cite{helstrom1967detection,holevo1973statistical} states that the optimal error probability $p_{e}(\rho
,\sigma,\beta)$ of symmetric quantum hypothesis testing is as follows:
\begin{align}
p_{e}(\rho,\sigma,\beta)  &  \coloneqq \min_{M_{1},M_{2}\geq0 } \left\{ \beta \operatorname{Tr}
[M_{2}\rho]+ (1-\beta) \operatorname{Tr}[M_{1}\sigma]:M_1 + M_2 = I \right\}
\\
&  =\frac{1}{2}\left(  1-\left\Vert \beta \rho-(1-\beta)\sigma\right\Vert _{1}\right) \label{eq:probability_error} .
\end{align}

With this in mind, we are assuming in this paradigm that there is a constant $ \varepsilon \in [0,1]$, and our goal is to determine the minimum value of $n$ such
that
\begin{equation}
p_{e}\!\left(\rho^{\otimes n},\sigma^{\otimes n}, \beta \right) \coloneqq \frac{1}{2}\left(  1-\left\Vert \beta \rho^{\otimes n}-(1-\beta) \sigma^{\otimes
n}\right\Vert _{1}\right) \leq \varepsilon.\label{eq:eps-n-relation}
\end{equation}
To this end, let us define 
\begin{equation}\label{eq:non_private_SC}
    \mathrm{SC}_{(\rho,\sigma)}(\varepsilon,\beta) \coloneqq \min\! \left\{ n\in \mathbb{N} : p_{e}\!\left(\rho^{\otimes n},\sigma^{\otimes n},\beta \right) \leq  \varepsilon \right\},
\end{equation}
and recall Theorem~7 and Corollary~8 of~\cite{cheng2024sample}.

\begin{theorem}[Sample Complexity of Symmetric Hypothesis Testing (Theorem~7 and Corollary~8 of~\cite{cheng2024sample}), {\cite[Theorem~2]{Nuradha_2025QueryComplexity}}] \label{thm:sample_C_no_private}
Let $\rho$ and $\sigma$ be states. Fix $\varepsilon \in [0,1]$ and $\beta \in (0,1)$. The following statements hold:
    \begin{enumerate}
        \item  If 
    $\rho \perp \sigma$ (i.e., $\rho \sigma= 0$),
    $\varepsilon\in [1/2,1]$, or $\exists\, s\in[0,1]$ such that $\varepsilon \geq \beta^s (1-\beta)^{1-s}$, then 
	\begin{align}
\label{eq:binary_symmetric1}
		\mathrm{SC}_{\left(\rho,\sigma\right)}(\varepsilon,\beta) = 1.
	\end{align}
	\item If $\rho = \sigma$ and $\min\{\beta, 1-\beta\} > \varepsilon \in [0,1/2)$, then
	\begin{align}
		\label{eq:binary_symmetric2}
		\mathrm{SC}_{\left(\rho,\sigma\right)}(\varepsilon,\beta)  = +\infty.
	\end{align}
    \item If the conditions mentioned above are excluded, then
\begin{equation}
  \max\!\left\{ \frac{\ln\!\left( \frac{\beta (1-\beta)}{\varepsilon (1-\varepsilon)} \right) }{ -\ln F(\rho,\sigma) } ,\frac{1-\frac{\varepsilon(1-\varepsilon)}{\beta (1-\beta)}}{ \left[d_{\mathrm{B}}(\rho,\sigma)\right]^2  } \right\} \leq  \mathrm{SC}_{(\rho,\sigma)}(\varepsilon,\beta)  \leq \left \lceil \frac{ 2\ln \!\left( \frac{\sqrt{ \beta (1-\beta)} }{ \varepsilon } \right) }{-\ln  F(\rho,\sigma)} \right\rceil,
\end{equation}
where $F(\rho,\sigma) \coloneqq \|\sqrt{\rho} \sqrt{\sigma}\|_1^2$ and 
$\left[d_B(\rho,\sigma)\right]^2 \coloneqq 2 \left(1- \sqrt{F(\rho,\sigma)}\right)$.
    \end{enumerate}
\end{theorem}

 In~\cite{cheng2024sample}, the sample complexity of hypothesis testing of noiseless quantum data has been studied. Following that, the sample complexity of hypothesis testing with access to only privatized (noisy) quantum states was studied in~\cite{nuradha2024contraction,Christoph2024sample}. In this case, privacy constraints are imposed by quantum local differential privacy~\cite{hirche2022quantum}. See~\cite{QDP_computation17,nuradha2023quantum} for further details on ensuring privacy for quantum data. The contraction coefficients of quantum divergences under private (noisy) channels have been used as one of the key technical tools in the analysis of private quantum hypothesis testing. 

In this work, we extend the study of the sample complexity of noisy quantum hypothesis testing when we are given noisy channel outputs of the original input states to discriminate. We utilize quantum Doeblin coefficients and the tools developed in this work to do so.

Let $\rho$ and $\sigma$ be states. Define, from~\cite{hirche2024quantumDivergences},
\begin{equation} \label{eq:Helinger_1_2}
    H_{1/2}\!\left(\rho \Vert \sigma \right) \coloneqq \frac{1}{2} \int_{1}^\infty \left[ E_\gamma(\rho \Vert \sigma) + E_\gamma(\sigma \Vert \rho) \right] \gamma^{-3/2}\,  \mathsf{d} \gamma.
\end{equation}

\begin{proposition}[Sample Complexity of Noisy Quantum Hypothesis Testing] \label{Prop:SC_noisy}
    Let $\rho$ and $\sigma$ be states, and let $\cN$ be a quantum channel. Fix $\varepsilon \in [0,1]$ and $\beta \in (0,1)$. The following statements hold:
    \begin{enumerate}
        \item \label{one_cond} If 
    $\cN(\rho) \perp \cN(\sigma)$ (i.e., $\cN(\rho) \cN(\sigma) = 0$),
    $\varepsilon\in [1/2,1]$, or $\exists\, s\in[0,1]$ such that $\varepsilon \geq \beta^s (1-\beta)^{1-s}$, then 
	\begin{align}
		\mathrm{SC}_{\left(\cN(\rho),\cN(\sigma)\right)}(\varepsilon,\beta) = 1.
	\end{align}
	\item \label{second_cond} If $\cN(\rho) = \cN(\sigma)$ and $\min\{\beta, 1-\beta\} > \varepsilon \in [0,1/2)$, then
	\begin{align}
		\mathrm{SC}_{\left(\cN(\rho),\cN(\sigma)\right)}(\varepsilon,\beta)  = +\infty.
	\end{align}
    \item If the first two conditions are excluded, then
    \begin{multline}
      \frac{1-\frac{\varepsilon(1-\varepsilon)}{\beta (1-\beta)}}{\left(1-\alpha(\cN)\right) H_{1/2}(\rho \Vert \sigma)}\leq   \mathrm{SC}_{\left(\cN(\rho),\cN(\sigma)\right)}(\varepsilon,\beta) \\
      \leq \left \lceil 2 \ln\!\left( \frac{\sqrt{\beta (1-\beta)}}{\varepsilon} \right) \!\left( \frac{1}{\left(1-\check{\alpha}(\cN) \right) T(\rho, \sigma)} \right)^2  \right\rceil, 
    \end{multline}
    where $\alpha(\cN)$ and $\check{\alpha}(\cN)$ are defined in~\eqref{eq:alpha_cN_def} and~\eqref{eq:reverse_Doeblin}, respectively \footnote{Note that $1-\alpha(\cN)$ can be replaced by $\eta_{\Tr}(\cN)$ for a tighter lower bound and $1-\check{\alpha}(\cN)$ by $\check{\eta}_{\Tr}(\cN)$ for a tighter upper bound, whenever they are analytically computable. For an example, consider the GAD channel (\cref{lem:GADC_contr_coef}).}. 
    
    \item If the first two conditions are not satisfied and
    \begin{equation}
    \lambda_{\rho,\sigma,\cN} \coloneqq \max\left\{ \lambda_{\min}(\cN(\rho)),\lambda_{\min}(\cN(\sigma)) \right\} > 0,    
    \end{equation}
    then
    \begin{equation}
  \mathrm{SC}_{\left(\cN(\rho),\cN(\sigma)\right)}(\varepsilon,\beta)  \geq \frac{ \lambda_{\rho,\sigma, \cA}}{4} \ln\!\left( \frac{\beta (1-\beta)}{\varepsilon (1-\varepsilon)} \right) \left(\frac{1}{ \left(1-\alpha(\cN)\right) T(\rho, \sigma)} \right)^2.
\end{equation}
    \end{enumerate}
\end{proposition}

\begin{proof}
    Items (1) and (2) follow directly by~\cref{thm:sample_C_no_private}.

    For the lower bound in Item (3), consider the following:
    \begin{align}
       \left[d_B\!\left( \cN(\rho), \cN(\sigma) \right)\right]^2 
  & =   2\left(1-\sqrt{F}\!\left(\cN(\rho),\cN(\sigma) \right)\right)  \\
  &\leq H_{1/2}\!\left( \cN(\rho) \Vert \cN(\sigma) \right) \\ 
  & \leq \eta_{\Tr}(\cN) H_{1/2}\!\left( \rho \Vert \sigma \right) \\
  &\leq \left(1-\alpha(\cN) \right)H_{1/2}\!\left( \rho \Vert \sigma \right),
    \end{align}
where the first inequality follows by~\cite[Eq.~(5.50)]{hirche2024quantumDivergences}, the second inequality by identifying that $H_{1/2}$ is an $f$-divergence and  its contraction coefficient is upper bounded by the contraction coefficient of the trace distance; and the last inequality because $\eta_{\Tr}(\cN) \leq 1-\alpha(\cN)$.
Substituting the above in the lower bound with the denominator $\left[d_{\mathrm{B}}(\rho,\sigma)\right]^2 $ in Item 3 of~\cref{thm:sample_C_no_private} concludes the lower bound.

For the upper bound in Item (3), consider that 
\begin{equation} \label{eq:item3_1}
    -\ln \sqrt{F}\!\left(\cN(\rho),\cN(\sigma) \right) \geq \frac{1}{2} \left[T\!\left(\cN(\rho),\cN(\sigma) \right)\right]^2,
\end{equation}
which follows because 
\begin{align}
    {-\ln \sqrt{F}(\rho,\sigma)} & \geq \frac{1}{2}{\left[d_{\mathrm{B}}(\rho,\sigma)\right]^2 } 
    \geq \frac{1}{2}{\left[T(\rho,\sigma)\right]^2 }.
\end{align}
The above inequalities follow because $-\ln(x) \geq 1- x$ for $x >0$, and the last one due to the well known Fuchs-van-de-Graaf inequalities.
Next, by employing the bound in \eqref{eq:expansion_trace} involving the reverse Doeblin coefficient, we have that 
\begin{equation}\label{eq:item3_2}
T\!\left(\cN(\rho),\cN(\sigma) \right) \geq \left( 1- \check{\alpha}(\cN) \right) T(\rho,\sigma).
\end{equation}
Combining~\eqref{eq:item3_1} and~\eqref{eq:item3_2}, together with the upper bound in Item~(3) of \cref{thm:sample_C_no_private}, completes the proof of the upper bound. 

Now, let us prove the claim in Item (4). First recall that 
\begin{equation}
    -\ln F(\rho,\sigma) \leq D(\rho \Vert \sigma).
\end{equation}
Also recall the assumption $\lambda_{\rho,\sigma,\cN} \coloneqq \max\left\{ \lambda_{\min}(\cN(\rho)),\lambda_{\min}(\cN(\sigma)) \right\} > 0$. Without loss of generality, suppose that $\lambda_{\rho,\sigma,\cN} = \lambda_{\min}(\cN(\sigma))$ (when it is otherwise, interchange $\rho$ and $\sigma)$. With that, we have that
\begin{align}
     -\ln F\!\left(\cN(\rho),\cN(\sigma) \right) & \leq D\!\left(\cN(\rho) \Vert \cN(\sigma) \right) \\
     & \leq \frac{4 \left[T\!\left( \cN(\rho), \cN(\sigma) \right) \right]^2}{\lambda_{\min} \! \left( \cN(\sigma) \right) } \\
     & \leq \frac{4 \left(1-\alpha(\cN) \right)^2} {\lambda_{\min} \! \left( \cN(\sigma) \right)} \left[T(\rho,\sigma)\right]^2,
\end{align}
where the second inequality follows by~\cite[Theorem~2]{audenaert2005continuity} and the monotonicity of norms such that $\left\|A\right\|_2 \leq \left\|A\right\|_1$; and the last inequality by $\eta_{\Tr} (\mathcal{N}) \leq 1-\alpha(\cN)$.
Plugging the above inequality in the lower bound with the denominator $-\ln F(\rho,\sigma)$ in Item (3) of~\cref{thm:sample_C_no_private} completes the proof.
\end{proof}

\begin{remark}[Noisy Hypothesis Testing of Quantum States with Reference Systems]
    Let $\rho_{RA}$ and $\sigma_{RA}$ be quantum states such that $\rho_R=\sigma_R$. Let noise channel $\cN$ act on system $A$.
    Then, using the same proof techniques in~\cref{Prop:SC_noisy}, we can obtain the following upper and lower bounds as follows:
  \begin{multline}
      \frac{1-\frac{\varepsilon(1-\varepsilon)}{\beta (1-\beta)}}{\left(1-\alpha_{+}(\cN)\right) H_{1/2}(\rho_{RA} \Vert \sigma_{RA})}\leq   \mathrm{SC}_{\left(\id_R \otimes \cN(\rho_{RA}),\id_R \otimes \cN(\sigma_{RA})\right)}(\varepsilon,\beta) \\
      \leq \left \lceil 2 \ln\!\left( \frac{\sqrt{\beta (1-\beta)}}{\varepsilon} \right) \!\left( \frac{1}{\left(1-\check{\alpha}(\cN) \right) T(\rho_{RA}, \sigma_{RA})} \right)^2  \right\rceil, 
    \end{multline}
    where $\alpha_{+}$ is defined in~\cref{def:alpha-plus-doeblin}, and $\check{\alpha}$ in~\eqref{eq:reverse_Doeblin}.
    The main technical tools employed in this claim are:
    \begin{equation}
      H_{1/2}\!\left( \id_R \otimes \cN (\rho_{RA}) \Vert \id_R \otimes \cN (\sigma_{RA}) \right)   \leq \left(1- \alpha_{+} (\cN)\right) H_{1/2}(\rho_{RA} \Vert \sigma_{RA}),
    \end{equation}
  which follows from~\cref{Cor:CC_f_div} by the choice $D_f= H_{1/2}$ in~\eqref{eq:Helinger_1_2},  and 
    \begin{equation}
       T\!\left( \id_R \otimes \cN (\rho_{RA}) \Vert \id_R \otimes \cN (\sigma_{RA}) \right)   \geq \left(1- \check{\alpha}(\cN) \right)T(\rho_{RA} \Vert \sigma_{RA}),
    \end{equation}
    which follows due to~\cref{prop:trace_distance_complete_exp_Doeblin_bound}.
\end{remark}

\begin{corollary}[Generalized Amplitude Damping Channels]
\label{cor:GADC-sample-complexity-bounds}
Let $\rho$ and $\sigma$ be qubit states, and let $\cN \equiv \cN_{p,\eta} $ be a generalized damping channel with parameters $\eta \in (0,1)$ and $p\in[0,1]$. Fix $\varepsilon \in (0,1)$ and $\beta \in (0,1)$. For the cases where the first two conditions of~\cref{Prop:SC_noisy} do not hold, we have the following upper and lower bound on the sample complexity:
\begin{equation}
 \frac{1-\frac{\varepsilon(1-\varepsilon)}{\beta (1-\beta)}}{ \sqrt{\eta} H_{1/2}(\rho \Vert \sigma)}\leq  \\ 
\mathrm{SC}_{\left(\cN(\rho),\cN(\sigma)\right)}(\varepsilon,\beta) \leq \left \lceil 2 \ln\!\left( \frac{\sqrt{\beta (1-\beta)}}{\varepsilon} \right) \!\left( \frac{1}{\eta T(\rho, \sigma)} \right)^2  \right\rceil.
\end{equation}   
\end{corollary}

\begin{proof}
    The proof follows by substituting $1-\check{\alpha}(\cN_{p,\eta}) =\eta$ for the generalized damping channels~(\cite[Proposition~C.1]{hirche2024quantum}) and $\eta_{\Tr}(\cN_{p,\eta}) =\sqrt{\eta}$ (\cref{lem:GADC_contr_coef}). For the lower bound, recall that the lower bounds in~\cref{Prop:SC_noisy} are tighter with the use of $\eta_{\Tr}(\cN)$ instead of $1-\alpha(\cN)$ since $\eta_{\Tr}(\cN) \leq 1-\alpha(\cN)$.
\end{proof}

\cref{cor:GADC-sample-complexity-bounds} provides an example of a noisy channel for which the sample complexity of noisy hypothesis testing  of the states $\rho$ and $\sigma$ can be finite with $\eta <1$, whenever $\cN_{p,\eta}(\rho) \neq \cN_{p,\eta}(\sigma)$.

\begin{remark}[Another Lower Bound for the special setting of GAD channel]
It is possible to obtain a lower bound that inversely depends on $\left[T(\rho,\sigma)\right]^2$ as in Item 4 of~\cref{Prop:SC_noisy} considering the minimum eigenvalues of the Choi operator of GAD channel:
\begin{equation}
    \lambda_{\min}\!\left(\Gamma^{\cN_{p, \eta}}\right) = \frac{1+\eta}{2}- \sqrt{\eta +(1-\eta)^2 (1-2p)^2}.
\end{equation}
This will lead to $\lambda_{\min} (\cN(\rho)) \geq \lambda_{\min}\!\left(\Gamma^{\cN_{p, \eta}}\right)$. This will be strictly greater than 1 when $p \notin \{0,1\}$. In particular for $p \in (0,1)$, we have that 
\begin{equation}
\mathrm{SC}_{\left(\cN_{p,\eta}(\rho),\cN_{p,\eta}(\sigma)\right)}(\varepsilon,\beta)  \geq \frac{ \lambda_{\min}\!\left(\Gamma^{\cN_{p, \eta}}\right)}{4} \ln\!\left( \frac{\beta (1-\beta)}{\varepsilon (1-\varepsilon)} \right) \left(\frac{1}{ \sqrt{\eta} T(\rho, \sigma)} \right)^2.
\end{equation}
\end{remark}

\subsection{Fairness of Noisy Quantum Models}

\label{sec:stability-fairness-q-learners}

Decision-making models are prone to unfair treatments to individuals based on different characteristics including their gender, race, etc. With the development of quantum learning models, there is also a risk of implicit biases and unfair treatments to quantum data generated by various sources. Thus, it is important to quantify fairness in quantum learning models. Most current quantum devices are noisy, and so we focus here on how fair are noisy quantum channels when used in a learning model.

Quantum fairness refers to treating all input states equally, meaning that all pairs of input states that are close in some distance metric (e.g., close in normalized trace distance) should yield similar outcomes when processed by a quantum channel~\cite{fairnessQ_verifying22}.
Define $\cA \coloneqq \cM \circ \cE$, which is a quantum-to-classical channel in which a quantum channel~$\cE$ is followed by a measurement channel comprised of a POVM $(M_i)_{i \in \cO}$. With that, quantum fairness is defined in~\cite{fairnessQ_verifying22} as follows.

\begin{definition}[$(\gamma,\beta)$-Fairness~\cite{fairnessQ_verifying22}]
 Let $\cA = \cM \circ \cE$, and let $\hat{D}(\cdot \| \cdot)$ and $d(\cdot \| \cdot)$ be distance metrics on $\cD(\cH)$ and $\cD(\cO)$, respectively. Fix $0 < \gamma,\beta \leq 1$. Then the decision model $\cA$ is $(\gamma,\beta)$ fair if for all $\rho,\sigma \in \cD(\cH)$ such  that $\hat{D}(\rho \| \sigma) \leq \gamma$,
\begin{equation}
d\!\left( \cA(\rho)\| \cA(\sigma) \right) \leq \beta.
\end{equation}
\end{definition}

We next show that noisy quantum channels provide fairness to quantum data when used as a quantum learning model.

\begin{proposition}[Fairness Guarantee from Noisy Channels]
\label{prop:Privacy-implying-fairness}
Suppose that $\hat{D}(\rho \| \sigma)= \frac{1}{2} \left \| \rho- \sigma \right\|_1 $ and $d\!\left( \cA(\rho)\| \cA(\sigma) \right)= \frac{1}{2} \sum_i \!\left| \Tr\!\left[M_i \cE(\rho-\sigma) \right]\right |$.
   We have that  $\cA=\{\cE, \{M_i\}_{i \in \cO} \}$ is $(\gamma, \beta_\gamma)$-fair for all $\rho,\sigma \in \cD(\cH)$ such  that $\hat{D}(\rho \| \sigma) \leq \gamma$, where 
    \begin{equation}
       \beta_\gamma \coloneqq  \gamma \left(1-\alpha(\cA) \right),
    \end{equation}
with $\alpha(\cN)$ defined in~\eqref{eq:alpha_cN_def}.
\end{proposition}

\begin{proof}
    The proof follows analogously to the proof of~\cite[Proposition~9]{nuradha2024contraction} by bounding the contraction of trace distance under the noisy channels in terms of the Doeblin coefficient of the channel~$\cA$ (in particular by $1-\alpha(\cA)$).
\end{proof}

\subsection{Mixing, Indistinguishability, and Decoupling Times of Quantum Processes}\label{subS:Mixing_time_of_Quantum_Processes}

By definition, contraction coefficients capture how a measure of dissimilarity that satisfies the data-processing inequality, i.e., a generalized divergence $\mathbb{D}$, is altered after the action of a channel. It is for this reason that they place fundamental limits on tasks that require distinguishing the outputs of a channel. As quantum channels are fundamentally descriptions of dynamical systems, contraction coefficients also capture at what speed, if at all, dynamical systems converge to a fixed point or become independent of the input. In this section, we show that the quantum Doeblin coefficient can be used to bound the mixing time and acts as a quantum generalization of a `coefficient of ergodicity.' Moreover, as the induced quantum Doeblin coefficient also captures these quantities, we introduce the `decoupling time,' a quantum generalization of mixing time that is specifically captured by complete contraction coefficients.

For clarity, we begin by setting certain terminology for discrete-time quantum Markov chains. A quantum \textit{homogeneous} Markov chain is a process where, at each time step, the same quantum channel $\cN_{A \to A}$ occurs; i.e., for all $n \in \mbb{N}$, the total process is as follows:
\begin{equation}
\cN^{n} \coloneq \bigcirc_{i \in [n]} \cN = \underbrace{\cN \circ \cdots \circ \cN}_{n \text{ times}} \ .
\end{equation}
We thus will talk of a `homogeneous Markov chain $\cN_{A \to A}$' as this completely specifies the Markov chain. A quantum \textit{inhomogeneous} Markov chain $(\cG_{n})_{n \in \mbb{N}}$ is identified by a sequence of quantum channels $(\cN^{i}_{A_{i} \to A_{i+1}})_{i \in \mbb{N}}$ where, for all $n \in \mbb{N}$, $\cG_{n} \coloneq \bigcirc_{i \in [n]} \cN^{i}$. Note that a homogeneous Markov chain is clearly a special case of an inhomogeneous one. These definitions directly generalize the notion of homogeneous and non-homogeneous Markov chains in the classical literature (see, e.g.,~\cite{seneta2006book, levin2017markov}). Finally, given an inhomogeneous Markov chain $(\cG_{n})_{n \in \mbb{N}}$ defined via $(\cN^{i}_{A_{i} \to A_{i+1}})_{i \in \mbb{N}}$, for all $k \leq m \in \mbb{N}$, we let $\cG_{k:m} \coloneq \bigcirc_{i = k}^{m} \cN^{i}$, which allows us to select sub-processes of the inhomogeneous Markov chain.

\subsubsection{Mixing Processes and Mixing Times}
We begin by defining notions of convergence, or `mixing,' for quantum Markov chains. We then show how mixing times can be bounded by contraction coefficients.
\begin{definition}(\cite{burgarth2013ergodic, George-2024ergodic}) \label{def:mixing-Markov-chain} We say a homogeneous Markov chain $\cN_{A \to A}$ is \textit{mixing} if it has a unique fixed point state $\omega \in \Density(A)$ and
\begin{align}
    \lim_{n \to \infty} \left\Vert \cN^{n}(\rho) - \omega \right\Vert_{1} = 0 \quad \forall \rho \in \Density(A) \ . 
\end{align}
    We say it is \textit{strongly mixing} if $\omega$ is furthermore full rank, which is also known as being primitive \cite{wolf2012quantum}. 
\end{definition}
We remark that value of strongly mixing is that it often allows to establish particularly powerful claims about convergence to the fixed point; see, e.g., \cite[Theorem~6.7]{wolf2012quantum}, \cite[Theorem~50 \& Theorem~72]{George-2024ergodic}, \cite[Theorem 8]{george2025unifiedapproachquantumcontraction}. We can generalize these definitions straightforwardly to inhomogeneous Markov chains.
\begin{definition} 
For an inhomogeneous Markov chain $(\cG_{n})_{n \in \mbb{N}}$, we say it is mixing if there exists $\omega \in \Density(A)$ such that
    \begin{align}
    \lim_{n \to \infty} \left\Vert \cG_{n}(\rho) - \omega \right\Vert_{1} = 0 \quad \forall \rho \in \Density(A) \ .
    \end{align}
We say it is strongly mixing if $\omega$ is furthermore full rank.
\end{definition}

Note that a mixing Markov chain destroys information about the initial state, but also accomplishes something stronger---namely, converges to a fixed state. We define the following weaker notions of mixing by considering only the former convergence property of destroying information about the initial state.
\begin{definition}\label{def:weakly-mixing}
    An inhomogeneous Markov chain $(\cG_{n})_{n \in \mbb{N}}$ is weakly mixing if 
    \begin{align}
        \lim_{n \to \infty} \left\Vert \cG_{n}(\rho) - \cG_{n}(\sigma) \right\Vert_{1} = 0 \quad \forall \rho,\sigma \in \Density(A) \ .
    \end{align}
    Moreover, we say an inhomogeneous Markov chain $(\cG_{n})_{n \in \mbb{N}}$ is always weakly mixing if for all $k \geq 1$, 
    \begin{align}
        \lim_{n \to \infty} \left\Vert \cG_{k:n}(\rho) - \cG_{k:n}(\sigma) \right\Vert_{1} = 0 \quad \forall \rho,\sigma \in \Density(A) \ ,
    \end{align}
    where $\cG_{k:n} \coloneq \bigcirc_{i = k}^{n} \cN^{i}$ is a sub-process of the inhomogeneous Markov chain.
\end{definition}
\noindent We remark that being `weakly mixing' matches the original \textit{operational} notion of being `weakly ergodic' \cite{hajnal1956ergodic}, and being `always weakly mixing' operationally aligns with the stricter notion of being `weakly ergodic' introduced subsequently \cite{hajnal1958weak}. However, neither formally generalizes the notion of weakly ergodic, which has to do with the convergence of the entries of the transition matrix describing the Markov chain (see \cite{seneta2006book} for an in-depth study of weak ergodicity). In effect, the value of being `always weakly mixing' is that it guarantees the process does not cease to be weakly mixing at some finite time.

\begin{figure}
    \centering
    \begin{minipage}{0.48\textwidth}
        \centering
        \includegraphics[width=0.95\textwidth]{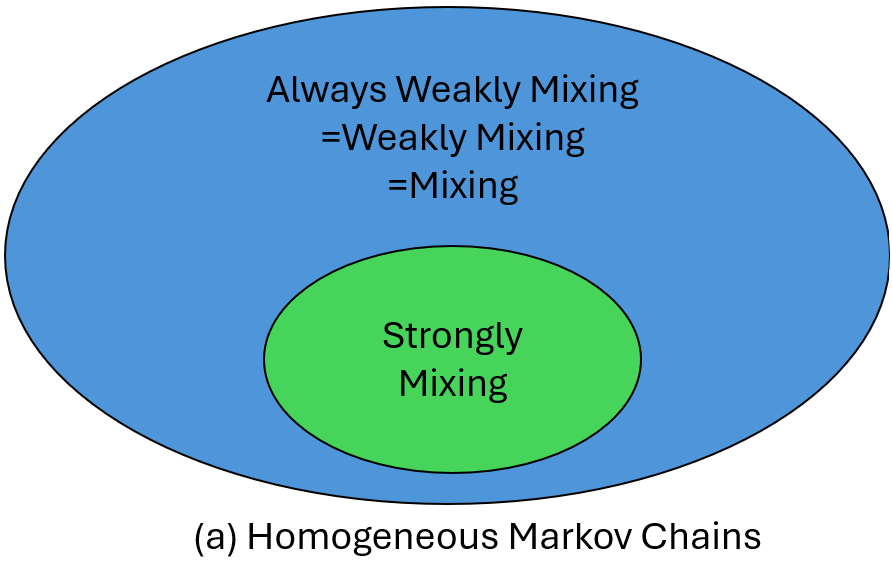}
    \end{minipage}\hfill
    \begin{minipage}{0.48\textwidth}
        \centering
        \includegraphics[width=0.95\textwidth]{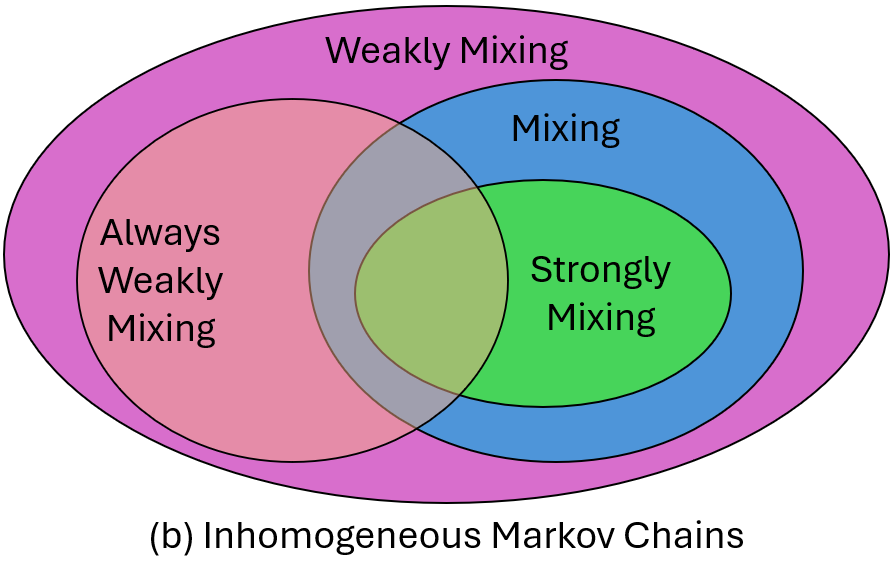}
    \end{minipage}
    \caption{The diagram of containments of notions of mixing for (a) homogeneous and (b) inhomogeneous Markov chains as established in \cref{prop:mixing-class-containment}.}
    \label{fig:mixing-time-containments}
\end{figure}

We also remark that the definitions of mixing measure ``strength" in terms of the convergence guarantees. We provide a complete characterization of the containments between the different notions of mixing (see \cref{fig:mixing-time-containments} for a depiction).
\begin{proposition}\label{prop:mixing-class-containment}
    For homogeneous Markov chains, we have the following set containments for the types of mixing:
    \begin{align}
        \operatorname{Strongly \; Mixing} \subsetneq \operatorname{Mixing}  = \operatorname{Weakly\;Mixing} = \operatorname{Always\;Weakly\;Mixing}   \ .
    \end{align}
    For inhomogeneous Markov chains,
    \begin{equation}
        \begin{matrix}
            \operatorname{Strongly \; Mixing} \subsetneq \operatorname{Mixing} \subsetneq \operatorname{Weakly\;Mixing} \\
            
            \operatorname{Always\;Weakly\;Mixing} \subsetneq \operatorname{Weakly\;Mixing} \\
            
            \operatorname{Mixing} \not \subseteq \operatorname{Always\; Weakly\;Mixing} \not \subseteq  \operatorname{Mixing} \\
            
            \operatorname{Strongly\;Mixing} \not \subseteq \operatorname{Always\; Weakly\;Mixing} \not \subseteq  \operatorname{Strongly\;Mixing} \\
            \operatorname{Always\;Weakly\;Mixing}\cap\operatorname{Mixing}\neq\emptyset \\
            \operatorname{Always\;Weakly\;Mixing}\cap\operatorname{Strongly\;Mixing}\neq\emptyset
        \end{matrix} \ .
    \end{equation}
\end{proposition}

\begin{proof}
   As the proof is long but straightforward, we relegate it to~\cref{app:Mixing_class_proof}.
\end{proof}

The above notions of mixing are all asymptotic. In practical applications, namely, sampling algorithms, one is generally interested in how many iterations of a mixing channel~$\cN$ with fixed point $\omega$ are necessary for every input to become $\delta$-close to $\omega$ with respect to $\mbb{D}$, which is known as a mixing time.
\begin{definition}\label{def:mixing-time}
    Given a divergence $\mbb{D}$, a channel $\cN$ with a unique fixed point $\omega$, and $\delta \geq 0$, the $\mbb{D}$-mixing time is defined as
    \begin{align}
        t_{\operatorname{mix}}^{\mbb{D}}(\cN,\delta) \coloneq \inf\!\left \{n \in \mbb{N}: \mbb{D}\!\left(\cN^{n}(\rho), \omega \right) \leq \delta \quad \forall \rho \in \Density(A) \right \} \ , 
    \end{align}
    which is equal to $+\infty$ when no such $n \in \mbb{N}$ exists.
\end{definition}
It is known that the mixing time is generally controlled by the input-dependent contraction coefficient (\cref{def:input-dep-contraction-coeffs}) \cite{Raginsky2016,George-2024ergodic,george2025unifiedapproachquantumcontraction}. However, this can be further relaxed to the contraction coefficient and the complete contraction coefficient via the same proof method.
\begin{proposition}\label{prop:mixing-time-bound-in-terms-of-contraction-coeff}
    Let $\cN_{A \to A}$ be a mixing channel with fixed point $\omega$. Then for a divergence~$\mbb{D}$ and $\delta \in (0,C_{\mbb{D}}(\omega)]$,
    \begin{align}
        t_{\operatorname{mix}}^{\mbb{D}}(\cN,\delta) \leq \left \lceil \frac{\log\!\left(C_{\mbb{D}}(\omega)/\delta\right)}{\log\!\left(1/\eta_{\mbb{D}}(\cN)\right)} \right \rceil, 
    \end{align}
    where $C_{\mbb{D}}(\omega) \coloneq \sup_{\rho \in \Density(A)} \mbb{D}(\rho \Vert \omega)$.
\end{proposition}

\begin{proof}
    Consider that
    \begin{equation}\label{eq:mixing-time-bound-proof}
    \begin{aligned}
        \mbb{D}\!\left(\cN^{n}(\rho) \Vert \omega \right) = \mbb{D}\!\left(\cN^{n}(\rho) \Vert \cN^{n}(\omega)\right) &\leq \eta_{\mbb{D}}(\cN) \, \mbb{D}(\cN^{n-1}(\rho) \Vert \cN^{n-1}(\omega)) \\
        &\leq [\eta_{\mbb{D}}(\cN)]^{n}\, \mbb{D}(\rho \Vert \omega),  
    \end{aligned}
    \end{equation}
    where the equality is due to the assumption that $\omega$ is a fixed point, the first and second inequalities use the definition of contraction coefficient from \eqref{eq:def-gen-contra-coef}, and the third inequality follows from \Cref{prop:complete-contraction-coeff-properties}. As it has to hold for all states, we apply the bound $\mbb{D}(\rho \Vert \omega) \leq C_{\mbb{D}}(\omega)$ in order to obtain an input-independent bound. Thus, it suffices to solve $[\eta_{\mbb{D}}(\cN)]^{n} C_{\mbb{D}}(\omega) \leq \delta$ for $n$. By algebraic manipulations, we find that $n \geq \log\!\left(C_{\mathbb{D}}(\omega)/\delta \right)/\log\!\left((\eta_{\mbb{D}}(\cN))^{-1}\right)$. Taking the ceiling function such that $n$ is a natural number completes the proof. 
\end{proof}

We remark that the above is somewhat limited, as often $C_{\mbb{D}}(\omega) = +\infty$, although in many cases this can be handled if $\omega$ is full rank so that a reverse Pinsker inequality can be applied \cite{George-2024ergodic,george2025unifiedapproachquantumcontraction}.\footnote{Alternatively, if one is willing to make the stronger assumption that their sampling procedure is guaranteed to be initialized with a specific $\rho$, they could simply compute $\mbb{D}(\rho \Vert \omega)$ to avoid the issue altogether.} However, for the case of the trace distance, we obtain the following corollary, which needs no rank constraint.
\begin{corollary}\label{cor:mixing-time-bound-via-q-Doeblin}
    Let $\cN_{A \to A}$ be a mixing channel with fixed point $\omega$. Then for all $\delta \in (0,1)$, 
    \begin{align}
        t_{\operatorname{mix}}^{T}(\cN,\delta) \leq \left \lceil \frac{\log(\delta^{-1})}{\log([1-\alpha_{I}(\cN)]^{-1})} \right \rceil  \leq \left \lceil \frac{\log(\delta^{-1})}{\log([1-\alpha(\cN)]^{-1})} \right \rceil \ . 
        \label{eq:mix-time-bnd-TD}
    \end{align}
\end{corollary}

\begin{proof}
    This follows from applying \Cref{prop:mixing-time-bound-in-terms-of-contraction-coeff}, as $C_{T}(\omega) \leq 1$ for every state $\omega$. Then one obtains the first inequality from \eqref{eq:trace-distance-CC-to-induced-doeblin} and the second then follows from \Cref{prop:induced-doeblin-larger-than-doeblin}.
\end{proof}

We remark that, while the final upper bound in \eqref{eq:mix-time-bnd-TD} is likely generally loose, it can still be useful. First, \cref{thm:qubit-faithfulness} implies that for time-homogeneous Markov chains defined via a qubit channel, either the system does not contract or \cref{cor:mixing-time-bound-via-q-Doeblin} obtains a non-trivial bound, so there exist clear use cases. Second, while concurrent work \cite{george2025unifiedapproachquantumcontraction} finds an efficient method for obtaining bounds with input-dependent contraction coefficients, it requires the fixed point to be full rank. While under those conditions it is possible that there are cases where that work's method obtains tighter mixing times due to the hierarchy of quantum contraction coefficients,
\cref{cor:mixing-time-bound-via-q-Doeblin} applies for channels with non-full rank fixed points and thus includes cases that cannot be addressed by previous works. As a concrete example, any generalized amplitude damping channel with $p=1$, $\cN_{1,\eta}$, has the fixed point $\dyad{0}$, which has unit rank, but $\alpha(\cN_{1,\eta}) > 0$ unless $\eta = 0$ by \cref{lem:Doeblin_GAD}.

\subsubsection{Indistinguishability Times}
As noted previously, it is known that mixing times are naturally addressed by input-dependent contraction coefficients. It is thus natural to ask what the contraction coefficient more directly captures. In the classical literature, contraction coefficients are identified as `coefficients of ergodicity' \cite{seneta2006book} as they capture the finite-time notion of being weakly mixing (\Cref{def:weakly-mixing}).\footnote{In fact, it is claimed in \cite{seneta1973historical} that Doeblin announced his coefficient was relevant to a variant of weak ergodicity, but he never published the results. Ref.~\cite{seneta1973historical} conjectures this could be because Doeblin died prematurely during World War II.} We briefly and formally show this to be the case in the quantum setting, as it will help motivate decoupling times, which are introduced in~\cref{Sec:Decoupling_times}.
\begin{definition}\label{def:distinguishability-time}
    Let $\mbb{D}$ be a divergence and $\cN$ a quantum channel. Then the indistinguishability time is defined for $\delta \geq 0$ as 
    \begin{align}
        t_{\operatorname{ind}}^{\mbb{D}}(\cN,\delta) \coloneqq \inf \! \left \{ n \in \mbb{N} : \mbb{D}\!\left(\cN^{n}(\rho) \Vert \cN^{n}(\sigma)\right) \leq \delta \, \quad \forall \rho,\sigma \in \Density(A) \right \} \ , 
    \end{align}
    which is equal to $+\infty$ when no such $n \in \mbb{N}$ exists.
\end{definition}
\noindent Note this definition is more demanding than the mixing time (\cref{def:mixing-time}) as it does not require the reference state to be fixed.

\begin{proposition}\label{prop:indist-time-bounds}
    \sloppy For an arbitrary inhomogeneous Markov chain $(\cG_{n})_{n \in \mbb{N}}$ defined by channels $(\cN^{i}_{A_{i} \to A_{i+1}})_{i \in \mbb{N}}$ and for all $\rho,\sigma \in \Density(A_{1})$,
    \begin{align}
        \mbb{D}(\cG_{n}(\rho) \Vert \cG_{n}(\sigma)) \leq \mbb{D}(\rho \Vert \sigma) \prod_{i \in [n]} \eta_{\mbb{D}}(\cN^{i}) \leq \mbb{D}(\rho \Vert \sigma) \prod_{i \in [n]} \eta^{c}_{\mbb{D}}(\cN^{i}) \ . 
    \end{align}
    
    In particular, for a homogeneous Markov chain $\cN_{A \to A}$ and $\delta \in (0,C_{\mbb{D}}]$,
    \begin{align}
        t_{\operatorname{ind}}^{\mbb{D}}(\cN,\delta) \leq \left \lceil \frac{\log(C_{\mbb{D}}/\delta)}{\log(1/\eta_{\mbb{D}}(\cN))} \right \rceil  \leq \left \lceil \frac{\log(C_{\mbb{D}}/\delta)}{\log(1/\eta^{c}_{\mbb{D}}(\cN))} \right \rceil \ ,  
    \end{align}
    where $C_{\mbb{D}} \coloneq \sup_{\rho,\sigma \in \Density(A)} \mbb{D}(\rho \Vert \sigma)$.
\end{proposition}
\begin{proof}
    By the same argument as in \eqref{eq:mixing-time-bound-proof},
    \begin{align}
        \mbb{D}(\cG_{n}(\rho) \Vert \cG_{n}(\sigma)) \leq \mbb{D}(\rho \Vert \sigma) \prod_{i \in [n]} \eta_{\mbb{D}}(\cN^{i}) \leq \mbb{D}(\rho \Vert \sigma) \prod_{i \in [n]} \eta^{c}_{\mbb{D}}(\cN^{i}) \ . 
    \end{align}
    In the case the Markov chain is homogeneous, the product simplifies to $\left[\eta_{\mbb{D}}(\cN)\right]^{n}$ as before. Using $\mbb{D}(\rho \Vert \sigma) \leq C_{\mbb{D}}$ and following the same argument as in the proof of \Cref{prop:mixing-time-bound-in-terms-of-contraction-coeff} completes the proof.
\end{proof}

While the issue is again that $C_{\mbb{D}}$ can often be $+\infty$, in the case of trace distance we obtain the following tractable corollary involving the induced Doeblin coefficient $\alpha_{I}$ from \cref{prop:equality-of-induced-doeblins}.
\begin{corollary}\label{cor:q-Doeblin-coeff-of-ergodicity}
    \sloppy For an arbitrary inhomogeneous Markov chain $(\cG_{n})_{n \in \mbb{N}}$ defined by channels $\left(\cN^{i}_{A_i \to A_{i+1}}\right)_{i \in \mbb{N}}$ and for all $\rho,\sigma \in \Density(A_{1})$,
    \begin{align}
        T(\cG_{n}(\rho) \Vert \cG_{n}(\sigma)) \leq \prod_{i \in [n]} (1-\alpha_{I}(\cN^{i})) \leq \prod_{i \in [n]} (1-\alpha(\cN^{i})) \ . 
    \end{align}
    
    Let $\cN_{A \to A}$ be  a time-homogeneous Markov chain. Then for all $\delta \in (0,1)$, 
    \begin{align}
        t_{\operatorname{ind}}^{T}(\cN,\delta) \leq \left \lceil \frac{\log(\delta^{-1})}{\log([1-\alpha_{I}(\cN)]^{-1})} \right \rceil  \leq \left \lceil \frac{\log(\delta^{-1})}{\log([1-\alpha(\cN)]^{-1})} \right \rceil \ . 
    \end{align}
\end{corollary}
\begin{proof}
    The proof follows similarly to the proof of \Cref{cor:mixing-time-bound-via-q-Doeblin}.
\end{proof}

We also observe the following.
\begin{proposition} \label{prop:inhomo_weakly}
    An inhomogeneous Markov chain $(\cG_{n})_{n \in \mbb{N}}$ defined by channels $\left(\cN^{i}_{A_i \to A_{i+1}}\right)_{i \in \mbb{N}}$ is always weakly mixing (\Cref{def:weakly-mixing}) if for all $k \in \mbb{N}$, $\prod_{i =k}^{\infty} (1-\alpha_{I}(\cN^{i})) = 0$.
\end{proposition}

\begin{proof}
    If for all $k \in \mbb{N}$, $\prod_{i =k}^{\infty} (1-\alpha_{I}(\cN^{i})) = 0$, then for all $\rho,\sigma \in \Density(A)$,
    \begin{align}
        \lim_{n \to \infty} \left\Vert \cG_{k:n}(\rho) - \cG_{k:n}(\sigma) \right\Vert_{1} \leq 2 \cdot \lim_{n \to \infty} \prod_{i =k}^{n} (1-\alpha_{I}(\cN^{i})) = 0 \ , 
    \end{align}
    which is the definition of always weakly mixing.
\end{proof}
Thus, by~\cref{prop:inhomo_weakly}, we conclude that (input-independent) contraction coefficients capture the notion of being weakly mixing.

\subsubsection{Decoupling Times} 

\label{Sec:Decoupling_times}

So far, we have shown how to use the quantum Doeblin coefficient $\alpha(\cN)$ to bound the speed of mixing and indistinguishability. However, we have also shown that these are captured by the induced quantum Doeblin coefficient $\alpha_{I}(\cN)$. This would suggest that the quantum Doeblin coefficient should be relevant for an even stronger demand, in an analogous manner to how indistinguishability time is a stronger demand than mixing time and operationally separates input-dependent and input-independent contraction coefficients. Here we show this is indeed the case as the complete contraction coefficient, and thus quantum Doeblin coefficient, capture the notion of being `weakly decoupling' and can bound the `decoupling time.'

We begin with natural definitions that strengthen each version of a channel from being mixing to being decoupling.
\begin{definition}\label{def:decoupling-Markov-chain}
    A homogeneous Markov chain $\cN_{A \to A}$ that is mixing with fixed point $\omega$ is decoupling if for every finite Hilbert space $R$,
    \begin{align}
    \lim_{n \to \infty} \left\Vert (\id_{R} \otimes \cN^{n})(\rho) - \rho_{R} \otimes \omega_{A} \right\Vert_{1} = 0 \quad \forall \rho \in \Density(R \otimes A) \ . 
\end{align}
    Moreover, we say it is strongly decoupling or completely primitive if $\omega$ is full rank.
\end{definition}
It is easy to see that the above definition captures decoupling as it makes the $R$ and $A$ systems independent, which is the key idea in decoupling \cite{Schumacher2002,hayden2008decoupling,dupuis2010decouplingapproachquantuminformation} (see also \cite{anshu2017quantum,colomer2023decoupling,cheng2023quantumbroadcastchannelsimulation}).
However, like being mixing, it includes the stronger assumption that it converges to a fixed state. Thus, we can define the weaker notion of weakly decoupling.
\begin{definition}\label{def:weakly-decoupling}
    An inhomogeneous Markov chain $(\cG_{n})_{n \in \mbb{N}}$ is weakly decoupling if for every finite Hilbert space $R$ and all $\rho_{RA},\sigma_{RA} \in \Density(R \otimes A)$ such that $\rho_{R} = \sigma_{R}$,
    \begin{align}\label{eq:weakly-decoupling}
        \lim_{n \to \infty} \left\Vert (\id_{R} \otimes \cG_{n})(\rho) - (\id_{R} \otimes \cG_{n})(\sigma) \right\Vert_{1} = 0 \ .
    \end{align}
    Moreover, we say an inhomogeneous Markov chain $(\cG_{n})_{n \in \mbb{N}}$ is always weakly decoupling if for all $k \geq 1$, for every finite Hilbert space $R$, and for all $\rho_{RA},\sigma_{RA} \in \Density(R \otimes A)$ such that $\rho_{R} = \sigma_{R}$,
    \begin{align}
        \lim_{n \to \infty} \left\Vert (\id_{R} \otimes \cG_{k:n})(\rho) - (\id_{R} \otimes \cG_{k:n})(\sigma) \right\Vert_{1} = 0 \quad \forall \rho,\sigma \in \Density(A) \ ,
    \end{align}
    where $\cG_{k:n} \coloneq \bigcirc_{i = k}^{n} \cN^{i}$ is a sub-process of the inhomogeneous Markov chain.
\end{definition}
\noindent As in the mixing case, we remark that being `always weakly decoupling' precludes that the process is decoupling only for some finite time.

We remark that the given definition for weakly decoupling does not make the decoupling explicit. Here, we show that it is equivalent to a definition that makes the decoupling explicit.
\begin{proposition}
     An inhomogeneous Markov chain $(\cG_{n})_{n \in \mbb{N}}$ being weakly decoupling is equivalent to the following condition, holding for every finite Hilbert space $R$ and all $\rho_{RA},\sigma_{RA} \in \Density(R \otimes A)$ such that $\rho_{R} = \sigma_{R}$,
    \begin{align}\label{eq:weakly-decoupling-explicit}
        \lim_{n \to \infty} \left\Vert (\id_{R} \otimes \cG_{n})(\rho_{RA}) - \sigma_{R} \otimes \cG_{n}(\sigma_{A}) \right\Vert_{1} = 0 \ ,
    \end{align}
    and similarly for always weakly decoupling.
\end{proposition}

\begin{proof}
    \cref{eq:weakly-decoupling} implies \eqref{eq:weakly-decoupling-explicit} by letting $\rho_{RA}$ be arbitrary and choosing $\sigma_{RA} = \sigma_{R} \otimes \sigma_{A} = \rho_{R} \otimes \sigma_{A}$. To prove the reverse implication, suppose that \eqref{eq:weakly-decoupling-explicit} holds for every finite Hilbert space $R$, and all $\hat{\rho}_{RA},\hat{\sigma}_{RA} \in \Density(R \otimes A)$ such that $\hat{\rho}_{R} = \hat{\sigma}_{R}$. Let $\rho_{RA},\sigma_{RA} \in \Density(R \otimes A)$ be such that $\rho_{R} = \sigma_{R}$. Then
    \begin{align}
        & \left\Vert (\id_{R} \otimes \cG_{n})(\rho_{RA}) - (\id_{R} \otimes \cG_{n})(\sigma_{RA}) \right\Vert_{1} \notag \\
        &= \left\Vert (\id_{R} \otimes \cG_{n})(\rho_{RA}) - \rho_{R} \otimes \cG_{n}(\rho_{A}) + \rho_{R} \otimes \cG_{n}(\rho_{A}) - (\id_{R} \otimes \cG_{n})(\sigma_{RA}) \right\Vert_{1} \\
        &\leq \left \Vert (\id_{R} \otimes \cG_{n})(\rho_{RA}) - \rho_{R} \otimes \cG_{n}(\rho_{A}) \right\Vert_{1} +  \left\Vert  (\id_{R} \otimes \cG_{n})(\sigma_{RA}) - \rho_{R} \otimes \cG_{n}(\rho_{A}) \right \Vert_{1} \ ,
    \end{align}
    where we have used the triangle inequality. Then,
    \begin{align}
         & \lim_{n \to \infty} \left\Vert (\id_{R} \otimes \cG_{n})(\rho_{RA}) - (\id_{R} \otimes \cG_{n})(\sigma_{RA}) \right\Vert_{1} \notag \\
         &\leq \lim_{n \to \infty} \left[ \left\Vert (\id_{R} \otimes \cG_{n})(\rho_{RA}) - \rho_{R} \otimes \cG_{n}(\rho_{A}) \right\Vert_{1} +  \left\Vert (\id_{R} \otimes \cG_{n})(\sigma_{RA}) -  \rho_{R} \otimes \cG_{n}(\rho_{A}) \right\Vert_{1} \right] \\
         &= 0 + 0 \ , 
    \end{align}
    where we have used the assumption that the limit exists for each term individually and is equal to zero (in the first case by letting $\hat{\rho} = \rho_{RA}$ and $\hat{\sigma} = \rho_{R} \otimes \rho_{A}$ and in the second case by letting $\hat{\rho} = \sigma_{RA}$ and $\hat{\sigma}_{RA} = \rho_{R} \otimes \rho_{A}$). The same argument works for always weakly decoupling. This completes the proof.
\end{proof}

\begin{figure}
    \centering
    \begin{minipage}{0.48\textwidth}
        \centering
        \includegraphics[width=0.95\textwidth]{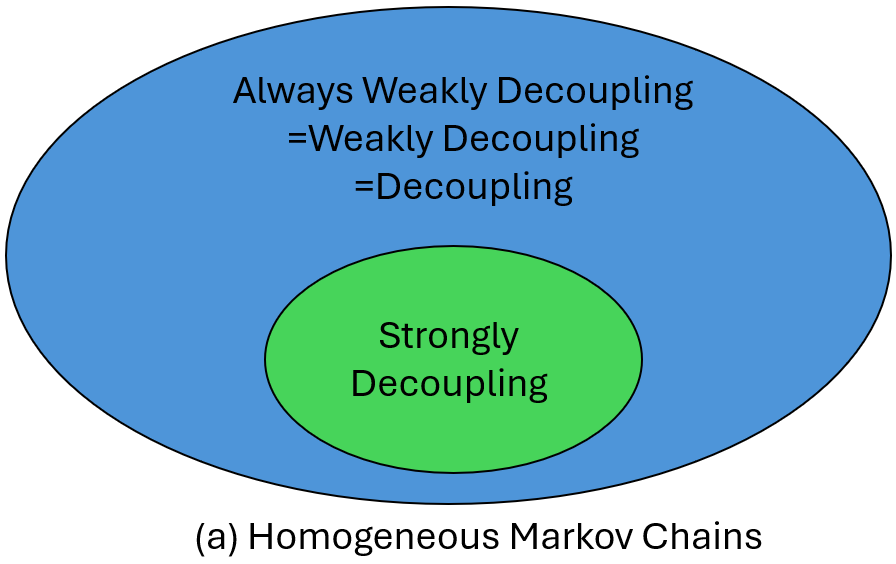}
    \end{minipage}\hfill
    \begin{minipage}{0.48\textwidth}
        \centering
        \includegraphics[width=0.95\textwidth]{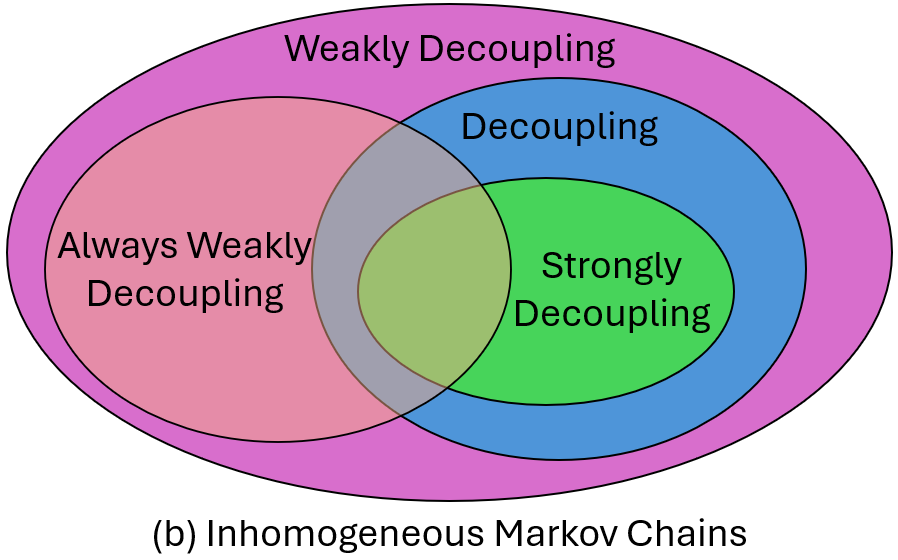}
    \end{minipage}
    \caption{The diagram of containments of notions of decoupling for (a) homogeneous and (b) inhomogeneous Markov chains as established in \cref{prop:decoupling-class-containment}. As may be seen by comparing to Fig.~\ref{fig:mixing-time-containments}, mixing and decoupling have analogous containments.}
    \label{fig:decoupling-time-containments}
\end{figure}
We end our discussion on the notions of decoupling by establishing the equivalent of \cref{prop:mixing-class-containment} for various definitions of decoupling, as shown in \cref{fig:decoupling-time-containments}.

\begin{proposition}\label{prop:decoupling-class-containment}
    For homogeneous Markov chains, we have the following set containments of the types of decoupling:
    \begin{align*}
        \operatorname{Strongly \; Decoupling} \subsetneq \operatorname{Decoupling}  = \operatorname{Weakly\;Decoupling} = \operatorname{Always\;Weakly\;Decoupling}   \ .
    \end{align*}
    For inhomogeneous Markov chains,
    \begin{equation}
        \begin{matrix}
            \operatorname{Strongly \; Decoupling} \subsetneq \operatorname{Decoupling} = \operatorname{Weakly\;Decoupling} \\
            
            \operatorname{Always\;Weakly\;Decoupling} \subsetneq \operatorname{Weakly\;Decoupling} \\
            
            \operatorname{Decoupling} \not \subseteq \operatorname{Always\; Weakly\;Decoupling} \not \subseteq  \operatorname{Decoupling} \\
            
            \operatorname{Strongly\;Decoupling} \not \subseteq \operatorname{Always\; Weakly\;Decoupling} \not \subseteq  \operatorname{Strongly\;Decoupling} \\
            \operatorname{Always\;Weakly\;Decoupling}\cap\operatorname{Decoupling}\neq\emptyset \\
            \operatorname{Always\;Weakly\;Decoupling}\cap\operatorname{Strongly\;Decoupling}\neq\emptyset
        \end{matrix} \ .
    \end{equation}
\end{proposition}
\begin{proof}
    See~\cref{app:decopuling_time_proof}.
\end{proof}

With the taxonomy of decoupling channels addressed, we define the notion of decoupling time and decoupled indistinguishability time for a divergence $\mbb{D}$. 
\begin{definition}\label{def:decoupling-time}
    Let $\mbb{D}$ be a divergence, and let $\cN_{A \to A}$ be a quantum channel with unique fixed point $\omega \in \Density(A)$. Then for all $\delta \geq 0$, the decoupling time of the homogeneous Markov chain $\cN$ is defined as
    \begin{align} \label{eq:decoupling_time}
        t_{\operatorname{dec}}^{\mbb{D}}(\cN,\delta) \coloneqq \inf \left\{ n \in \mbb{N} : \begin{matrix}   \mbb{D}((\id_{R} \otimes \cN^{n})(\rho_{RA}) \Vert \rho_{R} \otimes \omega_{A}) \leq \delta \, , \\
        \forall R \, , \, \forall \rho_{RA} \in  \Density(R \otimes A) \end{matrix} \right\} \ , 
    \end{align}
    which is equal to $+\infty$ when no such $n \in \mbb{N}$ exists. \\

    Similarly, the decoupled indistinguishability time is defined as
    \begin{align} \label{def:d_i_time}
        t_{\operatorname{d-i}}^{\mbb{D}}(\cN,\delta)
        &\coloneqq  \inf \left\{ n \in \mbb{N} : \begin{matrix}   \mbb{D}((\id_{R} \otimes \cN^{n})(\rho_{RA}) \Vert (\id_{R} \otimes \cN^{n})(\sigma_{RA})) \leq \delta \, , \\ 
          \forall R \, , \, \forall \rho_{RA},\sigma_{RA} \in \Density(R \otimes A) \ \textnormal{s.t.}  \ \rho_{R} = \sigma_{R}
          \end{matrix} \right\} \ , 
    \end{align}
    which is equal to $+\infty$ when no such $n \in \mbb{N}$ exists.
\end{definition}
\noindent Note these conditions are more demanding than mixing time or indistinguishability time, respectively.

We then obtain the following bounds in terms of complete contraction coefficients. We begin with the decoupling time.
\begin{proposition}\label{prop:decoupling-time-bounds}
     Let $\mbb{D}$ be a divergence, and let $\cN_{A \to A}$ be a quantum channel with unique fixed state $\omega \in \Density(A)$. Then, for all $\delta \in \! \left(0,\hat{C}_{\mbb{D}}(\omega)\right]$,
     \begin{align}
         t_{\operatorname{dec}}^{\mbb{D}}(\cN,\delta) \leq \left \lceil \frac{\log\!\left(\hat{C}_{\mbb{D}}(\omega)/\delta \right)}{\log\!\left(1/\eta^{c}_{\mbb{D}}(\cN,\omega)\right)} \right \rceil  \leq \left \lceil \frac{\log\!\left(\hat{C}_{\mbb{D}}(\omega)/\delta \right)}{\log\!\left(1/\eta^{c}_{\mbb{D}}(\cN) \right)} \right \rceil  ,
     \end{align}
     where $\hat{C}_{\mbb{D}}(\omega) \coloneq \sup_{\rho \in \Density(RA)} \mbb{D}(\rho_{RA} \Vert \rho_{R} \otimes \omega_{A})$. 
\end{proposition}

\begin{proof}
    The proof is effectively the same as those for mixing and indistinguishability times, but we provide it here, as the $R$ system is now relevant. Let $R$ be an arbitrary finite-dimensional Hilbert space, and let $\rho_{RA} \in \Density(R \otimes A)$. Then,
    \begin{align}
        \mbb{D}((\id_{R} \otimes \cN^{n})(\rho) \Vert \rho_{R} \otimes \omega) &= \mbb{D}((\id_{R} \otimes \cN^{n})(\rho) \Vert (\id_{R} \otimes \cN^{n})(\rho_{R} \otimes \omega)) \\
        &\leq \eta^{c}_{\mbb{D}}(\cN,\omega)\, \mbb{D}((\id_{R} \otimes \cN^{n-1})(\rho) \Vert (\id_{R} \otimes \cN^{n-1})(\rho_{R} \otimes \omega_{A})) \\ 
        &\leq \left[\eta_{\mbb{D}}^{c}(\cN,\omega)\right]^{n} \mbb{D}(\rho_{RA} \Vert \rho_{R} \otimes \omega_{A}) \ ,
    \end{align}
    where the equality follows from the assumption that $\omega$ is a fixed point of $\cN$, the first inequality follows because $\rho_{RA}$ and $\rho_{R} \otimes \omega_{A}$ have the same marginal on the $R$ system so one can use the definition of input-dependent complete contraction coefficient (\cref{def:input-dep-contraction-coeffs}) and the final inequality follows from applying the same argument $n-1$ times. Then, we use the upper bound  $\mbb{D}(\rho_{RA} \Vert \rho_{R} \otimes \omega_{A}) \leq \hat{C}_{\mbb{D}}(\omega)$. As this holds for an arbitrary $\rho_{RA}$, to obtain a bound on decoupling time, it suffices to solve the inequality $\left[\eta_{\mbb{D}}^{c}(\cN,\omega)\right]^{n}\hat{C}_{\mbb{D}}(\omega) \leq \delta$ for $n$. Using properties of the logarithm to do so, similarly as done in the proof of \cref{prop:mixing-time-bound-in-terms-of-contraction-coeff}, completes the proof for the input-dependent complete contraction coefficient, and then one may relax it to the input-independent variant.
\end{proof}

For the case of the trace distance, we obtain the following corollary:
\begin{corollary} \label{Cor:decopling_Doeblin}
     Let $\cN_{A \to A}$ be a quantum channel with unique fixed point $\omega \in \Density(A)$. Then
     \begin{align}
         t_{\operatorname{dec}}^{T}(\cN,\delta) \leq \left \lceil \frac{\log(\delta^{-1})}{\log([1-\alpha(\cN)]^{-1})} \right \rceil \ .
     \end{align}
\end{corollary}
\begin{proof}
    This follows from \cref{prop:decoupling-time-bounds}, that $\hat{C}_{T}(\omega) \leq 1$, and \cref{prop:trace_distance_complete_cont_Doeblin_bound}.
\end{proof}
\noindent Note that this means not only  have we  introduced the decoupling time, but we have also established a method for being able to efficiently compute an upper bound on it.

We now consider the decoupled indistinguishability time.
\begin{proposition}
    For an inhomogeneous Markov chain $(\cG_{n})_{n \in \mbb{N}}$ defined by channels \\ $(\cN^{i}_{A_{i} \to A_{i+1}})_{i \in \mbb{N}}$, for every finite dimensional Hilbert space $R$ and $\rho,\sigma \in \Density(R \otimes A_{1})$ such that $\rho_{R} = \sigma_{R}$, the following inequality holds:
    \begin{align}
        \mbb{D}((\id_{R} \otimes \cG_{n})(\rho) \Vert (\id_{R} \otimes \cG_{n})(\sigma)) \leq \mbb{D}(\rho \Vert \sigma)\prod_{i \in [n]} \eta^{c}_{\mbb{D}}(\cN^{i}) \ . 
    \end{align}
    
    In particular, for a homogeneous Markov chain $\cN_{A \to A}$ and $\delta \in (0,\hat{C}_{\mbb{D}}]$,
    \begin{align}
        t_{\operatorname{d-i}}^{\mbb{D}}(\cN,\delta) \leq \left \lceil \frac{\log\!\left(\hat{C}_{\mbb{D}}/\delta \right)}{\log\!\left(1/\eta^{c}_{\mbb{D}}(\cN)\right)} \right \rceil \ ,  
    \end{align}
    where $\hat{C}_{\mbb{D}} \coloneq \sup_{\rho,\sigma \in \Density(R \otimes A):
    \rho_{R} = \sigma_{R}} \mbb{D}(\rho \Vert \sigma)$. 
\end{proposition}
\begin{proof}
    The proof is the same as that for \cref{prop:indist-time-bounds}, except one needs to handle the $R$ system in the same manner as in \cref{prop:decoupling-time-bounds}, which results in using the complete contraction coefficient.
\end{proof}
Again, we obtain various special cases for the trace distance. We state the ones that we believe should be of most interest. The proofs follow by the same arguments that we used for other cases and are thus omitted.
\begin{corollary} \label{Cor:Distinguishability_time_Doeblin}
     Let $\cN_{A \to A}$ be define a time-homogeneous Markov chain. Then for all $\delta \in (0,1)$, 
    \begin{align}
        t_{\operatorname{d-i}}^{T}(\cN,\delta) \leq \left \lceil \frac{\log(\delta^{-1})}{\log([1-\alpha(\cN)]^{-1})} \right \rceil \ . 
    \end{align}
\end{corollary}

\begin{proposition}
    An inhomogeneous Markov chain $(\cG_{n})_{n \in \mbb{N}}$ defined by channels $(\cN^{i}_{A\to A_{i+1}})_{i \in \mbb{N}}$ is always weakly decoupling (\Cref{def:weakly-decoupling}) if for all $k \in \mbb{N}$, $\prod_{i =k}^{\infty} (1-\alpha_{I}(\cN^{i})) = 0$.
\end{proposition}

In summary, what we have seen in this section is that for every divergence $\mbb{D}$, channel $\cN$, and fixed $\delta \geq 0$, we have the ordering on operational tasks and identification of which contraction coefficient controls them (see \cref{fig:contraction-relations}). Moreover, we have shown that in the case $\mbb{D} = T$, we can control the finite-time convergence of each task using the quantum Doeblin coefficient, which is efficiently computable, and this is nontrivial as long as $\alpha(\cN) < 1$.
\begin{figure}
    \begin{center}
    \begin{tikzpicture}
        \draw
            (-3,0) node[minimum width = 2.5cm,circle,fill=blue!20] (mix) {$
            \begin{matrix}
                \eta_{\mbb{D}}(\cN,\omega) \\
            \downarrow \\
            t_{\text{mix}}^{\mbb{D}}(\cN,\delta)
            \end{matrix}$}
            ++(3,2) node[minimum width = 2.5cm,circle,fill=blue!20] (decouple) {$
            \begin{matrix}
                \eta^{c}_{\mbb{D}}(\cN,\omega) \\
            \downarrow \\
            t_{\text{dec}}^{\mbb{D}}(\cN,\delta)
            \end{matrix}$}
            ++(0,-4) node[minimum width = 2.5cm,circle,fill=blue!20] (indist) {$
            \begin{matrix}
                \eta_{\mbb{D}}(\cN) \\
            \downarrow \\
            t_{\text{ind}}^{\mbb{D}}(\cN,\delta)
            \end{matrix}$}
            ++(3,2) node[minimum width = 2.5cm, circle,fill=blue!20] (dec-indist) {$
            \begin{matrix}
                \eta^{c}_{\mbb{D}}(\cN) \\
            \downarrow \\
            t_{\text{d-i}}^{\mbb{D}}(\cN,\delta)
            \end{matrix}$}
            ++ (-1.45,1) node[node font = \huge, rotate=-20] (leq-1) {$\leq$}
            ++ (0,-2) node[node font = \huge, rotate=20] (leq-2) {$\leq$}
            ++ (-3,0) node[node font = \huge, rotate=-20] (leq-3) {$\leq$}
            ++ (0,2) node[node font = \huge, rotate=20] (leq-4) {$\leq$}
            ;
    \end{tikzpicture}
    \end{center}
    \caption{Depiction of relations between different tasks related to the convergence of quantum dynamical systems and the contraction coefficient, as established in this section. The $\leq$ sign denotes the inequality between the operational quantities as well as the inequality between the corresponding contraction coefficients. The symbol $\downarrow$ denotes that the given contraction coefficient controls the finite-time convergence of the corresponding task. Here, $\eta_{\mathbb{D}}(\cN, \omega), \eta_{\mathbb{D}}(\cN), \eta_{\mathbb{D}}^c(\cN, \omega)$, and $\eta_{\mathbb{D}}^c(\cN)$ are defined in~\eqref{eq:input_d_CC},~\eqref{eq:def-gen-contra-coef},~\eqref{eq:input_d_Com_CC}, and~\eqref{eq:Comp_CC}, respectively. Also, $t_{\text{mix}}^{\mbb{D}}(\cN,\delta)$, $t_{\text{ind}}^{\mbb{D}}(\cN,\delta)$, $t_{\text{dec}}^{\mbb{D}}(\cN,\delta)$, and $t_{\text{d-i}}^{\mbb{D}}(\cN,\delta)$ are defined in~\cref{def:mixing-time}, \cref{def:distinguishability-time}, \eqref{def:decoupling-time}, and \eqref{def:d_i_time}, respectively. }
    \label{fig:contraction-relations}
\end{figure}

\section{Conclusion}

\label{sec:conclusion}

\subsection{Summary}

In this paper, we comprehensively studied quantum Doeblin coefficients, two of which had already been defined and considered \cite{wolf2012quantum,hirche2024quantum} and others which we defined here. We showed that these coefficients satisfy several desirable properties including efficient computability using SDPs, concatenation, and, for one of them, multiplicativity. We also developed various interpretations of the two main quantum Doeblin coefficients (i.e., $\alpha(\cN)$ and $\alpha_I(\cN)$), including as minimal singlet fractions, exclusion values, reverse max-mutual and oveloH informations, reverse robustnesses, and  hypothesis testing reverse mutual and oveloH informations. Furthermore, we studied the complete contraction coefficient of quantum divergences and showed how Doeblin coefficients provide efficient computable bounds on them. We also analytically characterized the Doeblin coefficients of various practically relevant channels and extensively studied them for qubit channels, while utilizing the structure of qubit channels in terms of the Stokes parameterization. 
Finally, we showcased how the technical tools developed in this work find use in applications ranging from identifying the limitations on the trainability of quantum learning algorithms using parameterized quantum circuits when the circuits are noisy, the overhead of error mitigation protocols to learn expected values of observables with noisy data samples, the impact of noise in the hypothesis testing of quantum states, the fairness offered by noisy learning algorithms, and mixing, indistinguishability, and decoupling times of time-varying channels. In all of these applications, we provided efficiently computable bounds while expanding what was previously known, in terms of the circuit architectures, noise models, and so on.

\subsection{Open Questions}

There are various open questions that arise from our work. Here we list them without any particular ordering. 

\begin{enumerate}
    \item In~\cref{prop:CC_HS}, we showed that $\eta^c_\gamma(\cN) \leq 1- \alpha_+(\cN)$. It is an interesting open question to determine whether $\eta^c_\gamma(\cN) \leq 1- \alpha(\cN)$. Due to the Hermitian operator $X$ present in the optimization of~$\alpha(\cN)$, our proof technique in establishing the former cannot be used as is in the aforementioned problem.

    \item In this work, we showed that some Doeblin coefficients satisfy multiplicativity while others do not. It would be interesting to explore other possible Doeblin coefficients that might satisfy these desirable properties.

    \item We analysed Doeblin coefficients of various classes of channels and special channels that are of practical relevance. Studying how the symmetry of a given channel helps in obtaining analytical expressions for their Doeblin coefficients and their tensor product is of relevance (i.e., $\cN$ and $\cN^{\otimes n}$ for $n \in \mathbb{N}$).

    \item All applications of the quantum Doeblin coefficients in this work are related to the fact that they provide upper bounds on  trace-distance contraction coefficients. It would be interesting to find more direct applications of the quantum Doeblin coefficients, in the sense that they do not rely upon this kind of relation to trace-distance contraction coefficients. A recent such example was put forward in \cite{ji2025retrocausal}, in which the Doeblin information in \cref{prop:reverse-max-mutual-info} plays a role in characterizing the retrocausal capacity of a quantum channel.
\end{enumerate}

\section*{Acknowledgements}

We acknowledge helpful discussions with Nana Liu on noise-induced barren plateaus and the parameter shift rule. We also thank Hemant Mishra for helpful discussions.
Part of this work was completed during the conference ``Beyond IID in Information Theory,'' held at the University of Illinois Urbana-Champaign from July 29 to August 2, 2024, and supported by NSF Grant No.~2409823. Additional parts were completed during a visit of CH to the School of Electrical and Computer Engineering at Cornell University, whose admin staff he thanks for hospitality, and at the 2025 Quantum Resource Workshop in Jeju, Korea. We thank the organizers of both events, Beyond IID and Quantum Resources, for organizing these events and including time for research discussions during them. 

TN and MMW acknowledge support from NSF Grant No.~2329662. IG thanks Eric Chitambar for providing a copy of \cite{Chitambar-2021a}. IG is supported by the Ministry of Education, Singapore, through grant T2EP20124-0005. CH received funding by the Deutsche Forschungsgemeinschaft (DFG, German
Research Foundation) – 550206990. TN acknowledges support from the Dieter Schwarz Exchange Programme on Quantum Communication and Security at the Centre for Quantum Technologies, National University of Singapore, Singapore, during her research visit, when the second version of the arXiv post was finalized.

\bibliographystyle{quantum}
\bibliography{lib}

\appendix 

\section{Proof of \cref{thm:game-interpretation}}

\label{app:proof-game-interpretation-exclusion-value}

The proof is largely identical to that of~\cite[Theorem 10]{george2024cone}, which itself is a rather direct generalization of the proof of~\cite[Theorem 6]{Chitambar-2021a}. For completeness, we present it in full. Let $\cK(B\otimes B')$ be an arbitrary convex cone such that $\I \in \operatorname{relint}(\cK)$ and it is closed under local completely-positive (CP) maps on the $B'$ (i.e., $\id_{B} \otimes \psi_{B' \to A'}(X) \in \cK(A' \otimes B')$ if $X \in \cK(B \otimes B')$ and $\Phi,\psi$ are completely-positive. We note the $B'$ space will ultimately correspond to the $A$ space of the cone.

Let $\ket{\varphi}_{A'B'}$ be an arbitrary pure state.  Without loss of generality we can take $\ket{\varphi}$ to be maximally entangled as Nielsen's Theorem guarantees the existence of an LOCC transformation such that $\ket{\Phi_{d}}_{A'\wt{A'}}\to\ket{\varphi}_{A'B'}$, where we recall $\ket{\Phi_{d}}_{A'\wt{A'}}$ is the maximally entangled state where $A' \cong \wt{A}'$ and here we take $d \coloneq \vert A' \vert$. Specifically, Nielsen's theorem~\cite{Nielsen-1999} uses a measurement on Bob's side with Kraus operators $\{M^k\}_k$ and correcting unitaries $\{U^k\}_k$ on Alice's side such that $U_{A'}^k\otimes M_{\wt{A'}\to B'}^k\ket{\Phi_{d}}_{A'\wt{A}'}=\sqrt{p(k)}\ket{\varphi^{k}}_{A'B'}$. Consider the channel state exclusion game winning probability for cone $\cK(B \otimes B')$. That is,
\begin{equation}
\label{eq:EA-induced-doeblin-step-1}
\begin{aligned}
\sum_{x} P(x \vert x) 
\coloneq \sum_{x} \sum_{k} \Tr\Big[\Lambda^x_{BB'}\left(\cN_{A\to B}\circ\cE^x_{A'\to A}\otimes\id^{B'}\right) \left(p(k)\dyad{\varphi^{k}}\right)\Big] \ ,
\end{aligned}
\end{equation}
where by construction
$$ p(k)\dyad{\varphi^{k}} = [(U^k\otimes M^k)\Phi_{d}(U^k\otimes M^k)^\dagger]  \ , $$
and by assumption $\Lambda^{y}_{BB'} \in \cK(B \otimes B')$ for all $y$. Note $\{(I\otimes M^k)^\dagger \Lambda^x(I\otimes M^k)\}_{k,y}$ forms a POVM on $B\wt{A'}$. This follows from the fact that $\{M^k\}_k$ are Kraus operators for a CPTP map, and so the adjoint map, $\sum_k {M^{k}}^\dagger(\cdot)M^k$ is a unital map. Similarly, note that $\cU^k(\cdot)\coloneqq U^k(\cdot) {U^{k}}^{\dagger}$ denotes a unitary channel. Therefore, the collection $\{\cE^x \circ \cU^k\}_{x,k}$ is a family of encoders.  Therefore, we can re-express~\eqref{eq:EA-induced-doeblin-step-1} as
\begin{equation}
\label{eq:EA-induced-doeblin-step-2}
\begin{aligned}
\sum_{z} P(z|z)
= \sum_{z} \Tr\!\left[\hat{\Lambda}^z_{B\wt{A}'}\!\left(\cN_{A\to B}\circ\hat{\cE}^z_{A'\to A}\otimes\id_{\wt{A'}}[{\Phi_{d}}_{A'\wt{A}'}]\right)\right],
\end{aligned}
\end{equation}
where the $\hat{\cE}^z$ and $\hat{\Lambda}^z$ are the concatenated encoders and decoder, i.e.~we define $\cZ = \cX \times \cK$. Note that as $\cK(B \otimes \wt{A}')$ is closed under local CP maps, $\hat{\Lambda}^{z}_{B\wt{A}'} \in \cK(B \otimes \wt{A}')$ for all $z$. This shows that we can restrict our attention just to a shared maximally entangled state as we built a new strategy achieving the same value. Furthermore, without loss of generality, we can assume that $d_{A'}\geq d_{A}$, because the transformation $\ket{\Phi_{d}}_{A''\wt{A}''}\to\ket{\varphi}_{A'B'}$ is always possible for any $d_{A''}\geq d_{A'}$ because we could have just as well used the same argument with system $A''$ and arrived at ${\Phi_{d}}_{A''\wt{A}''}$ in~\eqref{eq:EA-induced-doeblin-step-2}.

We next take Kraus-operator decompositions $\hat{\cE}^z(\cdot)=\sum_{k}N^{z,k}(\cdot){N^{z,k}}^{\dagger}$ with each $N^{z,k}:\mbb{C}^{d_{A'}}\to\mbb{C}^{d_A}$.  Since $d_{A'}\geq d_A$, we can use the transpose trick/ricochet property $N^{z,k}\otimes I\ket{\Phi_{d}}_{A'\wt{A}'}=I\otimes {N^{z,k}}^{T}\ket{\Phi_{d}}_{A\wt{A}}$ to obtain
\begin{align}
\label{eq:EA-induced-doeblin-step-3}
\sum_{x} P(x \vert x)  = \frac{1}{d_{A'}}\sum_{z} \sum_{k} \Tr\!\left[\wt{P}^{z,k}\left(\cN_{A\to B}\otimes\id_{\wt{A}}\right)(\Phi_{d_{A}})\right]
=  \Tr[\Omega^{AB} \Gamma^\cN_{AB}],
\end{align}
where $\wt{P}^{z,k} \coloneqq  (I_{B}\otimes {\ol{N}^{z,k}})\hat{\Lambda}^z_{B\wt{A'}}(I_{B}\otimes {N^{z,k}}^{T})$ and we have swapped the ordering of the systems to match earlier notation, and
\begin{align}
\Omega_{AB}&=\frac{1}{d_{A'}}\sum_{z}\sum_{k}(\ol{N}^{z,k}\otimes I_B )\hat{\Lambda}^z_{A'B}({N^{z,k}}^{T}\otimes I_B)\\
&\coloneq \frac{1}{d_{A'}}\sum_z \widetilde{\cE}_{z \, A' \to A}\otimes\id_B\left(\hat{\Lambda}_z^{A'B}\right),
\end{align}
in which $\widetilde{\cE}_{z}(\cdot)\coloneqq \sum_{k}\ol{N}_{z,k}(\cdot)N_{z,k}^T$. Note, $\widetilde{\cE}_{z \, A' \to A} \otimes \id_{B}(\hat{\Lambda}^{z}_{A'B})\in \cK(A \otimes B)$ as it is again a local CP map $\wt{A}' \to A$ and the cone structure was invariant under CP maps on $\wt{A}'$ by assumption.\footnote{Note we re-ordered the spaces so we before had $\cK(B \otimes \wt{A}')$, which was then $\cK(\wt{A}' \otimes B)$, which is still invariant under the same space, the ordering merely changed.} As $\cK$ is a cone, sums of its elements are contained, so it follows that $\Omega_{AB} \in \cK(A \otimes B)$. Moreover, since each $\widetilde{\cE}_z$ is trace preserving, we have
\begin{align}
\Tr_A\Omega_{AB}&=\frac{1}{d_{A'}}\sum_z\Tr_A\!\left[\widetilde{\cE}_{z \, A'\to A}\otimes\id^B\left(\hat{\Lambda}^z_{A'B}\right)\right]\\
&=\frac{1}{d_{A'}}\sum_z\Tr_{A'}\left(\hat{\Lambda}^z_{A'B}\right)\\
&=\frac{1}{d_{A'}}\Tr_{A'}\left( I_{A'}\otimes I_B\right)= I_B.
\end{align}
Hence $\Tr_A\Omega_{AB}= I_B$ and $\Omega \in \cK$ is a necessary condition on the operator $\Omega_{AB}$ such that $\sum_{x} P(x|x) =\Tr[\Omega^{AB}\Gamma^{\cN}_{AB}]$. That is, any solution induces a feasible solution to~\eqref{eqn:inducedDoeblinDual_cone}, so it is a necessary condition. Now we make sure it is also sufficient.

\sloppy Consider an arbitrary positive semidefinite operator $\Omega_{AB}$ such that $\Tr_A\Omega_{AB}= I_B$ and $\Omega_{AB} \in \cK(A \otimes B)$. Consider the discrete Weyl operators on $A$, explicitly given by $U_{m,n}=\sum_{k=0}^{d_A-1}e^{i mk 2\pi/d_A}\ket{m\oplus n}\bra{m}$, 
where $m,n=0,\dots,d_A-1$ and addition is taken modulo $d_A$. It's well known that twirling the discrete Weyl operators results in the completely depolarizing map, e.g.,
\begin{equation}
\Delta(X)\coloneqq \frac{1}{d_A}\sum_{m,n}U_{m,n}(X)U_{m,n}^\dagger = \Tr[X]\mbb{I} \ .
\end{equation} 
Hence, 
\begin{equation}
    \Delta_A\otimes\id_B[\Omega_{AB}]=\mbb{I}_A\otimes \Tr_A\Omega_{AB}=\mbb{I}_A\otimes\mbb{I}_B.
\end{equation}
This implies that the set $\{ \mu(m,n) \coloneqq  \cU^A_{m,n}\otimes\id^B(\Omega^{AB})\}_{m,n}$ forms a valid POVM on $AB$.  Therefore, we can construct an entanglement-assisted protocol as follows.  Let Alice and Bob share a maximally entangled state $\ket{\tau_{d_A}}_{\wt{A}A}$.  Alice then applies the $\cU_{m,n}^T(\cdot)\coloneqq U_{m,n}^T(\cdot) U_{m,n}^*$ to the $A$ system and sends it through the channel $\cN$. Bob then performs the POVM $\{\mu(m,n)\}_{m,n}$ just described on systems $\wt{A}B$. The obtained score is
\begin{align}
 \sum_{m,n}P(m,n|m,n)
& = \frac{1}{d_A}\sum_{m,n}\Tr\!\left[\left(\cU^{\wt{A}}_{m,n}\otimes\id^B\left[\Omega_{\wt{A}B}\right]\right)(\id_{\wt{A}}\otimes\cN)  \left(\id_{\wt{A}}\otimes\cU^{T}_{m,n}\left[\tau_{d_A}\right]\right)\right]\\
& = \frac{1}{(d_A)^2}\sum_{m,n}\Tr\!\left[\Omega^{\wt{A}B}\left(\id^{\wt{A}}\otimes\cN\left[\tau_{d_A}\right]\right)\right]\\
& = \Tr[\Omega \Gamma^{\cN}].
\end{align}
The key idea in these equalities is that the unitary encoding $U_{m,n}$ performed on Alice's side is canceled by exactly one POVM element on Bob's side.

We therefore have necessary and sufficient conditions for $\cG^{\cK}_{\operatorname{QSE}}(\cN)$ to be given by $\max\left\{ \Tr[\Omega_{AB}\Gamma^{\cN}] :\, \Tr_{A}[\Omega] = I_{B} \,,\, \Omega_{AB} \in \cK \right\}$, which is the conic program given in~\eqref{eqn:inducedDoeblinDual_cone}. Noting that we assumed $I \in \operatorname{relint}(\cK)$ means that strong duality holds, this completes the proof.

\section{Proof of~\cref{thm:alpha_N_equality_Hypo_Test}} 

\label{Proof:thm:alpha_N_equality_Hypo_Test}

The proof given here is similar to the proof of~\cite[Proposition~7]{ji2024barycentric}. 
Consider that the desired equality in~\eqref{eq:reverse-hypothesis-testing-MI} is equivalent to
\begin{equation}
\frac{1}{d^{2}}\alpha(\mathcal{N})=\exp\!\left(  -\inf_{\tau\in
\operatorname{aff}(\mathcal{D})}D_{H}^{1-\frac{1}{d^{2}}}\!\left(\pi_{d}\otimes
\tau\Vert\Phi^{\mathcal{N}} \right)\right)  .
\end{equation}
Now observe that
\begin{align}
& \exp\!\left(  -\inf_{\tau\in\operatorname{aff}(\mathcal{D})}D_{H}^{1-\frac
{1}{d^{2}}}\!\left(\pi_{d}\otimes\tau\Vert\Phi^{\mathcal{N}}\right)\right)  \nonumber\\
& =\sup_{\tau\in\operatorname{aff}(\mathcal{D})}\exp\!\left(  -D_{H}^{1-\frac
{1}{d^{2}}}\!\left(\pi_{d}\otimes\tau\Vert\Phi^{\mathcal{N}}\right)\right)  \\
& =\sup_{\tau\in\operatorname{aff}(\mathcal{D})}\exp\!\left(  -\left[  -\ln
\inf_{\Lambda\geq0}\left\{  \operatorname{Tr}[\Lambda\Phi^{\mathcal{N}
}]:\operatorname{Tr}[\Lambda(\pi_{d}\otimes\tau)]\geq\frac{1}{d^{2}
},\ \Lambda\leq I\right\}  \right]  \right)  \\
& =\sup_{\tau\in\operatorname{aff}(\mathcal{D})}\inf_{\Lambda\geq0}\left\{
\operatorname{Tr}[\Lambda\Phi^{\mathcal{N}}]:\operatorname{Tr}[\Lambda(\pi
_{d}\otimes\tau)]\geq\frac{1}{d^{2}},\ \Lambda\leq I\right\}  .
\end{align}

We first prove the inequality
\begin{equation}
\exp\!\left(  -\inf_{\tau\in\operatorname{aff}(\mathcal{D})}D_{H}^{1-\frac
{1}{d^{2}}}\!\left(\pi_{d}\otimes\tau\Vert\Phi^{\mathcal{N}} \right)\right)  \leq\frac
{1}{d^{2}}\alpha(\mathcal{N}).
\end{equation}
Recall from \cref{prop:dual-SDP-doeblin-main} that the dual of $\alpha(\mathcal{N})$ is given by
\begin{equation}
\alpha(\mathcal{N})=\inf_{Y_{AB}\geq0}\left\{  \operatorname{Tr}[Y_{AB}
\Gamma_{AB}^{\mathcal{N}}]:\operatorname{Tr}_{A}[Y_{AB}]=I_{B}\right\}  ,
\end{equation}
which implies that
\begin{equation}
\frac{1}{d^{2}}\alpha(\mathcal{N})=\inf_{Y_{AB}\geq0}\left\{
\operatorname{Tr}\!\left[  \frac{Y_{AB}}{d}\Phi_{AB}^{\mathcal{N}}\right]
:\operatorname{Tr}_{A}[Y_{AB}]=I_{B}\right\}  .
\end{equation}
Let $Y_{AB}$ be such that $\operatorname{Tr}_{A}[Y_{AB}]=I_{B}$, and let
$\tau\in\operatorname{aff}(\mathcal{D})$; then
\begin{align}
\operatorname{Tr}\!\left[  \frac{Y_{AB}}{d}(\pi_{d}\otimes\tau)\right]    &
=\frac{1}{d^{2}}\operatorname{Tr}[\operatorname{Tr}_{A}[Y_{AB}]\tau]\\
& =\frac{1}{d^{2}}\operatorname{Tr}[I_{B}\tau]\\
& =\frac{1}{d^{2}}.
\end{align}
Also, consider that
\begin{equation}
\frac{Y_{AB}}{d}\leq I_{AB}
\end{equation}
because
\begin{align}
\frac{1}{d^{2}}\sum_{i=1}^{d^{2}}U_{A}^{i}\frac{Y_{AB}}{d}\left(  U_{A}
^{i}\right)  ^{\dag}  & =\pi_{A}\otimes\operatorname{Tr}_{A}\!\left[
\frac{Y_{AB}}{d}\right]  \\
& =\frac{1}{d}\pi_{A}\otimes I_{B},
\end{align}
where $\left\{  U_{A}^{i}\right\}  _{i=1}^{d^{2}}$ is the set of
Heisenberg--Weyl operators. This implies that
\begin{align}
\frac{Y_{AB}}{d}  & \leq\sum_{i=1}^{d^{2}}U_{A}^{i}\frac{Y_{AB}}{d}\left(
U_{A}^{i}\right)  ^{\dag}\\
& =d^{2}\frac{\pi_{A}}{d}\otimes I_{B}\\
& =I_{AB}.
\end{align}
Thus, $\frac{Y_{AB}}{d}$ is a legitimate measurement operator, and it
satisfies the constraint
\begin{equation}
\operatorname{Tr}\!\left[  \frac{Y_{AB}}{d}(\pi_{d}\otimes\tau)\right]
\geq\frac{1}{d^{2}}.
\end{equation}
Then we find that
\begin{equation}
\inf_{\Lambda\geq0}\left\{  \operatorname{Tr}[\Lambda\Phi^{\mathcal{N}
}]:\operatorname{Tr}[\Lambda(\pi_{d}\otimes\tau)]\geq\frac{1}{d^{2}}
,\ \Lambda\leq I\right\}  \leq\operatorname{Tr}\!\left[  \frac{Y_{AB}}{d}
\Phi_{AB}^{\mathcal{N}}\right]  .
\end{equation}
Taking an infimum over all $Y_{AB}\geq0$ satisfying the constraint
$\operatorname{Tr}_{A}[Y_{AB}]=I_{B}$ then implies that
\begin{equation}
\inf_{\Lambda\geq0}\left\{  \operatorname{Tr}[\Lambda\Phi^{\mathcal{N}
}]:\operatorname{Tr}[\Lambda(\pi_{d}\otimes\tau)]\geq\frac{1}{d^{2}}
,\ \Lambda\leq I\right\}  \leq\frac{1}{d^{2}}\alpha(\mathcal{N}).
\end{equation}
Since the inequality holds for all $\tau\in\operatorname{aff}(\mathcal{D})$,
we can take the supremum over all such $\tau$ to conclude that
\begin{equation}
\exp\!\left(  -\inf_{\tau\in\operatorname{aff}(\mathcal{D})}D_{H}^{1-\frac
{1}{d^{2}}}(\pi_{d}\otimes\tau\Vert\Phi^{\mathcal{N}})\right)  \leq\frac
{1}{d^{2}}\alpha(\mathcal{N}).
\end{equation}

To show the opposite inequality, consider that
\begin{align}
\frac{1}{d^{2}}\alpha(\mathcal{N})  & =\frac{1}{d^{2}}\sup_{X\in
\operatorname{Herm}}\left\{  \operatorname{Tr}[X]:I\otimes X\leq
\Gamma^{\mathcal{N}}\right\}  \\
& =\frac{1}{d^{2}}\sup_{X\in\operatorname{Herm}}\left\{  \operatorname{Tr}
[X]:\pi_{d}\otimes X\leq\Phi^{\mathcal{N}}\right\}  \\
& =\frac{1}{d^{2}}\sup_{\lambda\geq0,\tau\in\operatorname{aff}(\mathcal{D}
)}\left\{  \lambda:\lambda\left(  \pi_{d}\otimes\tau\right)  \leq
\Phi^{\mathcal{N}}\right\}  \\
& =\sup_{\lambda\geq0,\tau\in\operatorname{aff}(\mathcal{D})}\left\{
\frac{\lambda}{d^{2}}:\lambda\left(  \pi_{d}\otimes\tau\right)  \leq
\Phi^{\mathcal{N}}\right\}  \\
& =\sup_{\lambda\geq0,\tau\in\operatorname{aff}(\mathcal{D})}\inf_{\Lambda
\geq0}\left\{  \frac{\lambda}{d^{2}}+\operatorname{Tr}[\Lambda\left(
\Phi^{\mathcal{N}}-\lambda\left(  \pi_{d}\otimes\tau\right)  \right)
]\right\}  \\
& =\sup_{\lambda\geq0,\tau\in\operatorname{aff}(\mathcal{D})}\inf_{\Lambda
\geq0}\left\{  \operatorname{Tr}[\Lambda\Phi^{\mathcal{N}}]+\lambda\left(
\frac{1}{d^{2}}-\operatorname{Tr}[\Lambda\left(  \pi_{d}\otimes\tau\right)
]\right)  \right\}  \\
& \leq\sup_{\tau\in\operatorname{aff}(\mathcal{D})}\inf_{\Lambda\geq0}
\sup_{\lambda\geq0}\left\{  \operatorname{Tr}[\Lambda\Phi^{\mathcal{N}
}]+\lambda\left(  \frac{1}{d^{2}}-\operatorname{Tr}[\Lambda\left(  \pi
_{d}\otimes\tau\right)  ]\right)  \right\}  \\
& =\sup_{\tau\in\operatorname{aff}(\mathcal{D})}\inf_{\Lambda\geq0}\left\{
\operatorname{Tr}[\Lambda\Phi^{\mathcal{N}}]:\frac{1}{d^{2}}\leq
\operatorname{Tr}[\Lambda\left(  \pi_{d}\otimes\tau\right)  ]\right\}  \\
& \leq\sup_{\tau\in\operatorname{aff}(\mathcal{D})}\inf_{\Lambda\geq0}\left\{
\operatorname{Tr}[\Lambda\Phi^{\mathcal{N}}]:\frac{1}{d^{2}}\leq
\operatorname{Tr}[\Lambda\left(  \pi_{d}\otimes\tau\right)  ],\ \ \Lambda\leq
I\right\}  \\
& =\sup_{\tau\in\operatorname{aff}(\mathcal{D})}\exp\!\left(  -D_{H}^{1-\frac
{1}{d^{2}}}\!\left(\pi_{d}\otimes\tau\Vert\Phi^{\mathcal{N}}\right)\right)  \\
& =\exp\!\left(  -\inf_{\tau\in\operatorname{aff}(\mathcal{D})}D_{H}^{1-\frac
{1}{d^{2}}}\!\left(\pi_{d}\otimes\tau\Vert\Phi^{\mathcal{N}}\right)\right)  .
\end{align}
This concludes the proof.

\section{Proof of Proposition \ref{prop:concat-of-alpha-wang}}

\label{app:proof-of-alpha-wang-concat}

As stated in the main text, the proof effectively follows the same steps as used to establish \cref{prop:concatenation}.

\begin{proof}[Proof of \cref{lem:alpha-wang-CP-ordering-form}]
    Starting from the definition,
\begin{align}
    \alpha_{\wang}(\cN) &=  \max_{X_B \in \operatorname{Herm}} \left\{ \operatorname{Tr}[X_B] : - \Gamma^{\mathcal{N}}_{AB} \leq I_A \otimes X_B \leq \Gamma^{\mathcal{N}}_{AB}\right\} \\ 
    &= \max_{\substack{\lambda \in \mbb{R}, \\ \tau \in \aff(\cD)}} \left\{ \lambda : -\cN \leq \lambda\cR_{\tau} \leq \cN \right\}
\end{align}
where the second equality follows from the same argument that established \cref{eq:key-robustness-step}. Using that the complete positive ordering on channels means $\cM \leq \cN \iff \cM = \cN + \cG$ for some completely positive map $\cG$, we re-express the previous optimization as
\begin{align}
    \alpha_{\wang}(\cN) =\max_{\substack{\lambda \in \mbb{R}, \\ \tau \in \aff(\cD), \\ \cM',\cG' \in \operatorname{CP}}} \left\{ \lambda :\lambda\cR_{\tau} + \cM' = \cN \, , \, \cN = \cG' - \lambda \cR_{\tau} \right\} \ . 
\end{align}
Noting that the first constraint is the same as in \cref{prop:reverse-robustness}, by the same argument as around \cref{eq:reverse-robustness-normalization-argument}, we conclude that $\lambda \in [0,1]$ and $\lambda \cR_{\tau} + \cM' = \cN$ can be re-expressed as $\lambda \cR_{\tau}+\left(1-\lambda\right)\cM$ for $\cM$ being a quantum channel:
\begin{align}
    \alpha_{\wang}(\cN) =\max_{\substack{\lambda \in [0,1], \\ \tau \in \aff(\cD), \\ \cM \in \operatorname{CPTP}, \\ \cG' \in \operatorname{CP}}} \left\{ \lambda :\lambda\cR_{\tau} + (1-\lambda)\cM = \cN \, , \, \cN = \cG' - \lambda \cR_{\tau} \right\} \ . 
\end{align}
Using the second constraint, for every quantum state $\rho$,
\begin{align}
    1 = \Tr[\cN(\rho)] = \Tr[\cG'(\rho)] - \lambda \Rightarrow \Tr[\cG'(\rho)] = (1 + \lambda) \in [1,2] \ . 
\end{align}
Thus, $\frac{1}{1+\lambda}\cG'$ is always a CPTP map and we conclude
\begin{align}
    \alpha_{\wang}(\cN) =\max_{\substack{\lambda \in \left[0,1\right], \\ \tau \in \aff(\cD), \\ \cM,\cG \in \operatorname{CPTP}}} \left\{ \lambda :\lambda\cR_{\tau} + (1-\lambda)\cM = \cN \, , \, \cN = (1+\lambda)\cG - \lambda \cR_{\tau} \right\} \ . 
\end{align}
This completes the proof of the lemma.
\end{proof}

\begin{proof}[Proof of \cref{prop:concat-of-alpha-wang}]
First, by re-parameterizing Lemma \ref{lem:alpha-wang-CP-ordering-form}, we have
\begin{align}\label{eq:one-minus-alpha-wang-opt}
    1-\alpha_{\wang}(\cN) =\min_{\substack{\lambda \in [0,1], \\ \tau \in \aff(\cD), \\ \cM,\cG \in \operatorname{CPTP}}} \left\{ \lambda :(1-\lambda)\cR_{\tau} + \lambda\cM = \cN \, , \, \cN = (2-\lambda)\cG - (1-\lambda) \cR_{\tau} \right\} \ ,
\end{align}
which is what we use in this proof. Now consider quantum channels $\cN_{1},\cN_{2}$. Using \eqref{eq:one-minus-alpha-wang-opt}, let $(\lambda_{1},\tau_{1},\cM_{1},\cG_{1})$ and $(\lambda_{2},\tau_{2},\cM_{2},\cG_{2})$ be optimizers of $1-\alpha(\cN_{1})$ and $1-\alpha(\cN_{2})$ respectively. We will use these to construct a feasible point for $\left(1-\alpha_{\wang}(\cN_{2}\circ \cN_{1}) \right)$. Define
\begin{align}
    \lambda' \coloneq \lambda_{2}\lambda_{1} \quad \tau^{\prime}\coloneqq \frac{1-\lambda_{2}}{1-\lambda_{1}\lambda_{2}}\tau
_{2}+\frac{\lambda_{2}\left(  1-\lambda_{1}\right)  }{1-\lambda_{1}\lambda
_{2}}\mathcal{M}_{2}(\tau_{1}) \quad \cM' \coloneq \cM_{2} \circ \cM_{1} \ .
\end{align}
Using that the first constraint in \eqref{eq:one-minus-alpha-wang-opt} is the same as the sole constraint for $1-\alpha(\cN)$ given in \cref{eq:reverse-robustness-opposite}, by the same argument as in the proof of \cref{prop:concatenation}, we conclude that these choices are feasible for the first constraint of $1-\alpha_{\wang}(\cN_{2} \circ \cN_{1})$. Moreover, define
\begin{align}\label{eq:cGprime-def}
    \cG' \coloneq \frac{1}{2-\lambda'}[\cN_{2} \circ \cN_{1} + (1-\lambda')\cR_{\tau'}] \ , 
\end{align}
which is completely positive as it is the positive sum of completely positive maps and is trace preserving as $\cG'(\rho) = \frac{1}{2-\lambda'}[1 + (1-\lambda)] = 1$ for every quantum state $\rho$. Re-ordering \cref{eq:cGprime-def}, one finds $\cN_{2}\circ \cN_{1} = (2-\lambda')\cG' - (1-\lambda)\cR_{\tau'}$; i.e., the choice of $\cG'$ in \eqref{eq:cGprime-def} satisfies the second constraint in \cref{eq:one-minus-alpha-wang-opt}. Thus, $(\lambda',\tau',\cM',\cG')$ is feasible for $1-\alpha_{\wang}(\cN_{2} \circ \cN_{1})$. As it is a minimization, we have $1-\alpha_{\wang}(\cN_{2}\circ \cN_{1}) \leq \lambda' = \lambda_{2}\lambda_{1} = (1-\alpha_{\wang}(\cN_{2}))(1-\alpha_{\wang}(\cN_{1}))$. This completes the proof of the proposition.
\end{proof}

\section{Qubit Channel Parameterizations}

\label{sec:qubit-channel-parameterizations}

Here we recall or establish all the properties of qubit channel parameterizations that we will need.
\paragraph{Stokes Parameterization} First we recall the Stokes parameterization of a qubit linear map. As $\{I,\sigma_X , \sigma_Y, \sigma_Z\}$ form a basis of $\cL(\mbb{C}^{2})$, any $M \in \cL(\mbb{C}^{2})$ may be written as $M = w_{0}I + \vec{w} \cdot \vec{\sigma}$ where $\vec{w}$ is the Bloch vector. We may identify any qubit-to-qubit linear map, $\Phi \in \cT(\mbb{C}^{2},\mbb{C}^{2})$, with a transformation map $\mbb{T}_{\Phi} \in \mbb{R}^{4 \times 4}$. Then $\Phi(M)$ can be determined from the matrix multiplication $\begin{bmatrix} w_{0}' \\ \vec{w}' \end{bmatrix} = \mbb{T}_{\Phi}\begin{bmatrix} w_{0} \\ \vec{w} \end{bmatrix}$. The most important facts we will need about the transformation map are the following, which are all established or reported in~\cite{King-2000a}:
\begin{enumerate}
    \item If $\Phi$ is trace-preserving, then $\mbb{T}_{\Phi} = \begin{bmatrix} 1 & \vec{0}^{T} \\ \vec{t} & T \end{bmatrix}$ where $T \in \mbb{R}^{3 \times 3}$ is the `transfer matrix' and $\vec{t}$ is a three-dimensional column vector.
    \item If $\Phi$ is trace-preserving and unital, then $\vec{t} = 0$.
    \item If $\Phi$ is a trace-preserving, so it has the form given in the first item, then its complete positivity conditions are 
    \begin{align}\label{eq:positivity-conditions}
        \vert T_{11} \pm T_{22} \vert \leq \vert 1 + T_{33} \vert \ .
    \end{align} 
    \item If $\Phi$ is a trace-preserving and unital, then up to local pre- and post- processing $\mbb{T} = \operatorname{diag}(1,\lambda_{1},\lambda_{2},\lambda_{3})$ and it is CP if and only if $\vert \lambda_{1} \pm \lambda_{2} \vert \leq \vert 1 + \lambda_{3} \vert$. In particular, this means $(\lambda_{1},\lambda_{2},\lambda_{3})$ need to be in the tetrahedron defined by the extreme points 
    \begin{align}\label{eq:extreme-points}
        \{(1,1,1),(1,-1,-1),(-1,1,-1),(-1,-1,1)\} \ . 
    \end{align}
\end{enumerate}

\begin{example}\label{example:Bloch-matrix} ~ 
    \begin{enumerate}
        \item The depolarizing channel has $\lambda_{1} = \lambda_{2} = \lambda_{3} = 1 - q$ where $q \in [0,4/3]$.
        \item The extreme points are 
        \begin{equation}
        \begin{aligned}
            (1,1,1) \mapsto \operatorname{id}(\cdot) &\quad (1,-1,-1) \mapsto \sigma_{X} \cdot \sigma_{X}^{\dag} \\
            (-1,1,-1) \mapsto \sigma_{Y} \cdot \sigma_{Y}^{\dag} &\quad (-1,-1,1) \mapsto \sigma_{Z} \cdot \sigma_{Z}^{\dag} \ ,
        \end{aligned}
        \end{equation}
        which may be verified using $\sigma_{X}\sigma_{Y}\sigma_{X} = -\sigma_{Y}$ and similar identities.
    \end{enumerate}
\end{example}

\paragraph{Choi Operator Pauli Parameterization}

The other representation that will be useful will be that of the Choi operator of a linear map \cref{eq:choi_operator}. We recall that to be trace-preserving, we should have that $\Tr_{B}[\Gamma^{\Phi}] = I_{A}$, and similarly, to be unital, we should have that $\Tr_{A}[\Gamma^{\Phi}] = I_{B}$. From these criteria, it follows that for qubit-to-qubit linear maps,
\begin{align}
    \Phi \text{ is TP}  \quad \Leftrightarrow \quad  & \Gamma^{\Phi} = \frac{1}{2}[I_{A} \otimes I_{B} + I_{A} \otimes \vec{r} \cdot \vec{\sigma} + \sum_{i,j} a_{i,j} \sigma_{i} \otimes \sigma_{j}] \label{eq:Choi-TP-parameterize} \\
    \Phi \text{ is Unital} \quad \Leftrightarrow \quad & \Gamma^{\Phi} = \frac{1}{2}[I_{A} \otimes I_{B} + \vec{s} \cdot \vec{\sigma} \otimes I_{B} + \sum_{i,j} b_{i,j} \sigma_{i} \otimes \sigma_{j}] \ , \label{eq:Choi-Unital-parameterize}
\end{align}
where $\vec{s},\vec{r} \in \mbb{R}^{3}$ and $a_{i,j}, b_{i,j} \in \mbb{R}$ for all $i,j \in \{1,2,3\}$ and we define the matrix forms $A = [a_{i,j}]$, $B = [b_{i,j}]$ are again the `transfer matrices.' To see that the above hold, note that as $\{I,\sigma_X,\sigma_Y,\sigma_Z\}^{\times 2}$ form a basis of $\cL(\mbb{C}^{4})$ and the only terms we have dropped in each case would have to be zero to satisfy the partial trace conditions. 

The above shows a TP map is defined by $(\vec{r},A)$ and a unital map is defined by $(\vec{s},B)$. Note that~\eqref{eq:Choi-TP-parameterize} and~\eqref{eq:Choi-Unital-parameterize} define a bijection on these parameterizations. Namely we may define the bijection $(\vec{r},A) \leftrightarrow (\vec{s} = \vec{r}, B = A^{T})$. Note this corresponds to applying the `swap map' to convert a TP map $\cN$ to a unital one $\cM$, i.e., $\mbb{F}\Gamma^{\cN}\mbb{F}^{\dag} = \Gamma^{\cM}$ where $\mbb{F}:A \otimes B \to B \otimes A$ is the swap operator that may be defined via $\mbb{F}(\ket{a_{1}}\bra{a_{2}} \otimes \ket{b_{1}}\bra{b_{2}}) =  \ket{b_{1}}\bra{b_{2}} \otimes \ket{a_{1}}\bra{a_{2}}$ for all $\ket{a_{1}},\ket{a_{2}} \in A$, $\ket{b_{1}},\ket{b_{2}} \in B$. Note that this map is positivity-preserving, i.e., $\Gamma^{\cN} \geq 0 \Leftrightarrow \mbb{F}\Gamma^{\cN}\mbb{F}^{\dag} \geq 0$. This means we may immediately lift the positivity conditions of~\eqref{eq:positivity-conditions} to unital maps as the diagonal entries are left invariant under the transpose.
\begin{proposition}\label{prop:unital-positivity-constraints}
    Let $\Phi$ be a qubit-to-qubit unital map with the parameterization in~\eqref{eq:Choi-Unital-parameterize}. It is unital if and only if $\vert B_{11} \pm B_{22} \vert \leq \vert 1 + B_{33} \vert$.
\end{proposition}

We also state the following fact established in~\cite{Chitambar-2021a}.
\begin{proposition}[\cite{Chitambar-2021a}]\label{prop:simplified-choi-under-processing}
     For every trace-preserving qubit-to-qubit 
     map, there exist qubit unitaries $U$ and $V$ such that 
    \begin{align}\label{eq:post-p}
        \Gamma^{\cV \circ \Phi \circ \cU} = U \otimes V \Gamma^{\Phi} (U \otimes V)^{\dag} = \frac{1}{2}[I_{A} \otimes I_{B} + I_{A} \otimes \vecp{r} \cdot \vec{\sigma} + \sum_{k} t_{k}' \sigma_{k} \otimes \sigma_{k}] \ ,
    \end{align}
    where $\cU,\cV$ denote using $U,V$ as post and pre-processing respectively.
\end{proposition}

We provide a proof for completeness as it does not seem to follow from the Stokes parameterization in a straightforward manner. ~ \\
\begin{proof}
    Let $\Phi$ have the form in~\eqref{eq:Choi-TP-parameterize}. Then there exists a singular value decomposition of the transfer matrix $A = \sum_{k} t_{k}' \vec{a}_{k}\vec{b}_{k}^{T}$ where $\{\vec{a}_{k}\}_{k}, \{\vec{b}_{k}\}_{k}$ are orthonormal basis of $\mbb{R}^{3}$. Then there exist rotations $O_{1},O_{2} \in SO(3)$ such that $O_{1}:\{\vec{a}_{k}\}_{k} \mapsto \{e_{0},e_{1},\pm e_{2}\}_{k}$, $O_{2}:\{\vec{b_{k}}\} \mapsto \{e_{0},e_{1}, \pm e_{2}\}$, where the $\pm$ denotes that the rotation must respect the determinant of the initial set of vectors. Given the isomorphism $SO(3) \cong SU(2) \setminus \mbb{Z}_{2}$, this means there exist unitaries $U,V \in SU(2)$ such that 
    \begin{align}
        U(\vec{a}_{k} \cdot \vec{\sigma})U^{\dag} \otimes V(\vec{b}_{k} \cdot \vec{\sigma})V^{\dag} = \begin{cases}
            \sigma_{k} \otimes \sigma_{k} & k \in \{1,2\} \\
            \pm \sigma_{3} \otimes \sigma_{3} & k = 3
        \end{cases} \ .
    \end{align}
    Thus,
    \begin{align}
        U \otimes V \Gamma^{\Phi} (U \otimes V)^{\dag} = \frac{1}{2}[I_{A} \otimes I_{B} + I_{A} \otimes V(\vecp{r} \cdot \vec{\sigma}) + \sum_{k} t_{k'}' \sigma_{k} \otimes \sigma_{k}] \ ,
    \end{align}
    where $t_{k'}'$ captures the signedness and $\vecp{r}$ captures that the values of $\vec{r}$ have been shifted. The pre- and post-processing relation is a well-known relation on Choi matrices and may be verified via direct calculation.
\end{proof}

\section{Another Alternative Notion of Quantum Doeblin Coefficient}

\label{app:relate-to-extended-sandwiched}
\begin{definition}
Define 
\begin{equation}
\label{eq:c_N_def}
\alpha_{\fang}(\mathcal{N})\coloneqq \exp\!\left(-\inf_{\beta>1}\left\{  \inf_{\tau\in\operatorname{aff}
(\mathcal{D})}\widetilde{D}_{\beta}(\pi_{d}\otimes\tau\Vert\Phi^{\mathcal{N}
})+\frac{2}{\beta-1}\ln d\right\}\right)  ,
\end{equation}
where the extended sandwiched R\'enyi relative entropy is defined for $\tau \in \operatorname{aff}(\mathcal{D})$ and $\sigma \geq 0$ as follows~\cite[Eq.~(13)]{Wang_2020}:
\begin{equation}
    \widetilde{D}_{\beta}(\tau \Vert \sigma) \coloneqq \frac{\beta}{\beta -1} \ln \left \| \sigma^{\frac{1-\beta}{2\beta}} \tau  \sigma^{\frac{1-\beta}{2\beta}} \right \|_\beta
\end{equation}
with $\|A\|_p \coloneqq \left(\Tr\!\left[ \left(\sqrt{A A^\dag} \right)^p \right] \right)^{1/p}$.
\end{definition}

Several properties of the extended sandwiched R\'enyi relative entropy have been investigated in~\cite{Wang_2020} and~\cite[Section~III]{ji2024barycentric}.

\begin{proposition}
\label{prop:tighter_contraction_n_tensor}
Let $\cN$ be a quantum channel.
The following inequality holds for all $n\in \mathbb{N}$:
\begin{equation}
\label{eq:tighter_contraction_n_tensor}
\eta_{\operatorname{Tr}}(\mathcal{N}^{\otimes n}) \leq 1-\alpha(\mathcal{N}^{\otimes n})\leq1- \alpha_{\fang}(\mathcal{N})^n \leq n(1-\alpha_{\fang}(\mathcal{N})),
\end{equation}
where $c(\cN)$ is defined in~\eqref{eq:c_N_def}.
\end{proposition}

\begin{proof}
Applying~\cref{thm:alpha_N_equality_Hypo_Test} and~\cite[Lemma~6]{ji2024barycentric} for $\beta>1$,
we conclude that
\begin{align}
-\ln\!\left(  \frac{1}{d^{2}}\alpha(\mathcal{N})\right)   &  =\inf_{\tau
\in\operatorname{aff}(\mathcal{D})}D_{H}^{1-\frac{1}{d^{2}}}(\pi_{d}
\otimes\tau\Vert\Phi^{\mathcal{N}})\\
&  \leq\inf_{\tau\in\operatorname{aff}(\mathcal{D})}\widetilde{D}_{\beta}
(\pi_{d}\otimes\tau\Vert\Phi^{\mathcal{N}})+\frac{\beta}{\beta-1}\ln d^{2} \label{eq:D_H_to_Renyi}\\
&  =\inf_{\tau\in\operatorname{aff}(\mathcal{D})}\widetilde{D}_{\beta}(\pi
_{d}\otimes\tau\Vert\Phi^{\mathcal{N}})+\frac{2\beta}{\beta-1}\ln d.
\end{align}
We can rewrite this as
\begin{align}
-\ln\alpha(\mathcal{N})  &  \leq\inf_{\tau\in\operatorname{aff}(\mathcal{D}
)}\widetilde{D}_{\beta}(\pi_{d}\otimes\tau\Vert\Phi^{\mathcal{N}}
)+\frac{2\beta}{\beta-1}\ln d-2\ln d\\
&  =\inf_{\tau\in\operatorname{aff}(\mathcal{D})}\widetilde{D}_{\beta}(\pi
_{d}\otimes\tau\Vert\Phi^{\mathcal{N}})+\frac{2}{\beta-1}\ln d.
\end{align}
Since this inequality holds for all $\beta>1$, we conclude that
\begin{equation}
-\ln\alpha(\mathcal{N})\leq -\ln  \alpha_{\fang}(\mathcal{N}).
\label{eq:alpha-to-c-ineq-single-copy}
\end{equation}

Now applying this to a tensor-power channel $\mathcal{N}^{\otimes n}$, we find
that
\begin{align}
-\ln\alpha(\mathcal{N}^{\otimes n})  &  \leq
-\ln  \alpha_{\fang}(\mathcal{N}^{\otimes n}) \label{eq:alpha-to-c-ineq-1st}\\
& = \inf_{\beta>1} \left\{\inf_{\tau^{(n)}\in
\operatorname{aff}(\mathcal{D}^{(n)})}\widetilde{D}_{\beta}(\pi_{d}^{\otimes
n}\otimes\tau^{(n)}\Vert(  \Phi^{\mathcal{N}})  ^{\otimes n}
)+\frac{2}{\beta-1}\ln d^{n}\right\}\\
&  \leq\inf_{\beta>1} \left\{\inf_{\tau\in\operatorname{aff}(\mathcal{D})}\widetilde{D}_{\beta}
(\pi_{d}^{\otimes n}\otimes\tau^{\otimes n}\Vert(  \Phi^{\mathcal{N}
})  ^{\otimes n})+\frac{2}{\beta-1}\ln d^{n}\right\}\\
&  =n\cdot \inf_{\beta>1} \left\{  \inf_{\tau\in\operatorname{aff}(\mathcal{D})}\widetilde{D}
_{\beta}(\pi_{d}\otimes\tau\Vert\Phi^{\mathcal{N}})+\frac{2}{\beta-1}\ln
d\right\}  \\
& = n\left(-\ln  \alpha_{\fang}(\mathcal{N})\right).
\label{eq:alpha-to-c-ineq-last}
\end{align}
The established inequality is equivalent to the inequality $1-\alpha(\mathcal{N}^{\otimes n})\leq 1- \alpha_{\fang}(\mathcal{N})^n$, thus establishing the second inequality in~\eqref{eq:tighter_contraction_n_tensor}.

The first inequality in~\eqref{eq:tighter_contraction_n_tensor} follows because $
\eta_{\operatorname{Tr}}(\mathcal{N})\leq1-\alpha(\cN)$.
 The last inequality in~\eqref{eq:tighter_contraction_n_tensor} follows because $1-x^n \leq n(1-x)$ for all $n\in \mathbb{N}$ and $x\in[0,1]$.
\end{proof}

\begin{corollary}
    The quantity $\alpha_{\fang}(\mathcal{N})$ is weakly supermultiplicative; i.e., for all $n\in\mathbb{N}$, the following inequality holds:
    \begin{equation}
        \alpha_{\fang}(\mathcal{N})^n \leq \alpha_{\fang}(\mathcal{N}^{\otimes n}).
    \end{equation}
\end{corollary}

\begin{proof}
    This claim follows by rearranging terms in~\eqref{eq:alpha-to-c-ineq-1st}--\eqref{eq:alpha-to-c-ineq-last}. 
\end{proof}

\begin{proposition}\label{prop:c-N-bounded-between-doeblin-and-doeblin-wang}
For a quantum channel $\mathcal{N}$, the following inequalities hold:
\begin{equation}
    \alpha(\mathcal{N})\geq \alpha_{\fang}(\mathcal{N}) \geq \alpha_{\wang}(\mathcal{N}).
\end{equation}
\end{proposition}

\begin{proof}
    The first inequality was established in~\eqref{eq:alpha-to-c-ineq-single-copy}. The second inequality follows by taking the $\beta \to \infty$ limit in the objective function of $c(\mathcal{N}$ and noting that the limit gives $\alpha_{\wang}(\mathcal{N})$; i.e.,
    \begin{equation}
        \lim_{\beta \to \infty} \exp\!\left(-  \inf_{\tau\in\operatorname{aff}
(\mathcal{D})}\widetilde{D}_{\beta}(\pi_{d}\otimes\tau\Vert\Phi^{\mathcal{N}
})+\frac{2}{\beta-1}\ln d\right) = \alpha_{\wang}(\mathcal{N}).
    \end{equation}
    This concludes the proof.
\end{proof}

It is an open question to determine whether $\alpha_{\wang}(\mathcal{N})$ is equal to $\alpha_{\fang}(\mathcal{N})$.

\begin{proposition} \label{prop:c_N_with_existing}
Let $\cN$ be a quantum channel.
The following inequality holds:
\begin{equation}
-\ln \alpha_{\fang}(\mathcal{N})\leq\inf_{\sigma\in\mathcal{D}}D_{\max}(\pi_{d}\otimes
\sigma\Vert\Phi^{\mathcal{N}}).
\end{equation}
where $\alpha_{\fang}(\mathcal{N})$ is defined in~\eqref{eq:c_N_def}, $\pi_d$ is the maximally mixed state of dimension~$d$, and $\Phi^\cN$ is the Choi state of the channel $\cN$.
\end{proposition}

\begin{proof}
Let $\sigma\in\mathcal{D}$. Then
\begin{align}
-\ln \alpha_{\fang}(\mathcal{N})  & =\inf_{\beta>1}\left\{  \inf_{\tau\in\operatorname{aff}
(\mathcal{D})}\widetilde{D}_{\beta}(\pi_{d}\otimes\tau\Vert\Phi^{\mathcal{N}
})+\frac{2}{\beta-1}\ln d\right\}  \\
& \leq\inf_{\beta>1}\left\{  \widetilde{D}_{\beta}(\pi_{d}\otimes\sigma
\Vert\Phi^{\mathcal{N}})+\frac{2}{\beta-1}\ln d\right\}  \\
& \leq\lim_{\beta\rightarrow\infty}\left\{  \widetilde{D}_{\beta}(\pi
_{d}\otimes\sigma\Vert\Phi^{\mathcal{N}})+\frac{2}{\beta-1}\ln d\right\}  \\
& =D_{\max}(\pi_{d}\otimes\sigma\Vert\Phi^{\mathcal{N}}).
\end{align}
Since the inequality holds for all $\sigma\in\mathcal{D}$, the desired
inequality follows.
\end{proof}

\begin{remark}[Comparison to Existing Bounds]
     By~\cref{prop:c_N_with_existing}, we conclude that $\alpha_{\fang}(\mathcal{N})$ in~\eqref{eq:tighter_contraction_n_tensor} gives an improved bound compared to~\cite[Eq.~(III.61)]{hirche2024quantum}. 
\end{remark}

\section{Extra Proofs on Mixing and Decoupling Times of Quantum Processes}\label{app:Convergence-of-Quantum-Processes}

\subsection{Proof of \cref{prop:mixing-class-containment}} \label{app:Mixing_class_proof}
    (Strongly Mixing $\subsetneq$ Mixing): For both homogeneous and inhomogeneous Markov chains, by definition, if a Markov chain is strongly mixing, it is mixing. On the other hand, any replacement channel which prepares a pure state is mixing but not strongly mixing. This establishes the strict containment of strongly mixing in mixing. \\

    (Mixing $\subseteq$ Weakly Mixing): We prove the inhomogeneous case as it is more general. Let $(\cG_{n})_{n \in \mbb{N}}$ be mixing with state $\omega \in \Density(A)$. Then, for any $\rho,\sigma \in \Density(A)$,
    \begin{align}
        \Vert \cG_{n}(\rho) - \cG_{n}(\sigma) \Vert_{1} &= \Vert \cG_{n}(\rho) - \omega + \omega - \cG_{n}(\sigma) \Vert_{1} \\
        &\leq \Vert \cG_{n}(\rho) - \omega \Vert_{1} + \Vert \cG_{n}(\sigma) - \omega \Vert_{1} \ ,
    \end{align}
    where we have used the triangle inequality. It follows,
    \begin{align}
        \lim_{n \to \infty} \Vert \cG_{n}(\rho) - \cG_{n}(\sigma) \Vert_{1} \leq \lim_{n \to \infty} \left[ \Vert \cG_{n}(\rho) - \omega \Vert_{1} + \Vert \cG_{n}(\sigma) - \omega \Vert_{1} \right] = 0 \ ,
    \end{align}
    where we used each limit on the RHS exists and is zero by assumption. As $\rho,\sigma$ were arbitrary and we only assumed $(\cG_{n})_{n \in \mbb{N}}$ is mixing, this proves mixing implies weakly mixing, which proves the containment. \\ 
    
    (Always Weakly $\subseteq$ Weakly): Weakly mixing is the $k=1$ case of being always weakly mixing (\cref{def:weakly-mixing}), so being always weakly mixing implies weakly mixing. \\
    
    (Weakly $\supseteq$ Always Weakly for Homogeneous Markov Chains): For a homogeneous Markov chain $\cN_{A \to A}$, for any $k \leq n \in \mbb{N}$, $\cG_{k:n} = \bigcirc_{i = 1}^{n-k} \cN$, so weakly mixing implies always weakly mixing by definition.\\

    (Mixing $\supseteq$ Weakly Mixing for Homogeneous Markov Chains): Let $\cN_{A \to A}$ be a homogeneous Markov chain that is weakly mixing. By definition, for all $\rho,\sigma \in \Density(A)$, 
    \begin{align}
        \lim_{n \to \infty} \Vert \cN^{n}(\rho) - \cN^{n}(\sigma) \Vert_{1} = 0 \ .
    \end{align}
    It is known that any quantum channel has at least one fixed point state (see e.g. \cite[Theorem 4.24]{Watrous-Book}). Let $\omega \in \Density(A)$ denote this existing fixed point for $\cN$. Then by letting $\sigma$ in the above be $\omega$, we conclude for all $\rho \in \Density(A)$
    \begin{align}
         0 = \lim_{n \to \infty} \Vert \cN^{n}(\rho) - \cN^{n}(\omega) \Vert_{1} = \lim_{n \to \infty} \Vert \cN^{n}(\rho) - \omega \Vert_{1} \ ,
    \end{align}
    where the first equality is by the definition of weakly mixing and the second is that $\omega$ is a fixed point. By \cref{def:mixing-Markov-chain}, this means the homogeneous Markov chain is mixing. Thus, for a homogeneous Markov chain, being weakly mixing implies mixing, which establishes the containment. \\
    
    ( (Strongly) Mixing but not Always Weakly Mixing for Inhomogeneous Markov Chains): Consider $(\cG_{n})_{n \in \mbb{N}}$ where $\cN^{1}$ is any replacer channel and $\cN^{i} = \id_{A_{2} \to A_{2}}$ for all $i > 1$. Then $(\cG_{n})_{n \in \mbb{N}}$ is mixing by direct calculation, but for any $n \geq 2$, $\cG_{2:n}$ is not mixing by direct calculation, and thus not always weakly mixing. To separate from strongly mixing, let the replacer channel prepare a full rank state in $\Density(A_{2})$. \\

    (Always Weakly $\subsetneq$ Weakly for Inhomogeneous Markov Chains): The same example as the previous case is weakly mixing but not always weakly mixing. \\
    
    (Always Weakly Mixing but not (Strongly) Mixing for Inhomogeneous Markov Chains): Consider $(\cG_{n})_{n \in \mbb{N}}$ where 
    \begin{equation}
       \cN^{i} = \begin{cases} \cR^{\dyad{0}} & i \text{ is odd} \\ \cR^{\dyad{1}} & i \text{ is even} \end{cases}. 
    \end{equation}
     Recall that $\cR^\tau(\rho) = \Tr[\rho] \tau$ for all $\rho$. Then $(\cG_{n})_{n \in \mbb{N}}$ is not mixing as it never converges to a fixed point. As it is not mixing, it is not strongly mixing. However, it is always weakly mixing by direct calculation. \\

    ($\operatorname{Always\;Weakly\;Mixing}\cap\operatorname{Mixing}\neq\emptyset$ and $\operatorname{Always\;Weakly\;Mixing}\cap\operatorname{Strongly\;Mixing}\neq\emptyset$): These follow from the containments for homogeneous Markov Chains. \\

    (Mixing $\subsetneq$ Weakly Mixing for Inhomogeneous Markov Chains): The construction that separates Always Weakly Mixing and Mixing also works for this case.

\subsection{Proof of \cref{prop:decoupling-class-containment}} \label{app:decopuling_time_proof}
    In effect, this is the same proof as for Proposition \ref{prop:mixing-class-containment} as all the examples considered there use replacer channels, which are decoupling. However, for completeness, we include the proof.

    (Strongly Decoupling $\subsetneq$ Decoupling): For both homogeneous and inhomogeneous Markov chains, by definition, if a Markov chain is strongly decoupling, it is decoupling. On the other hand, any replacement channel which prepares a pure state is decoupling but not strongly decoupling. This establishes the strict containment of strongly mixing in mixing. \\

    (Decoupling $\subseteq$ Weakly Decoupling): We prove the inhomogeneous case as it is more general. Let $(\cG_{n})_{n \in \mbb{N}}$ be decoupling with state $\omega \in \Density(A)$. Then, for any $\rho,\sigma \in \Density(R \otimes A)$ such that $\rho_{R} = \sigma_{R}$,
    \begin{align}
        &\Vert (\id_{R} \otimes \cG_{n})(\rho) - (\id_{R} \otimes \cG_{n})(\sigma) \Vert_{1} \\
        &= \Vert (\id_{R} \otimes \cG_{n})(\rho) - \rho_{R} \otimes \omega + \rho_{R} \otimes \omega - (\id_{R} \otimes \cG_{n})(\sigma) \Vert_{1} \\
        &\leq \Vert (\id_{R} \otimes \cG_{n})(\rho) - \rho_{R} \otimes \omega \Vert_{1} + \Vert (\id_{R} \otimes \cG_{n})(\sigma) - \sigma_{R} \otimes \omega \Vert_{1} \ ,
    \end{align}
    where we have used the triangle inequality and that $\rho_{R} = \sigma_{R}$. It follows,
    \begin{align}
        &\lim_{n \to \infty} \Vert (\id_{R} \otimes \cG_{n})(\rho) - (\id_{R} \otimes \cG_{n})(\sigma) \Vert_{1} \\
        &\leq \lim_{n \to \infty} \left[ \Vert (\id_{R} \otimes \cG_{n})(\rho) - \rho_{R} \otimes \omega \Vert_{1} + \Vert (\id_{R} \otimes \cG_{n})(\sigma) - \sigma_{R} \otimes \omega \Vert_{1} \right] = 0 \ ,
    \end{align}
    where we used each limit on the RHS exists and is zero by assumption. As $\rho,\sigma$ were arbitrary other than $\rho_{R} = \sigma_{R}$ and we only assumed $(\cG_{n})_{n \in \mbb{N}}$ is decoupling, this proves decoupling implies weakly decoupling, which proves the containment. \\ 
    
    (Always Weakly $\subseteq$ Weakly): Weakly decoupling is the $k=1$ case of being always weakly decoupling (\cref{def:weakly-decoupling}), so being always weakly decoupling implies weakly decoupling. \\
    
    (Weakly $\supseteq$ Always Weakly for Homogeneous Markov Chains): For a homogeneous Markov chain $\cN_{A \to A}$, for any $k \leq n \in \mbb{N}$, $\cG_{k:n} = \bigcirc_{i = 1}^{n-k} \cN$, so weakly decoupling implies always weakly decoupling by definition.\\

    (Decoupling $\supseteq$ Weakly Decoupling for Homogeneous Markov Chains): Let $\cN_{A \to A}$ be a homogeneous Markov chain that is weakly decoupling. By definition, for all $\rho,\sigma \in \Density(R \otimes A)$ such that $\rho_{R} = \sigma_{R}$, 
    \begin{align}
        \lim_{n \to \infty} \Vert (\id_{R} \otimes \cN^{n})(\rho) - (\id_{R} \otimes \cN^{n})(\sigma) \Vert_{1} = 0 \ .
    \end{align}
    It is known that any quantum channel has at least one fixed point state (see e.g. \cite[Theorem 4.24]{Watrous-Book}). Let $\omega \in \Density(A)$ denote this existing fixed point for $\cN$. Then by letting $\rho \in \Density(R \otimes A)$ be arbitrary and $\sigma_{RA} = \rho_{R} \otimes \omega_{A}$ in the above, we conclude
    \begin{align}
         0 &= \lim_{n \to \infty} \Vert (\id_{R} \otimes \cN^{n})(\rho_{RA}) - (\id_{R} \otimes \cN^{n})(\rho_{R} \otimes \omega_{A}) \Vert_{1} \\
         &= \lim_{n \to \infty} \Vert (\id_{R} \otimes \cN^{n})(\rho) - \rho_{R} \otimes \omega_{A} \Vert_{1} \ ,
    \end{align}
    where the first equality is by the definition of weakly decoupling and the second is that $\omega$ is a fixed point. As this held for all $\rho \in \Density(R \otimes A)$, by \cref{def:decoupling-Markov-chain}, this means the homogeneous Markov chain is decoupling. Thus, for a homogeneous Markov chain, being weakly decoupling implies decoupling, which establishes the containment. \\
    
    ( (Strongly) Decoupling but not Always Weakly Decoupling for Inhomogeneous Markov Chains): Consider $(\cG_{n})_{n \in \mbb{N}}$ where $\cN^{1}$ is any replacer channel and $\cN^{i} = \id_{A_{2} \to A_{2}}$ for all $i > 1$. Then, as a replacer channel is decoupling, $(\cG_{n})_{n \in \mbb{N}}$ is decoupling by direct calculation, but for any $n \geq 2$, $\cG_{2:n}$ is not decoupling by direct calculation, and thus not always weakly decoupling. To separate from strongly decoupling, let the replacer channel prepare a full rank state in $\Density(A_{2})$. \\

    (Always Weakly $\subsetneq$ Weakly for Inhomogeneous Markov Chains): The same example as the previous case is weakly decoupling but not always weakly decoupling. \\
    
    (Always Weakly Decoupling but not (Strongly) Decoupling for Inhomogeneous Markov Chains): Consider $(\cG_{n})_{n \in \mbb{N}}$ where 
    \begin{equation}
       \cN^{i} = \begin{cases} \cR^{\dyad{0}} & i \text{ is odd} \\ \cR^{\dyad{1}} & i \text{ is even} \end{cases} , 
    \end{equation}
     where we recall that $\cR^\tau(\rho) = \Tr[\rho] \tau$ for all $\rho$. Then $(\cG_{n})_{n \in \mbb{N}}$ is not decoupling as it never converges to a fixed point. As it is not decoupling, it is not strongly decoupling. However, it is always weakly decoupling by direct calculation. \\

    ($\operatorname{Always\;Weakly\;Decoupling}\cap\operatorname{(Strong)\;Decoupling}\neq\emptyset$): These follow from the containments for homogeneous Markov Chains. \\

    (Decoupling $\subsetneq$ Weakly Decoupling for Inhomogeneous Markov Chains): The construction that separates Always Weakly Decoupling and Decoupling also works for this case.

\end{document}